\documentclass[aps,prx,reprint,superscriptaddress,nofootinbib]{revtex4-2}
\usepackage{blindtext}
\usepackage{lipsum}
\usepackage{graphics}
\usepackage{amsmath}
\usepackage{graphicx}
\usepackage{graphics}
\usepackage{amssymb}
\usepackage{amsthm}
\usepackage{verbatim}
\usepackage{physics}
\usepackage{float}
\usepackage{mathtools}
\usepackage{dsfont}
\usepackage{ulem}
\usepackage[colorlinks=true,citecolor=navyblue,linkcolor=navyblue,urlcolor=navyblue,pdftitle={Coherent-moving}]{hyperref}

\usepackage[dvipsnames]{xcolor}
\usepackage[export]{adjustbox}
\usepackage{framed}
\definecolor{shadecolor}{RGB}{224,224,224}

\def \q{{\mathbf q}}
\def \nn{\nonumber \\}

\definecolor{orange(ryb)}{HTML}{FFA500}
\definecolor{lightorange(ryb)}{HTML}{FFB300}
\definecolor{dodgerblue}{HTML}{1E90FF}
\definecolor{navyblue}{HTML}{0066CC}
\definecolor{royalblue}{HTML}{4169E1}
\definecolor{lightdodgerblue}{HTML}{17b3ff}
\definecolor{crimson}{HTML}{FF4C4C}%
\definecolor{pinkerton}{HTML}{EC368D}
\definecolor{forest}{HTML}{6DD189}
\definecolor{opblue}{RGB}{0,122,255}
\definecolor{latgray}{RGB}{156,156,156}
\definecolor{edgered}{RGB}{250,106,100}
\definecolor{plaqgray}{RGB}{235,235,235}
\definecolor{ribgreen}{RGB}{52,168,83}
\definecolor{endorange}{RGB}{245,155,35}
\definecolor{fbpurple}{RGB}{112,48,160}

\usepackage[compat=0.6]{yquant}

\usepackage{tikz}
\usepackage{tikz-cd}
\usetikzlibrary{arrows}
\usetikzlibrary{arrows.meta}
\usetikzlibrary{snakes}
\usetikzlibrary{intersections}
\usetikzlibrary{shapes.geometric}
\usetikzlibrary{decorations.pathmorphing, patterns,shapes}
\usetikzlibrary{decorations.markings}
\usetikzlibrary{calc}

\tikzset{
	partial ellipse/.style args={#1:#2:#3}{
		insert path={+ (#1:#3) arc (#1:#2:#3)}
	}
}

\tikzset{
	mid arrow/.style={postaction={decorate,decoration={
				markings,
				mark=at position .575 with {\arrow[#1]{stealth}}
	}}},
	near arrow/.style={postaction={decorate,decoration={
				markings,
				mark=at position .275 with {\arrow[#1]{stealth}}
	}}},
	far arrow/.style={postaction={decorate,decoration={
				markings,
				mark=at position .800 with {\arrow[#1]{stealth}}
	}}},
}

\pgfdeclarepatternformonly{south west lines}{\pgfqpoint{-0pt}{-0pt}}{\pgfqpoint{3pt}{3pt}}{\pgfqpoint{3pt}{3pt}}{
	\pgfsetlinewidth{0.4pt}
	\pgfpathmoveto{\pgfqpoint{0pt}{0pt}}
	\pgfpathlineto{\pgfqpoint{3pt}{3pt}}
	\pgfpathmoveto{\pgfqpoint{2.8pt}{-.2pt}}
	\pgfpathlineto{\pgfqpoint{3.2pt}{.2pt}}
	\pgfpathmoveto{\pgfqpoint{-.2pt}{2.8pt}}
	\pgfpathlineto{\pgfqpoint{.2pt}{3.2pt}}
	\pgfusepath{stroke}}

\tikzset{
  op/.style={text=opblue, fill=white, inner sep=1.5pt, font=\small},
  lat/.style={latgray, line width=1.1pt},
  skel/.style={latgray, line width=0.9pt},
  rsk/.style={edgered, line width=0.9pt},
  xl/.style={opblue, font=\small, fill=white, inner sep=1.2pt},
  pn/.style={black, font=\small},
  pq/.style={latgray, font=\small, inner sep=1pt},
  ctrl/.style={opblue, line width=0.9pt}}
\newcommand{\vtxLR}{%
\vcenter{\hbox{\begin{tikzpicture}[line width=0.9pt,
    ar/.style={postaction={decorate, decoration={markings,
      mark=at position ##1 with {\arrow{Stealth[length=1.6mm,width=1.3mm]}}}}}]
  \draw[latgray, ar=0.12, ar=0.9] (-1.0,0) -- (1.0,0);
  \draw[latgray, ar=0.12, ar=0.9] (0,-1.0) -- (0,1.0);
  \node[op] at (0,0.45) {$L^{g}$};
  \node[op] at (-0.45,0) {$R^{g}$};
  \node[op] at (0,-0.45) {$R^{g}$};
  \node[op] at (0.45,0) {$L^{g}$};
  \node[pq, above right=0pt] at (0,0) {$v$};
\end{tikzpicture}}}}
\newcommand{\plaqClosed}{%
\vcenter{\hbox{\begin{tikzpicture}[line cap=round, line join=round,
    skel/.style={latgray, line width=0.9pt},
    xl/.style={opblue, fill=white, inner sep=1pt}]
  \draw[skel] (-0.7,-0.7) rectangle (0.7,0.7);
  \node[xl] at (0,0.7)     {$\bar{\mathcal{Z}}^{\Gamma}_{kl}$};
  \node[xl] at (-0.7,0)    {$\bar{\mathcal{Z}}^{\Gamma}_{li}$};
  \node[xl] at (0.7,0)     {$\mathcal{Z}^{\Gamma}_{jk}$};
  \node[xl] at (0,-0.7)    {$\mathcal{Z}^{\Gamma}_{ij}$};
  \node[latgray] at (0,0) {$p$};
\end{tikzpicture}}}}
\newcommand{\plaqOpen}{%
\vcenter{\hbox{\begin{tikzpicture}[line cap=round, line join=round,
    skel/.style={latgray, line width=0.9pt},
    xl/.style={opblue, fill=white, inner sep=1pt}]
  \draw[skel] (-0.7,-0.7) rectangle (0.7,0.7);
  \node[xl] at (0,0.7)     {$\bar{\mathcal{Z}}^{\Gamma}_{kl}$};
  \node[xl] at (-0.7,0)    {$\bar{\mathcal{Z}}^{\Gamma}_{lm}$};
  \node[xl] at (0.7,0)     {$\mathcal{Z}^{\Gamma}_{jk}$};
  \node[xl] at (0,-0.7)    {$\mathcal{Z}^{\Gamma}_{ij}$};
  \node[latgray] at (0,0) {$p$};
  \fill[latgray] (-0.7,-0.7) circle (1.6pt);
  \node[latgray, anchor=south west, inner sep=1pt] at (-0.65,-0.67) {$v$};
\end{tikzpicture}}}}
\newcommand{\tvlat}{%
  \draw[latgray] (-1.0,0) -- (2.5,0);
  \draw[latgray] (0,-0.9) -- (0,0.9);
  \draw[latgray] (1.5,-0.9) -- (1.5,0.9);}

\tikzset{
  il/.style={black, font=\footnotesize, fill=white, inner sep=1pt},
  eq/.style={black, font=\small},
  arr/.style={-{Stealth[length=2mm,width=1.6mm]}, black, line width=0.7pt},
  leg/.style={latgray, line width=0.9pt, densely dotted}}
\newcommand{\dmlat}{%
  \draw[skel] (-1.5,-0.75) rectangle (1.5,0.75);
  \draw[rsk] (0,-0.75)--(0,0.75);}
\newcommand{\dmlab}{%
  \node[latgray, font=\small] at (-0.75,0) {$p_{1}$};
  \node[latgray, font=\small] at (0.75,0) {$p_{2}$};}

\newcommand{\tcdom}[1][2.5]{%
  \draw[skel] (-1.5,-0.75) rectangle (1.5,0.75);
  \foreach \xa in {-1.5,0,1.5}{\draw[skel] (\xa,0.75)--(\xa,1.75); \draw[skel] (\xa,-1.75)--(\xa,-0.75);}
  \draw[skel] (-#1,0.75)--(-1.5,0.75); \draw[skel] (1.5,0.75)--(#1,0.75);
  \draw[skel] (-#1,-0.75)--(-1.5,-0.75); \draw[skel] (1.5,-0.75)--(#1,-0.75);}
\newcommand{\tclat}{\tcdom[2.2] \draw[skel] (0,-0.75)--(0,0.75);}
\newcommand{\tcrew}{\tcdom[2.2] \draw[rsk] (0,-0.75) .. controls (0.225,-0.15) and (0.9,0.45) .. (1.5,0.75);}
\newcommand{\zg}[1]{$\mathcal{Z}^{\Gamma}_{#1}$}
\newcommand{\zgb}[1]{$\bar{\mathcal{Z}}^{\Gamma}_{#1}$}
\newcommand{\zgp}[1]{$\mathcal{Z}^{\Gamma'}_{#1}$}
\newcommand{\zgpb}[1]{$\bar{\mathcal{Z}}^{\Gamma'}_{#1}$}
\newcommand{\opLg}{$L^{g}$}
\newcommand{\opRg}{$R^{g}$}

\newtheorem{theoreminner}{Theorem}
\colorlet{thmshade}{gray!12}
\newenvironment{theorem}[1][]{%
  \colorlet{shadecolor}{thmshade}\begin{shaded}%
  \if\relax\detokenize{#1}\relax\begin{theoreminner}\else\begin{theoreminner}[#1]\fi}%
  {\end{theoreminner}\end{shaded}}
\newtheorem*{lemmainner}{Lemma}
\newenvironment{lemma}[1][]{%
  \colorlet{shadecolor}{thmshade}\begin{shaded}%
  \if\relax\detokenize{#1}\relax\begin{lemmainner}\else\begin{lemmainner}[#1]\fi}%
  {\end{lemmainner}\end{shaded}}

\begin{document}
\title{Rewire and reset are all you need: Minimal resources for manipulating non-Abelions}

\author{Chiu Fan Bowen Lo}
\affiliation{Department of Physics, Harvard University, Cambridge, MA 02138, USA}

\author{Rohith Sajith}
\affiliation{Department of Physics, Harvard University, Cambridge, MA 02138, USA}

\author{Anasuya Lyons}
\affiliation{Department of Physics, Harvard University, Cambridge, MA 02138, USA}

\author{Ashvin Vishwanath}
\affiliation{Department of Physics, Harvard University, Cambridge, MA 02138, USA}

\author{Ruben Verresen}
\affiliation{Pritzker School of Molecular Engineering, University of Chicago, Chicago, IL 60637, USA}

\begin{abstract}
Braiding and fusing anyons are fundamental to topological quantum computation, yet existing circuit protocols often employ ancillas and measurements. We determine the minimal circuit resources for performing such anyon operations in quantum double models of generic finite groups. Our approach braids anyons through a sequential unitary circuit without ancillas. Intuitively, this is obtained by deforming the underlying graph on which the anyons live. More generally, we introduce \textit{graph rewiring}, a diagrammatic calculus that assigns a geometric interpretation to individual two-qudit gates and tracks their action on stabilizers. For fusion, we prove a general no-go theorem excluding purely unitary protocols for fusing unknown anyon types without introducing ancillas. Remarkably, we show such ``input-agnostic'' fusion is enabled by simply using a qudit reset. Hence, contrary to folklore, ancilla resources are not needed to braid or fuse non-Abelian anyons (non-Abelions). This work has implications in experimental efforts in error correction: our fusion protocol highlights reset as an underutilized primitive for error correction without the spacetime-overhead of ancillas and measurements.

\end{abstract}

\maketitle
\addtocontents{toc}{\string\let\string\l@subsubsection\string\@gobbletwo}%
\tableofcontents

\section{Introduction}
\label{sec:introduction}

Non-abelian topological order \cite{kitaevFaulttolerantQuantumComputation2003, nayak_non-abelian_2008} is emerging as a promising avenue for quantum computation. 
This mode of computation harnesses the striking properties of the quasiparticles characteristic of non-Abelian topological order (TO), known as non-Abelian anyons (non-Abelions)\cite{Goldin85}, to implement logical operations \cite{kitaevFaulttolerantQuantumComputation2003,freedman_modular_2002}. 
Information can be stored in either the internal Hilbert space of many non-Abelions (a ``fusion space'' encoding) or in the degenerate ground state subspace. 
In both cases, non-Abelions are braided and fused to perform logical manipulations on this encoded information. 
Braiding is a process whereby non-Abelions are moved around each other in space, such that their trajectories form a braid in spacetime; fusion is the action of moving multiple non-Abelions to the same location, after which their fusion outcome can be read out using measurement.
Non-Abelian TO may thereby serve as a standalone architecture for universal quantum computation \cite{SU24,mochonAnyonsNonsolvableFinite2003,mochonAnyonComputersSmaller2004,kitaev_finite_group_2007,cui_universal_2015, Cui2015_metaplectic, lo2025universalquantumcomputations3, chen_universal_2025,manjunath2026universalquantumcomputationgroup}, or be used to power up stabilizer-based computation with non-Clifford operations intrinsic to the non-Abelian TO \cite{davydova2026, huang2025generatinglogicalmagicstates,sajith_non-clifford_2026, huang2026hybridlatticesurgerynonclifford, kobayashi_clifford_2026,hsin_automorphism_2026,kobayashi_EM_nonClifford_2026}.

The promise of non-Abelian TO for computation stems from the fact that braiding and fusion implement logical operations that are inherently non-Clifford, and hence do not have to be supplied by magic state distillation.
For some anyon theories, most famously the Fibonacci theory, braiding alone is already universal \cite{freedman_modular_2002, freedman2001topological, bonesteel2005braid}, and braiding may in turn be traded for topological charge measurements in measurement-only schemes \cite{bonderson2008measurementonly}.
For the quantum doubles $D(G)$ of finite group $G$, universality can be attained for many groups by augmenting braiding with fusion, including the smallest non-Abelian group $S_3$ \cite{mochonAnyonsNonsolvableFinite2003, mochonAnyonComputersSmaller2004, cui_universal_2015, lo2025universalquantumcomputations3, chen_universal_2025}. Furthermore, while Fibonacci anyons (which are not expressible as a finite-group quantum double model but represent instead the simplest string-net model \cite{levin2004stringnet}) may be universal through braiding alone \cite{freedman_modular_2002}, it requires complex braiding protocols to approximate simple gates \cite{cui_search_2019,burke_topological_2024}. In contrast, the universal gate set for finite-group quantum doubles give relatively simple compilation of conventional Clifford and non-Clifford logical gates \cite{mochonAnyonComputersSmaller2004, lo2025universalquantumcomputations3}. Therefore, in this work we focus on quantum doubles $D(G)$ of finite groups. 

Non-Abelian TO is equally promising for error correction.
As in Abelian topological codes \cite{kitaevFaulttolerantQuantumComputation2003, dennis_topological_2002}, logical information is encoded nonlocally and is thereby protected from local noise; crucially, decoders for non-Abelian anyon models achieve finite error thresholds despite the probabilistic fusion of anyonic syndromes \cite{brell2013thermalization, wootton_error_2014, wootton2015active, burton2015classical, dauphinais2016faulttolerant, schotte2020quantum, schotte2022faulttolerant,jing_intrinsic_2025,fuente_high-threshold_2026,jing_integer_2026,davydova2026,chirame_stabilizing_2025}, so that a universal gate set and error correction can be hosted by one and the same code \cite{lyons2026quantumcomputinganyonsfault}. In certain cases, the non-Abelian fusion rules can even intrinsically herald errors and improve the code thresholds in non-Abelian decoders \cite{jing_intrinsic_2025,fuente_high-threshold_2026,jing_integer_2026}. 

Conversely, non-Abelions and defects can be introduced into stabilizer codes to unlock logical gates that are otherwise inaccessible: twist defects endow the surface code with Ising-like anyons and hence Clifford gates \cite{barkeshli10,bombin2010topological, brown2017poking,barkeshli19}, while interfacing stabilizer codes with genuinely non-Abelian TO yields non-Clifford logical gates \cite{warman2026, huang2025generatinglogicalmagicstates, huang2026hybridlatticesurgerynonclifford, davydova2026, sajith_non-clifford_2026, manjunath2026universalquantumcomputationgroup}.
These prospects have become concrete with recent experimental realizations of non-Abelian defects \cite{xu_digital_2023,andersen2023nonabelian,iqbal_qutrit_2025} and non-Abelian TO and its anyons \cite{iqbalCreationNonAbelianTopological2023,xu2024nonabelian, minev2025stringnet,lo2026universal} on programmable quantum processors. This progress has been enabled by new insights into unitary circuits for preparing such states \cite{liuMethodsSimulatingStringnet2022,lensky2023graph}, in certain cases admitting a constant-depth protocol using measurement-and-feedback \cite{tantivasadakarn2021longrange, verresen2021efficiently, tantivasadakarn2022shortest, bravyiAdaptiveConstantdepthCircuits2022, ren_efficient_2025, lyons_protocols_2025}, the latter most recently culminating in the experimental demonstration of a universal gate set in $D(S_3)$ \cite{lo2026universal}.

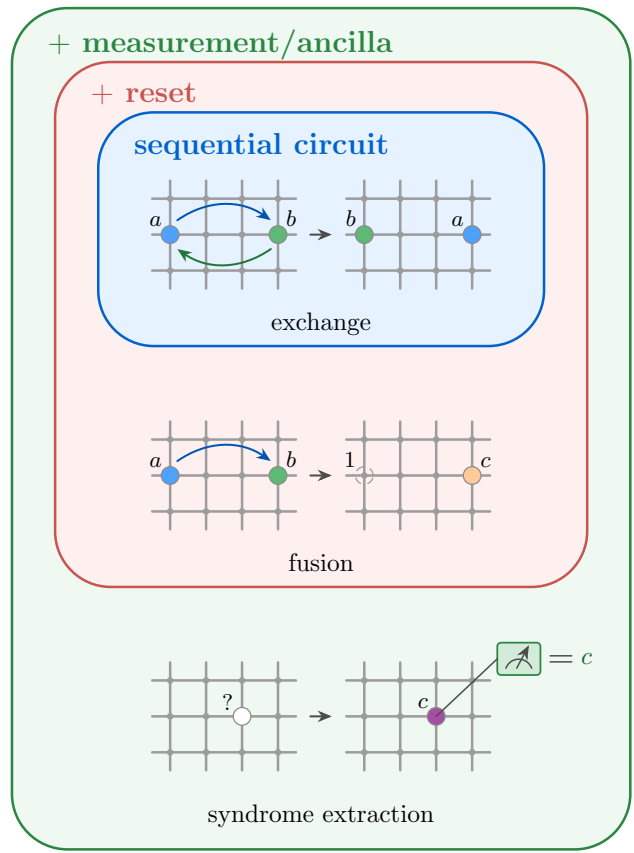
\begin{figure}[ht]
\centering
\adjustbox{max width=0.95\columnwidth}{%
\begin{tikzpicture}[line cap=round, line join=round,
    lat/.style={latgray, line width=1.0pt},
    red edge/.style={edgered, line width=1.6pt},
    any/.style={circle, inner sep=0pt, minimum size=7pt},
    vtx/.style={circle, fill=latgray, inner sep=0pt, minimum size=2.6pt},
    qd/.style={circle, fill=latgray!70, inner sep=0pt, minimum size=3.2pt},
    arr/.style={-{Stealth[length=2mm,width=1.6mm]}, black!70, line width=0.8pt},
    cls/.style={rounded corners=22pt, line width=1.0pt},
    clab/.style={font=\large\bfseries, anchor=north west, inner sep=6pt, xshift=8pt, yshift=-2pt},
    tlab/.style={font=\normalsize, anchor=north, inner sep=2pt}]
  \newcommand{\patch}[2]{%
    \begin{scope}[shift={(#1,#2)}]
      \foreach \y in {-0.5,0,0.5}{\draw[lat] (-1.0,\y)--(1.0,\y);}
      \foreach \x in {-0.75,-0.25,0.25,0.75}{\draw[lat] (\x,-0.75)--(\x,0.75);}
      \foreach \x in {-0.75,-0.25,0.25,0.75}{\foreach \y in {-0.5,0,0.5}{\node[vtx] at (\x,\y) {};}}
    \end{scope}}
  \newcommand{\patchred}[2]{\patch{#1}{#2}\draw[red edge] (#1-0.25,#2)--(#1+0.25,#2);}
  \newcommand{\anyon}[5]{\node[any, fill=#3, draw=black!45, line width=0.5pt] at (#1,#2) {}; \node[font=\small, anchor=#5] at (#1,#2) {#4};}
  \newcommand{\ghost}[2]{\node[any, draw=latgray, dashed, line width=0.5pt] at (#1,#2) {};}
  \newcommand{\iconpair}[3]{%
    \patch{#1-1.35}{#2}
    \anyon{#1-2.1}{#2}{opblue!70}{$a$}{south east}
    \anyon{#1-0.6}{#2}{ribgreen!80}{$b$}{south west}
    \draw[arr, opblue!70!black] (#1-2.0,#2+0.2) to[bend left=35] (#1-0.7,#2+0.2);
    \draw[arr] (#1-0.15,#2)--(#1+0.15,#2);
    \patch{#1+1.35}{#2}
    #3}
  \draw[cls, draw=ribgreen!80!black, fill=ribgreen!8]  (-4.3,-7.2) rectangle (4.3,4.5);
  \draw[cls, draw=edgered!80!black, fill=edgered!10]   (-3.7,-3.55) rectangle (3.7,3.75);
  \draw[cls, draw=opblue!80!black,  fill=opblue!10]    (-3.1,-0.2)  rectangle (3.1,3.05);
  \node[clab, opblue!80!black]  at (-3.1,3.05)  {sequential circuit};
  \node[clab, edgered!80!black] at (-3.7,3.75) {$+$ reset};
  \node[clab, ribgreen!80!black] at (-4.3,4.5) {$+$ measurement/ancilla};
  \patch{-1.35}{1.35}
  \anyon{-2.1}{1.35}{opblue!70}{$a$}{south east}
  \anyon{-0.6}{1.35}{ribgreen!80}{$b$}{south west}
  \draw[arr, opblue!70!black]  (-2.0,1.55) to[bend left=35] (-0.7,1.55);
  \draw[arr, ribgreen!70!black] (-0.7,1.15) to[bend left=35] (-2.0,1.15);
  \draw[arr] (-0.15,1.35)--(0.15,1.35);
  \patch{1.35}{1.35}
  \anyon{0.6}{1.35}{ribgreen!80}{$b$}{south east}
  \anyon{2.1}{1.35}{opblue!70}{$a$}{south east}
  \node[tlab] at (0,0.35) {exchange};
  \iconpair{0}{-2.0}{%
    \ghost{0.6}{-2.0}
    \node[font=\small, anchor=south east] at (0.6,-2.0) {$1$};
    \anyon{2.1}{-2.0}{orange!40}{$c$}{south west}}
  \node[tlab] at (0,-3.0) {fusion};
  \patch{-1.35}{-5.35}
  \node[any, fill=white, draw=black!45, line width=0.5pt] at (-1.1,-5.35) {};
  \node[font=\small, anchor=south] at (-1.3,-5.37) {$?$};
  \draw[arr] (-0.15,-5.35)--(0.15,-5.35);
  \patch{1.35}{-5.35}
  \anyon{1.6}{-5.35}{violet!70}{$c$}{south east}
  \node[draw=ribgreen!80!black, fill=ribgreen!22, rounded corners=1.5pt, line width=0.7pt, minimum width=0.6cm, minimum height=0.45cm, inner sep=0pt] (syn) at (2.75,-4.55) {};
  \draw[black!70, line width=0.6pt] (syn.center)++(-0.18,-0.13) arc (160:20:0.19);
  \draw[arr, black!70, line width=0.6pt] (syn.center)++(-0.03,-0.13)--++(0.17,0.32);
  \draw[black!70, line width=0.6pt] (1.6,-5.35) -- (2.45,-4.55);
  \draw[black!70, line width=0.6pt] ([yshift=1.2pt, xshift=3.5pt]syn.east) -- ++(0.3,0);
  \draw[black!70, line width=0.6pt] ([yshift=-1.2pt, xshift=3.5pt]syn.east) -- ++(0.3,0);
  \node[font=\normalsize, ribgreen!60!black, anchor=west] at (3.5,-4.55) {$c$};
  \node[tlab] at (0,-6.5) {syndrome extraction};
\end{tikzpicture}}
\caption{Circuit resource classes and the tasks they enable.
A sequential circuit (blue) suffices to braid anyons $a$ and $b$ at different sites.
Adding a reset primitive (red) allows input-agnostic fusion of $a$ and $b$ into $c$, leaving no anyon behind at the original site of $a$. This cannot be done with only unitary circuit without assuming knowledge on the input state.
Adding measurement or an ancilla (green) allows syndrome extraction: the anyon type $c$ at a site is measured and read out as a classical outcome.}
\label{fig:resource-classes}
\end{figure}

Given this promise, it is well-motivated to consider the minimal set of resources necessary to carry out operations in a non-Abelian quantum double. 
We will demonstrate that the necessary set of primitives are smaller than previously believed.

\begin{table*}[t]
  \centering
  \renewcommand{\arraystretch}{1.3}
  \setlength{\tabcolsep}{10pt}
  \begin{tabular}{l|l|ccccc}
  \hline\hline
  Operation & Section & rewiring & local pair creation & reset & measurement & ancilla \\
  \hline
  Create anyon pair   & Sec.~\ref{subsec:ribbon-rewiring}     & \checkmark & \checkmark &            &            &            \\
  Braiding    & Sec.~\ref{sec:braiding}     & \checkmark &            &            &            &            \\
  Fusion     & Sec.~\ref{sec:fusion}              & \checkmark &            & \checkmark &            &            \\
  (Fusion read-out) & Sec.~\ref{sec:fusion} & \checkmark &            &            & \checkmark & \checkmark \\
Augment graph      & Sec.~\ref{subsec:elementary-move}      & \checkmark &            &            &            & \checkmark \\
  \hline\hline
  \end{tabular}
  \caption{Resources required by each operation on anyons: unitary graph rewiring, local creation of an anyon pair, a non-unitary reset, projective measurement, and an ancilla; we highlight the section that discusses each. Graph rewiring is implemented using sequential unitary circuits. The ability to create anyon pairs with sequential circuits (without ancillas) has been discussed in Ref.~\onlinecite{lyons_protocols_2025}; the key novelty of the present work lies in the minimal resources for anyon braiding and fusion. As a useful sub-protocol, we also introduce how to minimally augment graphs with additional sites, which is similar in spirit to entanglement renormalization procedures \cite{aguado_vidal_2008,aguado_entanglement_2011}.}
  \label{tab:resources}
  \end{table*}

\subsection*{Key results and outline}

Our overarching thesis is that unitary gates and qubit reset are sufficient to implement braiding and fusion of non-Abelions. Notably, we do not require the use of adaptive circuits or ancilla, which have been used in existing schemes\footnote{The only apparent exceptions we are aware of concern anyon \textit{creation} (see e.g. Refs.~\onlinecite{iqbalCreationNonAbelianTopological2023} and \onlinecite{lyons_protocols_2025}) rather than \textit{movement}, the latter being considerably more challenging due to transporting an unknown internal state.}
to manipulate non-Abelions thus far in the literature \cite{aguadoCreationManipulationDetection2008,brennenSimulationsQuantumDouble2009,luo_simulation_2011,liuMethodsSimulatingStringnet2022,bravyiAdaptiveConstantdepthCircuits2022,jovanovicProposalDemonstrateNonabelian2024,chen_universal_2025,lo2026universal}. 
Measurement or ancilla is only required for the read-out of the fusion information, but not in the implementation of the act of fusion itself. Therefore, braiding, fusion, and syndrome measurement separate into distinct circuit resource complexity classes, as summarized in Fig.~\ref{fig:resource-classes} for an encapsulation of these resource requirements. Table~\ref{tab:resources} summarizes which resources more detailed operation requires.

We first introduce \textit{graph rewiring} in Sec.~\ref{sec:graph-rewiring}, a diagrammatic calculus that succinctly represent stabilizers of quantum double model of any finite group $G$. We illustrate this with the simple example of toric code in Sec.~\ref{subsec:tc-warmup}, pointing out the key observation that a single $CX$ gate in toric code has a well-defined geometric operation on the lattice. Following a quick review of the general construction of quantum double model on graphs in Sec.~\ref{subsec:qd-basics}, we introduce the graph rewiring in generality in Sec.~\ref{subsec:elementary-move}. The graph-rewiring picture extends remarkably simply from the toric code to quantum doubles of general groups: one replaces $CX$ with controlled group multiplication, multiplying the target edge's group label by the control's, with left or right multiplication and inverses determined by the edge orientations. The geometric action is unchanged and corresponds again to an elementary graph rewiring move. As an immediate application, we present in Sec.~\ref{subsec:ribbon-rewiring} how the ribbon operator used to create non-Abelions can be decomposed in terms of just two primitive: local pair creation and graph rewiring.

The idea of deforming the lattice to manipulate states has its conceptual roots in entanglement renormalization, where given a quantum double ground state, the lattice edges are disentangled from the topological bulk \cite{aguado_vidal_2008,aguado_entanglement_2011}. Entanglement renormalization gives rise to schemes that modify the underlying code \cite{zhu2020instantaneous,vasmer2022morphing} and an even more closely related scheme that modifies the edges to re-triangulate the underlying manifold. Such a method can be used to move Fibonacci anyons \cite{zhu2020universal}. However, these works treat the addition and removal of edges as the elementary operations, which necessarily involves ancilla. In contrast, our graph rewiring procedure does not involve an ancilla in general and is a more primitive move from which these operations can be built.  In another stream, there are works moving twist defects or Majoranas in the surface code: twist defects are manipulated by graph moves \cite{lensky2023graph,lensky_practical_2026}. This Majorana description provides a complementary viewpoint on the simple toric code example we explain in Sec.~\ref{subsec:tc-warmup}; however, the Majorana description is limited to twist defects in $\mathbb{Z}_2$ toric code or surface code, whereas our graph rewiring is directly applicable to quantum double of any group $G$. Finally, there are influences from the ZX calculus \cite{coecke_interacting_2008,coecke_interacting_2011}, which presents a convenient language to represent individual operators and circuit operation, in the Hopf algebra context especially \cite{collins2020hopffrobenius,majid2022zx,collins2024thesis}. The graph rewiring approach uses similar diagrammatic calculus but to keep track of the entire set of operators stabilizing states in quantum double. We provide a detailed commentary on the relations of graph rewiring with these works and the broader literature in Sec.~\ref{subsec:lit-review-comparison}. 

Then, in Sec.~\ref{sec:braiding}, to further demonstrate the power of graph rewiring, we construct a sequential circuits that implement exchange of charges or fluxes in general quantum double. The previously held intuition in the literature was that, when moving non-Abelions, an ancilla is necessary to keep track of the internal degree of freedom they possess. Contrary to this folklore, our braiding circuit of non-Abelions does not require the use of any ancilla. 

At this stage, one might wonder what is the extent of the power of unitary circuits (as represented by graph rewiring) to manipulate anyons. We show in Sec.~\ref{sec:no-go} that there is a sharp limit: no unitary circuit, even one assisted by a ``catalytic'' ancilla, can fuse anyons of unknown type while guaranteeing no remnant anyon is left behind. The obstruction follows from a simple counting argument: due to the coalescing nature of fusion, microscopic distinct input states have to map to the same output, which is incompatible with the orthogonality-preserving nature of a unitary (or even isometric) operator.
We note that if one allows for the ancilla to remain entangled at the end, then one can unitarily fuse, but this is not scalable, since each fusion event requires additonal ancilla, leads to a growing spatial overhead as fusion continues to happen in a computation.

Given the no-go theorem for unitary implementation of
fusion, one might conclude that measurement is necessary. However, we show that the dissipative \textit{reset} of a single edge, i.e. reinitializing it in a fixed state,  combined with graph rewiring is sufficient to fuse anyons without prior knowledge of their state. We specify protocols for fusing both charges and fluxes. This has interesting ramifications for practical error correction schemes, even for the surface code on superconducting qubits, because reset is a relatively cheap operation compared to measurement; we discuss this further in Sec.~\ref{sec:discussion}.

\section{Graph rewiring: a new diagrammatic calculus}
\label{sec:graph-rewiring}

Before the discussion of the concrete protocols for manipulating non-Abelions, we introduce \textit{graph rewiring}, a geometrically intuitive diagrammatic tool to keep track of the set of stabilizers of the state is modified by unitary circuits. Later, we use graph rewiring as an elementary primitive to develop circuit protocols manipulating anyons.

\subsection{Warm-up: rewiring the toric code}
\label{subsec:tc-warmup}

We first consider the simple example of toric code (or surface code) on the square lattice, understanding the action of a single $CX$ gate as a well-defined graph rewiring step that map toric code defined on one lattice to that defined on another.

The toric code stabilizers are given by:

\begin{equation}\label{eq:tc-stabilizers}
A_{v}=
\begin{tikzpicture}[line width=0.9pt, baseline={(0pt,-\the\fontdimen22\textfont2)}]
  \draw[latgray] (-0.9,0) -- (0.9,0);
  \draw[latgray] (0,-0.9) -- (0,0.9);
  \node[op] at (0,0.45) {$X$};
  \node[op] at (-0.45,0) {$X$};
  \node[op] at (0,-0.45) {$X$};
  \node[op] at (0.45,0) {$X$};
  \node[pq, above right=0pt] at (0,0) {$v$};
\end{tikzpicture}
,\qquad
B_{p}=
\begin{tikzpicture}[scale=0.8, line cap=round, line join=round, baseline={(0pt,-\the\fontdimen22\textfont2)}]
  \draw[skel] (-0.75,-0.75) rectangle (0.75,0.75);
  \node[latgray] at (0,0) {$p$};
  \node[xl] at (0,-0.75) {$Z$};
  \node[xl] at (0.75,0) {$Z$};
  \node[xl] at (0,0.75) {$Z$};
  \node[xl] at (-0.75,0) {$Z$};
\end{tikzpicture}
,
\end{equation}

for every vertices and plaquettes of the square lattice. Next, we apply the $CX$ gate on two consecutive edges as shown in Fig.~\ref{fig:tc-rewire}(a), and track how the stabilizers are modified.

\begin{figure}[ht]
\centering
\adjustbox{max width=\columnwidth}{%
\begin{tikzpicture}[x=0.095\columnwidth, y=0.095\columnwidth, line cap=round, line join=round,
    every node/.style={font=\small},
    ce/.style ={black, line width=1.2pt},
    te/.style ={edgered, line width=1.5pt},
    cr/.style ={opblue, line width=1.1pt},
    tg/.style ={circle, draw=opblue, line width=1.1pt, fill=white, inner sep=0pt, minimum size=7pt},
    pl/.style ={anchor=north west, inner sep=1pt},
  ]
  \newcommand{\tcLabels}[1]{%
    \node[pq, above left=0pt] at (0,0.75) {$v$}; \node[pq, above right=0pt] at (1.5,0.75) {$u$}; \node[pq, below left=0pt] at (0,-0.75) {$w$};
    \node[pq] at (-0.75,0) {$p_{1}$}; \node[pq] at (#1) {$p_{2}$};}
  \begin{scope}
    \tcdom
    \draw[ce] (0,0.75)--(1.5,0.75);
    \draw[te] (0,-0.75)--(0,0.75);
    \draw[cr] (0.75,0.75)--(0,0);
    \fill[opblue] (0.75,0.75) circle (2pt);
    \node[tg] (tgt) at (0,0) {};
    \draw[cr] (tgt.north)--(tgt.south) (tgt.west)--(tgt.east);
    \tcLabels{0.75,0}
    \node[pl] at (-2.4,1.7) {(a)};
  \end{scope}
  \node[font=\Large] at (2.85,0) {$\Rightarrow$};
  \begin{scope}[shift={(5.7,0)}]
    \tcdom
    \draw[ce] (0,0.75)--(1.5,0.75);
    \draw[te] (0,-0.75) .. controls (0.225,-0.15) and (0.9,0.45) .. (1.5,0.75);
    \tcLabels{0.9,-0.3}
    \node[pl] at (-2.4,1.7) {(b)};
  \end{scope}
\end{tikzpicture}}
\caption{The elementary rewiring move in the toric code.
(a) A $CX$ gate (drawn in blue) between two consecutive edges of the square lattice: the control (denoted by a solid dot) acts on the control edge (black) joining vertices $v$ and $u$, and the target (denoted by $\oplus$) acts on the target edge (red) joining vertices $w$ and $v$.
(b) The rewired graph: the end of the target edge at $v$ is detached and rewired to $u$.
The plaquettes $p_{1}$ and $p_{2}$ become a pentagon and a triangle respectively, and the vertices $v$ and $u$ become trivalent and pentavalent respectively, matching the modified stabilizers of Eq.~\eqref{eq:tc-modified} and Eq.~\eqref{eq:tc-modified-rewired}.}
\label{fig:tc-rewire}
\end{figure}
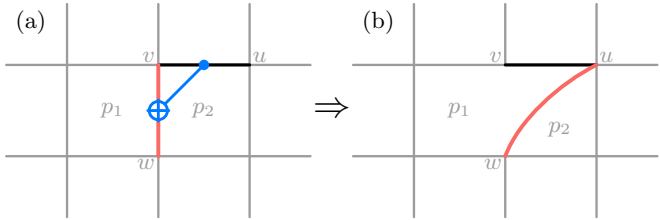

Using commutation relations between $CX$ and the $X,Z$ Pauli operators,
\begin{equation}\label{eq:ZX-commutation}
\vcenter{\hbox{
\begin{tikzpicture}
\begin{yquant}[horizontal]

[violet] qubit {} a[1];
[violet] qubit {} a[+1];

box {$X$} a[0];
cnot a[1] | a[0];

\end{yquant}
\end{tikzpicture}
}}
=
\vcenter{\hbox{
\begin{tikzpicture}
\begin{yquant}[horizontal]

[violet] qubit {} a[1];
[violet] qubit {} a[+1];

cnot a[1] | a[0];
box {$X$} a[0];
box {$X$} a[1];

\end{yquant}
\end{tikzpicture}
}}, \quad \quad
\vcenter{\hbox{
\begin{tikzpicture}
\begin{yquant}[horizontal]

[violet] qubit {} a[1];
[violet] qubit {} a[+1];

box {$Z$} a[1];
cnot a[1] | a[0];

\end{yquant}
\end{tikzpicture}
}}
=
\vcenter{\hbox{
\begin{tikzpicture}
\begin{yquant}[horizontal]

[violet] qubit {} a[1];
[violet] qubit {} a[+1];

cnot a[1] | a[0];
box {$Z$} a[0];
box {$Z$} a[1];

\end{yquant}
\end{tikzpicture}
}},
\end{equation}
we see that only 4 stabilizers do not commute with the $CX$ gate: two plaquette stabilizers that overlap with the target of the $CX$ gate, and two vertex stabilizers that overlap with the control of the $CX$ gate.
Given an original stabilizer $\mathcal{O}$, we denote the modified stabilizer by $\mathcal{O}'\equiv (CX)\,\mathcal{O}\,(CX)$.
They are modified as follows:

\begin{equation}\label{eq:tc-modified}
\begin{aligned}
A'_{v}&=
\begin{tikzpicture}[scale=0.65, line cap=round, line join=round, baseline={(0pt,-\the\fontdimen22\textfont2)}]
  \tclat \path (0,-2.6);
  \node[xl] at (-0.75,0.75) {$X$};
  \node[xl] at (0,1.55)     {$X$};
  \node[xl] at (0.75,0.75)  {$X$};
  \node[pq, above left=0pt] at (0,0.75) {$v$};
\end{tikzpicture}
&
A'_{u}&=
\begin{tikzpicture}[scale=0.65, line cap=round, line join=round, baseline={(0pt,-\the\fontdimen22\textfont2)}]
  \tclat \path (0,-2.6);
  \node[xl] at (0.75,0.75)  {$X$};
  \node[xl] at (0,0)        {$X$};
  \node[xl] at (1.5,0)      {$X$};
  \node[xl] at (1.5,1.55)   {$X$};
  \node[xl] at (2.1,0.75)   {$X$};
  \node[pq, above right=0pt] at (1.5,0.75) {$u$};
\end{tikzpicture}
\\
B'_{p_{1}}&=
\begin{tikzpicture}[scale=0.65, line cap=round, line join=round, baseline={(0pt,-\the\fontdimen22\textfont2)}]
  \tclat
  \node[pq] at (-0.75,0) {$p_{1}$};
  \node[xl] at (0,0)         {$Z$};
  \node[xl] at (0.75,0.75)   {$Z$};
  \node[xl] at (-0.75,0.75)  {$Z$};
  \node[xl] at (-1.5,0)      {$Z$};
  \node[xl] at (-0.75,-0.75) {$Z$};
\end{tikzpicture}
&
B'_{p_{2}}&=
\begin{tikzpicture}[scale=0.65, line cap=round, line join=round, baseline={(0pt,-\the\fontdimen22\textfont2)}]
  \tclat
  \node[pq] at (0.75,0) {$p_{2}$};
  \node[xl] at (0,0)         {$Z$};
  \node[xl] at (1.5,0)       {$Z$};
  \node[xl] at (0.75,-0.75)  {$Z$};
\end{tikzpicture}
\end{aligned}
\end{equation}
where the labels of the vertices and plaquettes follow Fig.~\ref{fig:tc-rewire}.

Here is the key observation: the modified stabilizers can be understood as those defined on a different lattice (which we refer to as the rewired lattice), as shown in Fig.~\ref{fig:tc-rewire}(b). The red edge, instead of connecting vertices $v$ and $w$ originally, is rewired to connect between vertices $u$ and $w$. 
On the rewired lattice, the four modified stabilizers of Eq.~\eqref{eq:tc-modified} are simply the vertex and plaquette stabilizers of the new vertices and plaquettes,
\begin{equation}\label{eq:tc-modified-rewired}
\begin{aligned}
A'_{v}&=
\begin{tikzpicture}[scale=0.65, line cap=round, line join=round, baseline={(0pt,-\the\fontdimen22\textfont2)}]
  \tcrew \path (0,-2.6);
  \node[xl] at (-0.75,0.75) {$X$};
  \node[xl] at (0,1.55)     {$X$};
  \node[xl] at (0.75,0.75)  {$X$};
  \node[pq, above left=0pt] at (0,0.75) {$v$};
\end{tikzpicture}
&
A'_{u}&=
\begin{tikzpicture}[scale=0.65, line cap=round, line join=round, baseline={(0pt,-\the\fontdimen22\textfont2)}]
  \tcrew \path (0,-2.6);
  \node[xl] at (0.75,0.75)  {$X$};
  \node[xl] at (0.6,0.11)   {$X$};
  \node[xl] at (1.5,0)      {$X$};
  \node[xl] at (1.5,1.55)   {$X$};
  \node[xl] at (2.1,0.75)   {$X$};
  \node[pq, above right=0pt] at (1.5,0.75) {$u$};
\end{tikzpicture}
\\
B'_{p_{1}}&=
\begin{tikzpicture}[scale=0.65, line cap=round, line join=round, baseline={(0pt,-\the\fontdimen22\textfont2)}]
  \tcrew
  \node[pq] at (-0.75,0) {$p_{1}$};
  \node[xl] at (0.6,0.11)    {$Z$};
  \node[xl] at (0.75,0.75)   {$Z$};
  \node[xl] at (-0.75,0.75)  {$Z$};
  \node[xl] at (-1.5,0)      {$Z$};
  \node[xl] at (-0.75,-0.75) {$Z$};
\end{tikzpicture}
&
B'_{p_{2}}&=
\begin{tikzpicture}[scale=0.65, line cap=round, line join=round, baseline={(0pt,-\the\fontdimen22\textfont2)}]
  \tcrew
  \node[pq] at (1.02,-0.3) {$p_{2}$};
  \node[xl] at (0.6,0.11)    {$Z$};
  \node[xl] at (1.5,0)       {$Z$};
  \node[xl] at (0.75,-0.75)  {$Z$};
\end{tikzpicture}
\end{aligned}
\end{equation}

Then, the form of the modified stabilizers given in Eq.~\eqref{eq:tc-modified} become transparent geometrically: for example, the 5-body plaquette stabilizer corresponds to the pentagon plaquette $p_1$, the 3-body vertex stabilizer corresponds to the trivalent vertex $v$. 
Given the rewired lattice, the corresponding stabilizers on the rewired lattice exactly coincide with those given in Eq.~\eqref{eq:tc-modified-rewired}.
Therefore, we refer to the action of the $CX$ gate as a \textit{rewiring} of the underlying lattice. 

We can infer a heuristic rule of the rewiring step: initially, the target edge (red) is connected to one end of the control edge (black); after the $CX$ gate action, the target edge is rewired to connect to the other end of the control edge. In the next few subsections, we show that above rewiring rule generalizes; in fact, it does not hinge on the Pauli commutation relations and hence is not restricted to only toric code or even stabilizer codes. Rather, it is applicable to the general setting of quantum double of any finite group $G$ on any directed graph. 

\subsection{Quantum double as graphs}
\label{subsec:qd-basics}

We start by defining the quantum double Hamiltonian of the finite group $G$ on a planar directed graph\footnote{In general, the planar condition can be relaxed: the quantum double can be defined on any 2-complex, which can be understood as a directed graph (such that vertex operators are well-defined) with a set of closed walks each attached with a plaquette (such that plaquette operator are defined).} \cite{kitaevFaulttolerantQuantumComputation2003}. We denote group elements in $G$ by $g$; in particular, the identity group element is denoted by $e$. The inverse of a group element $g$ is denoted as $\bar{g}$.

\subsubsection{Group qudit and generalized Pauli}

The Hilbert space is defined as follows: we put a qudit on each edge of the directed graph, such that the local Hilbert space on the edge is spanned by $\ket{g}$ for all $g\in G$; i.e. the dimension of the qudit is the size of the group $G$, with basis state labeled by each element of $G$ (we refer to this as the group element basis).

The Pauli $X$ from the $\mathbb{Z}_2$ case is generalized to group multiplications $L^g$ and $R^g$. Since $G$ is in general non-Abelian, we need to distinguish between the action of left and right multiplications as follows:
\begin{equation}
    L^g\ket{h} = \ket{gh}, \quad \quad \quad R^g\ket{h}= \ket{h \bar{g}},
    \label{eq:LR-def-main}
\end{equation}

The generalization of Pauli $Z$ from the $\mathbb{Z}_2$ case is the matrix product operator (MPO) $\mathcal{Z}_{ij}^\Gamma$, which acts on the group basis as follows: 
\begin{equation}
    \mathcal{Z}^\Gamma_{ij} \ket{g}= [\Gamma(g)]_{ij} \ket{g},
\end{equation}
where $[\Gamma(g)]_{ij}$ is the $(i,j)$ entry of the representation matrix $\Gamma(g)$ given the irrep $\Gamma$.

\subsubsection{Hamiltonian and ground states}

For each vertex $v$ of the directed graph, we associate vertex operators $A_v^g$ for every $g \in G$, such that its action on arbitrary vertex, e.g. the degree-5 vertex in panel (a) with both incoming and outgoing edges, is given given in panel (b) 
\begin{equation}
    A^g_v=\prod_{l:\;\mathrm{outgoing}} L^g_l \prod_{l:\;\mathrm{incoming}} R^g_l ,
    \label{eq:Av-convention}
\end{equation}

Next, for each plaquette $p$ of the directed graph, we associate a plaquette operator
\begin{equation}
    \begin{aligned}
B^e_p&=\delta_{\partial p,e}\\
&= \frac{1}{|G|}\sum_{\Gamma} d_{\Gamma}\sum_{\{i\}} \textrm{Tr} \left( \prod_{l: \text{ aligned}} \mathcal{Z}^\Gamma_{i_1i_2} \prod_{l: \text{ anti-aligned}} \bar{\mathcal{Z}}^\Gamma_{i_3i_4}\right),
\end{aligned}
\label{eq:Bp-convention}
\end{equation}

which ensures that the ordered\footnote{We choose the ordering convention to be counterclockwise around the plaquette.} product of group elements around the boundary of plaquette $p$ is the identity group element $e$; the second equality express $B^e_p$ in its MPO form, where $\textrm{Tr} (\cdot)$ means the ordered contraction of the internal indices $i$ mirroring that of the ordered product of group elements around the boundary of $p$.
See Eq.~\eqref{eq:Be-MPO} in Appendix~\ref{app:flux-projectors} for how the MPO form of $B^e_p$ is derived.

As a concrete example, we can consider the quantum double on the square lattice. We fix the edge orientation convention that horizontal and vertical edges point to the right and up respectively. Then, using Eq.~\eqref{eq:tc-stabilizers}, the vertex and plaquette operators take the form
\begin{equation}\label{eq:gs-stabilizers}
\begin{aligned}
A^{g}_{v}&=
\vtxLR
,\\[4pt]
B^{e}_{p}&=\frac{1}{|G|}\sum_{\Gamma} d_{\Gamma}\sum_{ijkl}\; \plaqClosed
,
\end{aligned}
\end{equation}
where the outgoing edges at $v$ carry $L^{g}$ and the incoming edges carry $R^{g}$, and $B^{e}_{p}=\delta_{\partial p,e}$ is written as a matrix product operator of the diagonal gates $\mathcal{Z}^{\Gamma}_{ij}$ ($\bar{\mathcal{Z}}^{\Gamma}_{ij}$ on edges traversed against their orientation) contracted counterclockwise around $p$; the gates are defined in Appendix~\ref{app:diagonal-gates}, Eq.~\eqref{eq:Z-def}.

When the quantum double is viewed as a non-Abelian lattice gauge theory, the Hamiltonian is
\begin{equation}\label{eq:quantum-double-H}
H=-\sum_{v}\sum_{g\in G}A^{g}_{v}-\sum_{p}B^{e}_{p},
\end{equation}
such that $A^g$ enacts a gauge transformation. Therefore, the condition
\begin{equation}\label{eq:charge-neutrality}
    A^g_v=+1, \quad \forall g \in G
\end{equation}
is the charge neutrality condition at vertex $v$, as the trivial charge is by definition gauge invariant. As an operator, $A_v=\sum_{g\in G} A^g_v$ is the vertex projector that projects to the trivial charge sector at vertex $v$. Similarly, the condition
\begin{equation}\label{eq:flux-neutrality}
    B^e_p = +1
\end{equation}
is the flux neutrality condition at plaquette $p$, which asserts the flatness of the gauge field within the ground state subspace. 
Neutrality condition of a contractible region (which is used later in Sec.\ref{sec:no-go}) is defined by the condition that Eq.~\eqref{eq:charge-neutrality} or Eq.~\eqref{eq:flux-neutrality} are satisfied in every site in the interior of the contractible region. 

We will sometimes refer to $A_v$ and $B^e_p$ projectors as stabilizers of the quantum double ground states, since they stabilize, i.e. act trivially on, the quantum double ground states. But note that this does not confine the discussion to the family of error correction code known as `stabilizer codes', which concerns stabilizer expressed in terms of Pauli operators. We will also be interested in the excited states of the quantum double model; these states will be stabilized by different projectors, which we introduce in the next section. 

\subsubsection{Anyon type and excited states}
\label{subsubsec:anyon-type}

Excited states (states containing anyons) of the quantum double model correspond to states that violate the charge and flux neutrality conditions in Eq.~\eqref{eq:charge-neutrality} and \eqref{eq:flux-neutrality} respectively. In a contractible region, pairs of anyons in the vacuum fusion channel can be created by applying ribbon operators, which are the non-Abelian generalization of string operators in the toric code. 

Mathematically, distinct anyon types of the quantum double of $G$ correspond to irreducible representations of the Drinfel'd double of $G$, denoted $D(G)$~\cite{dijkgraaf_quasi_1991,bais_quantum_1992,dewildpropitius_discrete_1995,kitaevFaulttolerantQuantumComputation2003}. We will not derive it here, but the irreps of $D(G)$ can be decomposed in terms of group data. In particular, each anyon type is labeled by a tuple $(C,\Gamma[Z(C)])$: the first data $C$ denotes a conjugacy class of $G$, and gives information about the flux type; for the second data, $Z(C)$ is the centralizer of some representative element $r \in C$ in $G$ (different choices of $r$ lead to isomorphic centralizers) and $\Gamma[Z(C)]$ is a irreducible representation (irrep) of $Z(C)$. This irrep information tells us the charge of a given anyon. 

Anyons with a nontrivial conjugacy class label $C$, but the trivial irrep label $[+]$, are called pure fluxes, while anyons with trivial conjugacy class $C = \{e\}$ but nontrivial irrep $\Gamma$ are called pure charges. 
We use $\Gamma$ as shorthand for the pure charge $(\{e\},\Gamma)$ and $C$ to label a magnetic conjugacy class.
The corresponding flux-class and vertex-representation projectors are
\begin{align}
    B^{C}_{p} &= \frac{|C|}{|G|}\sum_{\Gamma}\bar{\chi}^{\Gamma}(C)\sum_{ijkl}\; \plaqClosed
    , \label{eq:flux-projector}\\[4pt]
    P^{\Gamma}_{v} &= \frac{d_{\Gamma}}{|G|} \sum_{g\in G} \chi^{\Gamma}(\bar{g})\; \vtxLR
    , \label{eq:charge-projector}
\end{align}
where $\bar{\chi}^{\Gamma}(C)$ is the complex conjugate of the character of $\Gamma$ on the class $C$; the MPO form of $B^{C}_{p}$ is derived in Appendix~\ref{app:flux-projectors}, Eq.~\eqref{eq:BC-MPO}.
On locally flux-free states, $P^\Gamma_v$ projects onto the pure charge $\Gamma$.
For nontrivial flux, $B^C_p$ fixes only the conjugacy class; specifying a pure flux also requires trivial centralizer charge, as in Eq.~\eqref{eq:ff-site-projector}.

\subsubsection{Representing unentangled ancilla edges}

We can add edges or vertices to a given graph by including ancilla. Based on the state of the unentangled ancilla, the graph representation is different.
An unentangled ancilla qudit in the group element basis can be represented as a loop attached to any vertex of the lattice (Fig.~\ref{fig:ancilla-edges}(a)): it adds an additional plaquette (as a single loop) to the graph. If the ancilla is in the state $\ket{e}$, the loop encloses no flux, whereas in a state $\ket{g\neq e}$ there is a pair of equal and opposite fluxes inside and outside the loop.
An unentangled qudit in the representation basis can be visualized as an edge with one end attached to any vertex of the lattice (Fig.~\ref{fig:ancilla-edges}(b)): it adds an additional vertex to the graph.
In the state $\ket{+} = \sum_g \ket{g}$ there is no charge at the two ends of the edge, whereas in a state $\ket{\Gamma;i,j}$ there is a charge pair at its two ends.

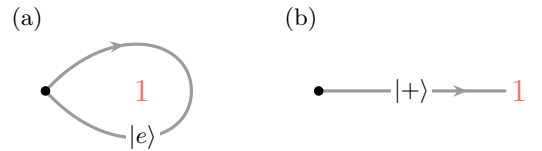
\begin{figure}[H]
\centering
\adjustbox{max width=\columnwidth}{%
\begin{tikzpicture}[x=0.095\columnwidth, y=0.095\columnwidth, line cap=round, line join=round,
    every node/.style={font=\small},
    oriented/.style={postaction={decorate, decoration={markings,
        mark=at position #1 with {\arrow{Stealth[length=1.9mm,width=1.5mm]}}}}},
    ae/.style ={latgray, line width=1.2pt},
    vx/.style ={circle, fill=black, inner sep=0pt, minimum size=3.6pt},
    sl/.style ={text=black, fill=white, inner sep=1pt},
    tr/.style ={edgered, font=\large, inner sep=1pt},
    pl/.style ={anchor=north west, inner sep=1pt},
  ]
  \begin{scope}
    \draw[ae, oriented=0.25] (0,1.2) .. controls (1.05,2.35) and (2.35,2.05) .. (2.35,1.2)
                      .. controls (2.35,0.35) and (1.05,0.05) .. node[sl, pos=0.5] {$\ket{e}$} (0,1.2);
    \node[tr] at (1.55,1.2) {$1$};
    \node[vx] at (0,1.2) {};
    \node[pl] at (-0.6,2.6) {(a)};
  \end{scope}
  \begin{scope}[shift={(4.4,0)}]
    \draw[ae, oriented=0.8] (0,1.2)--(3.0,1.2) node[sl, pos=0.5] {$\ket{+}$};
    \node[vx] at (0,1.2) {};
    \node[tr, anchor=west] at (3.05,1.2) {$1$};
    \node[pl] at (-0.6,2.6) {(b)};
  \end{scope}
\end{tikzpicture}}
\caption{Unentangled ancilla edges.
(a) An ancilla qudit in the group element basis is drawn as a loop (grey) based at a vertex of the lattice (black dot); in the state $\ket{e}$ the loop encloses the trivial flux $1$ (red).
(b) An ancilla qudit in the representation basis is drawn as an edge (grey) with one end at a vertex of the lattice (black dot) and the other end free; in the state $\ket{+}$ the free end carries the trivial charge $1$ (red).}
\label{fig:ancilla-edges}
\end{figure}

\subsection{A single entangling gate is an elementary move}
\label{subsec:elementary-move}

Now we consider the general case of the action of controlled group multiplication gates (which are generalizations of the $CX$ gate from the $\mathbb{Z}_2$ case) on quantum double defined on graphs. The action of controlled group multiplication gates on group element basis are defined as follows:
\begin{equation}
    CL\ket{h_1}\ket{h_2} = \ket{h_1}\ket{h_1 h_2}, \; 
    CR\ket{h_1}\ket{h_2} = \ket{h_1}\ket{h_2 \bar{h}_1}.
    \label{eq:CLCR-def-main}
\end{equation}
Note that unlike the $\mathbb{Z}_2$ case, $CL$ and $CR$ are not equal to their inverses; hence, there are four distinct controlled group multiplication gates in total.

We show that controlled group multiplication has a well-defined action on general directed graphs as long as the control and the target edges share a vertex: the endpoint of the target edge is shifted from one end of the control edge to the other, as illustrated in Fig.~\ref{fig:gr-move}.

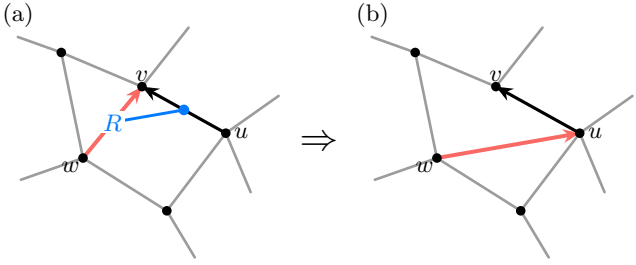
\begin{figure}[h]
\centering
\adjustbox{max width=\columnwidth}{%
\begin{tikzpicture}[x=0.095\columnwidth, y=0.095\columnwidth, line cap=round, line join=round,
    every node/.style={font=\small},
    ge/.style ={latgray, line width=1.0pt},
    ce/.style ={black, line width=1.2pt, -{Stealth[length=2.2mm,width=1.8mm]}},
    te/.style ={edgered, line width=1.5pt, -{Stealth[length=2.2mm,width=1.8mm]}},
    vx/.style ={circle, fill=black, inner sep=0pt, minimum size=3.6pt},
    cr/.style ={opblue, line width=1.1pt},
    rl/.style ={opblue, fill=white, inner sep=1pt},
    vl/.style ={inner sep=1pt},
    pl/.style ={anchor=north west, inner sep=1pt},
  ]
  \newcommand{\grPts}{%
    \coordinate (v) at (0,0);
    \coordinate (u) at (1.35,-0.75);
    \coordinate (w) at (-0.95,-1.15);
    \coordinate (b) at (0.4,-2.0);
    \coordinate (q) at (-1.3,0.55);}
  \newcommand{\grSpect}{%
    \draw[ge] (w)--(b); \draw[ge] (b)--(u); \draw[ge] (b)--(0.85,-2.75);
    \draw[ge] (v)--(q); \draw[ge] (q)--(w); \draw[ge] (q)--(-2.0,0.8);
    \draw[ge] (v)--(0.75,0.95);
    \draw[ge] (w)--(-1.95,-1.5);
    \draw[ge] (u)--(2.2,-0.15); \draw[ge] (u)--(1.75,-1.7);}
  \newcommand{\grDots}{\foreach \p in {v,u,w,b,q}{\node[vx] at (\p) {};}
    \node[vl, above=1pt] at (v) {$v$}; \node[vl, right=2pt] at (u) {$u$}; \node[vl, below left=0pt] at (w) {$w$};}
  \begin{scope}
    \grPts \grSpect
    \draw[ce] (u)--(v);
    \draw[te] (w)--(v);
    \draw[cr] ($(u)!0.5!(v)$)--($(w)!0.5!(v)$);
    \fill[opblue] ($(u)!0.5!(v)$) circle (2pt);
    \node[rl] at ($(w)!0.5!(v)$) {$R$};
    \grDots
    \node[pl] at (-2.3,1.4) {(a)};
  \end{scope}
  \node[font=\Large] at (2.85,-0.9) {$\Rightarrow$};
  \begin{scope}[shift={(5.7,0)}]
    \grPts \grSpect
    \draw[ce] (u)--(v);
    \draw[te] (w)--(u);
    \grDots
    \node[pl] at (-2.3,1.4) {(b)};
  \end{scope}
\end{tikzpicture}}
\caption{The elementary rewiring move in general quantum double of $G$.
(a) A controlled right multiplication $CR$ (blue) between two consecutive edges pointing into the vertex $v$: the (black) control edge [u,v] and the (red) target edge [w,v]; the dot marks the control and $R$ acts on the target.
(b) The rewired graph: the head of the target edge is detached from vertex $v$ and rewired to vertex $u$. 
All other edges and vertices are unchanged.}
\label{fig:gr-move}
\end{figure}

We denote the edges by their end points: initially the control edge is $[u,v]$, and the target edge is $[w,v]$; after rewiring, the rewired target edge is $[w,u]$.
In terms of state, let $g_{w,v}$ and $g_{u,v}$ be the group elements on $[w,v]$ and $[u,v]$ respectively. After the rewiring, the group element on $[w,u]$ becomes the product $g_{w,v} \bar{g}_{u,v}$, since $[w,u]$ spans the two-edge path $w\to v \to u$.

The two independent orientations of the control and target edge with respect to their connecting vertex determine which controlled group multiplication gate to act; the four orientation cases are tabulated in Fig.~\ref{fig:gr-four-cases}; in every case, the rewired edge state becomes the product of the two-edge path through $v$ that it replaces. 

We observe that the form of the ground state stabilizers (Eq.~\eqref{eq:gs-stabilizers}) and excited state stabilizers (Eq.~\eqref{eq:flux-projector} and~\eqref{eq:charge-projector}) all have the same generalized Pauli action on the lattice besides the different prefactor weights; therefore, to keep track of the action of controlled group multiplication gates, we just need to keep track of the action on the generalized Paulis. Hence, graph rewiring applies to both ground states and excited states. In Appendix~\ref{app:gr-move} we show that the gate maps every vertex and plaquette stabilizer (for either ground states or excited states) of the old graph to the corresponding stabilizer of the new graph. We note that in the Abelian $\mathbb{Z}_2$ case, all the controlled group multiplication gates reduce to the $CX$ gate, and we recover the case discussed in Sec.~\ref{subsec:tc-warmup}. 

\begin{figure}[ht]
  \centering
  \adjustbox{max width=\columnwidth}{%
  \begin{tikzpicture}[
      line cap=round, line join=round,
      cem/.style={black,   line width=1.2pt,
                  -{Stealth[length=2.4mm,width=1.9mm]}},
      tem/.style={edgered, line width=1.5pt,
                  -{Stealth[length=2.4mm,width=1.9mm]}},
      vx/.style ={circle, fill=black, inner sep=0pt, minimum size=4.5pt},
      gw/.style ={opblue, line width=1.1pt},
      gl/.style ={opblue, font=\normalsize, fill=white, inner sep=1.2pt},%
      vl/.style ={font=\small},
      rule/.style={black, line width=0.9pt},
      hd/.style ={font=\normalsize},
      lab/.style={font=\normalsize},
      imp/.style={font=\Large},
      el/.style ={font=\scriptsize, inner sep=1pt},
    ]
    \def\rowpitch{1.8}
    \def\latA{2.7}
    \def\latI{4.0}
    \def\latB{5.4}
    \def\xleft{-2.0}  \def\xmid{1.2}  \def\xright{6.55}
    \def\ytop{0.5}    \def\ybot{-6.7}
    \node[hd] at ({(\xleft+\xmid)/2}, 0.95)  {Gate};
    \node[hd] at ({(\xmid+\xright)/2}, 0.95) {Rewiring move};
    \draw[rule, line width=1.3pt] (\xleft,\ytop) -- (\xright,\ytop);
    \draw[rule, line width=1.3pt] (\xmid,1.3) -- (\xmid,\ybot);
    \foreach \r in {1,2,3}{\draw[rule] (\xleft,{-\r*\rowpitch+\ytop}) -- (\xright,{-\r*\rowpitch+\ytop});}
    \newcommand{\lmPts}{%
      \coordinate (v) at (0,0);
      \coordinate (u) at (0.8,-0.7);
      \coordinate (w) at (-1.0,-0.85);}
    \newcommand{\lmLabels}{%
      \node[vx] at (v) {};
      \node[vl, above=1pt] at (v) {$v$};
      \node[vl, below=0pt] at (u) {$u$};
      \node[vl, below=0pt] at (w) {$w$};}
    \newcommand{\lmBefore}[2]{%
      \lmPts
      \ifnum#1=1 \draw[cem] (u)--(v);\else\draw[cem] (v)--(u);\fi
      \ifnum#2=1 \draw[tem] (w)--(v);\else\draw[tem] (v)--(w);\fi}
    \newcommand{\lmAfter}[2]{%
      \lmPts
      \ifnum#1=1 \draw[cem] (u)--(v);\else\draw[cem] (v)--(u);\fi
      \ifnum#2=1 \draw[tem] (w)--(u);\else\draw[tem] (u)--(w);\fi}
    \newcommand{\lmCase}[6]{%
      \begin{scope}[yshift={-#1*\rowpitch cm}]
        \node[lab] at (-1.6,-0.3) {(#2)};
        \begin{scope}
          \lmBefore{#3}{#4}
          \coordinate (mc) at ($(u)!0.5!(v)$);
          \coordinate (mt) at ($(w)!0.5!(v)$);
          \draw[gw] (mc)--(mt);
          \fill[opblue] (mc) circle (2pt);
          \node[gl] at (mt) {#5};
          \lmLabels
        \end{scope}
        \begin{scope}[xshift=\latA cm]
          \lmBefore{#3}{#4}
          \node[el, above right=-1pt] at ($(u)!0.5!(v)$) {$h_c$};
          \node[el, above left=-1pt]  at ($(w)!0.5!(v)$) {$h_t$};
          \lmLabels
        \end{scope}
        \node[imp] at (\latI,-0.35) {$\Rightarrow$};
        \begin{scope}[xshift=\latB cm]
          \lmAfter{#3}{#4}
          \node[el, below=1pt] at ($(w)!0.5!(u)$) {#6};
          \lmLabels
        \end{scope}
      \end{scope}}
    \lmCase{0}{a}{1}{0}{$L$}{$h_c h_t$}
    \lmCase{1}{b}{0}{1}{$R^\dagger$}{$h_t h_c$}
    \lmCase{2}{c}{1}{1}{$R$}{$h_t \bar h_c$}
    \lmCase{3}{d}{0}{0}{$L^\dagger$}{$\bar h_c h_t$}
  \end{tikzpicture}}
  \caption{The four orientation cases of the elementary move.
  In the Gate column the control edge $c$ (black) joins $u$ and $v$, the target edge $t$ (red) joins $w$ and $v$, and the controlled multiplication (blue; dot on the control, label on the target) acts on the two edges.
  In the Rewiring move column the original graph is mapped to the final one, in which $t$ is detached from $v$ and reattached to $u$ with its orientation relative to $w$ preserved; edges are labelled by their group elements, and the reattached edge carries the holonomy of the two-edge path it replaces.
  (a) Aligned, $c$ into $v$ and $t$ out of $v$: $CL$.
  (b) Aligned, $c$ out of $v$ and $t$ into $v$: $CR^{\dagger}$.
  (c) Antialigned, both into $v$: $CR$.
  (d) Antialigned, both out of $v$: $CL^{\dagger}$.}
  \label{fig:gr-four-cases}
\end{figure}
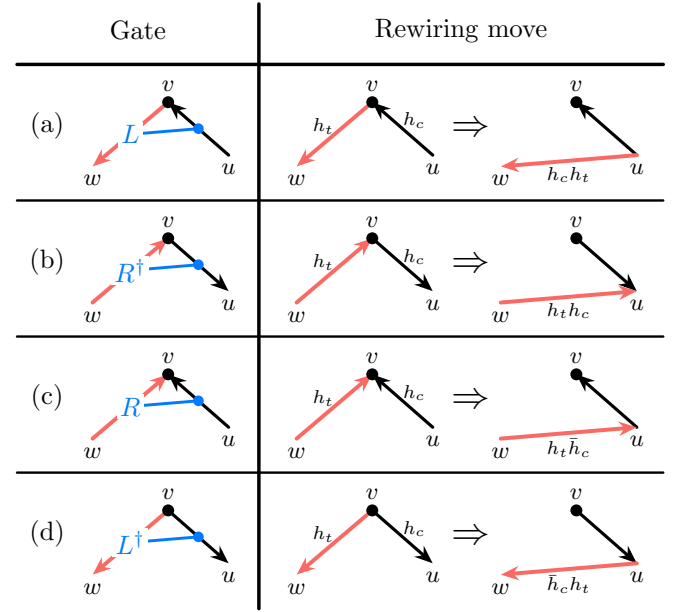

Graph rewiring also allows for graphs to be augmented by incorporating ancilla. Unentangled ancilla edges of Fig.~\ref{fig:ancilla-edges} can be entangled with the existing graph and modify the graph: elementary rewiring moves can split a vertex using an ancilla in $\ket{+}$ (Fig.~\ref{fig:ancilla-expand} (a)), and partition a plaquette with an ancilla in $\ket{e}$ (Fig.~\ref{fig:ancilla-expand} (b)). This can be interpreted as running an entanglement renormalization scheme backward, or in terms of an error correction code, introducing more redundancy in the encoding.

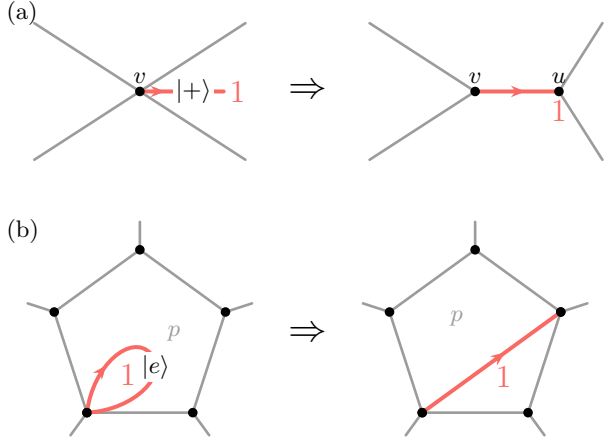
\begin{figure}[h]
\centering
\adjustbox{max width=\columnwidth}{%
\begin{tikzpicture}[x=0.095\columnwidth, y=0.095\columnwidth, line cap=round, line join=round,
    every node/.style={font=\small},
    ge/.style ={latgray, line width=1.0pt},
    oriented/.style={postaction={decorate, decoration={markings,
        mark=at position #1 with {\arrow{Stealth[length=1.9mm,width=1.5mm]}}}}},
    ae/.style ={edgered, line width=1.5pt},
    vx/.style ={circle, fill=black, inner sep=0pt, minimum size=3.6pt},
    sl/.style ={text=black, fill=white, inner sep=1pt},
    tr/.style ={edgered, font=\large, inner sep=1pt},
    vl/.style ={inner sep=1pt},
    pl/.style ={anchor=north west, inner sep=1pt},
  ]
  \newcommand{\exArmsL}{\draw[ge] (0,0)--(-1.7,1.1); \draw[ge] (0,0)--(-1.7,-1.1);}
  \begin{scope}
    \exArmsL
    \draw[ge] (0,0)--(1.7,1.1); \draw[ge] (0,0)--(1.7,-1.1);
    \draw[ae, oriented=0.3] (0,0)--(1.35,0) node[sl, pos=0.64] {$\ket{+}$};
    \node[vx] at (0,0) {}; \node[vl, above=2pt] at (0,0) {$v$};
    \node[tr, anchor=west] at (1.4,0) {$1$};
    \node[pl] at (-2.2,1.5) {(a)};
  \end{scope}
  \node[font=\Large] at (2.7,0) {$\Rightarrow$};
  \begin{scope}[shift={(5.4,0)}]
    \exArmsL
    \draw[ge] (1.35,0)--(2.05,1.1); \draw[ge] (1.35,0)--(2.05,-1.1);
    \draw[ae, oriented=0.6] (0,0)--(1.35,0);
    \node[vx] at (0,0) {}; \node[vl, above=2pt] at (0,0) {$v$};
    \node[vx] at (1.35,0) {}; \node[vl, above=2pt] at (1.35,0) {$u$};
    \node[tr, below=2pt] at (1.35,0) {$1$};
  \end{scope}
  \newcommand{\exPent}{%
    \coordinate (A) at (234:1.45); \coordinate (B) at (306:1.45); \coordinate (C) at (18:1.45); \coordinate (D) at (90:1.45); \coordinate (E) at (162:1.45);
    \draw[ge] (A)--(B)--(C)--(D)--(E)--cycle;
    \foreach \p in {A,B,C,D,E}{\draw[ge] (\p)--($(0,0)!1.31!(\p)$);}
    \foreach \p in {A,B,C,D,E}{\node[vx] at (\p) {};}}
  \begin{scope}[shift={(0,-4.0)}]
    \exPent
    \node[pq] at (0.55,0.1) {$p$};
    \begin{scope}[shift={(A)}, rotate=43]
      \draw[ae, oriented=0.25] (0,0) .. controls (0.6,0.5) and (1.4,0.45) .. (1.4,0)
                      .. controls (1.4,-0.45) and (0.6,-0.5) .. node[sl, pos=0.2] {$\ket{e}$} (0,0);
      \node[tr] at (0.9,0) {$1$};
    \end{scope}
    \node[vx] at (A) {};
    \node[pl] at (-2.2,2.0) {(b)};
  \end{scope}
  \node[font=\Large] at (2.7,-3.9) {$\Rightarrow$};
  \begin{scope}[shift={(5.4,-4.0)}]
    \exPent
    \draw[ae, oriented=0.6] (A)--(C);
    \node[vx] at (A) {}; \node[vx] at (C) {};
    \node[tr] at (0.46,-0.6) {$1$};
    \node[pq] at (-0.3,0.35) {$p$};
  \end{scope}
\end{tikzpicture}}
\caption{Expanding the graph with unentangled ancilla edges.
(a) An ancilla edge (red) in the state $\ket{+}$ is attached to a vertex $v$ of degree four (left).
Two elementary moves, controlled by the ancilla edge, reattach two of the edges at $v$ to its free end $u$ (right), so that $v$ is split into two vertices of degree three joined by the ancilla edge; the new vertex $u$ carries the trivial charge $1$.
(b) An ancilla edge in the state $\ket{e}$ is attached as a loop to a vertex of the pentagonal plaquette $p$ (left).
Two elementary moves, controlled by two consecutive edges of $p$, carry the free end of the loop to a next-nearest vertex (right), so that $p$ is partitioned by the diagonal into a triangle, which carries the trivial flux $1$, and a trapezoid.}
\label{fig:ancilla-expand}
\end{figure}

\subsection{Anyon creation: ribbon operator from graph rewiring}
\label{subsec:ribbon-rewiring}

Ribbon operators are used to create vacuum pair of non-Abelions on top of any input states. They are notoriously complicated operators \cite{kitaevFaulttolerantQuantumComputation2003,bombin_family_2008,yan_ribbon_2022}. Here we show that graph rewiring illuminate the structure of ribbon operator and provide an elementary derivation for unitary ribbon operators.

Recently, it was realized that unitary ribbon operators can be decomposed into a three-layer circuit \cite{lyons_protocols_2025}: 1) an ungauging layer to a simpler TO, 2) a finite-depth unitary layer, which creates anyon or defect pair in the simpler TO; and the regauging layer back to the original TO. 

Here, we reinterpret and rederive this result form a different perspective: where the system stays at the original quantum double, only defined on a different graph because of graph rewiring, so the whole circuit can be read off from the sequence of graph rewiring. We know that a single $L^g$ or $R^g$ operator can create a vacuum flux pair and a single $\mathcal{Z}^\Gamma_{ij}$ operator can create a vacuum charge pair. Therefore, the strategy is to rewire the graph to bring two far away sites together, then we can apply the single qudit operators to create the vacuum pair locally, then rewire the graph back. The rewiring turns out to coincide with the ungauging and regauging layers consist of elementary moves of Sec.~\ref{subsec:elementary-move}, thus providing an elementary derivation for the form of the unitary ribbon operator.

\paragraph{Charge ribbon.}
Suppose we want to create a $\Gamma, \bar{\Gamma}$ charge pair separated by four consecutive (purple) direct edges $a_{1},\dots,a_{4}$, as shown in Fig.~\ref{fig:gr-charge-ribbon}. The overall strategy is to create a geometrically non-local edge that connect between the vertices we want to create a vacuum charge pair. The charge ribbon circuit starts with sequential $CL$ gates, controlled by $a_{i}$ on $a_{i+1}$ for $i=1,2,3$, to create more and more geometrically non-local edges until a single non-local edge connect between the two endpoint of interest. Then, $\mathcal{Z}^{\Gamma}_{ij}$ acting on the longest non-local edge creates the charge $\bar{\Gamma}$ at its tail and $\Gamma$ at its head (by Eqs.~\eqref{eq:charge-projector}), so the pair is created directly at the two ends of the string, and the subsequent $CL^{\dagger}$ layer rewires to restore the square lattice. The combined unitary circuit in Fig.~\ref{fig:gr-charge-ribbon} coincides with the form of unitary charge ribbon operator given in Ref.~\onlinecite{lyons_protocols_2025}. 

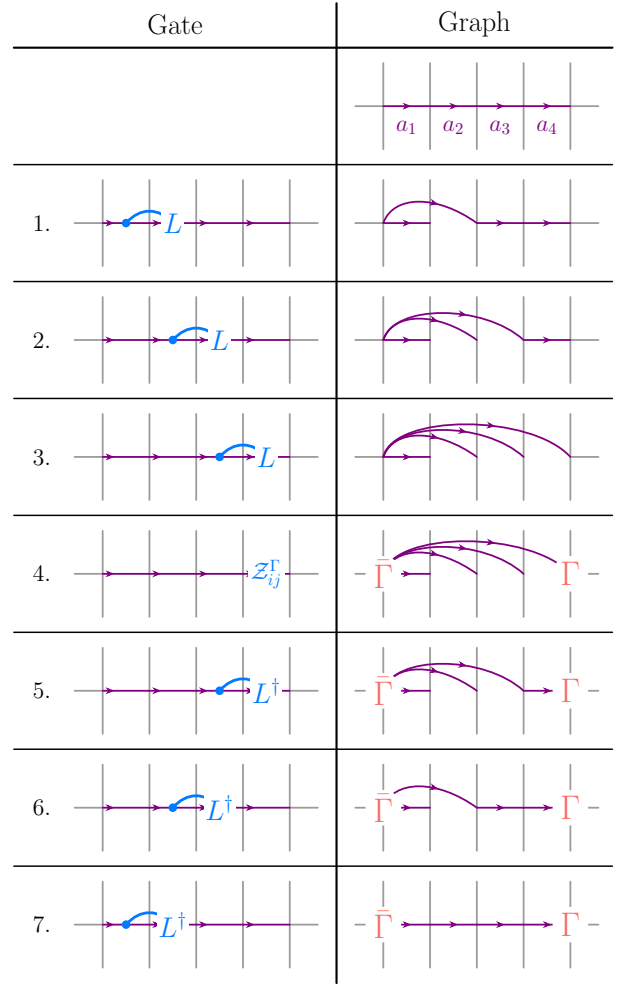
\begin{figure}[ht]
\centering
\adjustbox{max width=\columnwidth, max totalheight=0.55\textheight}{%
\begin{tikzpicture}[
      line cap=round, line join=round,
      oriented/.style 2 args={postaction={decorate, decoration={markings,
          mark=at position #1 with {\arrow[#2]{Stealth[length=1.9mm,width=1.5mm]}}}}},
      lat/.style   ={latgray, line width=1.1pt},
      dir/.style   ={violet,  line width=1.1pt, oriented={\rcArrPos}{violet}},
      dirarc/.style={violet,  line width=1.1pt, oriented={0.6}{violet}},
      rib/.style   ={opblue,  line width=1.7pt},
      rule/.style  ={black, line width=0.9pt},
      hd/.style    ={font=\LARGE},
      sn/.style    ={font=\Large},
      an/.style    ={edgered, font=\huge},
    ]
    \def\rowpitch{2.5}
    \def\stepX{-5.0}
    \def\latX{1.0}
    \def\xleft{-6.9}  \def\xmid{0.0}  \def\xright{5.9}
    \def\ytop{1.25}
    \def\rcArrPos{0.25}
    \newcommand{\rcVerts}{\foreach \i in {0,...,4}{\draw[lat] (\i,-0.92)--(\i,0.92);}}
    \newcommand{\rcStubs}{\draw[lat] (-0.6,0)--(0,0); \draw[lat] (4,0)--(4.6,0);}
    \newcommand{\rcStr}[1]{\draw[dir] (0,0)--(1,0);
      \ifnum#1<3 \pgfmathtruncatemacro{\rcFirst}{#1+1}\foreach \i in {\rcFirst,...,3}{\pgfmathtruncatemacro{\rcNext}{\i+1}\draw[dir] (\i,0)--(\rcNext,0);}\fi
      \ifnum#1>0 \foreach \i in {1,...,#1}{\pgfmathsetmacro{\rcEX}{\i+1}\pgfmathsetmacro{\rcHH}{0.45+0.16*\i}%
        \draw[dirarc] (0,0) .. controls (0.22,\rcHH) and ({\rcEX-0.85},\rcHH) .. (\rcEX,0);}\fi}
    \newcommand{\rcBase}{\rcVerts \rcStubs \rcStr{0}}
    \newcommand{\rcGate}[2]{\pgfmathsetmacro{\rcDX}{#1+0.5}%
      \draw[rib] (\rcDX,0) .. controls ({\rcDX+0.3},0.34) and ({\rcDX+0.7},0.34) .. ({\rcDX+1},0);
      \fill[opblue] (\rcDX,0) circle (2.6pt);
      \node[opblue, font=\LARGE, fill=white, inner sep=1pt] at ({\rcDX+1},0) {#2};}
    \newcommand{\rcAnyon}[2]{\fill[white] (#1,0) circle (0.38); \node[an] at (#1,0) {#2};}
    \newcommand{\rcAnyons}{\rcAnyon{0}{$\bar{\Gamma}$} \rcAnyon{4}{$\Gamma$}}
    \node[hd] at ({(\xleft+\xmid)/2}, {\ytop+0.5}) {Gate};
    \node[hd] at ({(\xmid+\xright)/2}, {\ytop+0.5}) {Graph};
    \draw[rule, line width=1.3pt] (\xleft,\ytop) -- (\xright,\ytop);
    \draw[rule, line width=1.3pt] (\xmid,{\ytop+0.95}) -- (\xmid,{-7*\rowpitch-\ytop});
    \foreach \r in {1,...,7}{\draw[rule] (\xleft,{-\r*\rowpitch+\ytop}) -- (\xright,{-\r*\rowpitch+\ytop});}
    \begin{scope}[shift={(\latX,0)}] \def\rcArrPos{0.6} \rcBase
      \foreach \i in {1,...,4}{\node[violet, font=\Large, inner sep=1pt] at ({\i-0.5},-0.45) {$a_{\i}$};} \end{scope}
    \foreach \r/\off in {1/0, 2/1, 3/2}{
      \node[sn] at (-6.3,{-\r*\rowpitch}) {\r.};
      \begin{scope}[shift={(\stepX,{-\r*\rowpitch})}] \rcBase \rcGate{\off}{$L$} \end{scope}
      \begin{scope}[shift={(\latX,{-\r*\rowpitch})}] \def\rcArrPos{0.6} \rcVerts \rcStubs \rcStr{\r} \end{scope}}
    \node[sn] at (-6.3,{-4*\rowpitch}) {4.};
    \begin{scope}[shift={(\stepX,{-4*\rowpitch})}] \rcBase
      \node[opblue, font=\Large, fill=white, inner sep=1.5pt] at (3.5,0) {$\mathcal{Z}^{\Gamma}_{ij}$}; \end{scope}
    \begin{scope}[shift={(\latX,{-4*\rowpitch})}] \def\rcArrPos{0.6} \rcVerts \rcStubs \rcStr{3} \rcAnyons \end{scope}
    \foreach \r/\off/\n in {5/2/2, 6/1/1, 7/0/0}{
      \node[sn] at (-6.3,{-\r*\rowpitch}) {\r.};
      \begin{scope}[shift={(\stepX,{-\r*\rowpitch})}] \rcBase \rcGate{\off}{$L^{\dagger}$} \end{scope}
      \begin{scope}[shift={(\latX,{-\r*\rowpitch})}] \def\rcArrPos{0.6} \rcVerts \rcStubs \rcStr{\n} \rcAnyons \end{scope}}
  \end{tikzpicture}}
\caption{Non-Abelian charge ribbon as graph rewiring, on four consecutive direct edges $a_{1},\dots,a_{4}$.
Each row shows the applied elementary move (Gate), drawn on the physical edges of the original lattice, and the rewired graph (Graph); arrowheads mark the edge orientation, chosen along the string (violet).
Steps 1--3 re-attach the tail of each successive edge to the starting vertex, so that after step 3 the edge $a_{4}$ links the two ends of the string and carries its holonomy $h_{a_{1}}h_{a_{2}}h_{a_{3}}h_{a_{4}}$.
Step 4 applies $\mathcal{Z}^{\Gamma}_{ij}$ to this edge, which creates a charge $\bar{\Gamma}$ at its tail and a charge $\Gamma$ at its head (red).
Steps 5--7 restore the original lattice.}
\label{fig:gr-charge-ribbon}
\end{figure}

\paragraph{Flux ribbon.}
Suppose we want to create a pair of fluxes labeled by $g,\bar{g}$ separated by three dual edges $\alpha_{1},\alpha_{2},\alpha_{3}$, as shown in Fig.~\ref{fig:gr-flux-ribbon}. Since these are non-Abelian fluxes, the flux label depends on the position of the base point, so we need to collect the dual edges to the same base point in order to apply $L^g$ on all the dual edges to create a local flux pair. Therefore, the strategy is to rewire the direct edges to create non-local edges as in the charge ribbon case, then the dual edges can be rewired to to the same vertex. The flux ribbon circuit starts with sequential $CL$ gates, controlled by $a_{i}$ and target on $\alpha_{i}$ and on $a_{i+1}$. After all the dual edges are connected to the same vertex at the end of step 3, then applying the single-edge operators $L^{g}$ on every dual edge create flux label $g, \bar{g}$ based at the same vertex. Finally, we rewire with $CL^\dagger$ to restore the original square lattice. The combined unitary circuit in Fig.~\ref{fig:gr-flux-ribbon} coincides with the form of unitary flux ribbon operator given in Ref.~\onlinecite{lyons_protocols_2025}.

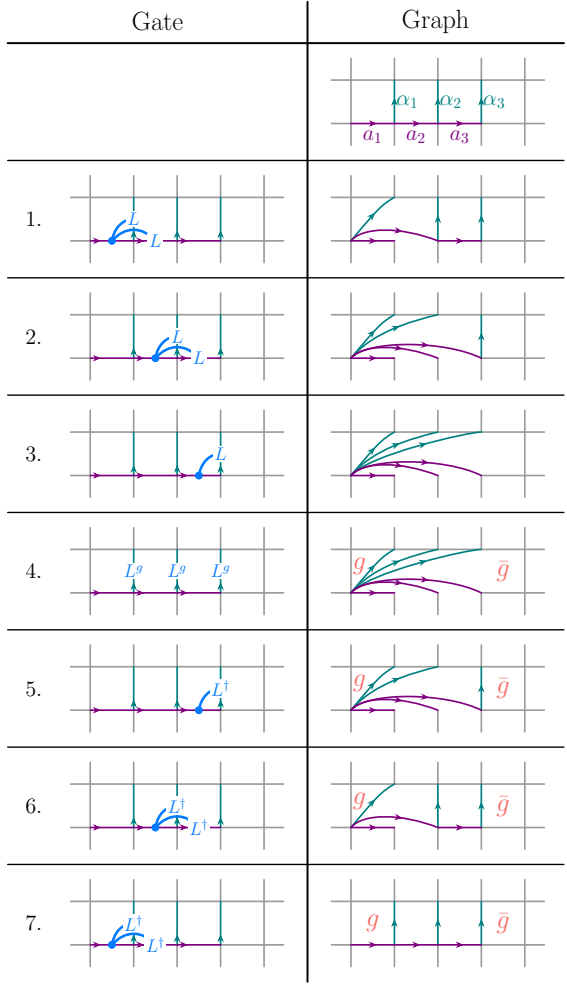
\begin{figure}[ht]
\centering
\adjustbox{max width=\columnwidth, max totalheight=0.55\textheight}{%
\begin{tikzpicture}[
      line cap=round, line join=round,
      oriented/.style 2 args={postaction={decorate, decoration={markings,
          mark=at position #1 with {\arrow[#2]{Stealth[length=1.9mm,width=1.5mm]}}}}},
      lat/.style    ={latgray, line width=1.1pt},
      dir/.style    ={violet,  line width=1.1pt, oriented={\rfArrPos}{violet}},
      dirarc/.style ={violet,  line width=1.1pt, oriented={0.6}{violet}},
      dual/.style   ={teal,    line width=1.1pt, oriented={\rfArrPos}{teal}},
      dualarc/.style={teal,    line width=1.1pt, oriented={0.6}{teal}},
      rib/.style    ={opblue,  line width=1.7pt},
      gl/.style     ={opblue,  font=\large},
      rule/.style   ={black, line width=0.9pt},
      hd/.style     ={font=\LARGE},
      sn/.style     ={font=\Large},
      an/.style     ={edgered, font=\LARGE},
    ]
    \def\rowpitch{2.7}
    \def\stepX{-5.0}
    \def\latX{1.0}
    \def\xleft{-6.9}  \def\xmid{0.0}  \def\xright{5.9}
    \def\ytop{1.35}
    \def\rfArrPos{0.25}
    \newcommand{\rfSpect}{\foreach \i in {0,...,4}{\draw[lat] (\i,-0.5)--(\i,0); \draw[lat] (\i,1)--(\i,1.5);}
      \draw[lat] (0,0)--(0,1); \draw[lat] (4,0)--(4,1);
      \draw[lat] (-0.45,1)--(4.45,1);
      \draw[lat] (-0.45,0)--(0,0); \draw[lat] (3,0)--(4.45,0);}
    \newcommand{\rfUps}[1]{\ifnum#1<3 \pgfmathtruncatemacro{\rfFirst}{#1+1}\foreach \i in {\rfFirst,...,3}{\draw[dual] (\i,0)--(\i,1);}\fi}
    \newcommand{\rfFan}[1]{\ifnum#1>0 \foreach \i in {1,...,#1}{\pgfmathsetmacro{\rfS}{1.4-0.2*\i}%
      \ifnum\i=1 \draw[dualarc] (0,0) .. controls ({0.5*\rfS},{0.6*\rfS}) and (0.7,0.827) .. (1,1); \fi
      \ifnum\i=2 \draw[dualarc] (0,0) .. controls ({0.5*\rfS},{0.6*\rfS}) and (1.4,0.86)  .. (2,1); \fi
      \ifnum\i=3 \draw[dualarc] (0,0) .. controls ({0.5*\rfS},{0.6*\rfS}) and (2.7,1.0)   .. (3,1); \fi}\fi}
    \newcommand{\rfStr}[1]{\draw[dir] (0,0)--(1,0);
      \ifnum#1<2 \pgfmathtruncatemacro{\rfFirst}{#1+1}\foreach \i in {\rfFirst,...,2}{\pgfmathtruncatemacro{\rfNext}{\i+1}\draw[dir] (\i,0)--(\rfNext,0);}\fi
      \ifnum#1>0 \foreach \i in {1,...,#1}{\pgfmathsetmacro{\rfEX}{\i+1}\pgfmathsetmacro{\rfS}{0.4+0.15*\i}%
        \draw[dirarc] (0,0) .. controls ({0.5*\rfS},{0.6*\rfS}) and ({\rfEX-0.85},{0.6*\rfS}) .. (\rfEX,0);}\fi}
    \newcommand{\rfBase}{\rfSpect \rfUps{0} \rfStr{0}}
    \newcommand{\rfStep}[3]{%
      \fill[opblue] ({#1-0.5},0) circle (2.6pt);
      \draw[rib] ({#1-0.5},0) .. controls ({#1-0.45},0.3) and ({#1-0.25},0.5) .. (#1,0.5);
      \node[gl, fill=white, inner sep=1pt] at (#1,0.5) {#2};
      \ifnum#3=1
        \draw[rib] ({#1-0.5},0) .. controls ({#1-0.2},0.34) and ({#1+0.2},0.34) .. ({#1+0.5},0);
        \node[gl, fill=white, inner sep=1pt] at ({#1+0.5},0) {#2};
      \fi}
    \newcommand{\rfLg}{\foreach \i in {1,...,3}{\node[gl, fill=white, inner sep=1pt] at (\i,0.5) {$L^{g}$};}}
    \newcommand{\rfAnyons}[2]{\node[an] at (#1,#2) {$g$}; \node[an] at (3.5,0.5) {$\bar{g}$};}
    \node[hd] at ({(\xleft+\xmid)/2}, {\ytop+0.5}) {Gate};
    \node[hd] at ({(\xmid+\xright)/2}, {\ytop+0.5}) {Graph};
    \draw[rule, line width=1.3pt] (\xleft,\ytop) -- (\xright,\ytop);
    \draw[rule, line width=1.3pt] (\xmid,{\ytop+0.95}) -- (\xmid,{-7*\rowpitch-\ytop});
    \foreach \r in {1,...,7}{\draw[rule] (\xleft,{-\r*\rowpitch+\ytop}) -- (\xright,{-\r*\rowpitch+\ytop});}
    \begin{scope}[shift={(\latX,-0.5)}] \def\rfArrPos{0.6} \rfBase
      \foreach \i in {1,...,3}{\node[violet, font=\Large, inner sep=1pt] at ({\i-0.5},-0.3) {$a_{\i}$};
        \node[teal, font=\Large, inner sep=1pt] at ({\i+0.3},0.5) {$\alpha_{\i}$};} \end{scope}
    \foreach \r/\nxt in {1/1, 2/1, 3/0}{
      \node[sn] at (-6.3,{-\r*\rowpitch}) {\r.};
      \begin{scope}[shift={(\stepX,{-\r*\rowpitch-0.5})}] \rfBase \rfStep{\r}{$L$}{\nxt} \end{scope}
      \pgfmathtruncatemacro{\rfN}{min(\r,2)}
      \begin{scope}[shift={(\latX,{-\r*\rowpitch-0.5})}] \def\rfArrPos{0.6} \rfSpect \rfUps{\r} \rfFan{\r} \rfStr{\rfN} \end{scope}}
    \node[sn] at (-6.3,{-4*\rowpitch}) {4.};
    \begin{scope}[shift={(\stepX,{-4*\rowpitch-0.5})}] \rfBase \rfLg \end{scope}
    \begin{scope}[shift={(\latX,{-4*\rowpitch-0.5})}] \def\rfArrPos{0.6} \rfSpect \rfUps{3} \rfFan{3} \rfStr{2} \rfAnyons{0.2}{0.62} \end{scope}
    \foreach \r/\i/\nxt/\k/\mx/\my in {5/3/0/2/0.2/0.62, 6/2/1/1/0.2/0.62, 7/1/1/0/0.5/0.5}{
      \node[sn] at (-6.3,{-\r*\rowpitch}) {\r.};
      \begin{scope}[shift={(\stepX,{-\r*\rowpitch-0.5})}] \rfBase \rfStep{\i}{$L^{\dagger}$}{\nxt} \end{scope}
      \pgfmathtruncatemacro{\rfN}{min(\k,2)}
      \begin{scope}[shift={(\latX,{-\r*\rowpitch-0.5})}] \def\rfArrPos{0.6} \rfSpect \rfUps{\k} \rfFan{\k} \rfStr{\rfN} \rfAnyons{\mx}{\my} \end{scope}}
  \end{tikzpicture}}
\caption{Non-Abelian flux ribbon as graph rewiring, on the three dual edges $\alpha_{1},\alpha_{2},\alpha_{3}$ (teal) attached to the direct string $a_{1},a_{2},a_{3}$ (violet).
Each row shows the applied elementary move (Gate), drawn on the physical edges of the original lattice, and the rewired graph (Graph); arrowheads mark the edge orientations, with the direct edges oriented along the string and the dual edges pointing away from it, so that every gate acts on aligned edges.
After step 3 all three dual edges emanate from the starting vertex and carry the holonomies $h_{a_{1}}\cdots h_{a_{i}}h_{\alpha_{i}}$.
Step 4 applies $L^{g}$ to every dual edge: the holonomy of each intermediate face is conjugated by $g$ and stays flat, while the two end plaquettes acquire the fluxes $g$ and $\bar{g}$ (red), both based at the starting vertex.
Steps 5--7 restore the original lattice.
For $G=\mathbb{Z}_{3}$ and $g=1$ the gates reduce to $C\mathcal{X}$ and the Pauli $\mathcal{X}$, and the circuit is the $\mathcal{X}$ string operator of the toric code.}
\label{fig:gr-flux-ribbon}
\end{figure}

\subsection{Comparison with prior work}\label{subsec:lit-review-comparison}

\subsubsection{Connection with entanglement renormalization}

Given a ground state of a quantum double, charge neutrality at a vertex (Eq.~\eqref{eq:charge-neutrality}) and flux neutrality on a plaquette (Eq.~\eqref{eq:flux-neutrality}) can be used to disentangle an edge from the ground state, as used in entanglement renormalization of quantum double states \cite{aguado_vidal_2008}.
The disentangling circuit consists of three controlled group multiplication gates, known as the P-moves and V-moves of Ref.~\cite{aguado_entanglement_2011}. Ref.~\cite{dennis_topological_2002} and \cite{aguado_entanglement_2011} regards these moves as the elementary operations. We show, however, that the elementary move of graph rewiring is \textit{individual} controlled group multiplication gate: each has a well-defined rewiring move on the graph that defines a qunatum double of $G$.

The idea of using Pachner moves (which re-triangulate a manifold~\cite{pachner1991pl}, serving as a primitive for entanglement renormalization) to re-triangulate a manifold to move anyons has precedent: 
it has been implemented as local unitaries on string-net and Turaev-Viro codes to move anyons and apply logical gates~\cite{koenig2010turaevviro,zhu2020universal,zhu2020instantaneous}. Perhaps due to the generality of string-net models, the local unitaries these prior works used is not a single elementary entangling gate. Nonetheless, due to the similar spirit of changing the underlying lattice or graph, it would be interesting for future work to see how graph rewiring can be suitably generalized to the string-net context. We also want to point out that the elementary move in graph rewiring is more elementary and hence subsumes the Pachner moves, as any triangulation of a manifold is itself a directed graph. In general, since graph is not restricted to triangulation of manifolds, the graph rewiring procedure readily apply to a general 2-dimensional CW complex (see footnote at the start of Sec.~\ref{subsec:qd-basics}).

\subsubsection{Connection with ZX calculus}

We now comment on the correspondence between a graph and the set of operators associated with the quantum double defined on the graph. 
This correspondence rests on the Hopf algebra structure of the group algebra $\mathbb{C}[G]$ carried by each edge and of its dual $\mathbb{C}^{G}$: the vertex operator $A^{g}_{v}$ multiplies a copy of $g$ into every incident edge, from the left or the right according to the orientation, and the plaquette operator $B^{e}_{p}$ evaluates the cointegral $\delta_{e}$ on the ordered product of the boundary edges.
Graph rewiring, as a diagrammatic calculus, is therefore closely related to the ZX calculus for Hopf algebras~\cite{collins2020hopffrobenius,majid2022zx,collins2024thesis}, in the form recently used for Hopf-algebra-valued spin chains in Ref.~\cite{lu2026hopfising}: $A^{g}_{v}$ and $B^{e}_{p}$ are the $X$-type and $Z$-type Hopf Pauli strings of $\mathbb{C}[G]$, the controlled group multiplications are the Hopf generalizations of $CX$, and the sliding, splitting and absorption moves of Appendices~\ref{app:mult-control-comm} and \ref{app:ctrl-diag} are its bialgebra and antipode rewrite rules.

The directed graph is the skeleton of the ZX diagram of the ground state, with a copy spider at every vertex and a multiplication spider on every edge; this is the tensor network representation of quantum double states~\cite{buerschaper_hierarchy_2013}. For $G=\mathbb{Z}_{2}$ it is the correspondence between CSS codes and phase-free ZX diagrams~\cite{kissinger_phasefree_2022} and, more recently, between general stabilizer codes and graphs~\cite{khesin2024universal}.
Our contribution is the converse use of this correspondence: rather than rewriting diagrams to prove operator identities, we use the known link between diagrams and operators and treat the graph itself as shorthand for the stabilizer group.

A single controlled group multiplication between two edges sharing a vertex maps the stabilizer group of one directed graph onto that of another, in which one endpoint of the target edge has slid along the control edge (Sec.~\ref{subsec:elementary-move}).
Because the move rests on the Hopf algebra structure of $\mathbb{C}[G]$ rather than on the Pauli group, it holds for any finite group $G$ and any planar directed graph, with the left or right action and the orientation of every edge tracked by the drawing.
The identification of a stabilizer state with its graph, previously available only for stabilizer codes~\cite{kissinger_phasefree_2022,khesin2024universal}, thereby extends to the non-Abelian codes $D(G)$, whose lattice moves had until now been composite gadgets that add or remove edges~\cite{aguado_entanglement_2011,dennis_topological_2002}.

\subsubsection{Connection with graph operations}

Deforming a stabilizer code with entangling gates has likewise been studied only in the Pauli setting~\cite{vasmer2022morphing,vuillot2019code,kubica2018ungauging,bombin2024unifying}, where every controlled group multiplication reduces to a $CX$ gate.
Closest to our elementary move are the Pauli-code graph deformations of Refs.~\cite{lensky2023graph,andersen2023nonabelian}, where a single two-qubit gate moves one edge and thereby moves a projective Ising anyon (equivalently, a Majorana twist defect), and the composite controlled-multiplication moves that add or remove an edge of a quantum double lattice~\cite{dennis_topological_2002,aguado_vidal_2008}. Retriangulation~\cite{koenig2010turaevviro,zhu2020universal,zhu2020instantaneous,minev2025stringnet} and measurement-based code deformation~\cite{brown2017poking} also keep every step a valid code, but do so either through multi-qubit moves or by removing qubits.
None of these assigns a lattice meaning to a single controlled multiplication between two existing edges of a non-Abelian $D(G)$ lattice, which is the elementary move we identified in Sec.~\ref{subsec:elementary-move}.

Since every graph rewiring move is unitary and preserves the numbers of edges, vertices and plaquettes, the tensor product Hilbert space and the stabilizer set stay well defined after each gate, so charges and fluxes can be located on the graph at every intermediate step.
This allows us to use graph rewiring as a design tool: the braiding and fusion circuits of Secs.~\ref{sec:braiding} and \ref{sec:fusion} are obtained by choosing a sequence of graphs, and the gates are read off from the moves.

\section{Braiding as graph rewiring}\label{sec:braiding}

In a topological quantum computation scheme, logical information is stored in the nonlocal internal Hilbert space of collection of non-Abelions \cite{kitaevFaulttolerantQuantumComputation2003,nayak_non-abelian_2008}; braiding non-Abelian anyons acts on the nonlocal internal Hilbert space, hence serves an important primitive to enable logical gates. To that end, we leverage graph rewiring to develop circuits that braid non-Abelions without the use of ancilla. Here, we provide concrete protocols for exchanging two charges (Sec.~\ref{subsec:charge-exchange}) and exchanging two fluxes (Sec.~\ref{subsec:flux-exchange}) in nearest neighbor sites. These circuits can be composed to exchange two dyons. To obtain a full braid, one just need to apply the exchange circuit twice. For exchanges involving anyons separated by longer distance, we just need to compose the exchange circuit to form a sequential circuit. No ancilla is used in the exchange circuits.

\subsection{Reduction to nearest neighbor exchange}

We first show that one can exchange and braid anyons arbitrarily far apart, without any knowledge, by using nearest neighbor exchanges as primitives. Hence, it suffices for us to construct concrete protocols for nearest neighbor exchanges in Sec.~\ref{subsec:charge-exchange} and Sec.~\ref{subsec:flux-exchange}.

\begin{figure}[ht]
\centering
\adjustbox{max width=0.7\columnwidth}{%
\begin{tikzpicture}[x=1cm, y=1cm, line cap=round, line join=round,
    every node/.style={font=\small},
    any/.style={circle, inner sep=0pt, minimum size=7pt, draw=black!45, line width=0.5pt},
    mid/.style={any, fill=latgray!60},
    hopA/.style={-{Stealth[length=1.8mm,width=1.4mm]}, opblue!70!black, line width=0.8pt},
    hopB/.style={-{Stealth[length=1.8mm,width=1.4mm]}, ribgreen!70!black, line width=0.8pt},
    hopM/.style={-{Stealth[length=1.8mm,width=1.4mm]}, latgray, line width=0.8pt}]
  \newcommand{\anyA}[3]{\node[any, fill=opblue!70]  at (#1,#2) {}; \node[anchor=#3, inner sep=4pt] at (#1,#2) {$a$};}
  \newcommand{\anyB}[3]{\node[any, fill=ribgreen!80] at (#1,#2) {}; \node[anchor=#3, inner sep=4pt] at (#1,#2) {$b$};}
  \newcommand{\anyM}[2]{\node[mid] at (#1,#2) {};}
  \newcommand{\cw}[4]{%
    \draw[#3] ($(#2+0.5,#1)+(160:0.5)$) arc[start angle=160, end angle=20,   radius=0.5];
    \draw[#4] ($(#2+0.5,#1)+(-20:0.5)$) arc[start angle=-20, end angle=-160, radius=0.5];}
  \newcommand{\ccw}[4]{%
    \draw[#3] ($(#2+0.5,#1)+(20:0.5)$)   arc[start angle=20,   end angle=160, radius=0.5];
    \draw[#4] ($(#2+0.5,#1)+(-160:0.5)$) arc[start angle=-160, end angle=-20, radius=0.5];}
  \node[anchor=east] at (-0.6,0) {1.};
  \anyA{0}{0}{east} \foreach \i in {1,...,4}{\anyM{\i}{0}} \anyB{5}{0}{west}
  \foreach \i in {0,...,3}{\cw{0}{\i}{hopA}{hopM}}
  \cw{0}{4}{hopA}{hopB}
  \node[anchor=east] at (-0.6,-2.0) {2.};
  \foreach \i in {0,...,3}{\anyM{\i}{-2.0}} \anyB{4}{-2.0}{south west} \anyA{5}{-2.0}{west}
  \foreach \i in {0,...,3}{\ccw{-2.0}{\i}{hopB}{hopM}}
  \node[anchor=east] at (-0.6,-4.0) {3.};
  \anyB{0}{-4.0}{east} \foreach \i in {1,...,4}{\anyM{\i}{-4.0}} \anyA{5}{-4.0}{west}
\end{tikzpicture}
}
\caption{Exchange of two far-apart anyons $a$ (blue) and $b$ (green) on a line through nearest neighbor exchanges: $a$ is carried to the right end by clockwise exchanges (step 1), then $b$ is carried to the left end by counterclockwise exchanges (step 2), leaving the anyons in the middle (grey) in place (step 3).}
\label{fig:far-exchange}
\end{figure}
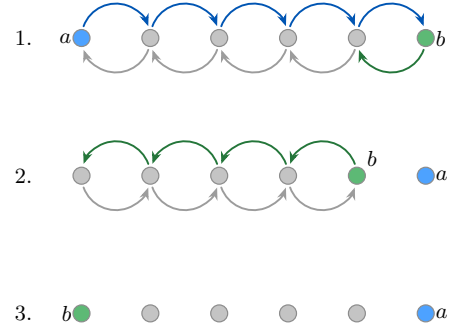

When exchanging between far away anyons (say they are on a line), we always make sure that the left and right one use different rotational direction exchange, such that applying twice doesn't result in braiding around the anyons in the middle (Fig.~\ref{fig:far-exchange}). In Fig.~\ref{fig:far-exchange}, both $a$ and $b$ pass on top of the intermediate anyons, which ensures that applying the exchange twice does not braid around them. We emphasize that we make no assumption about the intermediate space being free of anyons.

\subsection{Charge exchange circuit}
\label{subsec:charge-exchange}

\begin{figure}[ht]
\centering
\adjustbox{max width=\columnwidth, max totalheight=0.8\textheight}{%
\begin{tikzpicture}[
    line cap=round, line join=round,
    lat/.style ={latgray, line width=1.1pt},
    ga/.style ={postaction={decorate, decoration={markings,
                 mark=at position #1 with {\arrow{Stealth[length=2.2mm,width=1.9mm]}}}}},
    ga/.default=0.52,
    gar/.style={postaction={decorate, decoration={markings,
                 mark=at position #1 with {\arrowreversed{Stealth[length=2.2mm,width=1.9mm]}}}}},
    gar/.default=0.48,
    red edge/.style={edgered, line width=1.5pt},
    rule/.style={black, line width=0.9pt},
    mv/.style  ={opblue, line width=1.2pt},
    ml/.style  ={opblue, font=\LARGE, fill=white, inner sep=1.2pt},
    hd/.style  ={black, font=\LARGE},
    sn/.style  ={black, font=\huge},
  ]

  \def\rowpitch{4.0}
  \def\gcol{6.4}
  \def\tleft{-4.3}
  \def\tright{9.5}
  \def\tmid{3.2}
  \def\cshift{-1.95}
  \def\ceArrTip{1.7}

  \node[hd] at (0, 2.85) {Gate};
  \node[hd] at (\gcol, 2.85) {Graph};
  \draw[rule, line width=1.3pt] (\tleft,2.0) -- (\tright,2.0);
  \draw[rule, line width=1.3pt] (\tmid,3.7) -- (\tmid,{-4*\rowpitch-2.0});
  \foreach \r in {1,...,4}{\draw[rule] (\tleft,{-\r*\rowpitch+2.0}) -- (\tright,{-\r*\rowpitch+2.0});}

  \newcommand{\ceArrR}[2]{\draw[lat,-{Stealth[length=2.2mm,width=1.9mm]}] (#1-0.15,#2)--(#1+0.02,#2);}
  \newcommand{\ceArrU}[2]{\draw[lat,-{Stealth[length=2.2mm,width=1.9mm]}] (#1,#2-0.15)--(#1,#2+0.02);}
  \newcommand{\ceCtrl}{\fill[opblue] (1.5,0) circle (2.4pt);}
  \newcommand{\ceTopBot}{%
    \draw[lat] (-0.35, 1)--(3.35, 1);
    \draw[lat] (-0.35,-1)--(3.35,-1);
    \draw[lat] (0,-1)--(0,1);
    \draw[lat] (3,-1)--(3,1);
    \foreach \xc in {0,1,2,3}{\draw[lat] (\xc,1)--(\xc,1.35); \draw[lat] (\xc,-1)--(\xc,-1.35);}
    \foreach \xa in {0.5,1.5,2.5}{\ceArrR{\xa}{1} \ceArrR{\xa}{-1}}
    \foreach \xb in {0,3}{\ceArrU{\xb}{0.5} \ceArrU{\xb}{-0.5}}}
  \newcommand{\ceRedEdge}{%
    \draw[red edge] (1,0)--(2,0);
    \draw[red edge,-{Stealth[length=2.7mm,width=2.2mm]}] (1,0)--(\ceArrTip,0);}
  \newcommand{\ceDot}[1]{\fill[gray!95!black] (#1,0) circle (2.4pt);}
  \newcommand{\ceDotO}[1]{\draw[gray!95!black, line width=1pt, fill=white] (#1,0) circle (2.2pt);}
  \newcommand{\ceUpTwoToL}{\draw[lat,gar] (2,1) .. controls (1.8,0.56) and (1.3,0.26) .. (1,0);}
  \newcommand{\ceRightTwoToL}{\draw[lat,ga] (1,0) .. controls (1.3,0.26) and (2.4,0.5) .. (3,0);}
  \newcommand{\ceDownTwoToL}{\draw[lat,gar] (1,0) .. controls (1.3,-0.26) and (1.8,-0.56) .. (2,-1);}
  \newcommand{\ceUpOneToR}{\draw[lat,gar] (1,1) .. controls (1.2,0.56) and (1.7,0.26) .. (2,0);}
  \newcommand{\ceDownOneToR}{\draw[lat,ga] (1,-1) .. controls (1.2,-0.56) and (1.7,-0.26) .. (2,0);}
  \newcommand{\ceLeftOneToR}{\draw[lat,ga] (0,0) .. controls (0.5,-0.5) and (1.5,-0.45) .. (2,0);}

  \expandafter\newcommand\csname ceC1\endcsname{%
    \ceTopBot
    \draw[lat] (-0.35,0)--(3.35,0);
    \draw[lat] (1,-1)--(1,1);
    \draw[lat] (2,-1)--(2,1);
    \ceArrR{0.5}{0} \ceArrR{2.5}{0}
    \ceArrU{1}{0.5} \ceArrU{1}{-0.5} \ceArrU{2}{0.5} \ceArrU{2}{-0.5}
    \ceRedEdge
    \ceDot{2} \ceDotO{1}}
  \expandafter\newcommand\csname ceC2\endcsname{%
    \ceTopBot
    \draw[lat] (-0.35,0)--(1,0);
    \draw[lat] (1,-1)--(1,1);
    \draw[lat] (2,-1)--(2,0);
    \ceUpTwoToL
    \ceRightTwoToL
    \draw[lat] (3,0)--(3.35,0);
    \ceArrR{0.5}{0} \ceArrU{1}{0.5} \ceArrU{1}{-0.5} \ceArrU{2}{-0.5}
    \ceRedEdge
    \ceDot{2} \ceDotO{1}}
  \expandafter\newcommand\csname ceC3\endcsname{%
    \ceTopBot
    \draw[lat] (-0.35,0)--(0,0);
    \ceLeftOneToR
    \draw[lat] (1,0)--(1,1);
    \draw[lat] (2,-1)--(2,0);
    \ceUpTwoToL
    \ceRightTwoToL
    \ceDownOneToR
    \draw[lat] (3,0)--(3.35,0);
    \ceArrU{1}{0.5} \ceArrU{2}{-0.5}
    \ceRedEdge
    \ceDot{2} \ceDotO{1}}
  \expandafter\newcommand\csname ceC4\endcsname{%
    \ceTopBot
    \draw[lat] (-0.35,0)--(1,0);
    \draw[lat] (1,-1)--(1,0);
    \draw[lat] (2,0)--(3.35,0);
    \draw[lat] (2,0)--(2,1);
    \ceUpOneToR
    \ceDownTwoToL
    \ceArrR{0.5}{0} \ceArrR{2.5}{0} \ceArrU{1}{-0.5} \ceArrU{2}{0.5}
    \ceRedEdge
    \ceDot{1} \ceDotO{2}}
  \expandafter\newcommand\csname ceC5\endcsname{%
    \ceTopBot
    \draw[lat] (-0.35,0)--(3.35,0);
    \draw[lat] (1,-1)--(1,1);
    \draw[lat] (2,-1)--(2,1);
    \ceArrR{0.5}{0} \ceArrR{2.5}{0}
    \ceArrU{1}{0.5} \ceArrU{1}{-0.5} \ceArrU{2}{0.5} \ceArrU{2}{-0.5}
    \ceRedEdge
    \ceDot{1} \ceDotO{2}}

  \expandafter\newcommand\csname ceG1\endcsname{%
    \ceCtrl
    \draw[mv] (1.5,0) .. controls (1.65,0.3) and (1.82,0.46) .. (1.97,0.5);
    \node[ml] at (2,0.5) {$L$};
    \draw[mv] (1.5,0) .. controls (1.75,-0.45) and (2.25,-0.38) .. (2.47,-0.07);
    \node[ml] at (2.5,0) {$L$};}
  \expandafter\newcommand\csname ceG2\endcsname{%
    \ceCtrl
    \draw[mv] (1.5,0) .. controls (1.38,-0.3) and (1.18,-0.5) .. (1.03,-0.5);
    \node[ml] at (1,-0.5) {$R^{\dagger}$};
    \draw[mv] (0.53,0.05) .. controls (0.9,0.4) and (1.25,0.3) .. (1.5,0.02);
    \node[ml] at (0.5,0) {$R^{\dagger}$};}
  \expandafter\newcommand\csname ceG3\endcsname{%
    \node[ml] at (1.4,0) {$\mathcal{S}$};}
  \expandafter\newcommand\csname ceG4\endcsname{%
    \ceCtrl
    \draw[mv] (1.5,0) .. controls (1.3,0.12) and (1.2,0.4) .. (1.24,0.6);
    \node[ml] at (1.25,0.64) {$L$};
    \draw[mv] (1.5,0) .. controls (1.7,-0.12) and (1.8,-0.4) .. (1.76,-0.6);
    \node[ml] at (1.75,-0.64) {$R^{\dagger}$};}

  \foreach \r in {1,...,5}{
    \begin{scope}[shift={({\gcol+\cshift},{-(\r-1)*\rowpitch})}, scale=1.3]
      \csname ceC\r\endcsname
    \end{scope}}
  \begin{scope}[shift={({\gcol+\cshift},0)}, scale=1.3]
    \node[hd, anchor=south east, inner sep=1pt] at (0.94,0.08) {$v_{1}$};
    \node[hd, anchor=south west, inner sep=1pt] at (2.06,0.08) {$v_{2}$};
  \end{scope}
  \foreach \r in {2,...,5}{
    \pgfmathtruncatemacro{\p}{\r-1}
    \node[sn] at (-3.55,{-(\r-1)*\rowpitch+1.3}) {\p.};
    \def\ceArrTip{1.85}
    \begin{scope}[shift={(\cshift,{-(\r-1)*\rowpitch})}, scale=1.3]
      \csname ceC\p\endcsname
      \csname ceG\p\endcsname
    \end{scope}}

\end{tikzpicture}}
\caption{Charge exchange circuit as graph rewiring.
The circuit exchanges charges at two vertices connected by a direct edge (red) for any quantum double of group $G$.
Each row shows the elementary move (Gate), drawn as the gates (blue) acting on the graph of the previous row, and the resulting rewired graph (Graph); the grey dots track the two vertices to be exchanged.
Every controlled multiplication gate is controlled by the marked edge (dot) and targets the labeled edge.
At step 3 the inverse gate $\mathcal{S}$ reverses the orientation of the marked edge; at the same time we redraw the graph with the two vertices exchanged, such that the marked edge again points to the right.
After the protocol, the original lattice edges are restored with the same orientation, leaving only the two vertices swapped; hence, the two vertex stabilizers are swapped while all other stabilizers remain unchanged.}
\label{fig:charge_exchange_circuit}
\end{figure}
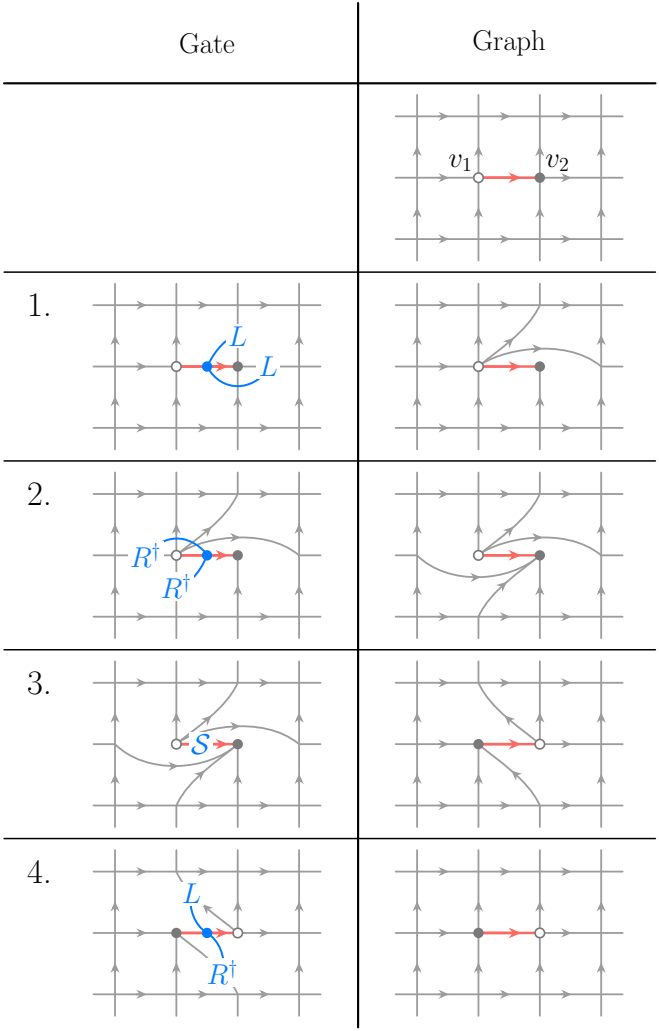

We construct a unitary circuit $U_\text{ex}$ that exchanges the charges at two vertices $v_{1}$ and $v_{2}$ connected by a direct edge $\ell$, for the quantum double of any finite group $G$.

\paragraph{The circuit action.}
The circuit and the accompanying rewiring of the graph are shown in Fig.~\ref{fig:charge_exchange_circuit}; the stabilizers are tracked step by step in Appendix~\ref{app:ce}. Applying the elementary moves of Sec.~\ref{subsec:elementary-move} (and carefully following the orientation rules given in Fig.~\ref{fig:gr-four-cases}),
the circuit rewires edges around the direct edge $\ell$ such that the overall action is swapping vertices $v_1$ and $v_2$. In additional to controlled group multiplication gate, a single-qudit inverse gate $\mathcal{S}_{\ell}$ of Eq.~\eqref{eq:inverse-gate-def} is needed to revert the direct edge orientation back to the original. In particular, from the step 2 to step 3 graphs in Fig.~\ref{fig:charge_exchange_circuit}, in addition to the action of $\mathcal{S}_{\ell}$ gate, the graph is continuously deformed such that the direct edge $l$ rotate clockwise by $180^\circ$ until $v_1$ and $v_2$ have swapped. This deformation is valid since it does not change the connectivity of the graph, and thereby the underlying stabilizers are invariant under this deformation. 

At the end of the circuit, the lattice is restored with the two endpoints exchanged, so the overall action on the set of stabilizers is just swapping the two vertex stabilizers $A^{g}_{v_{1}}$ and $A^{g}_{v_{2}}$,
\begin{equation}\label{eq:ce-swap}
A^{g}_{v_{1}}=
\begin{tikzpicture}[line width=0.9pt, baseline={(0pt,-\the\fontdimen22\textfont2)}]
  \draw[latgray] (-0.9,0) -- (1.8,0);
  \draw[latgray] (0,-0.9) -- (0,0.9);
  \draw[latgray] (0.9,-0.9) -- (0.9,0.9);
  \node[op] at (0,0.45) {$L^{g}$};
  \node[op] at (-0.45,0) {$R^{g}$};
  \node[op] at (0,-0.45) {$R^{g}$};
  \node[op] at (0.45,0) {$L^{g}$};
\end{tikzpicture}
\;\overset{\textstyle U_{\mathrm{ex}}}{\longleftrightarrow}\;
A^{g}_{v_{2}}=
\begin{tikzpicture}[line width=0.9pt, baseline={(0pt,-\the\fontdimen22\textfont2)}]
  \draw[latgray] (-0.9,0) -- (1.8,0);
  \draw[latgray] (0,-0.9) -- (0,0.9);
  \draw[latgray] (0.9,-0.9) -- (0.9,0.9);
  \node[op] at (0.45,0) {$R^{g}$};
  \node[op] at (0.9,0.45) {$L^{g}$};
  \node[op] at (1.35,0) {$L^{g}$};
  \node[op] at (0.9,-0.45) {$R^{g}$};
\end{tikzpicture}
\end{equation}
for every $g\in G$, and with them the charge types and internal states of Eqs.~\eqref{eq:P-charge-type} and \eqref{eq:P-internal} at the two vertices.
Every stabilizer other than the two vertex stabilizers is unchanged (Appendix~\ref{app:ce}).

\paragraph{Commuting gates and trivial mutual braiding.}
The six controlled multiplication gate commute with each other, since they are all controlled by the same edge $\ell$.
The only gate that does not commute with them is $\mathcal{S}_{\ell}$, and moving a controlled multiplication across it inverts the gate, $\mathcal{S}_{\ell}\,C_{\ell}L_{t}\,\mathcal{S}_{\ell}=C_{\ell}L^{\dagger}_{t}$ by Eq.~\eqref{eq:controlled-inverse}.
Fig.~\ref{fig:charge_exchange_circuit} therefore stands for a whole family of equivalent charge exchange circuits, where the equivalence is based on gate commutation relations. For example, the moves within steps 1--2 and within step 4 can be applied in any order, and any controlled multiplication gate can be shifted across $\mathcal{S}_{\ell}$ by inverting it.
Each ordering rewires the graph through a different sequence of intermediate graphs, some of which have a dangling edge.
In particular, one ordering first ungauges $v_{2}$ onto the direct edge and then regauges it at position $1$, so that the exchange acquires the interpretation of gauging as a sequential circuit~\cite{chen_sequential_2024} (Appendix~\ref{app:ce-ungauging}).
This freedom reflects that pure charges of $D(G)$ braid trivially with one another, so that the exchange need not distinguish a clockwise from a counterclockwise exchange\footnote{For quantum doubles of Hopf algebras with non-cocommutative vertex operators, the order of the controlled multiplications matters.}.

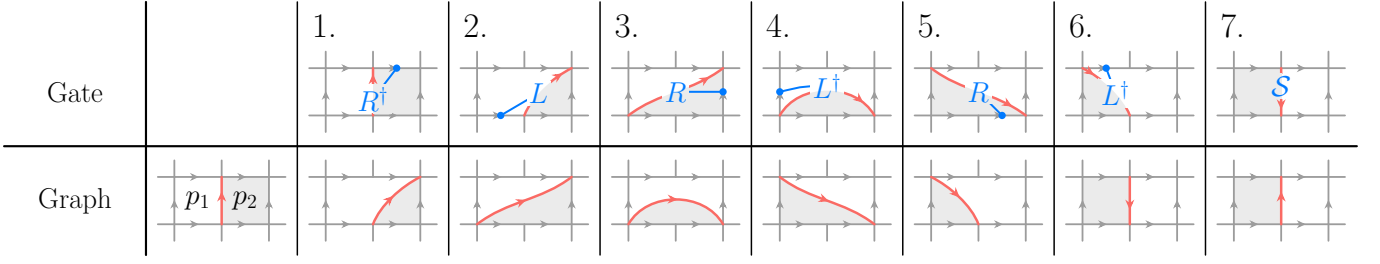
\begin{figure*}[t]
  \centering
  \adjustbox{max width=\textwidth}{%
  \begin{tikzpicture}[
      line cap=round, line join=round,
      lat/.style ={latgray, line width=1.1pt},
      gr/.style ={postaction={decorate, decoration={markings,
                   mark=at position #1 with {\arrow{Stealth[length=2.2mm,width=1.9mm]}}}}},
      gr/.default=0.52,
      red edge/.style={edgered, line width=1.5pt},
      rule/.style={black, line width=0.9pt},
      mv/.style  ={opblue, line width=1.2pt},
      ml/.style  ={opblue, font=\LARGE, inner sep=1.2pt},
      hd/.style  ={black, font=\LARGE},
      sn/.style  ={black, font=\huge},
    ]
  
    \def\colpitch{3.2}
    \def\hcol{-3.1}
    \def\tleft{-4.6}
    \def\tright{24.0}
    \def\ttop{1.7}
    \def\tmid{-1.35}
    \def\tbot{-3.65}
    \def\srow{-0.2}
    \def\grow{-2.5}
    \def\fbArrPos{0.52}
    \def\fbArrPosS{0.7}
  
    \node[hd] at (\hcol, \srow) {Gate};
    \node[hd] at (\hcol, \grow) {Graph};
    \draw[rule, line width=1.3pt] ({-\colpitch/2},\ttop) -- ({-\colpitch/2},\tbot);
    \draw[rule, line width=1.3pt] (\tleft,\tmid) -- (\tright,\tmid);
    \foreach \c in {1,...,7}{\draw[rule] ({(\c-0.5)*\colpitch},\ttop) -- ({(\c-0.5)*\colpitch},\tbot);}
  
    \newcommand{\fbArrR}[2]{\draw[lat,-{Stealth[length=1.9mm,width=1.7mm]}] (#1-0.15,#2)--(#1+0.02,#2);}
    \newcommand{\fbArrU}[2]{\draw[lat,-{Stealth[length=1.9mm,width=1.7mm]}] (#1,#2-0.15)--(#1,#2+0.02);}
    \newcommand{\fbCtrl}[2]{\fill[opblue] (#1,#2) circle (2.2pt);}
    \newcommand{\fbFrame}{%
      \draw[lat] (-1, 0.5)--(1, 0.5);
      \draw[lat] (-1,-0.5)--(1,-0.5);
      \draw[lat] (-1,-0.5)--(-1,0.5);
      \draw[lat] (1,-0.5)--(1,0.5);
      \foreach \xa in {-1,0,1}{\draw[lat] (\xa,0.5)--(\xa,0.85); \draw[lat] (\xa,-0.85)--(\xa,-0.5);}
      \draw[lat] (-1.35, 0.5)--(-1, 0.5); \draw[lat] (1, 0.5)--(1.35, 0.5);
      \draw[lat] (-1.35,-0.5)--(-1,-0.5); \draw[lat] (1,-0.5)--(1.35,-0.5);
      \foreach \xa in {-0.5,0.5}{\fbArrR{\xa}{0.5} \fbArrR{\xa}{-0.5}}
      \fbArrU{-1}{0} \fbArrU{1}{0}}
  
    \expandafter\newcommand\csname fbS1\endcsname{\fill[plaqgray] (0,-0.5) rectangle (1,0.5);}
    \expandafter\newcommand\csname fbS2\endcsname{\fill[plaqgray] (0,-0.5) .. controls (0.15,-0.1) and (0.6,0.3) .. (1,0.5) -- (1,-0.5) -- cycle;}
    \expandafter\newcommand\csname fbS3\endcsname{\fill[plaqgray] (-1,-0.5) .. controls (-0.35,0) and (0.35,0) .. (1,0.5) -- (1,-0.5) -- cycle;}
    \expandafter\newcommand\csname fbS4\endcsname{\fill[plaqgray] (-1,-0.5) .. controls (-0.6,0.2) and (0.6,0.2) .. (1,-0.5) -- cycle;}
    \expandafter\newcommand\csname fbS5\endcsname{\fill[plaqgray] (-1,0.5) .. controls (-0.35,0) and (0.35,0) .. (1,-0.5) -- (-1,-0.5) -- cycle;}
    \expandafter\newcommand\csname fbS6\endcsname{\fill[plaqgray] (-1,0.5) .. controls (-0.6,0.3) and (-0.15,-0.1) .. (0,-0.5) -- (-1,-0.5) -- cycle;}
    \expandafter\newcommand\csname fbS7\endcsname{\fill[plaqgray] (-1,-0.5) rectangle (0,0.5);}
    \expandafter\newcommand\csname fbS8\endcsname{\fill[plaqgray] (-1,-0.5) rectangle (0,0.5);}
  
    \expandafter\newcommand\csname fbF1\endcsname{%
      \csname fbS1\endcsname
      \fbFrame
      \draw[red edge,gr=\fbArrPosS] (0,-0.5)--(0,0.5);}
    \expandafter\newcommand\csname fbF2\endcsname{%
      \csname fbS2\endcsname
      \fbFrame
      \draw[red edge,gr=\fbArrPos] (0,-0.5) .. controls (0.15,-0.1) and (0.6,0.3) .. (1,0.5);}
    \expandafter\newcommand\csname fbF3\endcsname{%
      \csname fbS3\endcsname
      \fbFrame
      \draw[red edge,gr=\fbArrPos] (-1,-0.5) .. controls (-0.35,0) and (0.35,0) .. (1,0.5);}
    \expandafter\newcommand\csname fbF4\endcsname{%
      \csname fbS4\endcsname
      \fbFrame
      \draw[red edge,gr=\fbArrPos] (-1,-0.5) .. controls (-0.6,0.2) and (0.6,0.2) .. (1,-0.5);}
    \expandafter\newcommand\csname fbF5\endcsname{%
      \csname fbS5\endcsname
      \fbFrame
      \draw[red edge,gr=\fbArrPos] (-1,0.5) .. controls (-0.35,0) and (0.35,0) .. (1,-0.5);}
    \expandafter\newcommand\csname fbF6\endcsname{%
      \csname fbS6\endcsname
      \fbFrame
      \draw[red edge,gr=\fbArrPos] (-1,0.5) .. controls (-0.6,0.3) and (-0.15,-0.1) .. (0,-0.5);}
    \expandafter\newcommand\csname fbF7\endcsname{%
      \csname fbS7\endcsname
      \fbFrame
      \draw[red edge,gr=\fbArrPosS] (0,0.5)--(0,-0.5);}
    \expandafter\newcommand\csname fbF8\endcsname{%
      \csname fbS8\endcsname
      \fbFrame
      \draw[red edge,gr=\fbArrPosS] (0,-0.5)--(0,0.5);}
  
    \newcommand{\fbLabel}[3]{%
      \node[ml, text opacity=0] (fblbl) at (#2) {#3};
      \begin{scope}
        \clip (fblbl.south west) rectangle (fblbl.north east);
        \fill[white] (fblbl.south west) rectangle (fblbl.north east);
        \csname fbS#1\endcsname
      \end{scope}
      \node[ml] at (#2) {#3};}
  
    \expandafter\newcommand\csname fbH1\endcsname{%
      \draw[mv] (0,-0.15)--(0.5,0.5);
      \fbCtrl{0.5}{0.5}
      \fbLabel{1}{0,-0.15}{$R^{\dagger}$};}
    \expandafter\newcommand\csname fbH2\endcsname{%
      \draw[mv] (0.30,-0.03)--(-0.5,-0.5);
      \fbCtrl{-0.5}{-0.5}
      \fbLabel{2}{0.30,-0.03}{$L$};}
    \expandafter\newcommand\csname fbH3\endcsname{%
      \draw[mv] (0,0)--(1,0);
      \fbCtrl{1}{0}
      \fbLabel{3}{0,0}{$R$};}
    \expandafter\newcommand\csname fbH4\endcsname{%
      \draw[mv] (-1,0) .. controls (-0.75,0.05) and (-0.62,0.12) .. (-0.2,0.12);
      \fbCtrl{-1}{0}
      \fbLabel{4}{0,0.12}{$L^{\dagger}$};}
    \expandafter\newcommand\csname fbH5\endcsname{%
      \draw[mv] (0,0)--(0.5,-0.5);
      \fbCtrl{0.5}{-0.5}
      \fbLabel{5}{0,0}{$R$};}
    \expandafter\newcommand\csname fbH6\endcsname{%
      \draw[mv] (-0.30,-0.03)--(-0.5,0.5);
      \fbCtrl{-0.5}{0.5}
      \fbLabel{6}{-0.30,-0.03}{$L^{\dagger}$};}
    \expandafter\newcommand\csname fbH7\endcsname{%
      \fbLabel{7}{0,0.12}{$\mathcal{S}$};}
  
    \foreach \c in {1,...,8}{
      \begin{scope}[shift={({(\c-1)*\colpitch},\grow)}]
        \csname fbF\c\endcsname
      \end{scope}}
    \begin{scope}[shift={(0,\grow)}]
      \node[black, font=\LARGE] at (-0.5,0) {$p_{1}$};
      \node[black, font=\LARGE] at (0.5,0) {$p_{2}$};
    \end{scope}
    \foreach \c in {2,...,8}{
      \pgfmathtruncatemacro{\p}{\c-1}
      \node[sn, anchor=north west] at ({(\c-1.5)*\colpitch+0.2},{\ttop-0.15}) {\p.};
      \def\fbArrPos{0.85} \def\fbArrPosS{0.88}
      \ifnum\p=6 \def\fbArrPos{0.2}\fi
      \begin{scope}[shift={({(\c-1)*\colpitch},\srow)}]
        \csname fbF\p\endcsname
        \csname fbH\p\endcsname
      \end{scope}}
  
  \end{tikzpicture}}
  \caption{Flux exchange circuit as graph rewiring.
  The circuit clockwise exchanges the fluxes at two plaquettes separated by a dual edge (red, with the marked orientation) for any quantum double of group $G$.
  Each column shows the elementary move (Gate), drawn as the gate (blue) acting on the graph of the previous column, and the resulting rewired graph (Graph); the light grey shading tracks the deformed plaquette.
  Every gate targets the dual edge (label) and is controlled by one of its neighboring edges (dot).
  After step 6 the lattice returns to itself with the orientation of the dual edge reversed, which the inverse gate $\mathcal{S}$ at step 7 restores. After the protocol, the two plaquette stabilizers are swapped (as indicated by the swapped plaquettes) while all other stabilizers remain unchanged.}
  \label{fig:flux_exchange_circuit}
  \end{figure*}

\subsection{Flux exchange circuit}
\label{subsec:flux-exchange}

We now construct a unitary circuit $V_\mathrm{ex}$ that clockwise exchanges the fluxes at two plaquettes $p_{1}$ and $p_{2}$ sharing an edge $d$ (we refer to it as the dual edge), for the quantum double of any finite group $G$. To obtain the circuit for counterclockwise exchange, and the counterclockwise exchange, one just need to apply the inverse circuit.

\paragraph{The circuit action.}
The circuit and the accompanying rewiring of the graph are shown in Fig.~\ref{fig:flux_exchange_circuit}; the stabilizers are tracked step by step in Appendix~\ref{app:fb}.
Applying the elementary moves of Sec.~\ref{subsec:elementary-move} with the dual edge as the target (and following the orientation rules of Fig.~\ref{fig:gr-four-cases}), a $CR$ or $CR^\dagger$ slides the head of dual edge $d$ along the control edge, while a $CL$ or $CL^\dagger$ slides its tail. Combined, steps 1--6 rotate the dual edge clockwise by $180^\circ$ within the two plaquettes, carrying each plaquette into the other; the single-qudit inverse gate $\mathcal{S}_{d}$ of Eq.~\eqref{eq:inverse-gate-def} at step 7 then restores the orientation of $d$.
The lattice is thereby restored with the two plaquettes exchanged.
The corresponding action of the set of stabilizers is that the class projectors of Eq.~\eqref{eq:BC-MPO} at the two plaquettes are exchanged: in the MPO notation of Eq.~\eqref{eq:BC-MPO}, with the prefactor and the sums suppressed,
\begin{equation}\label{eq:fb-swap}
B^{C}_{p_{1}}=
\begin{tikzpicture}[scale=0.8, line cap=round, line join=round, baseline={(0pt,-\the\fontdimen22\textfont2)}, xl/.style={opblue, fill=white, inner sep=1pt}]
  \dmlat
  \node[latgray] at (-0.75,0) {$p_{1}$};
  \node[latgray] at (0.75,0) {$p_{2}$};
  \node[xl] at (-0.75,-0.75) {\zg{ij}};
  \node[xl] at (0,0)         {\zg{jk}};
  \node[xl] at (-0.75,0.75)  {\zgb{kl}};
  \node[xl] at (-1.5,0)      {\zgb{li}};
\end{tikzpicture}
\;\overset{\mbox{\large $V_{\mathrm{ex}}$}}{\longleftrightarrow}\;
B^{C}_{p_{2}}=
\begin{tikzpicture}[scale=0.8, line cap=round, line join=round, baseline={(0pt,-\the\fontdimen22\textfont2)}, xl/.style={opblue, fill=white, inner sep=1pt}]
  \dmlat
  \node[latgray] at (-0.75,0) {$p_{1}$};
  \node[latgray] at (0.75,0) {$p_{2}$};
  \node[xl] at (0.75,-0.75)  {\zg{ij}};
  \node[xl] at (1.5,0)       {\zg{jk}};
  \node[xl] at (0.75,0.75)   {\zgb{kl}};
  \node[xl] at (0,0)         {\zgb{li}};
\end{tikzpicture}
\end{equation}
for every conjugacy class $C$. All other stabilizers are unchanged (Appendix~\ref{app:fb}).

\paragraph{Commuting gates and nontrivial mutual braiding.}
The three sliding steps of the head of the dual edge (implemented with $CR$ or $CR^\dagger$) commute with the three sliding steps of the dual edge (implemented with $CL$ or $CL^\dagger$), so the two endpoints of the dual edge move independently of each other.
Fig.~\ref{fig:flux_exchange_circuit} therefore stands for a family of equivalent flux exchange circuits: the head slides and the tail slides can be interleaved in any way.
In particular, the ordering that first slides the head all the way around $p_{2}$ can be interpreted in terms of gauging (Appendix~\ref{app:fb-ungauging}).
On the other hand, the fixed ordering within the controlled right multiplications and within the controlled left multiplications reflects that fluxes of $D(G)$ braid nontrivially with one another, so that the circuit distinguishes a clockwise from a counterclockwise exchange, in contrast to the charge exchange circuit of Fig.~\ref{fig:charge_exchange_circuit}.

\paragraph{Action on the flux labels.}
We further clarify the action of the clockwise exchange at the level of the group element basis states, which allows us to track how the flux labels at both plaquettes based at vertex $u$ (Fig.~\ref{fig:fb-labels}), are modified by the circuit $V_\mathrm{ex}$.

Initially, the flux labels of $p_{1}$ and $p_{2}$ based at $u$, which we label as $w_1$ and $w_2$ respectively, are the counterclockwise holonomies,
\begin{equation}\label{eq:fb-initial-labels}
w_{1}=g_{d}\,\bar{g}_{2}\bar{g}_{3}g_{1},\qquad
w_{2}=g_{4}g_{5}\bar{g}_{6}\,\bar{g}_{d}.
\end{equation}

After evaluating the action of the circuit on the group basis state, the resulting flux labels are
\begin{align}
w_{1}'&=g_{d}'\,\bar{g}_{2}\bar{g}_{3}g_{1}=w_{2},\nonumber\\
w_{2}'&=g_{4}g_{5}\bar{g}_{6}\,\bar{g}_{d}'=w_{2}w_{1}\bar{w}_{2}.\label{eq:fb-labels}
\end{align}

Note that the total $w_{2}'w_{1}'=w_{2}w_{1}$ is invariant after the clockwise exchange, as expected. To summarize, the action of the exchange circuit $V_\mathrm{ex}$ on the flux labels is:
\begin{equation}\label{eq:exchange-flux-label}
\ket{w_{1},w_{2}}\mapsto\ket{w_{2},\,w_{2}w_{1}\bar{w}_{2}}.
\end{equation}

Applying the clockwise exchange circuit twice gives the full clockwise braiding
\begin{align}
\ket{w_{1},w_{2}}\mapsto\ket{w_{2},\,w_{2}w_{1}\bar{w}_{2}}&\mapsto\ket{t\,w_{1}\bar{t},\;t\,w_{2}\bar{t}},\label{eq:fb-monodromy}
\end{align}
which conjugates both labels by the total flux $t=w_{2}w_{1}$ of the pair based at $u$. This gives the R-symbol for fluxes in the quantum double.

\begin{figure}[ht]
\centering
\begin{tikzpicture}[x=0.12\columnwidth, y=0.12\columnwidth, line cap=round, line join=round,
    every node/.style={font=\small},
    dom/.style={latgray, line width=1.0pt},
    dual/.style={edgered, line width=1.6pt},
    flx/.style={opblue, line width=0.9pt,
                postaction={decorate, decoration={markings, mark=at position 0.5 with {\arrow{Stealth[length=2.2mm,width=1.8mm]}}}}},
    lab/.style={opblue, fill=white, inner sep=1pt},
    stp/.style={black, font=\footnotesize, align=center},
    pl/.style={anchor=south west, inner sep=1pt},
    el/.style={black, font=\footnotesize, inner sep=1pt}]
  \newcommand{\fbPerim}[1]{%
    \draw[dom] (#1-1.5,-0.75) rectangle (#1+1.5,0.75);
    \fill[black] (#1,-0.75) circle (1.5pt);
    \node[below=2pt] at (#1,-0.75) {$u$};}
  \newcommand{\fbDual}[1]{\draw[dual] (#1,-0.75)--(#1,0.75);}
  \newcommand{\fbLoops}[2]{%
    \draw[flx] plot[smooth, tension=0.7] coordinates {(0,-0.75) (-0.22,0.3) (-0.75,0.55) (-1.25,0.3) (-1.3,-0.2) (-0.7,-0.55) (0,-0.75)};
    \draw[flx] plot[smooth, tension=0.7] coordinates {(0,-0.75) (0.7,-0.55) (1.3,-0.2) (1.25,0.3) (0.75,0.55) (0.22,0.3) (0,-0.75)};
    \node[lab] at (-0.78,-0.1) {#1};
    \node[lab] at (0.78,-0.1) {#2};}
  \node[pl] at (-1.9,0.82) {(a)};
  \fbPerim{0}
  \fbDual{0}
  \fbLoops{$w_{1}$}{$w_{2}$}
  \node[latgray, font=\footnotesize] at (-0.78,0.2) {$p_{1}$};
  \node[latgray, font=\footnotesize] at (0.78,0.2) {$p_{2}$};
  \node[el, anchor=north] at (-0.75,-0.78) {$g_{1}$};
  \node[el, anchor=south] at (-0.75,0.78) {$g_{2}$};
  \node[el, anchor=east] at (-1.53,0) {$g_{3}$};
  \node[el, anchor=north] at (0.75,-0.78) {$g_{4}$};
  \node[el, anchor=west] at (1.53,0) {$g_{5}$};
  \node[el, anchor=south] at (0.75,0.78) {$g_{6}$};
  \node[el, anchor=west, text=edgered] at (0.05,0.6) {$g_{d}$};
  \draw[arr] (2.15,0) -- (3.05,0);
  \node[stp, above=1pt] at (2.6,0) {$V_{\mathrm{ex}}$};
  \begin{scope}[shift={(4.8,0)}]
    \node[pl] at (-1.5,0.82) {(b)};
    \fbPerim{0}
    \fbDual{0}
    \fbLoops{$w_{2}$}{$w_{2}w_{1}\bar{w}_{2}$}
  \end{scope}
\end{tikzpicture}
\caption{Flux labels through the flux exchange circuit.
(a) The fluxes $w_{1}$ and $w_{2}$ of $p_{1}$ and $p_{2}$ are the holonomies of the counterclockwise loops (blue) based at $u$, given in Eq.~\eqref{eq:fb-initial-labels}.
(b) After the flux clockwise exchange circuit $V_{\mathrm{ex}}$, $p_{1}$ carries $w_{2}$ and $p_{2}$ carries $w_{2}w_{1}\bar{w}_{2}$, as shown in Eq.~\eqref{eq:fb-labels}.}
\label{fig:fb-labels}
\end{figure}
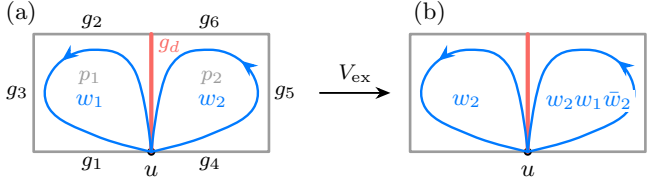

\section{No-go for unitary fusion}
\label{sec:no-go}

The exchange circuits of Sec.~\ref{sec:braiding} are purely unitary and need no ancilla, but they only ever exchange the contents of two sites: exchanging an anyon with the vacuum moves it precisely because the receiver is known to be empty. In this section we ask how far unitary circuits can be pushed. Can a unitary bring two anyons of unknown type onto a single site, and thereby fuse them? We show that it is impossible in general to unitarily fuse unknown Abelian or non-Abelian anyons, and identify the resource needed to evade this obstruction.

First, we clarify the conceptual difference between ribbon operators of Sec.~\ref{subsec:ribbon-rewiring}, and the operators required to move anyons \textit{coherently}. 
A string operator moves an anyon only with some knowledge about anyons at its endpoints. 
For example, in the toric code, a Pauli $Z$ string moves an $e$ anyon from site $A$ to an empty site $B$, but the same string can (i) create a pair of $e$ anyons from the vacuum, or (ii) move an $e$ anyon from $B$ to $A$ when only $B$ is occupied.
Thus a fixed string or ribbon can create, move, or fuse excitations depending on its input, but it does not by itself guarantee that the source is emptied for every unknown input; a different protocol is needed.

In Sec.~\ref{subsec:no-go}, we define what it means for a circuit to move and thereby fuse anyons without leaving remnant anyon.
The fusion protocol must work without knowing which input is present: it must always leave no anyon left at site $A$ and transfer the joint anyon content to $B$. We then show that no unitary circuit can move or fuse anyons without knowledge of the input, even when assisted by a ``catalytic'' ancilla which returns to its initial state. We then explain in Sec.~\ref{subsec:non-abelian-obstruction} why knowing the source anyon content does not generally remove the obstruction for non-Abelions. Finally, we devote Sec.~\ref{subsec:no-go-evading} to various ways around the no-go theorem, which underlie both the braiding circuits of Sec.~\ref{sec:braiding} and the fusion protocols of Sec.~\ref{sec:fusion}. 

\subsection{No input-agnostic fusion via unitary operators}\label{subsec:no-go}

Consider a quantum double model where we focus on two sites $A$ and $B$, decomposing the Hilbert space of this subregion as
\begin{equation}
\mathcal{H}=\bigoplus_{a,b}\mathcal{H}_{a,b},
\label{eq:sectors}
\end{equation}
where $\mathcal{H}_{a,b}$ is the space of internal states of an anyon of type $a$ at $A$ and of type $b$ at $B$, all other sites of the region being neutral (Sec.~\ref{subsec:qd-basics}); in $D(G)$ an anyon $(C,\Gamma)$ carries a local internal space of dimension $|C|\,d_{\Gamma}$, equal to its quantum dimension. Suppose we wish to fuse these two anyons, for instance onto site $B$. We require the fusion protocol to send every state in $\mathcal{H}_{a,b}$ to some state in $\mathcal{H}_{1,c}$, or more precisely into a direct sum over the fusion products $c$ appearing in the fusion $a \otimes b$. Crucially, we require that \textit{no anyons remain at $A$}, hence the trivial `1' label in $\mathcal{H}_{1,c}$. This is demonstrated in Fig.~\ref{fig:no-go}(a).

We call this general task---where we have not placed any restrictions on the initial $a$ and $b$ anyons appearing in Eq.~\eqref{eq:sectors}---input-agnostic fusion. This has a unitarity obstruction:

\begin{theorem}[No unitary fusion without input knowledge]\label{thm:no-go}
In a quantum double model $D(G)$ with non-trivial anyons, no unitary operator supported in a region containing the two sites can implement input-agnostic fusion.\\

This holds even when allowing for a \textit{catalytic} ancilla, i.e., an ancilla that starts in $\ket{0}$ and returns to $\ket{0}$ for every input such that it retains no record at the end. 
\end{theorem}

This follows from a simple counting argument. For instance, consider the two possible initial states in Fig.~\ref{fig:no-go}(b), i.e., $a=c$ and $b=1$ versus $a=1$ and $b=c$. By the above requirement, both would have to be mapped into the same Hilbert space. There are simply not enough linearly-independent states to accommodate this, which is in tension with the orthogonality of states if we use a unitary (or more generally isometric) map.

We can make this argument precise as follows. First observe that a unitary map from $\oplus_{a,b} \mathcal{H}_{a,b} \to \oplus_c \mathcal{H}_{1,c}$ would imply
\begin{equation}
\sum_{a,b} |\mathcal{H}_{a,b}| \leq \sum_c |\mathcal{H}_{1,c}| .\label{eq:inequality}
\end{equation}
In fact, the same inequality would hold if we use catalytic ancillas: in this case we have an isometry on our system qudits, which preserves the inner product, and in particular, linear independence. The rest of the lattice couples to the pair only through its total anyon type $c$, measured by a closed ribbon around the two sites \cite{kitaevFaulttolerantQuantumComputation2003}; a protocol supported near $A$ and $B$ conserves $c$ and leaves the rest untouched, so for each $c$ the rest contributes the same multiplicity to both sides of Eq.~\eqref{eq:inequality} and only internal states need to be counted. In $D(G)$ the internal space of the pair decomposes under its total type as $\mathcal{H}_{a,b}=\bigoplus_c V^{c}_{a,b}\otimes\mathcal{H}_{1,c}$, with $V^{c}_{a,b}$ the $N^{c}_{a,b}$-dimensional fusion space, so $|\mathcal{H}_{a,b}| = d_a d_b = \sum_c N^c_{a,b} d_c$ and $|\mathcal{H}_{1,c}| = d_c$. Hence, Eq.~\eqref{eq:inequality} implies
\begin{equation}
\sum_c \bigg( \sum_{a,b} N^c_{a,b} \bigg) d_c \leq \sum_c d_c.
\end{equation}
This is in direct contradiction with how for any non-trivial anyon theory, one has\footnote{For instance, if $c \neq 1$, we observe the sum is lower-bounded by $N_{1,c}^c + N_{c,1}^c = 2$, whereas if $c=1$ let $a$ be a non-trivial anyon such that we can bound the sum by $N^1_{1,1} + N^1_{a,\bar a} = 2$.} $\sum_{a,b} N^c_{a,b} > 1$ for every $c$.

\begin{figure}[h]
\centering
\adjustbox{max width=\columnwidth}{%
\begin{tikzpicture}[line cap=round, line join=round,
    every node/.style={font=\small},
    lat/.style={latgray, line width=1.0pt},
    vtx/.style={circle, fill=latgray, inner sep=0pt, minimum size=2.6pt},
    any/.style={circle, inner sep=0pt, minimum size=8pt, draw=black!45, line width=0.5pt},
    ghost/.style={circle, inner sep=0pt, minimum size=8pt, draw=latgray, dash pattern=on 1.4pt off 1.1pt, fill=white, line width=0.5pt},
    site/.style={text=black!60, font=\footnotesize, inner sep=1pt},
    arr/.style={-{Stealth[length=2.4mm,width=2.0mm]}, black!70, line width=1.0pt},
    set/.style={rounded corners=9pt, line width=0.9pt},
    pl/.style={anchor=north west, inner sep=1pt, font=\small}]
  \newcommand{\ngPatch}[2]{%
    \begin{scope}[shift={(#1,#2)}]
      \foreach \y in {-0.5,0,0.5}{\draw[lat] (-1.0,\y)--(1.0,\y);}
      \foreach \x in {-0.75,-0.25,0.25,0.75}{\draw[lat] (\x,-0.75)--(\x,0.75);}
      \foreach \x in {-0.75,-0.25,0.25,0.75}{\foreach \y in {-0.5,0,0.5}{\node[vtx] at (\x,\y) {};}}
      \node[site, anchor=north] at (-0.75,-0.78) {$A$};
      \node[site, anchor=north] at (0.75,-0.78) {$B$};
    \end{scope}}
  \newcommand{\ngAnyon}[5]{\node[any, fill=#3, label={[inner sep=0.5pt, label distance=1.5pt]#5:#4}] at (#1,#2) {};}
  \newcommand{\ngVac}[3]{\node[ghost, label={[inner sep=0.5pt, label distance=1.5pt]#3:$1$}] at (#1,#2) {};}
  \node[pl] at (-4.3,1.5) {(a)};
  \ngPatch{-2.0}{0}
  \ngAnyon{-2.75}{0}{opblue!70}{$a$}{above left}
  \ngAnyon{-1.25}{0}{ribgreen!80}{$b$}{above right}
  \draw[arr] (-0.7,0)--(0.7,0);
  \node[anchor=south] at (0,0.08) {fusion};
  \ngPatch{2.0}{0}
  \ngVac{1.25}{0}{above left}
  \node[any, fill=orange!40, label={[inner sep=1pt, fill=white, label distance=1.5pt]above:$a\otimes b$}] at (2.75,0) {};
  \node[anchor=north, inner sep=2pt] at (-2.0,-1.05) {$\mathcal{H}_{a,b}$};
  \node[anchor=north, inner sep=2pt] at (2.0,-1.05) {$\mathcal{H}_{1,a\otimes b}$};
  \begin{scope}[shift={(0,-5.2)}]
  \node[pl] at (-4.4,3.05) {(b)};
  \draw[set, draw=opblue!80!black, fill=opblue!8] (-4.1,0.15) rectangle (-0.7,2.45);
  \ngPatch{-2.4}{1.4}
  \ngAnyon{-3.15}{1.4}{violet!70}{$c$}{above left}
  \ngVac{-1.65}{1.4}{above right}
  \node[anchor=south, inner sep=2pt] at (-2.4,2.45) {$\mathcal{H}_{c,1}$};
  \draw[set, draw=opblue!80!black, fill=opblue!8] (-4.1,-2.45) rectangle (-0.7,-0.15);
  \ngPatch{-2.4}{-1.2}
  \ngVac{-3.15}{-1.2}{above left}
  \ngAnyon{-1.65}{-1.2}{violet!70}{$c$}{above right}
  \node[anchor=north, inner sep=2pt] at (-2.4,-2.45) {$\mathcal{H}_{1,c}$};
  \draw[set, draw=edgered!80!black, fill=edgered!10] (1.1,-1.05) rectangle (4.5,1.25);
  \ngPatch{2.8}{0.2}
  \ngVac{2.05}{0.2}{above left}
  \ngAnyon{3.55}{0.2}{violet!70}{$c$}{above right}
  \node[anchor=south, inner sep=2pt] at (2.8,1.25) {$\mathcal{H}_{1,c}$};
  \draw[arr] (-0.6,1.2) to[out=-25, in=155] (1.0,0.45);
  \draw[arr] (-0.6,-1.0) to[out=25, in=-155] (1.0,-0.25);
  \node[fill=white, inner sep=1.5pt] at (0.2,0.1) {fusion};
  \end{scope}
\end{tikzpicture}}
\caption{Idea of the no-go theorem.
(a) Fusion empties site $A$ and leaves the fusion product at site $B$.
(b) Both states with anyon $c$ at $A$ and trivial anyon at $B$ and states with anyon $c$ at $B$ and trivial anyon at $A$ are orthogonal inputs that fusion sends to the same set of states.
The two input sets together are twice as large as the output set, so no unitary can implement the map.}
\label{fig:no-go}
\end{figure}

\subsection{Knowing the source charge: Abelian vs non-Abelian}\label{subsec:non-abelian-obstruction}

Theorem~\ref{thm:no-go} applies to both Abelian and non-Abelian anyons. 
The distinction appears when the ``source'' site $A$ is known to contain an anyon $c$, while the state of the ``receiver'' site $B$ is unknown.

For Abelian $c$, an appropriately oriented unitary string deposits $\bar c$ at $A$ and $c$ at $B$, emptying the known source and producing the unique receiver type $c\otimes d$ for any input type $d$ at $B$. We can also directly relate this to the proof method of Theorem~\ref{thm:no-go}: if we know $a$, then repeating the proof leads to the restricted sum $\sum_b N^c_{a,b}$ which indeed equals unity in an Abelian theory---in this case we thus avoid a contradiction with the counting argument.

In contrast, in the non-Abelian case, this still leads to an obstruction. Indeed, even if we know the non-Abelian source anyon $a$, we can have $\sum_b N^c_{a,b} > 1$. This follows from the fact that $a \otimes \bar a$ contains some $c \neq 1$ and Frobenius reciprocity gives $N_{ac}^c = N^c_{a \bar a}$. Hence we cannot find a purely unitary protocol for fusion. In fact, the obstruction persists even if we know both $a$ and $b$ in a theory with non-trivial fusion multiplicity.

\subsection{Resources for fusion}
\label{subsec:no-go-evading}

To evade the assumptions of the input-agnostic no-go theorem, one can let an ancilla keep a record such that the non-injective mapping can be extended to a unitary in an extended Hilbert space, or restrict the allowed input states such that the mapping is no longer input-agnostic, or implement the non-injective mapping via reset. We discuss each of the three ways below.

\paragraph{Let an ancilla retain a record.}
For two specified input states, the map
\begin{align}
\ket{\psi_{c,1}}\otimes\ket{0}&\mapsto\ket{\psi_{1,c}}\otimes\ket{1},\nonumber\\
\ket{\psi_{1,c}}\otimes\ket{0}&\mapsto\ket{\psi_{1,c}}\otimes\ket{0}
\label{eq:which-path}
\end{align}
preserves orthogonality and can be extended to a unitary.
The ancilla distinguishes arrival from prior occupation, so it is not catalytic.
Discarding it dephases a superposition of these inputs; reading it produces a measurement record.
Previous approaches use ancillas, measurement, postselection, or feedforward in different combinations~\cite{aguadoCreationManipulationDetection2008,brennenSimulationsQuantumDouble2009,luo_simulation_2011,liuMethodsSimulatingStringnet2022,bravyiAdaptiveConstantdepthCircuits2022,jovanovicProposalDemonstrateNonabelian2024,chen_universal_2025,lo2026universal}. In some restricted cases where one has more knowledge about the input state, fusion can be implemented with ancilla; we provide the graph rewiring perspective on an example protocol in Appendix~\ref{app:ancilla-move}. In general, these incurs additional resources costs in terms of ancilla or measurement which, as we will show in Sec.~\ref{sec:fusion}, is not necessary.

\paragraph{Use knowledge of the receiver (unitary exchange).}
If the anyon content at receiver $B$ is known to be trivial, a unitary that exchanges sites $A$ and $B$ accomplishes the goal of moving an anyon at the source to the receiver site. In general, we can use such a unitary protocol to braid anyons non-trivially when we do not need to fuse them (move them to the same site). The charge and flux exchange circuits using graph rewiring (Sec.~\ref{sec:braiding}) are exactly of this kind: they exchange the contents of two sites without an ancilla, and move an anyon only into a site known to be empty.

\paragraph{Erase the record (reset).}
Alternatively, the lattice can undergo a non-unitary operation that discards information.
In Sec.~\ref{sec:fusion}, unitary gates bring the information to one edge, which is reset to a fixed pure state before the gates are reversed.
The reset is a non-unitary channel (it may map all states to $\ket{0}$ for example) so the lattice evolution need not preserve inner products.
For both charge and flux fusion, on the input sectors specified in Sec.~\ref{sec:fusion}, one edge reset suffices to empty the source and transfer the fusion-channel probabilities to the receiver, without input knowledge, an additional circuit ancilla, or feedforward beyond the reset.
The physical reset still exchanges information with an environment.
Theorem~\ref{thm:no-go} excludes a purely unitary implementation of this task with a catalytic ancilla; the protocols supply the required non-unitary step.

\section{Input-agnostic fusion with reset}
\label{sec:fusion}

We now implement the input-agnostic fusion: with no knowledge of the input, bringing two neighboring anyons of unknown type together into a single anyon, such that there is no remnant anyon left behind at one site.
By Theorem~\ref{thm:no-go}, this cannot be done unitarily, and the non-unitary resource we use is a single edge reset.
Each protocol below has the same three-step structure: 1) graph rewiring ungauges the source anyon onto the lattice edge joining the two sites, 2) that edge is reset, and 3) the graph is rewired back to its original form.
The source site ends in the vacuum, so by neutrality the total anyon type of the pair ends at the receiving site, with the probability of each fusion channel equal to that of the input.
We present the charge protocol in Sec.~\ref{subsec:charge-fusion} and the flux protocol in Sec.~\ref{subsec:flux-fusion}, and in each case track how the stabilizers of the input are mapped to those of the output. The reset does not in general preserve coherence between fusion channels (Sec.~\ref{subsec:cc-example}), but it preserves their probabilities exactly. The lost coherence is between anyon sectors of two adjacent sites, a local observable that is susceptible to noise in any case; information stored nonlocally in the fusion space of well-separated anyons is unaffected. Therefore, our fusion protocols suffices for error correction purposes where anyon fusion is used to return to the ground state.

We stress the conceptual distinction between the fusion protocols presented below and read-out of the fusion outcome: the fusion protocols do not leave a classical record selecting a fusion channel.
Reading out the fusion outcome is a distinct task that requires a projective measurement of the anyon type at the receiver, i.e.\ of the charge projectors $P^{\Gamma}_{v}$ of Eq.~\eqref{eq:charge-projector} or the flux projectors $B^{C}_{p}$ of Eq.~\eqref{eq:flux-projector}. The positive operator-valued measures (POVM) formed by either the charge or flux projectors can be implemented as projective measurement of an entangled ancilla, as routinely performed for syndrome extraction in error correction schemes.

\subsection{Charge fusion protocol}
\label{subsec:charge-fusion}

We first fuse charge anyons at two nearest-neighbour vertices $v_{1}$ and $v_{2}$ connected by a direct edge $\ell$, as sketched in Fig.~\ref{fig:cf-protocol}(a).
The protocol rewires the graph to expose $v_{1}$ at the dangling end of the direct edge $\ell$, resets $\ell$, then rewires back, so that $v_{1}$ is left in the vacuum and the two charges fuse at $v_{2}$ by charge neutrality.
We track how the vertex, closed loop ribbon and plaquette stabilizers are mapped in Sec.~\ref{subsubsec:charge-op-mapping}.

\begin{figure}[ht]
\centering
\adjustbox{max width=\columnwidth}{%
\begin{tikzpicture}[x=0.097\columnwidth, y=0.097\columnwidth, line cap=round, line join=round,
    every node/.style={font=\small},
    plat/.style={latgray, line width=1.0pt},
    red edge/.style={edgered, line width=1.75pt},
    any/.style={circle, inner sep=1pt, minimum size=16pt},
    vtx/.style={circle, fill=latgray, inner sep=0pt, minimum size=3.1pt},
    arr/.style={-{Stealth[length=3.4mm,width=3.0mm]}, black!70, line width=1.6pt},
    hop/.style={-{Stealth[length=1.9mm,width=1.6mm]}, opblue!70!black, line width=0.8pt},
    vl/.style={fill=white, inner sep=0.8pt},
    pl/.style={anchor=north west, inner sep=1pt},
    slat/.style={latgray, line width=1.0pt},
    sred/.style={edgered, line width=1.75pt},
    opb/.style={text=opblue, fill=white, inner sep=1.2pt},
    ctrlb/.style={opblue, line width=0.95pt},
    tag/.style={anchor=south, inner sep=1pt},
    wl/.style={anchor=base east, yshift=-2.25pt, inner sep=1pt, text=opblue},
    sarr/.style={-{Stealth[length=2.2mm,width=1.9mm]}, black!70, line width=1.0pt}]
  \newcommand{\cfPatch}[2]{%
    \begin{scope}[shift={(#1,#2)}]
      \foreach \y in {-1,0,1}{\draw[plat] (-2,\y)--(2,\y);}
      \foreach \x in {-1.5,-0.5,0.5,1.5}{\draw[plat] (\x,-1.5)--(\x,1.5);}
      \foreach \x in {-1.5,-0.5,0.5,1.5}{\foreach \y in {-1,0,1}{\node[vtx] at (\x,\y) {};}}
      \draw[red edge] (-0.5,0)--(0.5,0);
    \end{scope}}
  \newcommand{\cfAnyon}[5]{\node[any, fill=#3, draw=black!45, line width=0.6pt] (#5) at (#1,#2) {#4};}
  \newcommand{\cfGhost}[3]{\node[any, draw=latgray, dashed, fill=white, line width=0.6pt] at (#1,#2) {#3};}
  \newcommand{\cfFixed}{\draw[slat] (2,0)--(3,0); \draw[slat] (2,-1)--(2,1); \draw[sred] (1,0)--(2,0);}
  \newcommand{\cfLegUp}[1]{\ifnum#1=0 \draw[slat] (1,0)--(1,1); \else \draw[slat] (1,1) .. controls (1.0,0.55) and (1.5,0.5) .. (2,0); \fi}
  \newcommand{\cfLegLeft}[1]{\ifnum#1=0 \draw[slat] (0,0)--(1,0); \else \draw[slat] (0,0) .. controls (0.5,0.35) and (1.5,0.35) .. (2,0); \fi}
  \newcommand{\cfLegDown}[1]{\ifnum#1=0 \draw[slat] (1,0)--(1,-1); \else \draw[slat] (1,-1) .. controls (1.0,-0.55) and (1.5,-0.5) .. (2,0); \fi}
  \newcommand{\cfLat}[3]{\cfFixed \cfLegUp{#1} \cfLegLeft{#2} \cfLegDown{#3}}
  \newcommand{\cfVlab}{\node[anchor=north west, inner sep=1pt] at (1.05,-0.05) {$v_{1}$}; \node[anchor=north west, inner sep=1pt] at (2.05,-0.05) {$v_{2}$};}
  \newcommand{\cfDot}{\node[circle, fill=opblue, inner sep=0pt, minimum size=5.2pt] at (1.5,0) {};}
  \newcommand{\cfGateUp}[1]{\draw[ctrlb] (1.5,0) .. controls (1.38,0.3) and (1.2,0.46) .. (1.03,0.5); \cfDot \node[opb] at (1,0.5) {$#1$};}
  \newcommand{\cfGateLeft}[1]{\draw[ctrlb] (1.5,0) .. controls (1.3,0.28) and (0.85,0.28) .. (0.58,0.03); \cfDot \node[opb] at (0.5,0) {$#1$};}
  \newcommand{\cfGateDown}[1]{\draw[ctrlb] (1.5,0) .. controls (1.52,-0.4) and (1.35,-0.56) .. (1.03,-0.5); \cfDot \node[opb] at (1,-0.5) {$#1$};}
  \newcommand{\cfGateUpM}[1]{\draw[ctrlb] (1.5,0)--(1.31,0.52); \cfDot \node[opb] at (1.31,0.52) {$#1$};}
  \newcommand{\cfGateLeftM}[1]{\draw[ctrlb] (1.5,0)--(1.0,0.26); \cfDot \node[opb] at (1.0,0.26) {$#1$};}
  \newcommand{\cfGateDownM}[1]{\draw[ctrlb] (1.5,0)--(1.31,-0.52); \cfDot \node[opb] at (1.31,-0.52) {$#1$};}
  \newcommand{\cfPanel}[3]{\begin{scope}[shift={(#1,#2)}, scale=0.85] #3 \end{scope}}
  \node[pl] at (-5.25,1.85) {(a)};
  \cfPatch{-2.9}{0}
  \cfAnyon{-3.4}{0}{opblue!35}{$\Gamma_{1}$}{cfGa}
  \cfAnyon{-2.4}{0}{ribgreen!45}{$\Gamma_{2}$}{cfGb}
  \draw[hop] (cfGa.north) to[bend left=50] (cfGb.north);
  \node[vl, anchor=north east] at (-3.68,-0.3) {$v_{1}$};
  \node[vl, anchor=north west] at (-2.12,-0.3) {$v_{2}$};
  \node[vl, text=edgered, anchor=north] at (-2.9,-0.1) {$\ell$};
  \draw[arr] (-0.5,0)--(0.5,0);
  \node[anchor=south] at (0,0.16) {$\mathcal{E}$};
  \cfPatch{2.9}{0}
  \cfGhost{2.4}{0}{$1$}
  \cfAnyon{3.4}{0}{orange!40}{$\Lambda$}{cfLam}
  \node[pl] at (-5.25,-2.05) {(b)};
  \node[tag] at (0.375,-2.5) {Input};
  \cfPanel{-0.9}{-3.5}{\cfLat{0}{0}{0} \cfVlab}
  \node[tag] at (0.375,-4.9) {1.\ Rewire};
  \node[wl] at (-4.15,-5.9) {$W^{\dagger}=$};
  \cfPanel{-4.0}{-5.9}{\cfLat{0}{0}{0} \cfGateUp{L^{\dagger}}}
  \draw[sarr] (-1.45,-5.9)--(-1.05,-5.9);
  \cfPanel{-0.9}{-5.9}{\cfLat{1}{0}{0} \cfGateLeft{R^{\dagger}}}
  \draw[sarr] (1.65,-5.9)--(2.05,-5.9);
  \cfPanel{2.2}{-5.9}{\cfLat{1}{1}{0} \cfGateDown{R^{\dagger}}}
  \node[tag] at (0.375,-7.3) {2.\ reset};
  \cfPanel{-0.9}{-8.3}{\cfLat{1}{1}{1}}
  \node[draw=opblue, fill=white, rounded corners=1.7pt, line width=0.8pt, inner sep=2.2pt] (cfReset) at (-0.6,-8.94) {$\ket{+}$};
  \draw[opblue, line width=0.8pt, {Bar[width=4pt]}-{Stealth[length=2mm,width=1.6mm]}, shorten <=1pt] (cfReset.north east) -- (0.16,-8.33);
  \node[tag] at (0.375,-9.7) {3.\ Rewire back};
  \node[wl] at (-4.15,-10.7) {$W=$};
  \cfPanel{-4.0}{-10.7}{\cfLat{1}{1}{1} \cfGateDownM{R}}
  \draw[sarr] (-1.45,-10.7)--(-1.05,-10.7);
  \cfPanel{-0.9}{-10.7}{\cfLat{1}{1}{0} \cfGateLeftM{R}}
  \draw[sarr] (1.65,-10.7)--(2.05,-10.7);
  \cfPanel{2.2}{-10.7}{\cfLat{1}{0}{0} \cfGateUpM{L}}
  \node[tag] at (0.375,-12.1) {Output};
  \cfPanel{-0.9}{-13.1}{\cfLat{0}{0}{0}}
  \draw[arr] (-1.05,-3.5) to[out=180, in=90, looseness=1.0] (-2.725,-4.9);
  \draw[arr] (3.475,-6.9) to[out=-90, in=0, looseness=1.0] (1.8,-8.3);
  \draw[arr] (-1.05,-8.3) to[out=180, in=90, looseness=1.0] (-2.725,-9.7);
  \draw[arr] (3.475,-11.7) to[out=-90, in=0, looseness=1.0] (1.8,-13.1);
\end{tikzpicture}}
\caption{Charge fusion protocol.
(a) Two charges $\Gamma_{1}$ and $\Gamma_{2}$ at the vertices $v_{1}$ and $v_{2}$ connected by the direct edge $\ell$ (red) are fused by the channel $\mathcal{E}$ of Eq.~\eqref{eq:cf-channel}, which moves the charge of $v_{1}$ across the direct edge and leaves the vacuum (dashed) at $v_{1}$ and the fusion product $\Lambda\in\Gamma_{1}\otimes\Gamma_{2}$ at $v_{2}$.
(b) The protocol as a flowchart from the input lattice of (a) to the output lattice; each panel shows a gate or reset (blue) applied to the current graph. In step 1, each gate rewires its target edge from $v_{1}$ to $v_{2}$, so that after the third gate $\ell$ dangles from $v_{2}$ with $v_{1}$ at its free end. In step 2, the direct edge is reset to $\ket{+}$. In step 3, graph rewiring restores the graph.}
\label{fig:cf-protocol}
\end{figure}
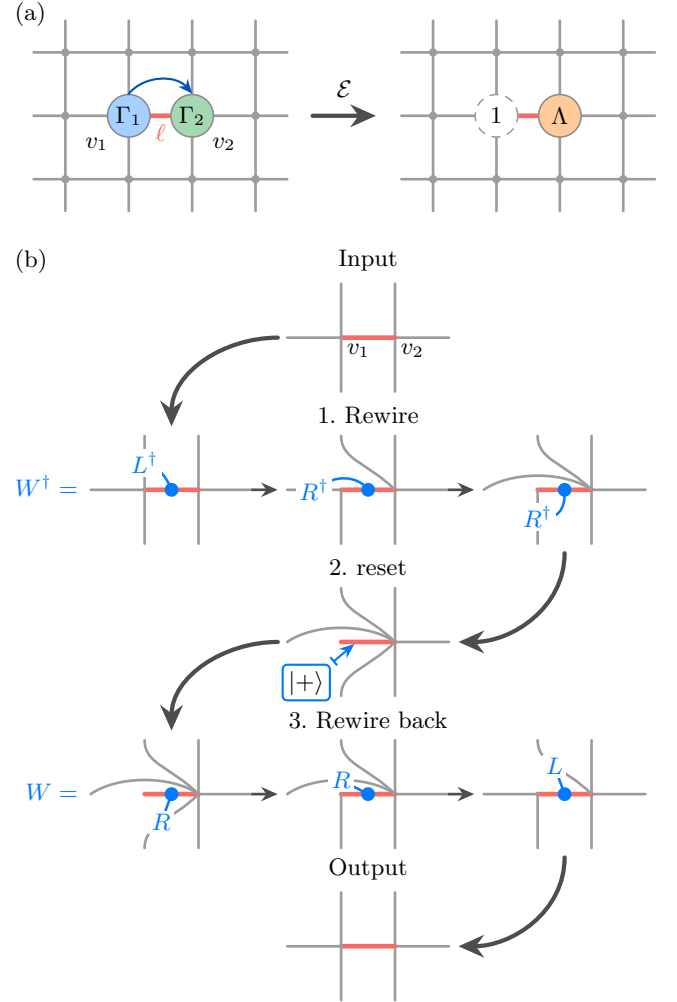

\subsubsection{The protocol}

Let $v_{1}$ and $v_{2}$ carry irreps $\Gamma_{1},\Gamma_{2}$ of $G$. The protocol consists of the three steps shown in Fig.~\ref{fig:cf-protocol}(b):

\begin{enumerate}
    \item rewire the graph until $v_{1}$ dangles from the direct edge
    \item reset the ``dangling'' direct edge to the $\ket{+}$ state
    \item rewire the graph back to the original
\end{enumerate}

We call $W^{\dagger}$ the combined circuit of the three gates of step 1, controlled multiplications of Eq.~\eqref{eq:CLCR-def} controlled by the direct edge $\ell$:
\begin{equation}\label{eq:cf-W-dagger}
W^{\dagger}=\sum_{g\in G}\;
\vcenter{\hbox{%
\begin{tikzpicture}[line width=0.9pt,baseline]
  \tvlat
  \node[op] at (0,0.45) {$L^{\bar{g}}$};
  \node[op] at (-0.5,0) {$R^{\bar{g}}$};
  \node[op] at (0,-0.45) {$R^{\bar{g}}$};
  \node[op] at (0.75,0) {$\ketbra{g}{g}$};
\end{tikzpicture}
}}
\end{equation}

Its action on the vertex stabilizer $A^g_{v_1}$ is given by:
\begin{equation}\label{eq:cf-W-defining}
W^\dagger\, A^{g}_{v_{1}} \,W=L^{g}_{\ell},
\end{equation} 
such that the vertex stabilizer of $v_{1}$ is supported only on the ``dangling'' direct edge.
We note that $W$ can be viewed a short sequential circuit~\cite{chen_sequential_2024} gauging the $G$ symmetry at $v_{1}$: $W^{\dagger}$ ungauges the Gauss law $A^{g}_{v_{1}}$ into the on-site $G$ action $L^{g}_{\ell}$ on the single edge $\ell$, and $W$ regauges it (Appendix~\ref{app:ce-ungauging}).

Then, in step 2, we reset the state on the direct edge $\ell$ to 
\begin{equation}
    \ket{+}=\frac{1}{\sqrt{|G|}}\sum_{g\in G}\ket{g},
\end{equation}
such that 
\begin{equation}\label{eq:Lp=1}
    L^{g}_{\ell}=+1\quad\forall g\in G
\end{equation}
is satisfied on the direct edge. 

Experimentally, a reset does not require a measurement: it can be implemented by dissipatively pumping the qudit into $\ket{+}$, or by swapping the direct edge with a fresh ancilla prepared in $\ket{+}$.
Nonetheless, for ease of deriving the stabilizer mapping in Sec.~\ref{subsubsec:charge-op-mapping}, we represent the reset as measuring the direct edge in the group element basis with outcome $\ket{x}$, applying $R^{x}$ (or $L^{\bar{x}}$) to bring the edge to the identity state $\ket{e}$ by Eq.~\eqref{eq:LR-def}, and then applying any unitary $U_{+}$ that satisfies $U_{+}\ket{e}=\ket{+}$ (for example, the inverse group Fourier transform $F^{\dagger}$ of Eq.~\eqref{eq:inv-gp-FT}, which maps the trivial-irrep basis state to $\ket{+}$, once $\ket{e}$ is identified with $\ket{\mathbf{1};1,1}$ in the computational encoding), which gives the Kraus operator
\begin{equation}\label{eq:cf-reset-kraus}
K^{(2)}_{x}=U_{+}R^{x}_{\ell}\ketbra{x}{x}_{\ell}=\ketbra{+}{x}_{\ell}.
\end{equation}
The outcome $x$ is used only for the flip inside the reset and is then discarded: the reset realizes the replacement channel $\rho\mapsto\sum_{x}\ketbra{+}{x}\rho\ketbra{x}{+}=\ketbra{+}{+}_{\ell}\otimes\Tr_{\ell}\rho$, independent of $x$.
We stress that the measure-and-flip decomposition is only a convenient Kraus representation of this unconditional reset: no outcome is read out or fed forward, the channel is the same for any physical realization of the reset, and the protocol therefore does not depend on measurement.

Finally, the inverse circuit $W$ of step 3 transforms the condition in Eq.~\eqref{eq:Lp=1} to
\begin{equation}
    A^{g}_{v_{1}}=+1\quad\forall g\in G
\end{equation}
ensuring that there is trivial charge at vertex $v_1$. 

The whole protocol is therefore the channel $\mathcal{E}=W\circ\bigl(\rho\mapsto\ketbra{+}{+}_{\ell}\otimes\Tr_{\ell}\rho\bigr)\circ W^{\dagger}$, with Kraus operators labelled by the discarded outcome $x$,
\begin{equation}\label{eq:cf-channel}
K_{x}=W\,\ketbra{+}{x}_{\ell}\,W^{\dagger},\qquad
\mathcal{E}(\rho)=\sum_{x\in G}K_{x}\rho K_{x}^{\dagger}.
\end{equation}

\subsubsection{Operator mappings}\label{subsubsec:charge-op-mapping}

We summarize how the stabilizers of the input state are carried through the protocol; the derivations are given in Appendix~\ref{app:cf}.
Here ``mapped'' means that a stabilizer (eigenvalue-$1$ projector) of the input state is carried to a stabilizer of the output state, i.e.\ the support of the state is propagated; the exact Heisenberg-picture identity is Eq.~\eqref{eq:cf-heisenberg} below.

\paragraph{Closed loop ribbon to single vertex projector.}
The closed loop ribbon operator projecting the pair $v_{1},v_{2}$ onto total charge $\Gamma$,
\begin{equation}\label{eq:cf-closed-ribbon}
P^{\Gamma}_{v_{1}v_{2}} = \frac{d_{\Gamma}}{|G|}\sum_{g,x\in G}\chi^{\Gamma}(\bar{g})\;
\vcenter{\hbox{%
\begin{tikzpicture}[line width=0.9pt,baseline]
  \tvlat
  \node[op] at (0,0.45) {$L^{g}$};
  \node[op] at (-0.5,0) {$R^{g}$};
  \node[op] at (0,-0.45) {$R^{g}$};
  \node[op] at (0.75,0) {$T^{x}$};
  \node[op] at (1.5,0.45) {$L^{\bar{x}gx}$};
  \node[op] at (2.0,0) {$L^{\bar{x}gx}$};
  \node[op] at (1.5,-0.45) {$R^{\bar{x}gx}$};
\end{tikzpicture}
}}
\end{equation}
is mapped to the single vertex charge projector of Eq.~\eqref{eq:charge-projector} at $v_{2}$,
\begin{equation}\label{eq:cf-single-vertex-v2}
P^{\Gamma}_{v_{2}} = \frac{d_{\Gamma}}{|G|}\sum_{g\in G}\chi^{\Gamma}(\bar{g})\;
\vcenter{\hbox{%
\begin{tikzpicture}[line width=0.9pt,baseline]
  \tvlat
  \node[op] at (0.75,0) {$R^{g}$};
  \node[op] at (1.5,0.45) {$L^{g}$};
  \node[op] at (2.0,0) {$L^{g}$};
  \node[op] at (1.5,-0.45) {$R^{g}$};
\end{tikzpicture}
}}
\end{equation}
together with the vacuum projector $P_{v_{1}}=\frac{1}{|G|}\sum_{g}A^{g}_{v_{1}}$ of Eq.~\eqref{eq:P-trivial} produced at $v_{1}$ (Appendix~\ref{app:cf-ribbon}).
For a single total charge $\Gamma$ the map is exact in the Heisenberg picture (Appendix~\ref{app:cf-ribbon}):
\begin{align}
\mathcal{E}^{\dagger}(P^{\Gamma}_{v_{2}})&=P^{\Gamma}_{v_{1}v_{2}},\label{eq:cf-heisenberg}\\
\Tr\bigl[P^{\Gamma}_{v_{2}}\,\mathcal{E}(\rho)\bigr]&=\Tr\bigl[P^{\Gamma}_{v_{1}v_{2}}\,\rho\bigr]\nonumber
\end{align}
for every input state $\rho$: the probability of each total charge $\Gamma$ of the pair is transferred exactly to $v_{2}$.

\paragraph{Single vertex projectors to tensor product ribbon.}
Conversely, consider charges $\Gamma_{1}$ at $v_{1}$ and $\Gamma_{2}$ at $v_{2}$, stabilized by the two single vertex charge projectors $P^{\Gamma_{1}}_{v_{1}}$ and $P^{\Gamma_{2}}_{v_{2}}$ of Eq.~\eqref{eq:charge-projector}.
The protocol maps the pair to the closed loop ribbon operator projected onto the tensor product representation $\Gamma_{1}\otimes\Gamma_{2}$, i.e.\ onto its fusion channels $\Lambda$ with $N^{\Lambda}_{\Gamma_{1}\Gamma_{2}}\neq0$,
\begin{equation}\label{eq:cf-tensor-ribbon}
P^{\Gamma_{1}\otimes\Gamma_{2}}_{v_{1}v_{2}}
\equiv\sum_{\Lambda\,:\,N^{\Lambda}_{\Gamma_{1}\Gamma_{2}}\neq0}P^{\Lambda}_{v_{1}v_{2}},
\end{equation}
which stabilizes the post-protocol state (Appendix~\ref{app:cf-vertex}).
Together with the vacuum projector produced at $v_{1}$, the lemma of Appendix~\ref{app:cf-lemma} reduces it to the single vertex support projector $\Pi_{v_{2}}=\sum_{\Lambda:N^{\Lambda}_{\Gamma_{1}\Gamma_{2}}\neq0}P^{\Lambda}_{v_{2}}$ at $v_{2}$.
$\Pi_{v_{2}}$ is a guaranteed support constraint on the output, namely that the charge at $v_{2}$ lies in the fusion channels of $\Gamma_{1}\otimes\Gamma_{2}$; it does not by itself fix the weights of the channels, which are given by Eq.~\eqref{eq:cf-heisenberg}.

\paragraph{Plaquette operators.}
The rewiring of step 1 removes the support of the two plaquettes flanking $\ell$ on the direct edge and only transiently extends the two plaquettes to the left of $v_{1}$ onto it, so the reset leaves all four unchanged and the rewiring back restores them.
All plaquette stabilizers are thus preserved by the protocol (Appendix~\ref{app:cf-plaquette}).

\subsection{Flux fusion protocol}
\label{subsec:flux-fusion}

We now fuse flux anyons at two nearest-neighbour plaquettes $p_{1}$ and $p_{2}$ sharing the dual edge $d$, as sketched in Fig.~\ref{fig:ff-protocol}(a).
The protocol rewires the graph until the dual edge $d$ is a loop carrying the flux of $p_{1}$, resets that edge, then rewires back, so that $p_{1}$ is left with trivial flux and the two fluxes fuse at $p_{2}$ by flux conservation.
We track how the plaquette, closed charge ribbon and closed dyonic ribbon stabilizers are mapped in Sec.~\ref{subsubsec:flux-op-mapping}.

\begin{figure}[ht]
\centering
\adjustbox{max width=\columnwidth}{%
\begin{tikzpicture}[x=0.097\columnwidth, y=0.097\columnwidth, line cap=round, line join=round,
    every node/.style={font=\small},
    plat/.style={latgray, line width=1.0pt},
    red edge/.style={edgered, line width=1.75pt},
    any/.style={circle, inner sep=1pt, minimum size=16pt},
    vtx/.style={circle, fill=latgray, inner sep=0pt, minimum size=3.1pt},
    arr/.style={-{Stealth[length=3.4mm,width=3.0mm]}, black!70, line width=1.6pt},
    hop/.style={-{Stealth[length=1.9mm,width=1.6mm]}, opblue!70!black, line width=0.8pt},
    vl/.style={fill=white, inner sep=0.8pt},
    pl/.style={anchor=north west, inner sep=1pt},
    slat/.style={latgray, line width=1.0pt},
    sred/.style={edgered, line width=1.75pt},
    opb/.style={text=opblue, fill=white, inner sep=1.2pt},
    ctrlb/.style={opblue, line width=0.95pt},
    tag/.style={anchor=south, inner sep=1pt},
    wl/.style={anchor=base east, yshift=-2.25pt, inner sep=1pt, text=opblue},
    elab/.style={latgray, font=\footnotesize, inner sep=0.5pt},
    sarr/.style={-{Stealth[length=2.2mm,width=1.9mm]}, black!70, line width=1.0pt}]
  \newcommand{\ffPatch}[2]{%
    \begin{scope}[shift={(#1,#2)}]
      \foreach \y in {-0.5,0.5}{\draw[plat] (-1.5,\y)--(1.5,\y);}
      \foreach \x in {-1,0,1}{\draw[plat] (\x,-1)--(\x,1);}
      \foreach \x in {-1,0,1}{\foreach \y in {-0.5,0.5}{\node[vtx] at (\x,\y) {};}}
      \draw[red edge] (0,-0.5)--(0,0.5);
    \end{scope}}
  \newcommand{\ffAnyon}[5]{\node[any, fill=#3, draw=black!45, line width=0.6pt] (#5) at (#1,#2) {#4};}
  \newcommand{\ffGhost}[3]{\node[any, draw=latgray, dashed, fill=white, line width=0.6pt] at (#1,#2) {#3};}
  \newcommand{\ffStepDom}{\draw[slat] (0,-0.5) rectangle (2,0.5);}
  \newcommand{\ffDualV}{\draw[sred] (1,-0.5)--(1,0.5);}
  \newcommand{\ffDualD}{\draw[sred] (1,-0.5)--(0,0.5);}
  \newcommand{\ffDualB}{\draw[sred] (1,-0.5) .. controls (0.7,0.1) and (0.3,0.1) .. (0,-0.5);}
  \newcommand{\ffDualL}{\begin{scope}[shift={(1,-0.5)}, rotate=135]
    \draw[sred] (0,0) .. controls (0.4,0.36) and (0.95,0.32) .. (0.95,0) .. controls (0.95,-0.32) and (0.4,-0.36) .. (0,0);\end{scope}}
  \newcommand{\ffPlab}{%
    \node[latgray, inner sep=1pt] at (0.5,0) {$p_{1}$};
    \node[latgray, inner sep=1pt] at (1.5,0) {$p_{2}$};}
  \newcommand{\ffDot}[2]{\node[circle, fill=opblue, inner sep=0pt, minimum size=5.2pt] at (#1,#2) {};}
  \newcommand{\ffGateB}[3]{\draw[ctrlb] (0.5,-0.5)--(#2,#3); \ffDot{0.5}{-0.5} \node[opb] at (#2,#3) {$#1$};}
  \newcommand{\ffGateL}[3]{\draw[ctrlb] (0,0)--(#2,#3);      \ffDot{0}{0}      \node[opb] at (#2,#3) {$#1$};}
  \newcommand{\ffGateT}[3]{\draw[ctrlb] (0.5,0.5)--(#2,#3);  \ffDot{0.5}{0.5}  \node[opb] at (#2,#3) {$#1$};}
  \newcommand{\ffPanel}[3]{%
    \begin{scope}[shift={(#1,#2)}, scale=1.05] \ffStepDom #3 \end{scope}}
  \node[pl] at (-5.25,1.3) {(a)};
  \ffPatch{-2.9}{0}
  \ffAnyon{-3.4}{0}{opblue!70}{$C_{1}$}{ffCa}
  \ffAnyon{-2.4}{0}{ribgreen!80}{$C_{2}$}{ffCb}
  \draw[hop] (ffCa.north) to[bend left=50] (ffCb.north);
  \node[anchor=north west, inner sep=1pt] at (-2.84,-0.56) {$u$};
  \node[text=edgered, anchor=south west] at (-2.9,0.58) {$d$};
  \draw[arr] (-0.5,0)--(0.5,0);
  \node[anchor=south] at (0,0.16) {$\Phi$};
  \ffPatch{2.9}{0}
  \ffGhost{2.4}{0}{$e$}
  \ffAnyon{3.4}{0}{orange!40}{$C$}{ffC}
  \node[pl] at (-5.25,-2.05) {(b)};
  \node[tag] at (0.2,-2.8) {Input};
  \begin{scope}[shift={(-0.85,-3.5)}, scale=1.05] \ffStepDom \ffDualV \ffPlab \end{scope}
  \node[tag] at (0.2,-4.7) {1.\ Rewire};
  \node[wl] at (-4.05,-5.4) {$M^{\dagger}=$};
  \ffPanel{-3.9}{-5.4}{\ffDualV \ffGateT{R}{1}{0}}
  \draw[sarr] (-1.65,-5.4)--(-1.05,-5.4);
  \ffPanel{-0.85}{-5.4}{\ffDualD \ffGateL{R}{0.5}{0}}
  \draw[sarr] (1.4,-5.4)--(2.0,-5.4);
  \ffPanel{2.2}{-5.4}{\ffDualB \ffGateB{R^{\dagger}}{0.5}{-0.05}}
  \node[tag] at (0.2,-6.6) {2.\ reset};
  \begin{scope}[shift={(-0.85,-7.3)}, scale=1.05] \ffStepDom \ffDualL \end{scope}
  \node[draw=opblue, fill=white, rounded corners=1.7pt, line width=0.8pt, inner sep=2.2pt] (ffReset) at (0.78,-7.2) {$\ket{e}$};
  \draw[opblue, line width=0.8pt, {Bar[width=4pt]}-{Stealth[length=2mm,width=1.6mm]}, shorten <=1pt] (ffReset.west) -- (-0.02,-7.18);
  \node[tag] at (0.2,-8.5) {3.\ Rewire back};
  \node[wl] at (-4.05,-9.2) {$M=$};
  \ffPanel{-3.9}{-9.2}{\ffDualL \ffGateB{R}{0.74}{0.12}}
  \draw[sarr] (-1.65,-9.2)--(-1.05,-9.2);
  \ffPanel{-0.85}{-9.2}{\ffDualB \ffGateL{R^{\dagger}}{0.5}{-0.05}}
  \draw[sarr] (1.4,-9.2)--(2.0,-9.2);
  \ffPanel{2.2}{-9.2}{\ffDualD \ffGateT{R^{\dagger}}{0.5}{0}}
  \node[tag] at (0.2,-10.4) {Output};
  \begin{scope}[shift={(-0.85,-11.1)}, scale=1.05] \ffStepDom \ffDualV \end{scope}
  \draw[arr] (-1.0,-3.5) to[out=180, in=90, looseness=1.0] (-2.85,-4.725);
  \draw[arr] (3.25,-6.075) to[out=-90, in=0, looseness=1.0] (1.4,-7.3);
  \draw[arr] (-1.0,-7.3) to[out=180, in=90, looseness=1.0] (-2.85,-8.525);
  \draw[arr] (3.25,-9.875) to[out=-90, in=0, looseness=1.0] (1.4,-11.1);
\end{tikzpicture}}
\caption{Flux fusion protocol.
(a) Two fluxes of classes $C_{1}$ and $C_{2}$ at the plaquettes $p_{1}$ and $p_{2}$ sharing the dual edge $d$ (red) are fused by the channel $\Phi$ of Eq.~\eqref{eq:flux-fusion-channel}, which moves the flux of $p_{1}$ across the dual edge and leaves trivial flux (dashed) at $p_{1}$ and the fusion product $C$, with $a^{C}_{C_{1}C_{2}}\neq0$, at $p_{2}$.
(b) The protocol as a flowchart from the input lattice of (a) to the output lattice; each panel shows a gate or reset (blue) applied to the current graph. In step 1, each gate carries the upper end of $d$ one step around $p_{1}$, so that after the third gate $d$ is a loop based at $u$ carrying the flux of $p_{1}$. In step 2, the loop edge is reset to $\ket{e}$. In step 3, graph rewiring restores the dual edge.}
\label{fig:ff-protocol}
\end{figure}
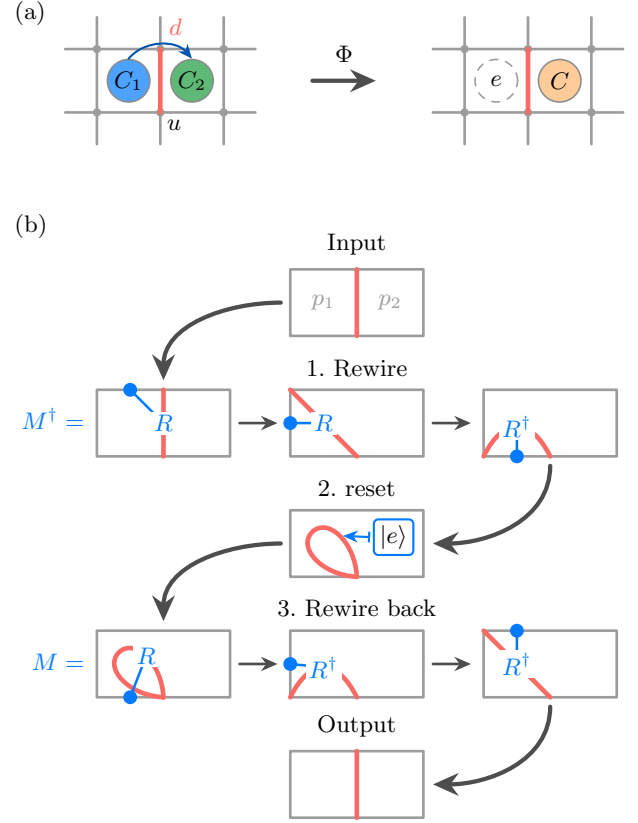

\subsubsection{The protocol}

Let plaquettes $p_{1}$ and $p_{2}$ carry fluxes in the conjugacy classes $C_{1}$ and $C_{2}$ of $G$, with all flux labels below based at the bottom vertex $u$ of the dual edge.
The protocol consists of the three steps shown in Fig.~\ref{fig:ff-protocol}(b):

\begin{enumerate}
    \item rewire the graph until the dual edge $d$ is a loop based at $u$ enclosing the flux of $p_{1}$
    \item reset the loop edge $d$ to the identity state $\ket{e}$
    \item rewire the graph back to the original
\end{enumerate}

We call $M^{\dagger}$ the combined circuit of the three gates of step 1, controlled right multiplications of Eq.~\eqref{eq:CLCR-def} targeting the dual edge $d$, which compose into a single controlled right multiplication, with $T^{x}=\ketbra{x}{x}$ the single-edge projector:
\begin{equation}\label{eq:ff-M-dagger}
M^{\dagger}=\sum_{a,b,c\in G}\;
\vcenter{\hbox{%
\begin{tikzpicture}[scale=1.0, line cap=round, line join=round, baseline]
  \dmlat
  \node[latgray, font=\small] at (-0.75,0.25) {$p_{1}$};
  \node[latgray, font=\small] at (0.75,0) {$p_{2}$};
  \fill[latgray] (0,-0.75) circle (1.6pt);
  \node[latgray, font=\small, anchor=north west, inner sep=1pt] at (0.02,-0.78) {$u$};
  \node[op] at (-0.75,-0.75) {$T^{a}$};
  \node[op] at (-1.5,0) {$T^{b}$};
  \node[op] at (-0.75,0.75) {$T^{c}$};
  \node[op] at (0,0) {$R^{\bar{a}bc}$};
\end{tikzpicture}
}}
\end{equation}
By Eq.~\eqref{eq:LR-def} it acts on the group basis as $M^{\dagger}\ket{g_{d}}=\ket{g_{d}\,\bar{c}\,\bar{b}\,a}=\ket{h_{p_{1},u}}$, where $h_{p_{1},u}=g_{d}\bar{c}\bar{b}a$ is the flux of $p_{1}$ based at $u$, read counterclockwise from $u$ along $d$, the top, the left and the bottom edge.
The order of the gates matters for non-Abelian $G$, since the three gates do not commute.

Its action on the labelled plaquette projector $B^{x}_{p_{1},u}=\delta_{h_{p_{1},u},x}$ of Eq.~\eqref{eq:Bg-MPO} and on the class projector $B^{C}_{p_{1}}$ of Eq.~\eqref{eq:BC-MPO} is
\begin{equation}\label{eq:ff-M-defining}
M^{\dagger}\,B^{x}_{p_{1},u}\,M=T^{x}_{d},\qquad
M^{\dagger}B^{C}_{p_{1}}M=\sum_{x\in C}T^{x}_{d},
\end{equation}
such that the flux of $p_{1}$ is carried by the loop edge $d$ alone (Appendix~\ref{app:M-gauging}).
Dually to $W$, $M$ is a short sequential circuit gauging the $\mathrm{Rep}(G)$ symmetry at $p_{1}$: $M^{\dagger}$ ungauges the plaquette projector into the single-edge projector $T^{x}_{d}$, and $M$ regauges it (Appendix~\ref{app:fb-ungauging}).

Then, in step 2, we reset the loop edge to $\ket{e}$, such that
\begin{equation}\label{eq:Te=1}
    T^{e}_{d}=+1
\end{equation}
is satisfied on the dual edge.
As in Sec.~\ref{subsec:charge-fusion}, the reset is an unconditional operation that requires no measurement; for bookkeeping we represent it as measuring $d$ in the group element basis with outcome $\ket{x}$ and flipping the edge to $\ket{e}$ with $L^{\bar{x}}_{d}$ of Eq.~\eqref{eq:LR-def}, which realizes the replacement channel $\rho\mapsto T^{e}_{d}\otimes\Tr_{d}\rho$ independently of the discarded outcome $x$.

Finally, the inverse circuit $M$ of step 3 transforms the condition in Eq.~\eqref{eq:Te=1} to
\begin{equation}
    B^{e}_{p_{1}}=+1,
\end{equation}
ensuring that there is trivial flux at plaquette $p_{1}$.

The whole protocol is therefore the channel
\begin{align}
\Phi(\rho)&=M\big[T^{e}_{d}\otimes\Tr_{d}(M^{\dagger}\rho M)\big]M^{\dagger}=\sum_{x\in G}K_{x}\rho K_{x}^{\dagger},\label{eq:flux-fusion-channel}\\
K_{x}&=M\ketbra{e}{x}_{d}M^{\dagger}=L^{\bar{x}}_{d}\,B^{x}_{p_{1},u},\nonumber
\end{align}
with Kraus operators labelled by the discarded outcome $x$ (Appendix~\ref{app:M-gauging}).
Each Kraus operator projects the flux of $p_{1}$ based at $u$ onto the value $x$ and then left-multiplies the dual edge by $\bar{x}$, which sets $h_{p_{1},u}$ to $e$ and, since the perimeter holonomy of $p_{1}\cup p_{2}$ does not involve $d$, multiplies the flux of $p_{2}$ by $x$.

\subsubsection{Operator mappings}\label{subsubsec:flux-op-mapping}

We summarize how the stabilizers of the input state are carried through the protocol; the derivations are given in Appendix~\ref{app:ff}.
As in Sec.~\ref{subsec:charge-fusion}, ``mapped'' means that a stabilizer of the input state is carried to a stabilizer of the output state; the exact Heisenberg-picture identities are Eq.~\eqref{eq:ff-heisenberg-loop} below.

\paragraph{Closed loop charge ribbon to single plaquette projector.}
The projector onto total flux class $C$ of the two plaquettes is the closed charge ribbon along its perimeter, the class projector of Eq.~\eqref{eq:BC-MPO} for the six-edge loop, with the tensors contracted cyclically counterclockwise starting from the bottom edge of $p_{1}$; the shared dual edge carries nothing:
\begin{equation}\label{eq:ff-closed-loop}
B^{C}_{p_{1}p_{2}}=\frac{|C|}{|G|}\sum_{\Gamma}\bar{\chi}^{\Gamma}(C)\sum_{ijklmn}\;
\vcenter{\hbox{%
\begin{tikzpicture}[scale=0.8, line cap=round, line join=round, baseline, xl/.style={opblue, fill=white, inner sep=1pt}]
  \dmlat
  \node[latgray] at (-0.75,0) {$p_{1}$};
  \node[latgray] at (0.75,0) {$p_{2}$};
  \node[xl] at (-0.75,-0.75) {\zg{ij}};
  \node[xl] at (0.75,-0.75)  {\zg{jk}};
  \node[xl] at (1.5,0)       {\zg{kl}};
  \node[xl] at (0.75,0.75)   {\zgb{lm}};
  \node[xl] at (-0.75,0.75)  {\zgb{mn}};
  \node[xl] at (-1.5,0)      {\zgb{ni}};
\end{tikzpicture}
}}
\end{equation}
It is mapped to the class projector of Eq.~\eqref{eq:BC-MPO} at the single plaquette $p_{2}$,
\begin{equation}\label{eq:ff-single-plaquette-p2}
B^{C}_{p_{2}}=\frac{|C|}{|G|}\sum_{\Gamma}\bar{\chi}^{\Gamma}(C)\sum_{jklm}\;
\vcenter{\hbox{%
\begin{tikzpicture}[scale=0.8, line cap=round, line join=round, baseline, xl/.style={opblue, fill=white, inner sep=1pt}]
  \dmlat
  \node[latgray] at (-0.75,0) {$p_{1}$};
  \node[latgray] at (0.75,0) {$p_{2}$};
  \node[xl] at (0.75,-0.75)  {\zg{jk}};
  \node[xl] at (1.5,0)       {\zg{kl}};
  \node[xl] at (0.75,0.75)   {\zgb{lm}};
  \node[xl] at (0,0)         {\zgb{mj}};
\end{tikzpicture}
}}
\end{equation}
together with the trivial flux projector $B^{e}_{p_{1}}$ of Eq.~\eqref{eq:Be-MPO} produced at $p_{1}$ (Appendix~\ref{app:ff-ribbon}).
The map is exact in the Heisenberg picture: for every input,
\begin{align}
\Phi^{*}(B^{C}_{p_{2}})&=B^{C}_{p_{1}p_{2}},\qquad
\Phi^{*}(B^{C}_{p_{1}p_{2}})=B^{C}_{p_{1}p_{2}},\label{eq:ff-heisenberg-loop}\\
\Phi^{*}(B^{e}_{p_{1}})&=\mathds{1},\nonumber
\end{align}
hence if the input is stabilized by $B^{C}_{p_{1}p_{2}}$, the output is stabilized by $B^{C}_{p_{2}}$ and by $B^{e}_{p_{1}}$, and in general $\Tr[B^{C}_{p_{2}}\Phi(\rho)]=\Tr[B^{C}_{p_{1}p_{2}}\rho]$: the probability of each total flux class $C$ of the pair is transferred exactly to $p_{2}$.

\paragraph{Single plaquette projectors to closed loop.}
Conversely, consider two fluxes of classes $C_{1}$ at $p_{1}$ and $C_{2}$ at $p_{2}$, stabilized by the two single plaquette class projectors $B^{C_{1}}_{p_{1}}$ and $B^{C_{2}}_{p_{2}}$ of Eq.~\eqref{eq:BC-MPO}.
Let $a^{C}_{C_{1}C_{2}}=\#\{(c_{1},c_{2})\in C_{1}\times C_{2}:c_{1}c_{2}=c\}$ for fixed $c\in C$ be the class-algebra constants, $\sigma_{C_{1}}\sigma_{C_{2}}=\sum_{C}a^{C}_{C_{1}C_{2}}\sigma_{C}$ with $\sigma_{C}=\sum_{g\in C}g$.
The protocol maps the pair to the closed charge ribbon projected onto the fusion channels of the two classes, i.e.\ onto the classes $C$ with $a^{C}_{C_{1}C_{2}}\neq0$,
\begin{equation}\label{eq:ff-Pi-p1p2}
\Pi_{p_{1}p_{2}}\equiv\sum_{C:\,a^{C}_{C_{1}C_{2}}\neq0}B^{C}_{p_{1}p_{2}},
\end{equation}
which stabilizes the post-protocol state (Appendix~\ref{app:ff-plaquette}).
Together with the trivial flux projector produced at $p_{1}$, the lemma of Appendix~\ref{app:ff-lemma} reduces it to the single plaquette support projector $\Pi_{p_{2}}=\sum_{C:\,a^{C}_{C_{1}C_{2}}\neq0}B^{C}_{p_{2}}$ at $p_{2}$.
$\Pi_{p_{2}}$ is a guaranteed support constraint on the output, namely that the flux at $p_{2}$ lies in the fusion channels of $C_{1}$ and $C_{2}$; it does not by itself fix the weights of the classes, which are given by Eq.~\eqref{eq:ff-heisenberg-loop}.
In labelled form, $\Phi^{*}(B^{x}_{p_{2},u})=B^{x}_{p_{1}p_{2},u}$, so fluxes with labels $x_{1}$ at $p_{1}$ and $x_{2}$ at $p_{2}$, both based at $u$, fuse deterministically to the label $x_{2}x_{1}$ at $p_{2}$, since the labels based at $u$ multiply, $h_{p_{1}p_{2},u}=h_{p_{2},u}h_{p_{1},u}$ (Eq.~\eqref{eq:ff-label-composition}); the order matters for non-Abelian $G$.
The constants $a^{C}_{C_{1}C_{2}}$ label magnetic sectors only; the centralizer charge is tracked by the dyonic ribbon.

\ifdefined\fdDyonDiagramMacrosLoaded\else
\def\fdDyonDiagramMacrosLoaded{1}
\makeatletter
\@ifundefined{color@opblue}{\definecolor{opblue}{RGB}{0,122,255}}{}
\@ifundefined{color@latgray}{\definecolor{latgray}{RGB}{156,156,156}}{}
\@ifundefined{color@edgered}{\definecolor{edgered}{RGB}{250,106,100}}{}
\@ifundefined{color@ribgreen}{\definecolor{ribgreen}{RGB}{52,168,83}}{}
\makeatother
\tikzset{
 fd lattice/.style={draw=latgray,line width=.75pt},
 fd direct/.style={draw=edgered,line width=1.1pt},
 fd operator/.style={text=opblue,fill=white,inner sep=1.15pt,font=\small},
 fd vertex/.style={text=latgray!80!black,font=\scriptsize,inner sep=1pt},
 fd flux/.style={draw=ribgreen,line width=.8pt,rounded corners=3pt,densely dashed},
 fd flux label/.style={text=ribgreen!65!black,fill=white,inner sep=1.3pt,font=\small},
 fd identity/.style={text=edgered!80!black,fill=white,inner sep=1pt,font=\small}
}
\newcommand{\fdDyonSkeleton}{%
 \draw[fd lattice] (-1.55,-.85) rectangle (1.55,.85);
 \draw[fd direct] (0,-.85)--(0,.85);
 \foreach \xx in {-1.55,0,1.55}{
   \draw[fd lattice] (\xx,-.85)--(\xx,-1.85);
   \draw[fd lattice] (\xx,.85)--(\xx,1.85);}
 \foreach \yy in {-.85,.85}{
   \draw[fd lattice] (-2.6,\yy)--(-1.55,\yy);
   \draw[fd lattice] (1.55,\yy)--(2.6,\yy);}
 \foreach \xx in {-1.55,0,1.55}{\foreach \yy in {-.85,.85}{
   \fill[latgray] (\xx,\yy) circle(1.25pt);}}
 \node[fd vertex,anchor=south east] at (-.02,-.84) {$u$};
 \node[fd vertex,anchor=north east] at (-.02,.84) {$w$};
 \node[fd vertex,anchor=south west] at (-1.51,-.81) {$v_1$};
 \node[fd vertex,anchor=south east] at (1.51,-.81) {$v_3$};
 \node[fd vertex,anchor=north east] at (1.51,.81) {$v_4$};
 \node[fd vertex,anchor=north west] at (-1.51,.81) {$v_6$};
 \node[fd vertex] at (-.77,-.23) {$p_1$};
 \node[fd vertex] at (.77,-.23) {$p_2$};
}
\newcommand{\fdDyonOuter}{%
 \node[fd operator] at (-2.11,-.85) {$R^{z_1}$};
 \node[fd operator] at (-1.55,-1.43) {$R^{z_1}$};
 \node[fd operator] at (0,-1.43) {$R^{z}$};
 \node[fd operator] at (1.55,-1.43) {$R^{z_3}$};
 \node[fd operator] at (2.11,-.85) {$L^{z_3}$};
 \node[fd operator] at (2.11,.85) {$L^{z_4}$};
 \node[fd operator] at (1.55,1.43) {$L^{z_4}$};
 \node[fd operator] at (0,1.43) {$L^{z_w}$};
 \node[fd operator] at (-1.55,1.43) {$L^{z_6}$};
 \node[fd operator] at (-2.11,.85) {$R^{z_6}$};
}
\newcommand{\fdDyonPerimeterMPO}{%
 \node[fd operator] at (-.77,-.85) {$\mathcal Z^\Gamma_{ij}$};
 \node[fd operator] at (.77,-.85) {$\mathcal Z^\Gamma_{j'k}$};
 \node[fd operator] at (1.55,0) {$\mathcal Z^\Gamma_{kl}$};
 \node[fd operator] at (.77,.85) {$\bar{\mathcal Z}^\Gamma_{lm}$};
 \node[fd operator] at (-.77,.85) {$\bar{\mathcal Z}^\Gamma_{mn}$};
 \node[fd operator] at (-1.55,0) {$\bar{\mathcal Z}^\Gamma_{ni}$};
}
\newcommand{\fdDyonBoundary}[2][1]{%
 \begin{tikzpicture}[x=#1cm,y=#1cm,line cap=round,line join=round,
    baseline=(current bounding box.center)]
  \fdDyonSkeleton
  \fdDyonOuter
  \ifnum#2=0\relax
   \draw[fd flux] (-1.24,-.53) rectangle (1.24,.53);
   \node[fd flux label] at (0,.43) {$B^h_{p_1p_2,u}$};
   \node[fd identity] at (0,0) {$I_d$};
  \else\ifnum#2>4\relax
   \draw[fd flux] (-1.24,-.53) rectangle (1.24,.53);
   \node[fd flux label] at (0,.43) {$B^h_{p_1p_2,u}$};
   \ifnum#2=5\relax \node[fd identity] at (0,0) {$I_d$};
   \else \node[fd operator] at (0,0) {$T_d^e$};\fi
  \else
   \ifnum#2=4\relax
    \node[fd operator] at (-.77,-.85)
      {\shortstack{$\mathcal Z^\Gamma_{ij}$\\[-2pt]$\mathcal Z^{\Gamma'}_{\delta\alpha}$}};
    \node[fd operator] at (.77,-.85) {$\mathcal Z^\Gamma_{j'k}$};
    \node[fd operator] at (1.55,0) {$\mathcal Z^\Gamma_{kl}$};
    \node[fd operator] at (.77,.85) {$\bar{\mathcal Z}^\Gamma_{lm}$};
    \node[fd operator] at (-.77,.85)
      {\shortstack{$\bar{\mathcal Z}^\Gamma_{mn}$\\[-2pt]$\bar{\mathcal Z}^{\Gamma'}_{\beta\gamma}$}};
    \node[fd operator] at (-1.55,0)
      {\shortstack{$\bar{\mathcal Z}^\Gamma_{ni}$\\[-2pt]$\bar{\mathcal Z}^{\Gamma'}_{\gamma\delta}$}};
    \node[fd operator] at (0,0) {$\mathcal Z^{\Gamma'}_{\alpha\beta}$};
   \else
    \fdDyonPerimeterMPO
    \ifnum#2=3\relax
     \node[fd operator] at (0,0) {$\mathcal Z^{\Gamma'}_{\alpha\alpha}$};
    \else \node[fd identity] at (0,0) {$I_d$};\fi
   \fi
  \fi\fi
 \end{tikzpicture}%
}
\newcommand{\fdDyonLocal}[2][1]{%
 \begin{tikzpicture}[x=#1cm,y=#1cm,line cap=round,line join=round,
    baseline=(current bounding box.center)]
  \fdDyonSkeleton
  \node[fd operator] at (-.77,-.85) {$R^z$};
  \node[fd operator] at (0,-1.43) {$R^z$};
  \ifnum#2=0\relax
   \draw[fd flux] (.25,-.53) rectangle (1.24,.53);
   \node[fd flux label] at (.75,.28) {$B^h_{p_2,u}$};
   \node[fd operator] at (0,0) {$L^z$};
   \node[fd operator] at (.77,-.85) {$L^z$};
  \else
   \node[fd operator] at (.77,-.85) {$\mathcal Z^\Gamma_{j'k}L^z$};
   \node[fd operator] at (1.55,0) {$\mathcal Z^\Gamma_{kl}$};
   \node[fd operator] at (.77,.85) {$\bar{\mathcal Z}^\Gamma_{lm}$};
   \node[fd operator] at (0,0) {$\bar{\mathcal Z}^\Gamma_{mj}L^z$};
  \fi
 \end{tikzpicture}%
}
\newcommand{\fdDyonVertex}[2][1]{%
 \begin{tikzpicture}[x=#1cm,y=#1cm,line cap=round,line join=round,
    baseline=(current bounding box.center)]
  \draw[fd lattice] (-1.45,-.55) rectangle (1.45,.55);
  \draw[fd direct] (0,-.55)--(0,.55);
  \draw[fd lattice] (0,-.55)--(0,-1.25);
  \draw[fd lattice] (0,.55)--(0,1.25);
  \node[fd vertex] at (-.72,0) {$p_1$};
  \node[fd vertex] at (.72,0) {$p_2$};
  \node[fd vertex,anchor=south east] at (0,-.55) {$u$};
  \node[fd vertex,anchor=north east] at (0,.55) {$w$};
  \ifnum#2<2\relax
   \node[fd operator] at (-.72,-.55) {$R^g$};
   \node[fd operator] at (.72,-.55) {$L^g$};
   \node[fd operator] at (0,-.98) {$R^g$};
   \ifnum#2=0\relax \node[fd operator] at (0,.1) {$L^g$};
   \else \node[fd operator] at (0,.1) {$C^g$};\fi
  \else
   \node[fd operator] at (-.72,.55) {$R^g$};
   \node[fd operator] at (.72,.55) {$L^g$};
   \node[fd operator] at (0,.98) {$L^g$};
   \ifnum#2=2\relax \node[fd operator] at (0,-.1) {$R^g$};
   \else \node[fd identity] at (0,-.1) {$I_d$};\fi
  \fi
 \end{tikzpicture}%
}
\fi

\paragraph{Closed dyonic ribbon to local site projector.}
Two fluxes can fuse to a dyon, whose type $(C,\pi)$ carries both a magnetic flux class $C$ and an irrep $\pi$ of the centralizer $Z(h)$ of $h\in C$.
Let $K^{(C,\pi)}_\sigma$ project onto the total anyon type enclosed by the closed ribbon $\sigma$ around $p_1\cup p_2$, based at $u$, and assume that the five perimeter vertices other than $u$ are charge free, a constraint we denote by the projector $\Pi_{\mathrm{in}}$.
In lattice-operator form,
\begin{equation}
 \begin{gathered}
 K^{(C,\pi)}_\sigma
 =\sum_{h\in C}\frac{d_\pi}{|Z(h)|}
       \sum_{z\in Z(h)}\bar\chi_\pi^{(h)}(z)\;\times\\[-2pt]
 \vcenter{\hbox{\fdDyonBoundary{0}}},
 \end{gathered}
 \label{eq:dg-main-input}
\end{equation}
where the green contour denotes the labelled magnetic projector $B^{h}_{p_{1}p_{2},u}$ and the blue edge labels the electric action, the centralizer element $z$ transported along the perimeter to the outer legs, which does not touch the dual edge $d$ (red); the transport and character conventions are given in Appendix~\ref{app:ff-dyon}.
An input stabilized by $K^{(C,\pi)}_\sigma$ is mapped to an output stabilized by $B^e_{p_1}$ and by the local site projector
\begin{equation}
 \begin{gathered}
 D^{(C,\pi)}_{(p_2,u)}
 =\sum_{h\in C}\frac{d_\pi}{|Z(h)|}
       \sum_{z\in Z(h)}\bar\chi_\pi^{(h)}(z)\;\times\\[-2pt]
 \vcenter{\hbox{\fdDyonLocal{0}}},
 \end{gathered}
 \label{eq:dg-main-output}
\end{equation}
whose diagram is $B^h_{p_2,u}A_u^z$, while the five neutral-vertex constraints are preserved.
The map is again exact in the Heisenberg picture on the input subspace,
\begin{equation}
 \Phi^*\!\left(D^{(C,\pi)}_{(p_2,u)}\right)\Pi_{\mathrm{in}}
   =K^{(C,\pi)}_\sigma\Pi_{\mathrm{in}},
 \label{eq:ff-dyon-claim}
\end{equation}
so the entire label $(C,\pi)$ becomes the anyon type at the site $(p_2,u)$: for a genuine dyon, $C\neq\{e\}$ and $\pi$ nontrivial, $p_2$ retains flux and $u$ carries a vertex violation, while $p_1$ is flux free.

\subsection{Concrete example of decoherence}
\label{subsec:cc-example}

Here, we consider a concrete example of fusing non-Abelian $\mathbf{2}$ charges in $D(S_3)$ to show how reset can remove the coherence across fusion channels while preserving their probabilities.

\paragraph{Four-anyon setup.}
Consider four pure electric anyons of type $\mathbf{2}$ in $D(S_3)$, arranged on a $2\times2$ block as in Fig.~\ref{fig:cc-setup}.
The upper row, $v_1,v_2$, holds the pair to be fused; the lower row, $v_3,v_4$, is a spectator pair whose charge labels and pair-fusion channel are preserved by the protocol.
All other vertices and all plaquettes are in their vacuum sectors.
We compare two states with these same local anyon types, but different ways of coupling them.
In $\ket{\Psi_{\mathbf{1}}}$, each pair fuses to vacuum.
In $\ket{\Psi_{\mathbf{1}'}}$, each pair instead fuses to the sign charge $\mathbf{1}'$.
The latter is allowed because $\mathbf{2}\otimes\mathbf{2}=\mathbf{1}\oplus\mathbf{1}'\oplus\mathbf{2}$.
Since $\mathbf{1}'\otimes\mathbf{1}'=\mathbf{1}$, the spectator pair compensates the first pair's charge: both four-anyon states have total vacuum charge.

\begin{figure}[ht]
\centering
\begin{tikzpicture}[x=0.112\columnwidth,y=0.112\columnwidth,
    line cap=round,line join=round,
    cc edge/.style={latgray,line width=0.8pt},
    cc pair/.style={latgray,densely dashed,rounded corners=5pt,line width=0.65pt},
    cc charge/.style={circle,draw=opblue!75!black,fill=opblue!12,
      line width=0.65pt,inner sep=0pt,minimum size=12pt,font=\small},
    cc vacuum/.style={circle,draw=latgray,fill=white,densely dashed,
      line width=0.65pt,inner sep=0pt,minimum size=12pt,font=\small},
    cc label/.style={font=\small,inner sep=1pt},
    cc panel/.style={font=\small,inner sep=0pt}]
  \foreach \dx in {0,4.8}{
    \begin{scope}[shift={(\dx,0)}]
      \foreach \yy in {0,-1.6}{
        \draw[cc edge] (-0.42,\yy)--(2.02,\yy);}
      \foreach \xx in {0,1.6}{
        \draw[cc edge] (\xx,0.40)--(\xx,-2.00);}
      \foreach \xx/\yy/\vv in {0/0/1,1.6/0/2,0/-1.6/3,1.6/-1.6/4}{
        \node[cc label,anchor=north,fill=white]
          at (\xx,\yy-0.39) {$v_{\vv}$};}
      \draw[cc pair] (-0.35,-1.88) rectangle (1.95,-1.26);
      \node[cc label,fill=white] at (0.8,-1.05) {$\lambda$};
      \node[cc label] at (0.8,-2.55) {total $\mathbf{1}$};
    \end{scope}
  }
  \node[cc panel] at (0.8,1.02) {(a) Input};
  \draw[cc pair] (-0.35,-0.28) rectangle (1.95,0.34);
  \node[cc label] at (0.8,0.56) {$\lambda$};
  \draw[edgered,line width=1.65pt] (0,0)--(1.6,0);
  \node[cc label,text=edgered,fill=white] at (0.8,-0.18) {$\ell$};
  \foreach \xx/\yy in {0/0,1.6/0,0/-1.6,1.6/-1.6}{
    \node[cc charge] at (\xx,\yy) {$\mathbf{2}$};}
  \draw[-{Stealth[length=1.7mm,width=1.4mm]},line width=0.8pt]
    (2.45,-0.80)--(3.95,-0.80);
  \node[cc label] at (3.2,-0.49) {$\mathcal E$};
  \node[cc panel] at (5.6,1.02) {(b) After fusion};
  \node[cc vacuum] at (4.8,0) {$\mathbf{1}$};
  \node[cc charge] at (6.4,0) {$\lambda$};
  \foreach \xx in {4.8,6.4}{
    \node[cc charge] at (\xx,-1.6) {$\mathbf{2}$};}
\end{tikzpicture}
\caption{Input (left) and output (right) for one component $\lambda=\mathbf{1}$ or $\mathbf{1}'$.
The upper pair is fused across $\ell$; the lower pair is a spectator.
Dashed contours label pair charges.
The output has vacuum at $v_1$ and charge $\lambda$ at $v_2$, with total charge $\mathbf{1}$ in both panels.
A site labelled $\mathbf{1}$ is unexcited; the pairs need not be far apart.}
\label{fig:cc-setup}
\end{figure}
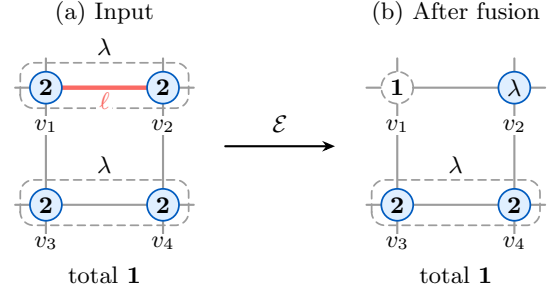

The local internal states are fixed by the explicit preparation in Appendix~\ref{app:cc-details}, which also verifies the four charge labels and the pair-fusion channels.
In particular, the two states are normalized and orthogonal.
We use their coherent superposition
\begin{equation}
 \ket{\Psi}=\alpha\ket{\Psi_{\mathbf{1}}}
             +\beta\ket{\Psi_{\mathbf{1}'}},\qquad
 |\alpha|^2+|\beta|^2=1.
 \label{eq:cc-superposition}
\end{equation}
This tests coherence between different fusion outcomes without changing either the incoming anyon types or the total charge.
The two branches belong to the same global vacuum sector, so their relative phase is a coherence of the joint four-anyon state allowed by charge superselection.

\paragraph{Rewiring and reset.}
The first rewiring step moves the outer edges of $v_1$ to $v_2$, leaving the shared edge $\ell$ as a dangling qudit.
The pair's total charge is now carried by the merged vertex.
For the chosen preparation, however, the dangling qudit also retains a record of the fusion component: its two internal states are orthogonal even though they have the same charge label.
Resetting this qudit to $\ket{+}$ discards that record.
Reversing the rewiring then restores the original graph, leaving $v_1$ in vacuum and the fused charge at $v_2$.
The rewiring $W^\dagger$ correlates the two fusion components with orthogonal states of the shared edge, so their relative phase is encoded in correlations between that edge and the remaining lattice.
Resetting the edge traces out this record: the orthogonality of the discarded states annihilates the off-diagonal term between fusion channels, while their unit norms leave the diagonal probabilities unchanged.

Let $\ket{\Omega_\lambda}$ denote the normalized output for the component $\ket{\Psi_\lambda}$.
It has charge $\lambda$ at $v_2$, with both spectator charges still of type $\mathbf{2}$ and no magnetic flux.
The complete fusion channel gives
\begin{align}
 \mathcal E\bigl(\ketbra{\Psi}{\Psi}\bigr)
  &=|\alpha|^2\ketbra{\Omega_{\mathbf{1}}}{\Omega_{\mathbf{1}}}
    \notag\\[-1pt]
  &\quad+|\beta|^2\ketbra{\Omega_{\mathbf{1}'}}{\Omega_{\mathbf{1}'}}.
 \label{eq:cc-output-main}
\end{align}
Thus the fusion probabilities remain $(p_{\mathbf{1}},p_{\mathbf{1}'},p_{\mathbf{2}})=(|\alpha|^2,|\beta|^2,0)$, but the relative phase is lost.
Physically, the first rewiring steps transfers the which-component record to degrees of freedom that is then reset and hence discarded.

While in this setup the relative phase between fusion channels was completely lost, it is possible to construct initial states where the relative phase partially survives under the protocol, as we detail in Appendix~\ref{app:cc-details}.

\section{Discussion}
\label{sec:discussion}

We have revealed the minimal set of operations necessary for manipulating non-Abelions: we first introduced the framework of graph rewiring to understand and derive unitary protocols for coherent moving and braiding of anyons. We then demonstrated that some amount of non-unitarity is needed to perform anyon fusion without knowledge of the input state. Reset was identified as a key non-unitary primitive, which, when combined with graph rewiring, enabled anyon fusion---in contrast to a purely unitary protocol prohibited by the no-go theorem. If the resulting fusion channel needs to be determined, then projective measurement for reading-out completes the toolkit. In contrast to the conventional belief, no adaptive measurements or ancillas are strictly required to braid non-Abelions or to fuse them: measurement enters only when the fusion outcome is to be read out (although adaptive operations may be useful for reducing circuit depth in special cases).

The rewiring and reset formalism is a conceptual framework that simplifies and unifies much of the prior literature on non-Abelian TO and non-Abelions. In addition to its conceptual utility, it makes the design of explicit circuits for use in experiment more straightforward and less costly than previous methods. Our results thus advance both the theoretical and experimental efforts in non-Abelian topological quantum computation.

A natural direction for future work is to explore how graph rewiring can provide a geometric perspective on sequential circuits and gauging on a lattice. The gauging map to derive a generic form for anyon creation operators \cite{lyons_protocols_2025} has been shown here as particular cases of graph rewiring. It would be interesting to clarify whether gauging on the lattice may in general be represented by graph rewiring operation. Furthermore, the gauging map has been understood elsewhere as a particular example of a sequential circuit \cite{chen_sequential_2024}, which are circuits capable of mapping between distinct phases of matter. Graph rewiring may also provide another viewpoint on such sequential circuits. More broadly, entangling circuits have been used to establish equivalences between codes: local Clifford circuits reduce every two-dimensional translation-invariant Pauli stabilizer code to copies of the toric code~\cite{bombin2012universal,haah2021classification}, with the unfolding of the color code as a prominent example~\cite{kubica2015unfolding}; finite-depth circuits map Pauli stabilizer models of twisted quantum doubles onto string-net ground states~\cite{ellison2022pauli}; and morphing relates codes by disentangling subsets of qubits~\cite{vasmer2022morphing}. Since graph rewiring assigns a transparent lattice picture to each such gate, it is a promising visual tool for elucidating other equivalences of this kind.

\subsection*{Implications for practical error correction}

Replacing measurement and feedback with reset can reduce control overhead: fast, high-fidelity unconditional reset is available in superconducting qubits \cite{magnard_fast_2018,kim_fast_2025,chen_fast_2026}, while neutral-atom platforms support optical pumping \cite{bluvstein_fault-tolerant_2026} and continuous loading of initialized atoms \cite{chiu_continuous_2025-1}. Local dissipation has long been proposed for preparing toric-code ground states \cite{verstraete_quantum_2009,weimer_rydberg_2010}, with proof-of-principle digital implementations using resettable ancillas \cite{barreiro_open-system_2011,muller_simulating_2011}.
Measurement-free error correction likewise uses ancillas for coherent syndrome processing \cite{fujii_measurement-free_2014,crow_improved_2016,heusen_measurement-free_2024,butt_demonstration_2026,park_enhancing_2025}.
Our protocols in Sec.~\ref{sec:fusion} implement local charge and flux fusion using a constant-size circuit and one reset of a code qudit, without separate ancillas or measurement records, such that
each application deterministically removes the input anyon and leaves behind a vacuum site.
These primitives apply to quantum doubles of arbitrary finite groups and may reduce the hardware overhead of dissipative preparation; whether suitable update schedules also shorten relaxation times remains open.

Fault-tolerant storage requires more than preparation: eliminating excitations can still enact a logical error (even for local rules if given sufficiently long time).
A nonzero memory threshold requires the recoverable lifetime to diverge with system size at fixed noise strength below threshold.
Recent progress includes a fault-tolerant local automaton for the toric code \cite{balasubramanian_local_2026}, local message-passing decoders \cite{lake_fast_2025,lake_local_2025}, and autonomous dissipative quantum cellular automata based on concatenated codes \cite{dunnweber_quantum_2026}.
Our fusion primitives may simplify local correction circuits even for the toric code.
More broadly, the same unitary and reset operations implement fusion in non-Abelian quantum doubles, suggesting a route toward extending local decoding to these models.
Ref.~\cite{chirame_stabilizing_2025} already stabilize non-Abelian topological order through local Lindbladian dynamics, but require perfectly heralded errors.
Combining our input-agnostic fusion channels with local decoding rules that tolerate unheralded noise is an exciting direction for future research.

\section{Acknowledgments}

C.F.B.L. acknowledge fruitful discussions with Nik Gjonbalaj, Dominik Kufel, Dachuan Lu, Rahul Sahay, Jeet Shah, and Yarden Sheffer. C.F.B.L., R.S., and A.L. acknowledge support from the National Science Foundation Graduate Research Fellowship Program (grant no. DGE 2140743). C.F.B.L., A.L., and A.V. are partly supported by the Simons Collaboration on Ultra Quantum Matter, which is a grant from the Simons Foundation (618615). This material is based on the work supported by the Defense Advanced Research Projects Agency (DARPA) MeasQuIT program under agreement no. HR00112490361 (R.V. and A.V.). This research was supported in part by grant NSF PHY-2309135 to the Kavli Institute for Theoretical Physics (KITP). C.F.B.L. and R.V. thank KITP for hospitality during the program "Engineering Fractionalization in Quantum Many-Body Systems", during which part of the work was completed. 

\bibliography{references}

\newpage
\onecolumngrid
\appendix
\addtocontents{toc}{\string\let\string\l@subsection\string\@gobbletwo}%

\section{Basis and gate definitions and conventions}

We collect the definitions and conventions that are scattered across the main text for ease of reference.

\subsection{Group element and representation basis}

The local Hilbert space on each qudit edge is $|G|$ dimensional. There are two convenient basis for the local Hilbert space: the group element basis and the representation basis, which are generalizations of the $Z$ basis and $X$ basis in toric code (i.e. when $G=\mathbb{Z}_2$).

The group basis is given by the regular representation of $G$, spanned by $\ket{g}$ for all $g \in G$. For example, when $G=S_3 = \langle s,r| s^2=r^3=e, srs=r^2\rangle$, which has 6 distinct group elements generated by a 120-degree rotation $r$ and a reflection $s$, there are 6 group element basis states:
\begin{equation}
    \ket{e}, \quad \ket{r}, \quad \ket{r^2}, \quad \ket{s}, \quad \ket{rs}, \quad \ket{r^2s}.
\end{equation}

The representation basis state is given by $\ket{\Gamma; i,j}$, where $\Gamma$ is an irreducible representation (irrep) of $G$, and $i,j$ are indices between $1$ and the dimension of the irrep $d_\Gamma$. For example, for $G=S_3$, there are three irreps: $\mathbf{1},\mathbf{1'},\mathbf{2}$, where $\mathbf{2}$ is a two-dimensional irrep; this results in $6$ representation basis states:
\begin{equation}
    \ket{\mathbf{1}; 1,1}, \quad \ket{\mathbf{1'}; 1,1}, \quad \ket{\mathbf{2}; 1,1}, \quad \ket{\mathbf{2}; 1,2}, \quad \ket{\mathbf{2}; 2,1}, \quad \ket{\mathbf{2}; 2,2}.
\end{equation}

In $\mathbb{Z}_2$ case, the Hadamard gate relate between the $Z$ and $X$ basis. 
Similarly, for the general $G$ case, there is a group Fourier transform (and its inverse) that relate between the group element basis and the representation basis. 
The group Fourier transform $F$ that maps group element basis to representation basis is given by \cite{bastidas_unification_2025}:
\begin{equation}\label{eq:gp-FT}
    F=\sum_{\Gamma} \sum_{i j} \sum_{g \in G} \sqrt{\frac{d_{\Gamma}}{|G|}} \Gamma_{ij}(g)|\Gamma, i j\rangle\langle g|
\end{equation}
where $\Gamma_{ij}(g)$ is the $(i,j)$ entry of the representation matrix $\Gamma(g)$ in the $\Gamma$ irrep.

And the inverse group Fourier transform $F^\dagger$ is given by
\begin{equation}\label{eq:inv-gp-FT}
    F^\dagger = \sum_{\Gamma} \sum_{i j} \sum_{g \in G} \sqrt{\frac{d_{\Gamma}}{|G|}} \Gamma^*_{ij}(g)\ket{g} \bra{\Gamma, i j} 
\end{equation}
where $\Gamma^*_{ij}(g)$ is the complex conjugation of $\Gamma_{ij}(g)$.

\subsection{Group multiplication gates and the inverse gate}
\label{app:mult-gates}

The action of left and right group multiplication gates are:
\begin{equation}
    L^g\ket{h} = \ket{gh}, \quad \quad \quad R^g\ket{h}= \ket{h \bar{g}}
    \label{eq:LR-def}
\end{equation}

with their respective inverses:
\begin{equation}
    (L^g)^\dagger\ket{h} = L^{\bar{g}}\ket{h} = \ket{\bar{g}h}, \quad \quad \quad (R^g)^\dagger \ket{h} = R^{\bar{g}}\ket{h}= \ket{h g}
    \label{eq:LR-inverse}
\end{equation}

The action of controlled group multiplication gates are as follows:
\begin{equation}
    CL\ket{h_1}\ket{h_2} = \ket{h_1}\ket{h_1 h_2}, \quad \quad \quad CR\ket{h_1}\ket{h_2} = \ket{h_1}\ket{h_2 \bar{h}_1}
    \label{eq:CLCR-def}
\end{equation}

where the first qudit is the control, and the second is the target, such that $CL$($CR$) implements left(right) group multiplication on the target qudit. 

The action of conjugation by group element $g$ are
\begin{equation}
    C^g\ket{h} = L^gR^g\ket{h} = \ket{gh\bar{g}}
    \label{eq:conjugation-def}
\end{equation}

The inverse gate (also known as the antipode; the name and notation is inherited from Hopf algebra)
\begin{equation}
    \mathcal{S}\ket{h} = \ket{\bar{h}}
    \label{eq:inverse-gate-def}
\end{equation}

The inverse gate provides another way to express the inverse of controlled group multiplication gates:
\begin{equation}
    C_cL_t^\dagger = S_c (C_cL_t) S_c^\dagger,  \quad \quad \quad 
    C_cR_t^\dagger = S_c (C_cR_t) S_c^\dagger,  
    \label{eq:controlled-inverse}
\end{equation}

\begin{equation}
\vcenter{\hbox{
\begin{tikzpicture}
\begin{yquant}[horizontal]

[violet] qubit {} a[1];
[violet] qubit {} a[+1];

box {$L^\dagger$} a[1] | a[0];

\end{yquant}
\end{tikzpicture}
}}
=
\vcenter{\hbox{
\begin{tikzpicture}
\begin{yquant}[horizontal]

[violet] qubit {} a[1];
[violet] qubit {} a[+1];

box {$S$} a[0];
box {$L$} a[1] | a[0];
box {$S$} a[0];

\end{yquant}
\end{tikzpicture}
}}, \quad \quad \quad 
\vcenter{\hbox{
\begin{tikzpicture}
\begin{yquant}[horizontal]

[violet] qubit {} a[1];
[violet] qubit {} a[+1];

box {$R^\dagger$} a[1] | a[0];

\end{yquant}
\end{tikzpicture}
}}
=
\vcenter{\hbox{
\begin{tikzpicture}
\begin{yquant}[horizontal]

[violet] qubit {} a[1];
[violet] qubit {} a[+1];

box {$S$} a[0];
box {$R$} a[1] | a[0];
box {$S$} a[0];

\end{yquant}
\end{tikzpicture}
}}
\end{equation}

Conjugating the target instead of the control exchanges left and right multiplication, since $\mathcal{S}R^{g}\mathcal{S}^{\dagger}=L^{g}$ by Eqs.~\eqref{eq:LR-def} and \eqref{eq:inverse-gate-def}:
\begin{equation}
    C_cL_t = S_t (C_cR_t) S_t^\dagger,  \quad \quad \quad
    C_cL_t^\dagger = S_t (C_cR_t^\dagger) S_t^\dagger,
    \label{eq:controlled-target-inverse}
\end{equation}
where $S_t$ is the inverse gate on the target qudit.

\subsection{Diagonal gates}
\label{app:diagonal-gates}

The group-valued Z operator is defined as
\begin{equation}
    \mathcal{Z}^\Gamma_{ij} = \sum_{g\in G} [\Gamma(g)]_{ij} \ketbra{g},
    \label{eq:Z-def}
\end{equation}

which is an matrix product operator (MPO) that requires specification of the indices $i,j$. Its conjugate is given by
\begin{equation}
    \bar{\mathcal{Z}}^{\Gamma}_{ik}= \sum_g [\Gamma(\bar g)]_{ik}\,|g\rangle\langle g| = S \mathcal{Z}^{\Gamma}_{ik} S^{\dagger}
    \label{eq:Zbar-def}
\end{equation}

such that the two are inverse to each other upon contracting the bond index:
\begin{align}
    \sum_{j}\mathcal{Z}^{\Gamma}_{ij} \bar{\mathcal{Z}}^{\Gamma}_{jk}
    &= \sum_{h,g} \sum_{j}[\Gamma(h)]_{ij} [\Gamma(\bar{g})]_{jk} \ketbra{h}\ketbra{g}
    = \sum_{h} [\Gamma(h\bar{h})]_{ik} \ketbra{h}
    = \delta_{ik},
    \label{eq:Z-orthogonality}
\end{align}
and similarly,
\begin{equation}
    \sum_j\bar{\mathcal{Z}}^\Gamma_{ij}\mathcal{Z}^\Gamma_{jk} = \delta_{ik}
    \label{eq:Zbar-orthogonality}
\end{equation}

Throughout, $\Gamma$ runs over a complete set of inequivalent unitary irreducible representations of the finite group $G$, $d_{\Gamma}=\dim\Gamma$, and $\chi^{\Gamma}(C)$ denotes the character evaluated at any representative of the conjugacy class $C$.
The bar on a tensor means evaluation on the inverse group element, which for a unitary $\Gamma$ conjugates and transposes the matrix entries, whereas the bar on a character value denotes plain complex conjugation.

\subsection{Quantum double projectors}

\subsubsection{Charge projectors at vertices}

\paragraph{Projector to trivial charge.}

The projector to the trivial charge sector at a vertex $v$ is given by
\begin{equation}
P_{v} =\frac{1}{|G|}\sum_{g}A_v^{g}.
\label{eq:P-trivial}
\end{equation}

\paragraph{Projector to a charge type.}

The projector for each anyon type is given by
\begin{equation}
P^{\Gamma}_{v} =\frac{d^{\Gamma}}{|G|}\sum_{g}\overline{\chi_{\Gamma}}(g)A_v^{g}= \frac{d^{\Gamma}}{|G|} \sum_{g \in G} \chi_{\Gamma}(\bar{g}) A_{v}^{g}.
\label{eq:P-charge-type}
\end{equation}

\paragraph{Projector to a charge internal label.}

We can resolve it into each internal state of the anyon, and there are $d^{\Gamma}$ for each anyon type
\begin{equation}
P^{\Gamma}_{i} = \frac{d^{\Gamma}}{|G|} \sum_{g \in G} \Gamma_{ii}(\bar{g}) A_{v}^{g}
\label{eq:P-internal}
\end{equation}

Above projectors can be shown to be orthogonal to each other, thus they form a POVM for the internal state space of non-Abelian charges. 

\textit{Proof.}
$$
P^{\Gamma}_{i}P^{\Gamma}_{j} = \left( \frac{d^{\Gamma}}{|G|} \right)^{2} \sum_{g,h} \Gamma_{ii} (g^{-1})\Gamma_{jj}(h^{-1}) A^{gh}_{v} 
$$
we first reindex $h\to g^{-1}h$,
$$
P^{\Gamma}_{i}P^{\Gamma}_{j} = \left( \frac{d^{\Gamma}}{|G|} \right)^{2} \sum_{g,h} \Gamma_{ii} (g^{-1})\Gamma_{jj}(h^{-1}g) A^{h}_{v} 
$$
Then we split $\Gamma_{jj}(h^{-1}g) =\sum_{k} \Gamma_{jk}(h^{-1})\Gamma_{kj}(g)$ as a homomorphism, and we massage the first term $\Gamma_{ii}(g^{-1}) = \Gamma ^{\dagger}_{ii}(g)$.
Now we can use the Great Orthogonality Theorem, which states that for unitary irrep,
$$
\sum_{g} \Gamma^{\dagger}_{ij}(g) \Gamma_{lm}(g) = \frac{|G|}{d_{\Gamma}}\delta_{il} \delta_{jm}
$$
Substituting back, we get
$$
P^{\Gamma}_{i} P^{\Gamma}_{j} = \frac{d^{\Gamma}}{|G|} \sum_{h} \sum_{k} \delta_{ik} \delta_{ij} \Gamma_{jk}(h^{-1}) A^{h}_{v} = \delta_{ij} \frac{d^{\Gamma}}{|G|} \sum_{h}  \Gamma_{ji}(h^{-1}) A^{h}_{v} = \delta_{ij} P^{\Gamma}_{j}
$$
So indeed they are orthogonal.

\subsubsection{Flux projectors at plaquettes}
\label{app:flux-projectors}

Horizontal edges are oriented to the right and vertical edges upward, and the flux through a plaquette $p$ is the holonomy $h_{p}=h_{\mathrm{bottom}}\,h_{\mathrm{right}}\,\bar{h}_{\mathrm{top}}\,\bar{h}_{\mathrm{left}}$, where edges traversed against their orientation enter through their inverse.

\paragraph{Projector to trivial flux.}
The plaquette operator $B^{e}_{p}$ projecting to the trivial flux is defined as
\begin{equation}
    B^{e}_{p}=\delta_{h_{p},e}.
\end{equation}
$B^{e}_{p}$ can be expressed in a matrix product operator (MPO) form,
\begin{align}
    B^{e}_{p}
    &= \delta_{h_{p},e}
    = \frac{1}{|G|}\sum_{\Gamma} d_{\Gamma}\,\chi^{\Gamma}(h_{p})\nonumber\\
    &= \frac{1}{|G|}\sum_{\Gamma} d_{\Gamma}\sum_{ijkl}\; \plaqClosed
    \label{eq:Be-MPO}
\end{align}
where for the second equality we use the character orthogonality relations for any finite group $G$~\cite{serre_linear_1977}; for the third equality we decompose the character function into the MPOs $\mathcal{Z}^{\Gamma}_{ij}$ ($\bar{\mathcal{Z}}^{\Gamma}_{ij}$ on edges traversed against their orientation) with support on each edge~\cite{brell_generalized_2015,lyons_protocols_2025}.
The MPOs $\mathcal{Z}^{\Gamma}_{ij}$ (or $\bar{\mathcal{Z}}^{\Gamma}_{ij}$ for anti-aligned edges) are contracted cyclically counterclockwise around the plaquette, starting from the bottom edge.

\paragraph{Projector to an arbitrary flux label.}
The plaquette operator $B^{g}_{p,v}$ projecting to a flux $g$ can be derived by modifying $B^{e}_{p}$.
Instead of evaluating the trace on the flux label $h_{p,v}$ of the plaquette (based at a vertex $v$, here chosen to be the bottom-left corner), we evaluate it on $\bar{g}h_{p,v}$; this is done by opening the trace at the base point $v$ with an insertion of $\Gamma(\bar{g})$, which modifies the delta function to $h_{p,v}=g$:
\begin{align}
    B^{g}_{p,v}
    &= \delta_{h_{p,v},g}
    = \frac{1}{|G|}\sum_{\Gamma} d_{\Gamma}\,\chi^{\Gamma}(\bar{g}h_{p,v}) \nonumber\\
    &= \frac{1}{|G|}\sum_{\Gamma} d_{\Gamma}\sum_{ijklm}[\Gamma(\bar{g})]_{mi}\; \plaqOpen
    \label{eq:Bg-MPO}
\end{align}
For $g=e$ the insertion is the identity matrix, the trace closes, and one recovers $B^{e}_{p}$ above.

\paragraph{Projector to a flux type.}
The projector $B^{C}_{p}$ onto fluxes in a conjugacy class $C$ is the sum of $B^{c}_{p,v}$ over $c\in C$.
Since we sum over all elements of the conjugacy class, the dependence on the base point drops out; the base point is only needed when keeping track of a flux label, which is not gauge invariant.
Summing the insertion $\Gamma(\bar{c})$ over the class gives a multiple of the identity by Schur's lemma~\cite{serre_linear_1977}, $\sum_{c\in C}\Gamma(\bar{c})=\frac{|C|\,\bar{\chi}^{\Gamma}(C)}{d_{\Gamma}}\,\mathds{1}$, where $\bar{\chi}^{\Gamma}(C)$ denotes the complex conjugate of the character $\chi^{\Gamma}(C)$, so the trace closes again and the class projector is the closed plaquette weighted by the complex-conjugated character:
\begin{align}
    B^{C}_{p}
    &= \sum_{c\in C}B^{c}_{p,v}
    = \frac{|C|}{|G|}\sum_{\Gamma}\bar{\chi}^{\Gamma}(C)\,\chi^{\Gamma}(h_{p})\nonumber\\
    &= \frac{|C|}{|G|}\sum_{\Gamma}\bar{\chi}^{\Gamma}(C)\sum_{ijkl}\; \plaqClosed
    \label{eq:BC-MPO}
\end{align}
Only the prefactor distinguishes $B^{C}_{p}$ from $B^{e}_{p}$, which is the special case $C=\{e\}$ with $\chi^{\Gamma}(e)=d_{\Gamma}$; the edge content is the same fixed-irrep contraction $W_{\Gamma}(p)=\chi^{\Gamma}(h_{p})$ for every class.
Therefore any transformation that is established for each fixed-irrep contraction $W_{\Gamma}(p)$ separately, as all the tensor manipulations in this appendix are, carries over from $B^{e}_{p}$ to every $B^{C}_{p}$ by reweighting the sum over irreps.
Knowing only the image of the single combination $B^{e}_{p}=\frac{1}{|G|}\sum_{\Gamma}d_{\Gamma}W_{\Gamma}(p)$ would not suffice: for $G=\mathbb{Z}_{3}$ the unitary that inverts every edge fixes $B^{e}_{p}$ but exchanges $B^{\omega}_{p}$ and $B^{\bar{\omega}}_{p}$.

\section{Gate commutation relations}

\subsection{Group multiplication with control gates}
\label{app:mult-control-comm}

When $L^g$ acts on the control qudit:

\begin{equation}
\label{eq:comm-L-control}
\vcenter{\hbox{
\begin{tikzpicture}
\begin{yquant}[horizontal]

[violet] qubit {$\ket{h_1}$} a[1];
[violet] qubit {$\ket{h_2}$} a[+1];

box {$L^g$} a[0];
box {$L$} a[1] | a[0];

[violet] output {$\ket{gh_1}$} a[0];
[violet] output {$\ket{gh_1h_2}$} a[1];

\end{yquant}
\end{tikzpicture}
}}
=
\vcenter{\hbox{
\begin{tikzpicture}
\begin{yquant}[horizontal]

[violet] qubit {} a[1];
[violet] qubit {} a[+1];

box {$L$} a[1] | a[0];
box {$L^g$} a[0];
box {$L^g$} a[1];

\end{yquant}
\end{tikzpicture}
}}, \quad \quad \quad 
\vcenter{\hbox{
\begin{tikzpicture}
\begin{yquant}[horizontal]

[violet] qubit {$\ket{h_1}$} a[1];
[violet] qubit {$\ket{h_2}$} a[+1];

box {$L^g$} a[0];
box {$R$} a[1] | a[0];

[violet] output {$\ket{gh_1}$} a[0];
[violet] output {$\ket{h_2 \bar{h}_1\bar{g}}$} a[1];

\end{yquant}
\end{tikzpicture}
}}
=
\vcenter{\hbox{
\begin{tikzpicture}
\begin{yquant}[horizontal]

[violet] qubit {} a[1];
[violet] qubit {} a[+1];

box {$R$} a[1] | a[0];
box {$L^g$} a[0];
box {$R^g$} a[1];

\end{yquant}
\end{tikzpicture}
}}
\end{equation}

When $L^g$ acts on the target qudit (note that the right panel shows that $L^g$ commute with right multiplication as expected):

\begin{equation}
\label{eq:comm-L-target}
\vcenter{\hbox{
\begin{tikzpicture}
\begin{yquant}[horizontal]

[violet] qubit {$\ket{h_1}$} a[1];
[violet] qubit {$\ket{h_2}$} a[+1];

box {$L^g$} a[1];
box {$L$} a[1] | a[0];

[violet] output {$\ket{h_1}$} a[0];
[violet] output {$\ket{h_1 gh_2}$} a[1];

\end{yquant}
\end{tikzpicture}
}}
=\sum_{h_1}
\vcenter{\hbox{
\begin{tikzpicture}
\begin{yquant}[horizontal]

[violet] qubit {} a[1];
[violet] qubit {} a[+1];

box {$L$} a[1] | a[0];
box {$T^{h_1}$} a[0];
box {$L^{h_1g \bar{h}_1}$} a[1];

\end{yquant}
\end{tikzpicture}
}}, \quad \quad \quad 
\vcenter{\hbox{
\begin{tikzpicture}
\begin{yquant}[horizontal]

[violet] qubit {$\ket{h_1}$} a[1];
[violet] qubit {$\ket{h_2}$} a[+1];

box {$L^g$} a[1];
box {$R$} a[1] | a[0];

[violet] output {$\ket{h_1}$} a[0];
[violet] output {$\ket{gh_2 \bar{h}_1}$} a[1];

\end{yquant}
\end{tikzpicture}
}}
=
\vcenter{\hbox{
\begin{tikzpicture}
\begin{yquant}[horizontal]

[violet] qubit {} a[1];
[violet] qubit {} a[+1];

box {$R$} a[1] | a[0];
box {$L^g$} a[1];

\end{yquant}
\end{tikzpicture}
}}
\end{equation}

The notation for the equation on the left is that for an operator identity. To understand it purely in terms of circuit element,  the $\sum_{h_1}$ should be grouped together with the action of $T^{h_1}$ on control qudit and $L^{h_1g \bar{h}_1}$on target qudit to form a controlled unitary gate.

When $R^g$ acts on the control qudit:

\begin{equation}
\label{eq:comm-R-control}
\vcenter{\hbox{
\begin{tikzpicture}
\begin{yquant}[horizontal]

[violet] qubit {$\ket{h_1}$} a[1];
[violet] qubit {$\ket{h_2}$} a[+1];

box {$R^g$} a[0];
box {$L$} a[1] | a[0];

[violet] output {$\ket{h_1 \bar{g}}$} a[0];
[violet] output {$\ket{h_1\bar{g}h_2}$} a[1];

\end{yquant}
\end{tikzpicture}
}}
=\sum_{h_1}
\vcenter{\hbox{
\begin{tikzpicture}
\begin{yquant}[horizontal]

[violet] qubit {} a[1];
[violet] qubit {} a[+1];

box {$L$} a[1] | a[0];
box {$T^{h_1}$} a[0];
box {$R^g$} a[0];
box {$L^{h_1 \bar{g}\bar{h}_1}$} a[1];

\end{yquant}
\end{tikzpicture}
}}, \quad \quad \quad 
\vcenter{\hbox{
\begin{tikzpicture}
\begin{yquant}[horizontal]

[violet] qubit {$\ket{h_1}$} a[1];
[violet] qubit {$\ket{h_2}$} a[+1];

box {$R^g$} a[0];
box {$R$} a[1] | a[0];

[violet] output {$\ket{h_1\bar{g}}$} a[0];
[violet] output {$\ket{h_2 g \bar{h}_1}$} a[1];

\end{yquant}
\end{tikzpicture}
}}
=\sum_{h_1}
\vcenter{\hbox{
\begin{tikzpicture}
\begin{yquant}[horizontal]

[violet] qubit {} a[1];
[violet] qubit {} a[+1];

box {$R$} a[1] | a[0];
box {$T^{h_1}$} a[0];
box {$R^g$} a[0];
box {$R^{h_1\bar{g}\bar{h}_1}$} a[1];

\end{yquant}
\end{tikzpicture}
}}
\end{equation}

When $R^g$ acts on the target qudit (note that the left panel shows that $R^g$ commute with left multiplication as expected):

\begin{equation}
\label{eq:comm-R-target}
\vcenter{\hbox{
\begin{tikzpicture}
\begin{yquant}[horizontal]

[violet] qubit {$\ket{h_1}$} a[1];
[violet] qubit {$\ket{h_2}$} a[+1];

box {$R^g$} a[1];
box {$L$} a[1] | a[0];

[violet] output {$\ket{h_1}$} a[0];
[violet] output {$\ket{h_1h_2 \bar{g}}$} a[1];

\end{yquant}
\end{tikzpicture}
}}
=
\vcenter{\hbox{
\begin{tikzpicture}
\begin{yquant}[horizontal]

[violet] qubit {} a[1];
[violet] qubit {} a[+1];

box {$L$} a[1] | a[0];
box {$R^{g}$} a[1];

\end{yquant}
\end{tikzpicture}
}}, \quad \quad \quad 
\vcenter{\hbox{
\begin{tikzpicture}
\begin{yquant}[horizontal]

[violet] qubit {$\ket{h_1}$} a[1];
[violet] qubit {$\ket{h_2}$} a[+1];

box {$R^g$} a[1];
box {$R$} a[1] | a[0];

[violet] output {$\ket{h_1}$} a[0];
[violet] output {$\ket{h_2 \bar{g} \bar{h}_1}$} a[1];

\end{yquant}
\end{tikzpicture}
}}
=\sum_{h_1}
\vcenter{\hbox{
\begin{tikzpicture}
\begin{yquant}[horizontal]

[violet] qubit {} a[1];
[violet] qubit {} a[+1];

box {$R$} a[1] | a[0];
box {$T^{h_1}$} a[0];
box {$R^{h_1 g \bar{h}_1}$} a[1];

\end{yquant}
\end{tikzpicture}
}}
\end{equation}

\paragraph{Sliding rules for vertex operators.}
\label{app:slide-rules}
For a controlled multiplication with control $c$ and target $t$, rearranging Eqs.~\eqref{eq:comm-L-control}--\eqref{eq:comm-R-target}, and conjugating with $\mathcal{S}_{c}$, which inverts the gate by Eq.~\eqref{eq:controlled-inverse} and exchanges $L^{g}_{c}$ and $R^{g}_{c}$ by Eq.~\eqref{eq:inverse-gate-def}, gives
\begin{align}
CL\,(R^{g}_{c}L^{g}_{t})\,CL^{\dagger}&=R^{g}_{c}, & CL\,L^{g}_{c}\,CL^{\dagger}&=L^{g}_{c}L^{g}_{t},\label{eq:ce-slide-CL}\\
CR\,(R^{g}_{c}R^{g}_{t})\,CR^{\dagger}&=R^{g}_{c}, & CR\,L^{g}_{c}\,CR^{\dagger}&=L^{g}_{c}R^{g}_{t},\label{eq:ce-slide-CR}\\
CL^{\dagger}(L^{g}_{c}L^{g}_{t})\,CL&=L^{g}_{c}, & CL^{\dagger}\,R^{g}_{c}\,CL&=R^{g}_{c}L^{g}_{t},\label{eq:ce-slide-CLd}\\
CR^{\dagger}(L^{g}_{c}R^{g}_{t})\,CR&=L^{g}_{c}, & CR^{\dagger}\,R^{g}_{c}\,CR&=R^{g}_{c}R^{g}_{t}.\label{eq:ce-slide-CRd}
\end{align}
Their interpretation as the elementary move of Sec.~\ref{subsec:elementary-move} is given in Appendix~\ref{app:gr-move}.

\subsection{Group multiplication with diagonal gates}

Since $L^g$ and $R^g$ permute the group basis, conjugating a diagonal gate by them reshuffles the weights $[\Gamma(h)]_{ij}$. Because $\Gamma$ is a homomorphism, the result is again a $\mathcal{Z}^\Gamma$ tensor, with the group multiplication converted into a $c$-number rotation of one bond index.

\paragraph{Left multiplication.}
\begin{equation}
    L^{g}\,\mathcal{Z}^\Gamma_{ij}\,(L^{g})^\dagger = \sum_k [\Gamma(\bar{g})]_{ik}\,\mathcal{Z}^\Gamma_{kj}
    \label{eq:L-Z}
\end{equation}
\begin{proof}
\begin{align}
L^{g}\,\mathcal{Z}^\Gamma_{ij}\,(L^{g})^\dagger
&= \sum_h [\Gamma(h)]_{ij}\,\ketbra{gh} 
= \sum_h [\Gamma(\bar{g} h)]_{ij}\,\ketbra{h} \nn
&= \sum_h \sum_k [\Gamma(\bar{g})]_{ik}[\Gamma(h)]_{kj}\,\ketbra{h}
= \sum_k [\Gamma(\bar{g})]_{ik}\,\mathcal{Z}^\Gamma_{kj},
\end{align}
where in the second equality we reindex $h \to \bar{g} h$.
\end{proof}

\paragraph{Right multiplication.}
\begin{equation}
    R^{g}\,\mathcal{Z}^\Gamma_{ij}\,(R^{g})^\dagger = \sum_k \mathcal{Z}^\Gamma_{ik}\,[\Gamma(g)]_{kj}
    \label{eq:R-Z}
\end{equation}
\begin{proof}
\begin{align}
R^{g}\,\mathcal{Z}^\Gamma_{ij}\,(R^{g})^\dagger
&= \sum_h [\Gamma(h)]_{ij}\,\ketbra{h\bar{g}}
= \sum_h [\Gamma(hg)]_{ij}\,\ketbra{h} \nn
&= \sum_h \sum_k [\Gamma(h)]_{ik}[\Gamma(g)]_{kj}\,\ketbra{h}
= \sum_k \mathcal{Z}^\Gamma_{ik}\,[\Gamma(g)]_{kj},
\end{align}
where in the second equality we reindex $h \to hg$.
\end{proof}

In matrix notation on the bond indices, Eqs.~\eqref{eq:L-Z} and \eqref{eq:R-Z} read $L^g \mathcal{Z}^\Gamma (L^g)^\dagger = \Gamma(\bar{g})\cdot\mathcal{Z}^\Gamma$ and $R^g \mathcal{Z}^\Gamma (R^g)^\dagger = \mathcal{Z}^\Gamma\cdot\Gamma(g)$: left (right) multiplication on the physical leg rotates the left (right) virtual leg of the MPO tensor.

\subsection{Controlled group multiplication with diagonal gates}
\label{app:ctrl-diag}

Throughout, $c$ denotes the control register and $t$ the target register, and an operator with a register subscript acts only on that register.
The controlled multiplication gates $CL$ and $CR$ are defined in Appendix~\ref{app:mult-gates}, and the diagonal gates $\mathcal{Z}^\Gamma$ and $\bar{\mathcal{Z}}^\Gamma$ in Appendix~\ref{app:diagonal-gates}.
Conjugating a diagonal gate on the target by $CL$ or $CR$ applies Eq.~\eqref{eq:L-Z} or \eqref{eq:R-Z} with $g$ replaced by the control's group element. The $c$-number rotation $\Gamma(g)$ then becomes a diagonal operator on the control, so the single tensor on the target \textit{splits} into a contracted pair straddling both registers: controlled multiplication elongates a charge string by one site.

\paragraph{$\mathcal{Z}^\Gamma$ on target, conjugated by $CL$.}
\begin{equation}
    CL\;\mathcal{Z}^\Gamma_{ij,t}\;CL^\dagger = \sum_k \bar{\mathcal{Z}}^{\Gamma}_{ik,c}\;\mathcal{Z}^\Gamma_{kj,t}
    \label{eq:CL-Z}
\end{equation}
\begin{proof}
\begin{align}
CL\;\mathcal{Z}^\Gamma_{ij,t}\;CL^\dagger
&= \sum_{g,h} [\Gamma(h)]_{ij}\, \ketbra{g}_c \otimes \ketbra{gh}_t 
= \sum_{g,h}[\Gamma(\bar{g} h)]_{ij}\,\ketbra{g}_c \otimes \ketbra{h}_t \nn
&= \sum_{k} \sum_g [\Gamma(\bar{g})]_{ik}\,\ketbra{g}_c\;\sum_h [\Gamma(h)]_{kj}\,\ketbra{h}_t
= \sum_k \bar{\mathcal{Z}}^{\Gamma}_{ik,c}\;\mathcal{Z}^\Gamma_{kj,t},
\end{align}
where in the second equality we reindex $h \to \bar{g} h$ on the target.
\end{proof}

\paragraph{$\mathcal{Z}^\Gamma$ on target, conjugated by $CR$.}
\begin{equation}
    CR\;\mathcal{Z}^\Gamma_{ij,t}\;CR^\dagger = \sum_k \mathcal{Z}^\Gamma_{ik,t}\;\mathcal{Z}^\Gamma_{kj,c}
    \label{eq:CR-Z}
\end{equation}
\begin{proof}
\begin{align}
CR\;\mathcal{Z}^\Gamma_{ij,t}\;CR^\dagger
&= \sum_{g,h} [\Gamma(h)]_{ij}\, \ketbra{g}_c \otimes \ketbra{h\bar{g}}_t
= \sum_{g,h}[\Gamma(hg)]_{ij}\,\ketbra{g}_c \otimes \ketbra{h}_t \nn
&= \sum_k \sum_h [\Gamma(h)]_{ik}\,\ketbra{h}_t\;\sum_g [\Gamma(g)]_{kj}\,\ketbra{g}_c
= \sum_k \mathcal{Z}^\Gamma_{ik,t}\;\mathcal{Z}^\Gamma_{kj,c},
\end{align}
where in the second equality we reindex $h \to hg$ on the target.
\end{proof}

\paragraph{$\bar{\mathcal{Z}}^\Gamma$ on target, conjugated by $CL$.}
\begin{equation}
    CL\;\bar{\mathcal{Z}}^\Gamma_{ij,t}\;CL^\dagger = \sum_k \bar{\mathcal{Z}}^{\Gamma}_{ik,t}\;\mathcal{Z}^\Gamma_{kj,c}
    \label{eq:CL-Zbar}
\end{equation}
Relative to Eq.~\eqref{eq:CL-Z}, the tensor emitted onto the control is unbarred and carries the \textit{right} bond index: a $\bar{\mathcal{Z}}^\Gamma$ string is oriented oppositely, so it elongates on the other end.
\begin{proof}
\begin{align}
CL\;\bar{\mathcal{Z}}^\Gamma_{ij,t}\;CL^\dagger
&= \sum_{g,h} [\Gamma(\bar{h})]_{ij}\, \ketbra{g}_c \otimes \ketbra{gh}_t
= \sum_{g,h}[\Gamma(\bar{h} g)]_{ij}\,\ketbra{g}_c \otimes \ketbra{h}_t \nn
&= \sum_{k} \sum_h [\Gamma(\bar{h})]_{ik}\,\ketbra{h}_t\;\sum_g [\Gamma(g)]_{kj}\,\ketbra{g}_c
= \sum_k \bar{\mathcal{Z}}^{\Gamma}_{ik,t}\;\mathcal{Z}^\Gamma_{kj,c},
\end{align}
where in the second equality we reindex $h \to \bar{g} h$ and use $\overline{\bar{g} h} = \bar{h} g$.
\end{proof}

\paragraph{$\bar{\mathcal{Z}}^\Gamma$ on target, conjugated by $CR$.}
\begin{equation}
    CR\;\bar{\mathcal{Z}}^\Gamma_{ij,t}\;CR^\dagger = \sum_k \bar{\mathcal{Z}}^{\Gamma}_{ik,c}\;\bar{\mathcal{Z}}^\Gamma_{kj,t}
    \label{eq:CR-Zbar}
\end{equation}
\begin{proof}
\begin{align}
CR\;\bar{\mathcal{Z}}^\Gamma_{ij,t}\;CR^\dagger
&= \sum_{g,h} [\Gamma(\bar{h})]_{ij}\, \ketbra{g}_c \otimes \ketbra{h\bar{g}}_t
= \sum_{g,h}[\Gamma(\bar{g}\bar{h})]_{ij}\,\ketbra{g}_c \otimes \ketbra{h}_t \nn
&= \sum_{k} \sum_g [\Gamma(\bar{g})]_{ik}\,\ketbra{g}_c\;\sum_h [\Gamma(\bar{h})]_{kj}\,\ketbra{h}_t
= \sum_k \bar{\mathcal{Z}}^{\Gamma}_{ik,c}\;\bar{\mathcal{Z}}^\Gamma_{kj,t},
\end{align}
where in the second equality we reindex $h \to hg$ and use $\overline{hg} = \bar{g}\bar{h}$.
\end{proof}

\paragraph{$\mathcal{Z}^\Gamma$ on control, conjugated by $CL$ or $CR$.}
\label{app:Z-on-control}
A diagonal gate sitting on the control is left invariant:
\begin{equation}
    CL\;\mathcal{Z}^\Gamma_{ij,c}\;CL^\dagger = \mathcal{Z}^\Gamma_{ij,c}
\end{equation}
\begin{proof}
$\mathcal{Z}^\Gamma_{ij}$ is diagonal in the group basis, so it commutes with the control projector $\ketbra{g}_c$; the target factor $L^g_t$ acts on a different register. Hence $CL$ commutes with $\mathcal{Z}^\Gamma_{ij,c}$ and conjugation is trivial.
\end{proof}
The same argument applies verbatim to $\bar{\mathcal{Z}}^\Gamma_{ij,c}$ and to conjugation by $CR$: any diagonal operator on the control is invariant under all controlled multiplications. Operationally this fixes the directionality of the splitting moves above: controlled multiplication transports charge tensors from the target onto the control, never the other way.

\paragraph{$\sum_j \mathcal{Z}^\Gamma_{ij,c}\mathcal{Z}^\Gamma_{jk,t}$, conjugated by $CL$.}
The remaining four identities are absorption moves: running the splitting moves backwards, a contracted pair of tensors straddling both registers retracts onto a single tensor on the target, provided the control tensor is the inverse (in the sense of Eq.~\eqref{eq:Z-orthogonality}) of the one the corresponding splitting move would emit.
\begin{equation}
    CL\;\sum_{j}\mathcal{Z}^\Gamma_{ij,c}\, \mathcal{Z}^\Gamma_{jk,t}\;CL^\dagger = \mathcal{Z}^\Gamma_{ik,t}
    \label{eq:CL-absorb-Z}
\end{equation}
\begin{proof}
\begin{align}
CL\;\sum_{j}\mathcal{Z}^\Gamma_{ij,c}\, \mathcal{Z}^\Gamma_{jk,t}\;CL^\dagger
&= \sum_{g,h}\sum_{j} [\Gamma(g)]_{ij} [\Gamma(h)]_{jk}\, \ketbra{g}_c \otimes \ketbra{gh}_t
= \sum_{g,h}\sum_{j}[\Gamma(g)]_{ij}[\Gamma(\bar{g} h)]_{jk}\,\ketbra{g}_c \otimes \ketbra{h}_t \nn
&= \sum_{g,h} [\Gamma(h)]_{ik}\,\ketbra{g}_c \otimes \ketbra{h}_t
= \mathcal{Z}^\Gamma_{ik,t},
\end{align}
where in the second equality we reindex $h \to \bar{g} h$ on the target, and the sum over $g$ on the control gives the identity.
\end{proof}
Consistency check: this also follows from the splitting rules. The control factor is invariant, the target factor splits as in Eq.~\eqref{eq:CL-Z}, $\mathcal{Z}^\Gamma_{jk,t} \to \sum_m \bar{\mathcal{Z}}^\Gamma_{jm,c}\,\mathcal{Z}^\Gamma_{mk,t}$, and the orthogonality property \eqref{eq:Z-orthogonality} on the control contracts $\sum_j \mathcal{Z}^\Gamma_{ij,c}\bar{\mathcal{Z}}^\Gamma_{jm,c} = \delta_{im}$.

\paragraph{$\sum_j \bar{\mathcal{Z}}^\Gamma_{ij,t}\bar{\mathcal{Z}}^\Gamma_{jk,c}$, conjugated by $CL$.}
\begin{equation}
    CL\;\sum_{j}\bar{\mathcal{Z}}^\Gamma_{ij,t}\, \bar{\mathcal{Z}}^\Gamma_{jk,c}\;CL^\dagger = \bar{\mathcal{Z}}^\Gamma_{ik,t}
    \label{eq:CL-absorb-Zbar}
\end{equation}
Note the mirrored register assignment relative to Eq.~\eqref{eq:CL-absorb-Z}: the left index $i$ now lives on the target and the right index $k$ on the control, consistent with the opposite orientation of the $\bar{\mathcal{Z}}^\Gamma$ string.
\begin{proof}
\begin{align}
CL\;\sum_{j}\bar{\mathcal{Z}}^\Gamma_{ij,t}\, \bar{\mathcal{Z}}^\Gamma_{jk,c}\;CL^\dagger
&= \sum_{g,h}\sum_{j} [\Gamma(\bar{h})]_{ij} [\Gamma(\bar{g})]_{jk}\, \ketbra{g}_c \otimes \ketbra{gh}_t
= \sum_{g,h}\sum_{j}[\Gamma(\bar{h} g)]_{ij}[\Gamma(\bar{g})]_{jk}\,\ketbra{g}_c \otimes \ketbra{h}_t \nn
&= \sum_{g,h} [\Gamma(\bar{h} g \bar{g})]_{ik}\,\ketbra{g}_c \otimes \ketbra{h}_t
= \bar{\mathcal{Z}}^\Gamma_{ik,t},
\end{align}
where in the second equality we reindex $h \to \bar{g} h$ on the target and use $\overline{\bar{g} h} = \bar{h} g$.
\end{proof}

\paragraph{$\sum_j \mathcal{Z}^\Gamma_{ij,t}\bar{\mathcal{Z}}^\Gamma_{jk,c}$, conjugated by $CR$.}
\begin{equation}
    CR\;\sum_{j}\mathcal{Z}^\Gamma_{ij,t}\, \bar{\mathcal{Z}}^\Gamma_{jk,c}\;CR^\dagger = \mathcal{Z}^\Gamma_{ik,t}
    \label{eq:CR-absorb-Z}
\end{equation}
Unlike the $CL$ case, the control tensor here is barred while the result is unbarred: since $CR$ emits an unbarred $\mathcal{Z}^\Gamma_{mj,c}$ when splitting (Eq.~\eqref{eq:CR-Z}), the absorbing tensor on the control must be $\bar{\mathcal{Z}}^\Gamma$ for the orthogonality contraction to fire.
\begin{proof}
\begin{align}
CR\;\sum_{j}\mathcal{Z}^\Gamma_{ij,t}\, \bar{\mathcal{Z}}^\Gamma_{jk,c}\;CR^\dagger
&= \sum_{g,h}\sum_{j} [\Gamma(h)]_{ij} [\Gamma(\bar{g})]_{jk}\, \ketbra{g}_c \otimes \ketbra{h\bar{g}}_t
= \sum_{g,h}\sum_{j}[\Gamma(hg)]_{ij}[\Gamma(\bar{g})]_{jk}\,\ketbra{g}_c \otimes \ketbra{h}_t \nn
&= \sum_{g,h} [\Gamma(hg\bar{g})]_{ik}\,\ketbra{g}_c \otimes \ketbra{h}_t
= \mathcal{Z}^\Gamma_{ik,t},
\end{align}
where in the second equality we reindex $h \to hg$ on the target.
\end{proof}

\paragraph{$\sum_j \mathcal{Z}^\Gamma_{ij,c}\bar{\mathcal{Z}}^\Gamma_{jk,t}$, conjugated by $CR$.}
\begin{equation}
    CR\;\sum_{j}\mathcal{Z}^\Gamma_{ij,c}\, \bar{\mathcal{Z}}^\Gamma_{jk,t}\;CR^\dagger = \bar{\mathcal{Z}}^\Gamma_{ik,t}
    \label{eq:CR-absorb-Zbar}
\end{equation}
This is the barred-string absorption for $CR$: the same mixed $\mathcal{Z}^\Gamma\bar{\mathcal{Z}}^\Gamma$ combination as in Eq.~\eqref{eq:CR-absorb-Z}, but with the registers swapped, retracts a $\bar{\mathcal{Z}}^\Gamma$ string instead.
\begin{proof}
\begin{align}
CR\;\sum_{j}\mathcal{Z}^\Gamma_{ij,c}\, \bar{\mathcal{Z}}^\Gamma_{jk,t}\;CR^\dagger
&= \sum_{g,h}\sum_{j} [\Gamma(g)]_{ij} [\Gamma(\bar{h})]_{jk}\, \ketbra{g}_c \otimes \ketbra{h\bar{g}}_t
= \sum_{g,h}\sum_{j}[\Gamma(g)]_{ij}[\Gamma(\bar{g} \bar{h})]_{jk}\,\ketbra{g}_c \otimes \ketbra{h}_t \nn
&= \sum_{g,h} [\Gamma(g\bar{g}\bar{h})]_{ik}\,\ketbra{g}_c \otimes \ketbra{h}_t
= \bar{\mathcal{Z}}^\Gamma_{ik,t},
\end{align}
where in the second equality we reindex $h \to hg$ on the target and use $\overline{hg} = \bar{g}\bar{h}$.
\end{proof}

\paragraph{Summary of splitting and absorption moves.}

\begin{table}[H]
\centering
\renewcommand{\arraystretch}{1.6}
\begin{tabular}{c|c|c}
\hline\hline
Conjugation & Input & Output \\
\hline
$CL$ & $\mathcal{Z}^\Gamma_{ij,t}$ & $\sum_k \bar{\mathcal{Z}}^\Gamma_{ik,c}\,\mathcal{Z}^\Gamma_{kj,t}$ \\
$CR$ & $\mathcal{Z}^\Gamma_{ij,t}$ & $\sum_k \mathcal{Z}^\Gamma_{ik,t}\,\mathcal{Z}^\Gamma_{kj,c}$ \\
$CL$ & $\bar{\mathcal{Z}}^\Gamma_{ij,t}$ & $\sum_k \bar{\mathcal{Z}}^\Gamma_{ik,t}\,\mathcal{Z}^\Gamma_{kj,c}$ \\
$CR$ & $\bar{\mathcal{Z}}^\Gamma_{ij,t}$ & $\sum_k \bar{\mathcal{Z}}^\Gamma_{ik,c}\,\bar{\mathcal{Z}}^\Gamma_{kj,t}$ \\
\hline\hline
\end{tabular}
\caption{Splitting moves: a single tensor on the target elongates across both registers.}
\label{tab:splitting}
\end{table}

\begin{table}[H]
\centering
\renewcommand{\arraystretch}{1.6}
\begin{tabular}{c|c|c}
\hline\hline
Conjugation & Input & Output \\
\hline
$CL$ & $\sum_j \mathcal{Z}^\Gamma_{ij,c}\,\mathcal{Z}^\Gamma_{jk,t}$ & $\mathcal{Z}^\Gamma_{ik,t}$ \\
$CL$ & $\sum_j \bar{\mathcal{Z}}^\Gamma_{ij,t}\,\bar{\mathcal{Z}}^\Gamma_{jk,c}$ & $\bar{\mathcal{Z}}^\Gamma_{ik,t}$ \\
$CR$ & $\sum_j \mathcal{Z}^\Gamma_{ij,t}\,\bar{\mathcal{Z}}^\Gamma_{jk,c}$ & $\mathcal{Z}^\Gamma_{ik,t}$ \\
$CR$ & $\sum_j \mathcal{Z}^\Gamma_{ij,c}\,\bar{\mathcal{Z}}^\Gamma_{jk,t}$ & $\bar{\mathcal{Z}}^\Gamma_{ik,t}$ \\
\hline\hline
\end{tabular}
\caption{Absorption moves: a contracted pair across both registers retracts onto the target.}
\label{tab:absorption}
\end{table}

Reading rule: the control tensor in each row of Table~\ref{tab:absorption} is exactly the inverse, under the orthogonality property \eqref{eq:Z-orthogonality}, of the tensor that the corresponding row of Table~\ref{tab:splitting} emits onto the control. Each absorption is the splitting run backwards, so splitting followed by absorption with the matching controlled multiplication is the identity move on the string. The combinations are not symmetric under freely swapping registers: for instance $\sum_j \bar{\mathcal{Z}}^\Gamma_{ij,c}\,\mathcal{Z}^\Gamma_{jk,t}$ does not collapse under $CL$, because the emitted control tensor $\bar{\mathcal{Z}}^\Gamma$ then contracts on the wrong side.

\paragraph{Inverted gates.}
\label{app:ctrl-inverse-Z}
Conjugating Eqs.~\eqref{eq:CL-Z}--\eqref{eq:CR-Zbar} with $\mathcal{S}_{c}$, which inverts the gate by Eq.~\eqref{eq:controlled-inverse} and exchanges $\mathcal{Z}^{\Gamma}_{c}$ and $\bar{\mathcal{Z}}^{\Gamma}_{c}$ by Eq.~\eqref{eq:Zbar-def}, gives the splitting moves of the inverted gates,
\begin{align}
CL^{\dagger}\,\mathcal{Z}^{\Gamma}_{ij,t}\,CL&=\sum_{k}\mathcal{Z}^{\Gamma}_{ik,c}\,\mathcal{Z}^{\Gamma}_{kj,t}, &
CL^{\dagger}\,\bar{\mathcal{Z}}^{\Gamma}_{ij,t}\,CL&=\sum_{k}\bar{\mathcal{Z}}^{\Gamma}_{ik,t}\,\bar{\mathcal{Z}}^{\Gamma}_{kj,c},\nonumber\\
CR^{\dagger}\,\mathcal{Z}^{\Gamma}_{ij,t}\,CR&=\sum_{k}\mathcal{Z}^{\Gamma}_{ik,t}\,\bar{\mathcal{Z}}^{\Gamma}_{kj,c}, &
CR^{\dagger}\,\bar{\mathcal{Z}}^{\Gamma}_{ij,t}\,CR&=\sum_{k}\mathcal{Z}^{\Gamma}_{ik,c}\,\bar{\mathcal{Z}}^{\Gamma}_{kj,t}.\label{eq:ctrl-inverse-Z}
\end{align}
That is, $CL^{\dagger}$ and $CR^{\dagger}$ act as $CL$ and $CR$ in Table~\ref{tab:splitting} with the bar on the tensor emitted onto the control toggled, and the same holds for the absorption moves of Table~\ref{tab:absorption}.

\section{Graph rewiring stabilizer derivation}

\subsection{The elementary move}
\label{app:gr-move}

We show that the controlled multiplication $U$ assigned by Fig.~\ref{fig:gr-four-cases} maps every vertex and plaquette stabilizer of the old graph to the corresponding stabilizer of the new graph.
In the notation of Sec.~\ref{subsec:elementary-move}, the control edge $c$ joins $u$ and $v$, the target edge $t$ joins $w$ and $v$, and the move reattaches $t$ to $u$ as the edge $t'$.

\paragraph{Aligned case.}
Figure~\ref{fig:gr-sliding} shows the action of the move on the stabilizers in case (a) of Fig.~\ref{fig:gr-four-cases}, $c$ incoming and $t$ outgoing at $v$, where the gate is $CL$.
The vertex operators transform by the sliding rule of Eq.~\eqref{eq:ce-slide-CL}: $v$ loses the target edge and $u$ gains it, while $w$, which acts on $t$ by $R^{g}_{t}$, is unchanged by Eq.~\eqref{eq:comm-R-target}, and every other vertex operator has no support on $c$ or $t$.
The plaquette operators transform by the splitting moves of Table~\ref{tab:splitting} and, in the plaquette containing both $c$ and $t$, the absorption moves of Table~\ref{tab:absorption}; since diagonal operators on the control are invariant (Appendix~\ref{app:Z-on-control}), only the two plaquettes containing $t$ change.
The result is the set of vertex and plaquette operators of the rewired graph.
Since the charge and flux projectors of Eqs.~\eqref{eq:charge-projector} and \eqref{eq:flux-projector} differ from $A^{g}_{v}$ and $B^{e}_{p}$ only by $c$-number weights, $U$ also maps every state of definite anyon content on the old graph to that on the new one.

\paragraph{The other three cases from the inverse gate.}
The inverse gate $\mathcal{S}_{e}$ reverses the orientation of edge $e$ in every stabilizer, exchanging $L^{g}_{e}\leftrightarrow R^{g}_{e}$ and $\mathcal{Z}^{\Gamma}_{e}\leftrightarrow\bar{\mathcal{Z}}^{\Gamma}_{e}$ [Eqs.~\eqref{eq:LR-def}, \eqref{eq:inverse-gate-def} and \eqref{eq:Zbar-def}].
Conjugating Fig.~\ref{fig:gr-sliding} with $\mathcal{S}_{c}$, $\mathcal{S}_{t}$ or both therefore gives the remaining cases of Fig.~\ref{fig:gr-four-cases}, with $CL\to\mathcal{S}_{c}CL\mathcal{S}_{c}=CL^{\dagger}$ [case (d), Eq.~\eqref{eq:controlled-inverse}], $CL\to\mathcal{S}_{t}CL\mathcal{S}_{t}=CR$ [case (c)] and $CL\to CR^{\dagger}$ [case (b)]; the same conjugations turn Eq.~\eqref{eq:ce-slide-CL} into Eqs.~\eqref{eq:ce-slide-CR}--\eqref{eq:ce-slide-CRd} and toggle the bar on the inserted plaquette tensor, Eq.~\eqref{eq:ctrl-inverse-Z}.
With the opposite convention to Eq.~\eqref{eq:Av-convention}, $CL\leftrightarrow CR^{\dagger}$ and $CR\leftrightarrow CL^{\dagger}$ are exchanged.

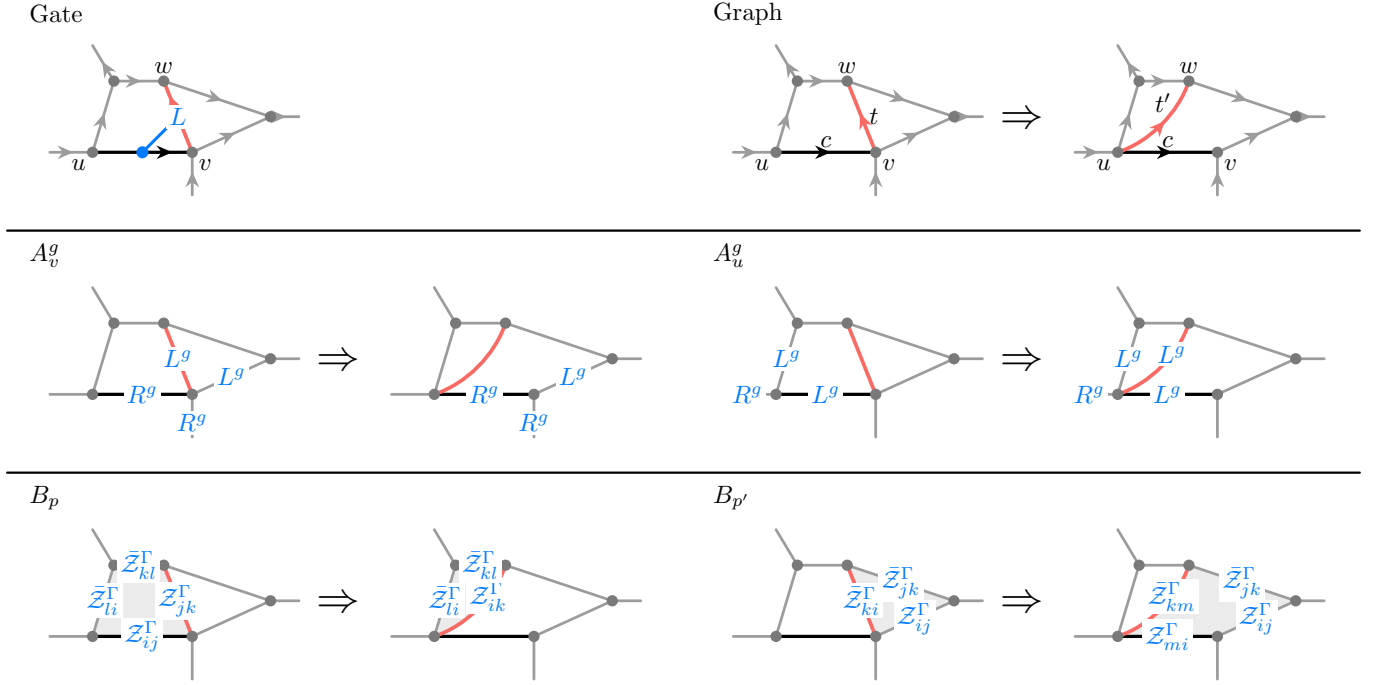
\begin{figure}[ht]
\centering
\adjustbox{max width=\columnwidth}{%
\begin{tikzpicture}[
    line cap=round, line join=round,
    lat/.style ={latgray, line width=1.1pt},
    cem/.style ={black, line width=1.3pt},
    tem/.style ={edgered, line width=1.5pt},
    ga/.style  ={postaction={decorate, decoration={markings,
                 mark=at position #1 with {\arrow{Stealth[length=2.2mm,width=1.9mm]}}}}},
    ga/.default=0.55,
    vx/.style  ={circle, fill=gray!95!black, inner sep=0pt, minimum size=4.5pt},
    vl/.style  ={black, font=\normalsize, inner sep=1pt},
    mv/.style  ={opblue, line width=1.2pt},
    ml/.style  ={opblue, font=\normalsize, fill=white, inner sep=1.2pt},
    xl/.style  ={opblue, font=\normalsize, fill=white, inner sep=1.2pt},
    hd/.style  ={black, font=\normalsize},
    imp/.style ={black, font=\LARGE},
    rule/.style={black, line width=0.9pt},
  ]
  \def\colpitch{9.6}
  \def\rowpitch{3.4}
  \def\xB{0.7}
  \def\xI{4.15}
  \def\xA{5.5}
  \newcommand{\slPatch}[1]{%
    \draw[lat,#1] (-0.6,0)--(0,0);
    \draw[lat,#1] (0,0)--(0.3,1);
    \draw[lat,#1] (0.3,1)--(0,1.5);
    \draw[lat,#1] (0.3,1)--(1,1);
    \draw[lat,#1] (1,1)--(2.5,0.5);
    \draw[lat,#1] (1.4,0)--(2.5,0.5);
    \draw[lat,#1] (2.5,0.5)--(2.9,0.5);
    \draw[lat,#1] (1.4,-0.6)--(1.4,0);}
  \newcommand{\slDots}{\node[vx] at (0,0) {}; \node[vx] at (1.4,0) {}; \node[vx] at (1,1) {};
    \node[vx] at (0.3,1) {}; \node[vx] at (2.5,0.5) {};}
  \newcommand{\slVerts}{%
    \slDots
    \node[vl, below left=1.5pt]  at (0,0) {$u$};
    \node[vl, below right=1.5pt] at (1.4,0) {$v$};
    \node[vl, above=2pt] at (1,1) {$w$};}
  \newcommand{\slC}[1]{\draw[cem,#1] (0,0)--(1.4,0);}
  \newcommand{\slT}[1]{\draw[tem,#1] (1.4,0)--(1,1);}
  \newcommand{\slTp}[1]{\draw[tem,#1] (0,0) .. controls (0.45,0.15) and (0.85,0.55) .. (1,1);}
  \newcommand{\slBefore}{\slPatch{} \slC{} \slT{} \slDots}
  \newcommand{\slAfter}{\slPatch{} \slC{} \slTp{} \slDots}
  \newcommand{\slShadeP}{\fill[plaqgray] (0,0)--(1.4,0)--(1,1)--(0.3,1)--cycle;}
  \newcommand{\slShadePA}{\fill[plaqgray] (0,0) .. controls (0.45,0.15) and (0.85,0.55) .. (1,1) -- (0.3,1) -- cycle;}
  \newcommand{\slShadeQ}{\fill[plaqgray] (1.4,0)--(2.5,0.5)--(1,1)--cycle;}
  \newcommand{\slShadeQA}{\fill[plaqgray] (1.4,0)--(2.5,0.5)--(1,1) .. controls (0.85,0.55) and (0.45,0.15) .. (0,0) -- cycle;}
  \newcommand{\slCell}[3]{%
    \node[hd, anchor=west] at (-0.3,1.95) {#1};
    \begin{scope}[shift={(\xB,0)}] #2 \end{scope}
    \node[imp] at (\xI,0.45) {$\Rightarrow$};
    \begin{scope}[shift={(\xA,0)}] #3 \end{scope}}

  \node[hd, anchor=west] at (-0.3,1.95) {Gate};
  \begin{scope}[shift={(\xB,0)}]
    \slPatch{ga} \slC{ga=0.8} \slT{ga=0.8}
    \draw[mv] (0.7,0)--(1.2,0.5);
    \fill[opblue] (0.7,0) circle (2.4pt);
    \node[ml] at (1.2,0.5) {$L$};
    \slVerts
  \end{scope}
  \begin{scope}[shift={(\colpitch,0)}]
    \slCell{Graph}
      {\slPatch{ga} \slC{ga} \slT{ga} \slVerts
       \node[vl, above=1pt] at (0.7,0) {$c$}; \node[vl, right=2pt] at (1.2,0.5) {$t$};}
      {\slPatch{ga} \slC{ga} \slTp{ga} \slVerts
       \node[vl, above=1pt] at (0.7,0) {$c$}; \node[vl, above left=-1pt] at (0.76,0.56) {$t'$};}
  \end{scope}

  \begin{scope}[shift={(0,-\rowpitch)}]
    \slCell{$A^{g}_{v}$}
      {\slBefore
       \node[xl] at (0.7,0) {$R^{g}$}; \node[xl] at (1.2,0.5) {$L^{g}$};
       \node[xl] at (1.95,0.25) {$L^{g}$}; \node[xl] at (1.4,-0.4) {$R^{g}$};}
      {\slAfter
       \node[xl] at (0.7,0) {$R^{g}$};
       \node[xl] at (1.95,0.25) {$L^{g}$}; \node[xl] at (1.4,-0.4) {$R^{g}$};}
  \end{scope}
  \begin{scope}[shift={(\colpitch,-\rowpitch)}]
    \slCell{$A^{g}_{u}$}
      {\slBefore
       \node[xl] at (0.7,0) {$L^{g}$}; \node[xl] at (0.15,0.5) {$L^{g}$};
       \node[xl] at (-0.4,0) {$R^{g}$};}
      {\slAfter
       \node[xl] at (0.7,0) {$L^{g}$}; \node[xl] at (0.15,0.5) {$L^{g}$};
       \node[xl] at (-0.4,0) {$R^{g}$};
       \node[xl] at (0.76,0.56) {$L^{g}$};}
  \end{scope}

  \begin{scope}[shift={(0,-2*\rowpitch)}]
    \slCell{$B_{p}$}
      {\slShadeP \slBefore
       \node[xl] at (0.7,0) {\zg{ij}}; \node[xl] at (1.2,0.5) {\zg{jk}};
       \node[xl] at (0.65,1) {\zgb{kl}}; \node[xl] at (0.15,0.5) {\zgb{li}};}
      {\slShadePA \slAfter
       \node[xl] at (0.76,0.56) {\zg{ik}};
       \node[xl] at (0.65,1) {\zgb{kl}}; \node[xl] at (0.15,0.5) {\zgb{li}};}
  \end{scope}
  \begin{scope}[shift={(\colpitch,-2*\rowpitch)}]
    \slCell{$B_{p'}$}
      {\slShadeQ \slBefore
       \node[xl] at (1.95,0.25) {\zg{ij}}; \node[xl] at (1.75,0.75) {\zgb{jk}};
       \node[xl] at (1.2,0.5) {\zgb{ki}};}
      {\slShadeQA \slAfter
       \node[xl] at (1.95,0.25) {\zg{ij}}; \node[xl] at (1.75,0.75) {\zgb{jk}};
       \node[xl] at (0.76,0.56) {\zgb{km}}; \node[xl] at (0.7,0) {\zg{mi}};}
  \end{scope}

  \foreach \r in {1,2}{\draw[rule] (-0.5,{-\r*\rowpitch+2.3}) -- ({\colpitch+8.9},{-\r*\rowpitch+2.3});}
\end{tikzpicture}}
\caption{Sliding rules for one elementary move on an arbitrary planar graph, in the aligned case (a) of Fig.~\ref{fig:gr-four-cases}: the control $c$ (black) is incoming and the target $t$ (red) outgoing at $v$, and the gate is $CL$ (Gate), which detaches $t$ from $v$ and reattaches it to $u$ as $t'$ (Graph).
$A^{g}_{v}$, $A^{g}_{u}$: the vertex at which $c$ and $t$ meet loses the target edge and the vertex at the other end of $c$ gains it, Eq.~\eqref{eq:ce-slide-CL}; every other vertex operator is unchanged.
$B_{p}$, $B_{p'}$: the plaquette operators, drawn as the fixed-irrep contractions of Eq.~\eqref{eq:Be-MPO} with the prefactor suppressed and contracted counterclockwise, change only in the two plaquettes containing $t$.
The gate inserts a tensor of $c$ next to the tensor of $t$ (Table~\ref{tab:splitting}); in $p$, which contains both edges, the inserted tensor contracts with the tensor already on $c$ to the identity (Table~\ref{tab:absorption}), so $p$ loses its support on $c$ and one side (here a trapezoid becoming a triangle), while $p'$ keeps the inserted tensor and gains $c$ as a side (a triangle becoming a quadrilateral).
These are the faces of the rewired graph.}
\label{fig:gr-sliding}
\end{figure}

\subsection{Charge exchange circuit}
\label{app:ce}

We track the stabilizers through the charge exchange circuit of Sec.~\ref{subsec:charge-exchange} in Fig.~\ref{fig:ce-stabilizers}.
The direct edge $\ell$ points from $v_{1}$ to $v_{2}$, the other three edges at each endpoint are labelled by their direction on the square lattice, and the six plaquettes are numbered as in the first row of the figure, with $2$ and $5$ flanking $\ell$.
The first row lists the initial vertex and plaquette operators of Eqs.~\eqref{eq:Av-convention} and \eqref{eq:Be-MPO}; in particular $v_{1}$ acts on $\ell$ by $L^{g}_{\ell}$ and $v_{2}$ by $R^{g}_{\ell}$.

\subsubsection{Stabilizer tracking}
\label{app:ce-tracking}

Figure~\ref{fig:ce-stabilizers} applies the sliding rules of Appendix~\ref{app:gr-move} to the seven elementary moves of Fig.~\ref{fig:charge_exchange_circuit}, with the direct edge $\ell$ as the control throughout.
The only step that is not a slide is the inverse gate $\mathcal{S}_{\ell}$ at move 5, which reverses the orientation of $\ell$ in every stabilizer by Eqs.~\eqref{eq:inverse-gate-def} and \eqref{eq:Zbar-def}; since $v_{1}$ then holds exactly the three edges above $\ell$ and $v_{2}$ the three below, the graph can be redrawn with the two endpoints exchanged without any edge crossing another, which changes no stabilizer.
The last row is Eq.~\eqref{eq:ce-swap}: the two vertex operators are exchanged and the six plaquette operators are unchanged.

\begin{figure}[t]
	\centering
	\adjustbox{max width=\columnwidth, max totalheight=0.92\textheight}{%
    \begin{tikzpicture}[
        line cap=round, line join=round,
        lat/.style ={latgray, line width=1.1pt},
        ga/.style ={postaction={decorate, decoration={markings,
                     mark=at position 0.52 with {\arrow{Stealth[length=2.2mm,width=1.9mm]}}}}},
        gar/.style={postaction={decorate, decoration={markings,
                     mark=at position 0.48 with {\arrowreversed{Stealth[length=2.2mm,width=1.9mm]}}}}},
        red edge/.style={edgered, line width=1.5pt},
        skel/.style={latgray, line width=0.9pt},
        rsk/.style ={edgered, line width=0.9pt},
        rule/.style={black, line width=0.9pt},
        ar/.style  ={black, line width=1.1pt,
                     -{Stealth[length=3.2mm,width=2.5mm]}},
        mv/.style  ={opblue, line width=1.2pt},
        xl/.style  ={opblue, font=\LARGE, fill=white, inner sep=1.4pt},
        ml/.style  ={opblue, font=\LARGE, fill=white, inner sep=1pt},
        pn/.style  ={black, font=\Large},
      ]

      \def\rowpitch{8.4}
      \def\colshift{32.6}
      \def\latX{18.725}
      \def\stabX{-0.6}
      \def\pA{1.15}
      \def\vH{3.7}
      \def\pB{7.05}
      \def\vF{9.05}
      \def\pC{13.1}

      \newcommand{\ecTopBot}{%
        \draw[lat] (-0.35, 1)--(3.35, 1);
        \draw[lat] (-0.35,-1)--(3.35,-1);
        \draw[lat] (0,-1)--(0,1);
        \draw[lat] (3,-1)--(3,1);
        \foreach \xc in {0,1,2,3}{\draw[lat] (\xc,1)--(\xc,1.35); \draw[lat] (\xc,-1)--(\xc,-1.35);}
        \foreach \xa in {0.5,1.5,2.5}{\ecArrR{\xa}{1} \ecArrR{\xa}{-1}}
        \foreach \xb in {0,3}{\ecArrU{\xb}{0.5} \ecArrU{\xb}{-0.5}}}
      \newcommand{\ecArrR}[2]{\draw[lat,-{Stealth[length=2.2mm,width=1.9mm]}] (#1-0.15,#2)--(#1+0.02,#2);}
      \newcommand{\ecArrU}[2]{\draw[lat,-{Stealth[length=2.2mm,width=1.9mm]}] (#1,#2-0.15)--(#1,#2+0.02);}
      \newcommand{\ecRedEdge}{%
        \draw[red edge] (1,0)--(2,0);
        \draw[red edge,-{Stealth[length=2.7mm,width=2.2mm]}] (1,0)--(1.7,0);}
      \newcommand{\ecDot}[1]{\fill[gray!95!black] (#1,0) circle (2.4pt);}
      \newcommand{\ecDotO}[1]{\draw[gray!95!black, line width=1pt, fill=white] (#1,0) circle (2.2pt);}
      \newcommand{\ecUpTwoToL}{\draw[lat,gar] (2,1) .. controls (1.8,0.56) and (1.3,0.26) .. (1,0);}
      \newcommand{\ecRightTwoToL}{\draw[lat,ga] (1,0) .. controls (1.3,0.26) and (2.4,0.5) .. (3,0);}
      \newcommand{\ecDownTwoToL}{\draw[lat,gar] (1,0) .. controls (1.3,-0.26) and (1.8,-0.56) .. (2,-1);}
      \newcommand{\ecUpOneToR}{\draw[lat,gar] (1,1) .. controls (1.2,0.56) and (1.7,0.26) .. (2,0);}
      \newcommand{\ecDownOneToR}{\draw[lat,ga] (1,-1) .. controls (1.2,-0.56) and (1.7,-0.26) .. (2,0);}
      \newcommand{\ecLeftOneToR}{\draw[lat,ga] (0,0) .. controls (0.5,-0.5) and (1.5,-0.45) .. (2,0);}
      \newcommand{\ecPlaqLabels}[3]{%
        \node[xl] at (#1,#2+0.85)   {\zgb{kl}};
        \node[xl] at (#1-0.85,#2)   {\zgb{li}};
        \node[xl] at (#1+0.85,#2)   {\zg{jk}};
        \node[xl] at (#1,#2-0.85)   {\zg{ij}};
        \node[black, font=\Large] at (#1,#2) {#3};}
      \newcommand{\ecPlaq}[3]{%
        \draw[skel] (#1-0.85,#2-0.85) rectangle (#1+0.85,#2+0.85);
        \ecPlaqLabels{#1}{#2}{#3}}
      \newcommand{\ecPlaqRB}[3]{%
        \draw[skel] (#1-0.85,#2-0.85) rectangle (#1+0.85,#2+0.85);
        \draw[rsk] (#1-0.85,#2-0.85)--(#1+0.85,#2-0.85);
        \ecPlaqLabels{#1}{#2}{#3}}
      \newcommand{\ecPlaqRT}[3]{%
        \draw[skel] (#1-0.85,#2-0.85) rectangle (#1+0.85,#2+0.85);
        \draw[rsk] (#1-0.85,#2+0.85)--(#1+0.85,#2+0.85);
        \ecPlaqLabels{#1}{#2}{#3}}
      \newcommand{\ecPlaqNB}[3]{%
        \draw[skel] (#1-0.85,#2-0.85) rectangle (#1+0.85,#2+0.85);
        \node[xl] at (#1,#2+0.85)   {\zgb{kl}};
        \node[xl] at (#1-0.85,#2)   {\zgb{li}};
        \node[xl] at (#1+0.85,#2)   {\zg{ik}};
        \node[black, font=\Large] at (#1,#2) {#3};}
      \newcommand{\ecPlaqNT}[3]{%
        \draw[skel] (#1-0.85,#2-0.85) rectangle (#1+0.85,#2+0.85);
        \node[xl] at (#1-0.85,#2)   {\zgb{li}};
        \node[xl] at (#1+0.85,#2)   {\zg{jl}};
        \node[xl] at (#1,#2-0.85)   {\zg{ij}};
        \node[black, font=\Large] at (#1,#2) {#3};}
      \newcommand{\ecPlaqFour}[7]{%
        \draw[skel] (#1-0.85,#2-0.85) rectangle (#1+0.85,#2+0.85);
        \node[xl] at (#1,#2+0.85)   {#4};
        \node[xl] at (#1-0.85,#2)   {#5};
        \node[xl] at (#1+0.85,#2)   {#6};
        \node[xl] at (#1,#2-0.85)   {#7};
        \node[black, font=\Large] at (#1,#2) {#3};}
      \newcommand{\ecPlaqExt}[5]{%
        \draw[rsk] (#1+#3*0.85,#2+#4*0.85) -- (#1+#3*2.15,#2+#4*0.85);
        \node[xl] at (#1+#3*1.5,#2+#4*0.85) {#5};}
      \newcommand{\ecLab}[3]{\if\relax\detokenize{#3}\relax\else\node[xl] at (#1,#2) {#3};\fi}
      \newcommand{\ecVop}[8]{%
        \draw[skel] (#1-1.3,0)--(#1+2.7,0);
        \draw[skel] (#1,-1.15)--(#1,1.15);
        \draw[skel] (#1+1.4,-1.15)--(#1+1.4,1.15);
        \draw[rsk] (#1,0)--(#1+1.4,0);
        \ecLab{#1-0.65}{0}{#2} \ecLab{#1}{0.6}{#3} \ecLab{#1}{-0.6}{#4}
        \ecLab{#1+0.7}{0}{#5}
        \ecLab{#1+1.4}{0.6}{#6} \ecLab{#1+1.4}{-0.6}{#7} \ecLab{#1+2.05}{0}{#8}}

      \foreach \cs in {0,\colshift}{
        \node[font=\Huge] at ({-3.675+\cs}, 4.9) {Gate};
        \node[font=\Huge] at ({8.3+\cs}, 4.9)  {Stabilizers};
        \node[font=\Huge] at ({20.675+\cs}, 4.9) {Graph};
        \draw[rule, line width=1.3pt] ({-6.45+\cs},4.2) -- ({23.85+\cs},4.2);
        \draw[rule, line width=1.3pt] ({-0.9+\cs},5.3) -- ({-0.9+\cs},{-3*\rowpitch-4.2});
        \draw[rule, line width=1.3pt] ({17.5+\cs},5.3) -- ({17.5+\cs},{-3*\rowpitch-4.2});
        \foreach \r in {1,2,3}{
          \draw[rule] ({-6.45+\cs},{-\r*\rowpitch+4.2}) -- ({23.85+\cs},{-\r*\rowpitch+4.2});}}

      \begin{scope}[shift={(\latX,0)}, scale=1.3]
        \ecTopBot
        \draw[lat] (-0.35,0)--(3.35,0);
        \draw[lat] (1,-1)--(1,1);
        \draw[lat] (2,-1)--(2,1);
        \ecArrR{0.5}{0} \ecArrR{2.5}{0}
        \ecArrU{1}{0.5} \ecArrU{1}{-0.5} \ecArrU{2}{0.5} \ecArrU{2}{-0.5}
        \ecRedEdge
        \node[pn,font=\large] at (0.5, 0.5) {4};
        \node[pn,font=\large] at (1.5, 0.5) {5};
        \node[pn,font=\large] at (2.5, 0.5) {6};
        \node[pn,font=\large] at (0.5,-0.5) {1};
        \node[pn,font=\large] at (1.5,-0.5) {2};
        \node[pn,font=\large] at (2.5,-0.5) {3};
        \ecDot{2} \ecDotO{1}
      \end{scope}
      \begin{scope}[shift={(\stabX,0)}, scale=1.25]
        \ecPlaq{\pA}{1.65}{4}  \ecPlaq{\pA}{-1.65}{1}
        \ecVop{\vH}{\opRg}{\opLg}{\opRg}{\opLg}{}{}{}
        \ecPlaqRB{\pB}{1.65}{5}  \ecPlaqRT{\pB}{-1.65}{2}
        \ecVop{\vF}{}{}{}{\opRg}{\opLg}{\opRg}{\opLg}
        \ecPlaq{\pC}{1.65}{6}  \ecPlaq{\pC}{-1.65}{3}
      \end{scope}

      \begin{scope}[shift={(\latX,-\rowpitch)}, scale=1.3]
        \ecTopBot
        \draw[lat] (-0.35,0)--(3.35,0);
        \draw[lat] (1,-1)--(1,1);
        \draw[lat] (2,-1)--(2,0);
        \ecUpTwoToL
        \ecArrR{0.5}{0} \ecArrR{2.5}{0}
        \ecArrU{1}{0.5} \ecArrU{1}{-0.5} \ecArrU{2}{-0.5}
        \ecRedEdge
        \ecDot{2} \ecDotO{1}
      \end{scope}
      \begin{scope}[shift={(\stabX,-\rowpitch)}, scale=1.25]
        \ecPlaq{\pA}{1.65}{}  \ecPlaq{\pA}{-1.65}{}
        \ecVop{\vH}{\opRg}{\opLg}{\opRg}{\opLg}{\opLg}{}{}
        \ecPlaqNB{\pB}{1.65}{}  \ecPlaqRT{\pB}{-1.65}{}
        \ecVop{\vF}{}{}{}{\opRg}{}{\opRg}{\opLg}
        \ecPlaqFour{\pC}{1.65}{}{\zgb{kl}}{\zgb{lm}}{\zg{jk}}{\zg{ij}}  \ecPlaqExt{\pC}{1.65}{-1}{-1}{\zg{mi}}
        \ecPlaq{\pC}{-1.65}{}
      \end{scope}

      \begin{scope}[shift={(\latX,-2*\rowpitch)}, scale=1.3]
        \ecTopBot
        \draw[lat] (-0.35,0)--(1,0);
        \draw[lat] (1,-1)--(1,1);
        \draw[lat] (2,-1)--(2,0);
        \ecUpTwoToL
        \ecRightTwoToL
        \draw[lat] (3,0)--(3.35,0);
        \ecArrR{0.5}{0} \ecArrU{1}{0.5} \ecArrU{1}{-0.5} \ecArrU{2}{-0.5}
        \ecRedEdge
        \ecDot{2} \ecDotO{1}
      \end{scope}
      \begin{scope}[shift={(\stabX,-2*\rowpitch)}, scale=1.25]
        \ecPlaq{\pA}{1.65}{}  \ecPlaq{\pA}{-1.65}{}
        \ecVop{\vH}{\opRg}{\opLg}{\opRg}{\opLg}{\opLg}{}{\opLg}
        \ecPlaqNB{\pB}{1.65}{}  \ecPlaqRT{\pB}{-1.65}{}
        \ecVop{\vF}{}{}{}{\opRg}{}{\opRg}{}
        \ecPlaq{\pC}{1.65}{}
        \ecPlaqFour{\pC}{-1.65}{}{\zgb{km}}{\zgb{li}}{\zg{jk}}{\zg{ij}}  \ecPlaqExt{\pC}{-1.65}{-1}{1}{\zg{ml}}
      \end{scope}

      \begin{scope}[shift={(\latX,-3*\rowpitch)}, scale=1.3]
        \ecTopBot
        \draw[lat] (-0.35,0)--(1,0);
        \draw[lat] (1,0)--(1,1);
        \draw[lat] (2,-1)--(2,0);
        \ecUpTwoToL
        \ecRightTwoToL
        \ecDownOneToR
        \draw[lat] (3,0)--(3.35,0);
        \ecArrR{0.5}{0} \ecArrU{1}{0.5} \ecArrU{2}{-0.5}
        \ecRedEdge
        \ecDot{2} \ecDotO{1}
      \end{scope}
      \begin{scope}[shift={(\stabX,-3*\rowpitch)}, scale=1.25]
        \ecPlaq{\pA}{1.65}{}
        \ecPlaqFour{\pA}{-1.65}{}{\zgb{ml}}{\zgb{li}}{\zg{jk}}{\zg{ij}}  \ecPlaqExt{\pA}{-1.65}{1}{1}{\zgb{km}}
        \ecVop{\vH}{\opRg}{\opLg}{}{\opLg}{\opLg}{}{\opLg}
        \ecPlaqNB{\pB}{1.65}{}  \ecPlaqNT{\pB}{-1.65}{}
        \ecVop{\vF}{}{}{\opRg}{\opRg}{}{\opRg}{}
        \ecPlaq{\pC}{1.65}{}
        \ecPlaqFour{\pC}{-1.65}{}{\zgb{km}}{\zgb{li}}{\zg{jk}}{\zg{ij}}  \ecPlaqExt{\pC}{-1.65}{-1}{1}{\zg{ml}}
      \end{scope}

      \begin{scope}[shift={(\colshift,{4*\rowpitch})}]
      \begin{scope}[shift={(\latX,-4*\rowpitch)}, scale=1.3]
        \ecTopBot
        \draw[lat] (-0.35,0)--(0,0);
        \ecLeftOneToR
        \draw[lat] (1,0)--(1,1);
        \draw[lat] (2,-1)--(2,0);
        \ecUpTwoToL
        \ecRightTwoToL
        \ecDownOneToR
        \draw[lat] (3,0)--(3.35,0);
        \ecArrU{1}{0.5} \ecArrU{2}{-0.5}
        \ecRedEdge
        \ecDot{2} \ecDotO{1}
      \end{scope}
      \begin{scope}[shift={(\stabX,-4*\rowpitch)}, scale=1.25]
        \ecPlaqFour{\pA}{1.65}{}{\zgb{kl}}{\zgb{li}}{\zg{mk}}{\zg{ij}}  \ecPlaqExt{\pA}{1.65}{1}{-1}{\zgb{jm}}
        \ecPlaq{\pA}{-1.65}{}
        \ecVop{\vH}{}{\opLg}{}{\opLg}{\opLg}{}{\opLg}
        \ecPlaqNB{\pB}{1.65}{}  \ecPlaqNT{\pB}{-1.65}{}
        \ecVop{\vF}{\opRg}{}{\opRg}{\opRg}{}{\opRg}{}
        \ecPlaq{\pC}{1.65}{}
        \ecPlaqFour{\pC}{-1.65}{}{\zgb{km}}{\zgb{li}}{\zg{jk}}{\zg{ij}}  \ecPlaqExt{\pC}{-1.65}{-1}{1}{\zg{ml}}
      \end{scope}

      \begin{scope}[shift={(\latX,-5*\rowpitch)}, scale=1.3]
        \ecTopBot
        \draw[lat] (-0.35,0)--(1,0);
        \draw[lat] (1,-1)--(1,0);
        \draw[lat] (2,0)--(3.35,0);
        \draw[lat] (2,0)--(2,1);
        \ecUpOneToR
        \ecDownTwoToL
        \ecArrR{0.5}{0} \ecArrR{2.5}{0} \ecArrU{1}{-0.5} \ecArrU{2}{0.5}
        \ecRedEdge
        \ecDot{1} \ecDotO{2}
      \end{scope}
      \begin{scope}[shift={(\stabX,-5*\rowpitch)}, scale=1.25]
        \ecPlaqFour{\pA}{1.65}{}{\zgb{kl}}{\zgb{li}}{\zg{mk}}{\zg{ij}}  \ecPlaqExt{\pA}{1.65}{1}{-1}{\zg{jm}}
        \ecPlaq{\pA}{-1.65}{}
        \ecVop{\vH}{}{\opLg}{}{\opRg}{\opLg}{}{\opLg}
        \ecPlaqNB{\pB}{1.65}{}  \ecPlaqNT{\pB}{-1.65}{}
        \ecVop{\vF}{\opRg}{}{\opRg}{\opLg}{}{\opRg}{}
        \ecPlaq{\pC}{1.65}{}
        \ecPlaqFour{\pC}{-1.65}{}{\zgb{km}}{\zgb{li}}{\zg{jk}}{\zg{ij}}  \ecPlaqExt{\pC}{-1.65}{-1}{1}{\zgb{ml}}
      \end{scope}

      \begin{scope}[shift={(\latX,-6*\rowpitch)}, scale=1.3]
        \ecTopBot
        \draw[lat] (-0.35,0)--(1,0);
        \draw[lat] (1,-1)--(1,1);
        \draw[lat] (2,0)--(3.35,0);
        \draw[lat] (2,0)--(2,1);
        \ecDownTwoToL
        \ecArrR{0.5}{0} \ecArrR{2.5}{0}
        \ecArrU{1}{0.5} \ecArrU{1}{-0.5} \ecArrU{2}{0.5}
        \ecRedEdge
        \ecDot{1} \ecDotO{2}
      \end{scope}
      \begin{scope}[shift={(\stabX,-6*\rowpitch)}, scale=1.25]
        \ecPlaq{\pA}{1.65}{}  \ecPlaq{\pA}{-1.65}{}
        \ecVop{\vH}{}{}{}{\opRg}{\opLg}{}{\opLg}
        \ecPlaqRB{\pB}{1.65}{}  \ecPlaqNT{\pB}{-1.65}{}
        \ecVop{\vF}{\opRg}{\opLg}{\opRg}{\opLg}{}{\opRg}{}
        \ecPlaq{\pC}{1.65}{}
        \ecPlaqFour{\pC}{-1.65}{}{\zgb{km}}{\zgb{li}}{\zg{jk}}{\zg{ij}}  \ecPlaqExt{\pC}{-1.65}{-1}{1}{\zgb{ml}}
      \end{scope}

      \begin{scope}[shift={(\latX,-7*\rowpitch)}, scale=1.3]
        \ecTopBot
        \draw[lat] (-0.35,0)--(3.35,0);
        \draw[lat] (1,-1)--(1,1);
        \draw[lat] (2,-1)--(2,1);
        \ecArrR{0.5}{0} \ecArrR{2.5}{0}
        \ecArrU{1}{0.5} \ecArrU{1}{-0.5} \ecArrU{2}{0.5} \ecArrU{2}{-0.5}
        \ecRedEdge
        \ecDot{1} \ecDotO{2}
      \end{scope}
      \begin{scope}[shift={(\stabX,-7*\rowpitch)}, scale=1.25]
        \ecPlaq{\pA}{1.65}{}  \ecPlaq{\pA}{-1.65}{}
        \ecVop{\vH}{}{}{}{\opRg}{\opLg}{\opRg}{\opLg}
        \ecPlaqRB{\pB}{1.65}{}  \ecPlaqRT{\pB}{-1.65}{}
        \ecVop{\vF}{\opRg}{\opLg}{\opRg}{\opLg}{}{}{}
        \ecPlaq{\pC}{1.65}{}  \ecPlaq{\pC}{-1.65}{}
      \end{scope}

      \end{scope}

      \newcommand{\ecMini}{%
        \draw[lat] (0,0)--(3,0);
        \draw[lat] (1,-0.85)--(1,0.85);
        \draw[lat] (2,-0.85)--(2,0.85);
        \draw[red edge] (1,0)--(2,0);}
      \newcommand{\ecMvDot}[2]{\fill[opblue] (#1,#2) circle (2.2pt);}

      \node[font=\Huge] at (-3.675,{-1*\rowpitch+2.55}) {1.};
      \begin{scope}[shift={(-6.0,{-1*\rowpitch})},scale=1.55]
        \ecMini \ecMvDot{1.5}{0}
        \draw[mv] (1.5,0) .. controls (1.65,0.26) and (1.8,0.38) .. (1.97,0.41);
        \node[ml] at (2,0.42) {$L$};
      \end{scope}
      \node[font=\Huge] at (-3.675,{-2*\rowpitch+2.55}) {2.};
      \begin{scope}[shift={(-6.0,{-2*\rowpitch})},scale=1.55]
        \ecMini \ecMvDot{1.5}{0}
        \draw[mv] (1.5,0) .. controls (1.75,0.45) and (2.25,0.38) .. (2.47,0.07);
        \node[ml] at (2.5,0) {$L$};
      \end{scope}
      \node[font=\Huge] at (-3.675,{-3*\rowpitch+2.55}) {3.};
      \begin{scope}[shift={(-6.0,{-3*\rowpitch})},scale=1.55]
        \ecMini \ecMvDot{1.5}{0}
        \draw[mv] (1.5,0) .. controls (1.38,-0.3) and (1.18,-0.48) .. (1,-0.55);
        \node[ml] at (1,-0.425) {$R^{\dagger}$};
      \end{scope}
      \begin{scope}[shift={(\colshift,{4*\rowpitch})}]
      \node[font=\Huge] at (-3.675,{-4*\rowpitch+2.55}) {4.};
      \begin{scope}[shift={(-6.0,{-4*\rowpitch})},scale=1.55]
        \ecMini \ecMvDot{1.5}{0}
        \draw[mv] (0.62,0.05) .. controls (0.9,0.4) and (1.25,0.3) .. (1.5,0.02);
        \node[ml] at (0.5,0) {$R^{\dagger}$};
      \end{scope}
      \node[font=\Huge] at (-3.675,{-5*\rowpitch+2.55}) {5.};
      \begin{scope}[shift={(-6.0,{-5*\rowpitch})},scale=1.55]
        \ecMini
        \node[ml] at (1.5,0) {$\mathcal{S}$};
      \end{scope}
      \node[font=\Huge] at (-3.675,{-6*\rowpitch+2.55}) {6.};
      \begin{scope}[shift={(-6.0,{-6*\rowpitch})},scale=1.55]
        \ecMini \ecMvDot{1.5}{0}
        \draw[mv] (1.5,0) .. controls (1.35,0.26) and (1.2,0.38) .. (1.03,0.41);
        \node[ml] at (1,0.42) {$L$};
      \end{scope}
      \node[font=\Huge] at (-3.675,{-7*\rowpitch+2.55}) {7.};
      \begin{scope}[shift={(-6.0,{-7*\rowpitch})},scale=1.55]
        \ecMini \ecMvDot{1.5}{0}
        \draw[mv] (1.5,0) .. controls (1.62,-0.3) and (1.82,-0.48) .. (2,-0.55);
        \node[ml] at (2,-0.425) {$R^{\dagger}$};
      \end{scope}
      \end{scope}

    \end{tikzpicture}}
	\caption{Stabilizer tracking for the charge exchange circuit of Fig.~\ref{fig:charge_exchange_circuit}, whose steps 1--4 are moves 1--2, 3--4, 5 and 6--7 here.
	Each row lists the move (Step), the stabilizers with their support on the direct edge in red (Stabilizers), and the rewired graph (Graph); the grey dots track $v_{1}$ (hollow) and $v_{2}$ (filled), and the six plaquettes are numbered in the first row.
	Every controlled multiplication is controlled by the direct edge (blue dot) and slides the labelled edge across it.
	At move 5 the inverse gate $\mathcal{S}$ reverses the direct edge and the graph is redrawn with the two endpoints exchanged.}
	\label{fig:ce-stabilizers}
\end{figure}
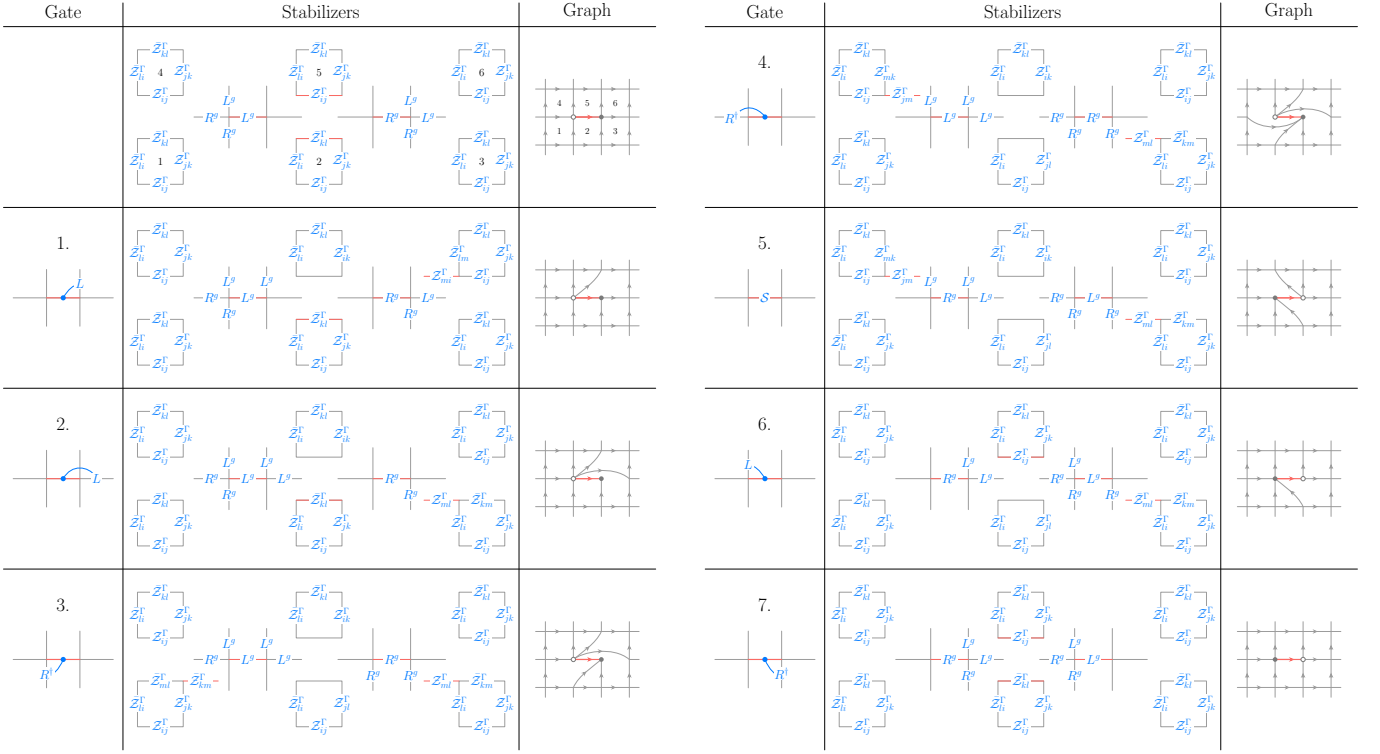

\subsubsection{Ungauging ordering}
\label{app:ce-ungauging}

The six controlled multiplications commute and $\mathcal{S}_{\ell}$ inverts any gate moved across it by Eq.~\eqref{eq:controlled-inverse}, so the circuit can be reordered freely.
Moving the two $CR^{\dagger}$ gates of step 2 to after $\mathcal{S}_{\ell}$ and $C_{\ell}R^{\dagger}_{\text{down}_{2}}$ of step 4 to before it gives
\begin{equation}\label{eq:ce-unitary}
U_{\mathrm{ex}}=W\,\mathcal{S}_{\ell}\,W_{2},\qquad
W_{2}\equiv C_{\ell}L_{\text{up}_{2}}\,C_{\ell}L_{\text{right}_{2}}\,C_{\ell}R_{\text{down}_{2}},\qquad
W\equiv C_{\ell}L_{\text{up}_{1}}\,C_{\ell}R_{\text{left}_{1}}\,C_{\ell}R_{\text{down}_{1}},
\end{equation}
where $W$ is the $G$ gauging circuit of the charge fusion protocol of Sec.~\ref{subsec:charge-fusion}: by the right identities of Eqs.~\eqref{eq:ce-slide-CL} and \eqref{eq:ce-slide-CR}, it builds the vertex operator at position $1$ out of a left multiplication on the direct edge alone,
\begin{equation}\label{eq:ce-gauging}
W\left(\;
\vcenter{\hbox{%
\begin{tikzpicture}[line width=0.9pt,baseline]
  \tvlat
  \node[op] at (0.75,0) {$L^{g}$};
\end{tikzpicture}
}}
\;\right)W^{\dagger}=\;
\vcenter{\hbox{%
\begin{tikzpicture}[line width=0.9pt,baseline]
  \tvlat
  \node[op] at (0,0.45) {$L^{g}$};
  \node[op] at (-0.5,0) {$R^{g}$};
  \node[op] at (0,-0.45) {$R^{g}$};
  \node[op] at (0.75,0) {$L^{g}$};
\end{tikzpicture}
}}
=A^{g}_{v_{1}},
\end{equation}
which is Eq.~\eqref{eq:cf-W-defining} read backwards: $W^{\dagger}$ ungauges $v_{1}$ onto the direct edge and $W$ regauges it.
$W_{2}$ acts on the edges of $v_{2}$ alone and, by Eqs.~\eqref{eq:ce-slide-CL} and \eqref{eq:ce-slide-CR},
\begin{equation}\label{eq:ce-ungauged}
W_{2}\,A^{g}_{v_{2}}\,W_{2}^{\dagger}=\;
\vcenter{\hbox{%
\begin{tikzpicture}[line width=0.9pt,baseline]
  \tvlat
  \node[op] at (0.75,0) {$R^{g}$};
\end{tikzpicture}
}}
,\qquad
W_{2}\,A^{g}_{v_{1}}\,W_{2}^{\dagger}=\;
\vcenter{\hbox{%
\begin{tikzpicture}[line width=0.9pt,baseline]
  \tvlat
  \node[op] at (0,0.45) {$L^{g}$};
  \node[op] at (-0.5,0) {$R^{g}$};
  \node[op] at (0,-0.45) {$R^{g}$};
  \node[op] at (0.75,0) {$L^{g}$};
  \node[op] at (1.5,0.45) {$L^{g}$};
  \node[op] at (2.0,0) {$L^{g}$};
  \node[op] at (1.5,-0.45) {$R^{g}$};
\end{tikzpicture}
}}
\end{equation}
The direct edge now dangles from the graph and carries the charge of $v_{2}$ on a single qudit, mirroring the ungauging of $v_{1}$ in the charge fusion protocol, Eqs.~\eqref{eq:cf-Av1-ungauged} and \eqref{eq:cf-Av2-ungauged}; no plaquette has support on it.
The inverse gate exchanges $L^{g}_{\ell}$ and $R^{g}_{\ell}$, and since no plaquette constrains the dangling edge the graph can be redrawn with $v_{2}$ at position $1$ and $v_{1}$ at position $2$; the regauging $W$ then builds the vertex operator at position $1$ out of the dangling edge by Eq.~\eqref{eq:ce-gauging}.
The charge exchange is thus an ungauging, an inverse gate on the ungauged edge, and a regauging: the three-layer structure of the ribbon operators of Sec.~\ref{subsec:ribbon-rewiring} with $\mathcal{S}_{\ell}$ as the single-edge operator, except that $W\neq W_{2}^{\dagger}$, which is what exchanges the two endpoints instead of restoring them, and the ungauge--reset--regauge structure of the charge fusion channel of Eq.~\eqref{eq:cf-channel} with the inverse gate in place of the reset.
Ungauging a path of vertices one after another onto the successive direct edges leaves those edges dangling, each carrying the charge of its vertex; a $\mathrm{SWAP}$ of two dangling edges then exchanges the charges and regauging in the reverse order restores the lattice, which realizes the exchange of charges as a sequential circuit~\cite{chen_sequential_2024}, cf.\ Fig.~\ref{fig:resource-classes}.

\subsubsection{Local neutrality disentangles the direct edge}
\label{app:ce-neutrality}

By Eq.~\eqref{eq:ce-ungauged}, the charge projectors of Eqs.~\eqref{eq:P-trivial}--\eqref{eq:P-internal} at $v_{2}$ are mapped to the same expressions with $A^{g}_{v_{2}}$ replaced by $R^{g}_{\ell}$.
If $v_{2}$ carries no charge, as every vertex does in the ground state, the vacuum projector $P_{v_{2}}$ is mapped to $\frac{1}{|G|}\sum_{g}R^{g}_{\ell}=\ketbra{+}{+}_{\ell}$ with $\ket{+}=\frac{1}{\sqrt{|G|}}\sum_{g}\ket{g}$: the direct edge is disentangled into a product state, $W_{2}$ is the three-gate vertex disentangler of entanglement renormalization~\cite{aguado_vidal_2008,aguado_entanglement_2011}, and charge neutrality at $v_{2}$ is the condition for it to disentangle.
The inverse gate acts trivially on $\ket{+}$, and $W$ then builds a neutral vertex at position $1$ out of it: exchanging a charge with the vacuum moves the charge across the direct edge.

\subsection{Flux exchange circuit}
\label{app:fb}

We track the stabilizers through the flux exchange circuit of Sec.~\ref{subsec:flux-exchange} in Fig.~\ref{fig:fb-stabilizers}.
The plaquette operators are drawn as the fixed-irrep contractions of Eq.~\eqref{eq:Be-MPO} with the prefactor suppressed, so the results hold for every class projector $B^{C}_{p}$ of Eq.~\eqref{eq:BC-MPO}; of the vertex operators, only those at the two endpoints of the dual edge are drawn.

\subsubsection{Stabilizer tracking}
\label{app:fb-tracking}

Figure~\ref{fig:fb-stabilizers} applies the sliding rules of Appendix~\ref{app:gr-move} to the seven moves of Fig.~\ref{fig:flux_exchange_circuit}, with the dual edge $d$ as the target throughout and one of its neighbouring edges as the control: $CR$ or $CR^{\dagger}$ slides the head of $d$ and $CL$ or $CL^{\dagger}$ its tail, each slide moving the factor on $d$ between the two vertex operators and the tensor on $d$ between the two plaquette operators that it affects.
After the six slides the lattice is restored with the two plaquettes exchanged and $d$ reversed, so every stabilizer is its final form with $L^{g}_{d}\leftrightarrow R^{g}_{d}$ and $\mathcal{Z}^{\Gamma}_{d}\leftrightarrow\bar{\mathcal{Z}}^{\Gamma}_{d}$, which the inverse gate $\mathcal{S}_{d}$ at step 7 undoes by Eqs.~\eqref{eq:inverse-gate-def} and \eqref{eq:Zbar-def}.
The last row is Eq.~\eqref{eq:fb-swap}: the two plaquette operators are exchanged, up to relabelling of the summed indices, and all other stabilizers are unchanged.

\begin{figure}[ht]
\centering
\adjustbox{max width=\textwidth, max totalheight=0.9\textheight}{%
\begin{tikzpicture}[
    line cap=round, line join=round,
    lat/.style ={latgray, line width=1.1pt},
    ga/.style ={postaction={decorate, decoration={markings,
                 mark=at position #1 with {\arrow{Stealth[length=2.2mm,width=1.9mm]}}}}},
    ga/.default=0.52,
    gr/.style ={postaction={decorate, decoration={markings,
                 mark=at position #1 with {\arrow{Stealth[length=2.7mm,width=2.2mm]}}}}},
    gr/.default=0.52,
    red edge/.style={edgered, line width=1.5pt},
    skel/.style={latgray, line width=0.9pt},
    rsk/.style ={edgered, line width=0.9pt},
    rule/.style={black, line width=0.9pt},
    mv/.style  ={opblue, line width=1.2pt},
    xl/.style  ={opblue, font=\Large, fill=white, inner sep=1.4pt},
    ml/.style  ={opblue, font=\Large, fill=white, inner sep=1pt},
  ]

  \def\rowpitch{7.0}
  \def\colshift{22.9}
  \def\latX{11.65}
  \def\stabX{-1.2}
  \def\blshift{0.3}
  \def\brshift{-0.3}

  \newcommand{\fmArrR}[2]{\draw[lat,-{Stealth[length=2.2mm,width=1.9mm]}] (#1-0.15,#2)--(#1+0.02,#2);}
  \newcommand{\fmArrU}[2]{\draw[lat,-{Stealth[length=2.2mm,width=1.9mm]}] (#1,#2-0.15)--(#1,#2+0.02);}
  \newcommand{\fmFrame}{%
    \draw[lat] (-1, 0.5)--(1, 0.5);
    \draw[lat] (-1,-0.5)--(1,-0.5);
    \draw[lat] (-1,-0.5)--(-1,0.5);
    \draw[lat] (1,-0.5)--(1,0.5);
    \foreach \xa in {-1,0,1}{\draw[lat] (\xa,0.5)--(\xa,1.2); \draw[lat] (\xa,-1.2)--(\xa,-0.5);}
    \draw[lat] (-1.7, 0.5)--(-1, 0.5); \draw[lat] (1, 0.5)--(1.7, 0.5);
    \draw[lat] (-1.7,-0.5)--(-1,-0.5); \draw[lat] (1,-0.5)--(1.7,-0.5);
    \foreach \xa in {-1.35,-0.5,0.5,1.35}{\fmArrR{\xa}{0.5} \fmArrR{\xa}{-0.5}}
    \fmArrU{-1}{0} \fmArrU{1}{0}
    \foreach \xa in {-1,0,1}{\fmArrU{\xa}{0.85} \fmArrU{\xa}{-0.85}}}
  \newcommand{\fmSq}[2]{\draw[skel] (#1-0.85,#2-0.85) rectangle (#1+0.85,#2+0.85);}
  \newcommand{\fmSqR}[2]{%
    \fmSq{#1}{#2}
    \draw[rsk] (#1+0.85,#2-0.85)--(#1+0.85,#2+0.85);}
  \newcommand{\fmSqL}[2]{%
    \fmSq{#1}{#2}
    \draw[rsk] (#1-0.85,#2-0.85)--(#1-0.85,#2+0.85);}
  \newcommand{\fmDom}[2]{%
    \draw[skel] (#1-1.7,#2-0.85) rectangle (#1+1.7,#2+0.85);
    \draw[rsk] (#1,#2-0.85)--(#1,#2+0.85);}
  \newcommand{\fmLegTR}[1]{%
    \draw[skel] (1.85,0.85)--(3.55,0.85);
    \node[xl] at (2.7,0.85)  {#1};}
  \newcommand{\fmLegR}[1]{%
    \draw[skel] (3.55,0.85)--(3.55,-0.85);
    \node[xl] at (3.55,0)    {#1};}
  \newcommand{\fmLegBL}[1]{%
    \draw[skel] (8.15,-0.85)--(6.45,-0.85);
    \node[xl] at (7.3,-0.85) {#1};}
  \newcommand{\fmLegL}[1]{%
    \draw[skel] (6.45,-0.85)--(6.45,0.85);
    \node[xl] at (6.45,0)    {#1};}
  \newcommand{\fmCrossSkel}[2]{%
    \draw[skel] (#1-1.3,#2)--(#1+1.3,#2);
    \draw[skel] (#1,#2-1.15)--(#1,#2+1.15);}
  \newcommand{\fmCrossD}[3]{%
    \fmCrossSkel{#1}{#2}
    \draw[rsk] (#1,#2)--(#1,#2-1.15);
    \node[xl] at (#1,#2+0.6)    {\opLg};
    \node[xl] at (#1-0.68,#2)   {\opRg};
    \node[xl] at (#1+0.68,#2)   {\opLg};
    \node[xl] at (#1,#2-0.6)    {#3};}
  \newcommand{\fmCrossU}[3]{%
    \fmCrossSkel{#1}{#2}
    \draw[rsk] (#1,#2)--(#1,#2+1.15);
    \node[xl] at (#1,#2+0.6)    {#3};
    \node[xl] at (#1-0.68,#2)   {\opRg};
    \node[xl] at (#1+0.68,#2)   {\opLg};
    \node[xl] at (#1,#2-0.6)    {\opRg};}
  \newcommand{\fmCrossNoDn}[2]{%
    \fmCrossSkel{#1}{#2}
    \draw[rsk] (#1,#2)--(#1,#2-1.15);
    \node[xl] at (#1,#2+0.6)    {\opLg};
    \node[xl] at (#1-0.68,#2)   {\opRg};
    \node[xl] at (#1+0.68,#2)   {\opLg};}
  \newcommand{\fmCrossNoUp}[2]{%
    \fmCrossSkel{#1}{#2}
    \draw[rsk] (#1,#2)--(#1,#2+1.15);
    \node[xl] at (#1-0.68,#2)   {\opRg};
    \node[xl] at (#1+0.68,#2)   {\opLg};
    \node[xl] at (#1,#2-0.6)    {\opRg};}

  \foreach \dx in {0,\colshift}{
    \node[font=\Huge] at ({-3.85+\dx}, 4.05) {Gate};
    \node[font=\Huge] at ({11.65+\dx}, 4.05) {Graph};
    \node[font=\Huge] at ({3.6+\dx}, 4.05)  {Stabilizers};
    \draw[rule, line width=1.3pt] ({-6.1+\dx},3.4) -- ({14.5+\dx},3.4);
    \draw[rule, line width=1.3pt] ({-1.6+\dx},4.45) -- ({-1.6+\dx},{-3*\rowpitch-3.4});
    \draw[rule, line width=1.3pt] ({8.8+\dx},4.45) -- ({8.8+\dx},{-3*\rowpitch-3.4});
    \foreach \r in {1,2,3}{
      \draw[rule] ({-6.1+\dx},{-\r*\rowpitch+3.4}) -- ({14.5+\dx},{-\r*\rowpitch+3.4});}}

  \begin{scope}[shift={(\latX,0)}, scale=1.45]
    \fill[plaqgray] (0,-0.5) rectangle (1,0.5);
    \fmFrame
    \draw[red edge] (0,-0.5)--(0,0.5);
    \draw[red edge,-{Stealth[length=2.7mm,width=2.2mm]}] (0,-0.5)--(0,0.2);
  \end{scope}
  \begin{scope}[shift={(\stabX,0)}]
    \begin{scope}[shift={(\blshift,0)}]
      \fmSqR{1.0}{0}
      \node[xl] at (1.0,0.85)  {\zgb{kl}};
      \node[xl] at (0.15,0)    {\zgb{li}};
      \node[xl] at (1.85,0)    {\zg{jk}};
      \node[xl] at (1.0,-0.85) {\zg{ij}};
    \end{scope}
    \fmCrossD{5.0}{1.8}{\opRg}
    \fmCrossU{5.0}{-1.8}{\opLg}
    \begin{scope}[shift={(\brshift,0)}]
      \fmSqL{9.0}{0}
      \node[xl] at (9.0,0.85)  {\zgb{kl}};
      \node[xl] at (8.15,0)    {\zgb{li}};
      \node[xl] at (9.85,0)    {\zg{jk}};
      \node[xl] at (9.0,-0.85) {\zg{ij}};
    \end{scope}
  \end{scope}

  \begin{scope}[shift={(\latX,-\rowpitch)}, scale=1.45]
    \fill[plaqgray] (0,-0.5) .. controls (0.15,-0.1) and (0.6,0.3) .. (1,0.5)
                    -- (1,-0.5) -- cycle;
    \fmFrame
    \draw[red edge,gr] (0,-0.5) .. controls (0.15,-0.1) and (0.6,0.3) .. (1,0.5);
  \end{scope}
  \begin{scope}[shift={(\stabX,-\rowpitch)}]
    \begin{scope}[shift={(\blshift,0)}]
      \fmSqR{1.0}{0}
      \node[xl] at (1.0,0.85)  {\zgb{kl}};
      \node[xl] at (0.15,0)    {\zgb{li}};
      \node[xl] at (1.85,0)    {\zg{jm}};
      \node[xl] at (1.0,-0.85) {\zg{ij}};
      \fmLegTR{\zgb{mk}}
    \end{scope}
    \fmCrossNoDn{5.0}{1.8}
    \fmCrossU{5.0}{-1.8}{\opLg}
    \begin{scope}[shift={(\brshift,0)}]
      \fmSqL{9.0}{0}
      \node[xl] at (8.15,0)    {\zgb{ki}};
      \node[xl] at (9.85,0)    {\zg{jk}};
      \node[xl] at (9.0,-0.85) {\zg{ij}};
    \end{scope}
  \end{scope}

  \begin{scope}[shift={(\latX,-2*\rowpitch)}, scale=1.45]
    \fill[plaqgray] (-1,-0.5) .. controls (-0.35,0) and (0.35,0) .. (1,0.5)
                    -- (1,-0.5) -- cycle;
    \fmFrame
    \draw[red edge,gr] (-1,-0.5) .. controls (-0.35,0) and (0.35,0) .. (1,0.5);
  \end{scope}
  \begin{scope}[shift={(\stabX,-2*\rowpitch)}]
    \begin{scope}[shift={(\blshift,0)}]
      \fmSqR{1.0}{0}
      \node[xl] at (1.0,0.85)  {\zgb{kl}};
      \node[xl] at (0.15,0)    {\zgb{li}};
      \node[xl] at (1.85,0)    {\zg{im}};
      \fmLegTR{\zgb{mk}}
    \end{scope}
    \fmCrossNoDn{5.0}{1.8}
    \fmCrossNoUp{5.0}{-1.8}
    \begin{scope}[shift={(\brshift,0)}]
      \fmSqL{9.0}{0}
      \node[xl] at (8.15,0)    {\zgb{kn}};
      \node[xl] at (9.85,0)    {\zg{jk}};
      \node[xl] at (9.0,-0.85) {\zg{ij}};
      \fmLegBL{\zg{ni}}
    \end{scope}
  \end{scope}

  \begin{scope}[shift={(\latX,-3*\rowpitch)}, scale=1.45]
    \fill[plaqgray] (-1,-0.5) .. controls (-0.7,0.55) and (0.7,0.55) .. (1,-0.5) -- cycle;
    \fmFrame
    \draw[red edge,gr=0.55] (-1,-0.5) .. controls (-0.7,0.55) and (0.7,0.55) .. (1,-0.5);
  \end{scope}
  \begin{scope}[shift={(\stabX,-3*\rowpitch)}]
    \begin{scope}[shift={(\blshift,0)}]
      \fmSqR{1.0}{0}
      \node[xl] at (1.0,0.85)  {\zgb{kl}};
      \node[xl] at (0.15,0)    {\zgb{li}};
      \node[xl] at (1.85,0)    {\zg{io}};
      \fmLegTR{\zgb{mk}}
      \fmLegR{\zg{om}}
    \end{scope}
    \fmCrossNoDn{5.0}{1.8}
    \fmCrossNoUp{5.0}{-1.8}
    \begin{scope}[shift={(\brshift,0)}]
      \fmSqL{9.0}{0}
      \node[xl] at (8.15,0)    {\zgb{jn}};
      \node[xl] at (9.0,-0.85) {\zg{ij}};
      \fmLegBL{\zg{ni}}
    \end{scope}
  \end{scope}

  \begin{scope}[shift={(\colshift,{4*\rowpitch})}]
  \begin{scope}[shift={(\latX,-4*\rowpitch)}, scale=1.45]
    \fill[plaqgray] (-1,0.5) .. controls (-0.35,0) and (0.35,0) .. (1,-0.5)
                    -- (-1,-0.5) -- cycle;
    \fmFrame
    \draw[red edge,gr] (-1,0.5) .. controls (-0.35,0) and (0.35,0) .. (1,-0.5);
  \end{scope}
  \begin{scope}[shift={(\stabX,-4*\rowpitch)}]
    \begin{scope}[shift={(\blshift,0)}]
      \fmSqR{1.0}{0}
      \node[xl] at (1.0,0.85)  {\zgb{kl}};
      \node[xl] at (1.85,0)    {\zg{lo}};
      \fmLegTR{\zgb{mk}}
      \fmLegR{\zg{om}}
    \end{scope}
    \fmCrossNoDn{5.0}{1.8}
    \fmCrossNoUp{5.0}{-1.8}
    \begin{scope}[shift={(\brshift,0)}]
      \fmSqL{9.0}{0}
      \node[xl] at (8.15,0)    {\zgb{jp}};
      \node[xl] at (9.0,-0.85) {\zg{ij}};
      \fmLegBL{\zg{ni}}
      \fmLegL{\zgb{pn}}
    \end{scope}
  \end{scope}

  \begin{scope}[shift={(\latX,-5*\rowpitch)}, scale=1.45]
    \fill[plaqgray] (-1,0.5) .. controls (-0.6,0.3) and (-0.15,-0.1) .. (0,-0.5)
                    -- (-1,-0.5) -- cycle;
    \fmFrame
    \draw[red edge,gr] (-1,0.5) .. controls (-0.6,0.3) and (-0.15,-0.1) .. (0,-0.5);
  \end{scope}
  \begin{scope}[shift={(\stabX,-5*\rowpitch)}]
    \begin{scope}[shift={(\blshift,0)}]
      \fmDom{1.85}{0}
      \node[xl] at (1.0,0.85)  {\zgb{kl}};
      \node[xl] at (1.85,0)    {\zg{lq}};
      \node[xl] at (2.7,0.85)  {\zgb{mk}};
      \node[xl] at (3.55,0)    {\zg{om}};
      \node[xl] at (2.7,-0.85) {\zg{qo}};
    \end{scope}
    \fmCrossNoDn{5.0}{1.8}
    \fmCrossU{5.0}{-1.8}{\opRg}
    \begin{scope}[shift={(\brshift,0)}]
      \fmSqL{9.0}{0}
      \node[xl] at (8.15,0)    {\zgb{ip}};
      \fmLegBL{\zg{ni}}
      \fmLegL{\zgb{pn}}
    \end{scope}
  \end{scope}

  \begin{scope}[shift={(\latX,-6*\rowpitch)}, scale=1.45]
    \fill[plaqgray] (-1,-0.5) rectangle (0,0.5);
    \fmFrame
    \draw[red edge] (0,-0.5)--(0,0.5);
    \draw[red edge,-{Stealth[length=2.7mm,width=2.2mm]}] (0,0.5)--(0,-0.2);
  \end{scope}
  \begin{scope}[shift={(\stabX,-6*\rowpitch)}]
    \begin{scope}[shift={(\blshift,0)}]
      \fmDom{1.85}{0}
      \node[xl] at (1.85,0)    {\zg{kq}};
      \node[xl] at (2.7,0.85)  {\zgb{mk}};
      \node[xl] at (3.55,0)    {\zg{om}};
      \node[xl] at (2.7,-0.85) {\zg{qo}};
    \end{scope}
    \fmCrossD{5.0}{1.8}{\opLg}
    \fmCrossU{5.0}{-1.8}{\opRg}
    \begin{scope}[shift={(\brshift,0)}]
      \fmDom{8.15}{0}
      \node[xl] at (7.3,0.85)  {\zgb{rp}};
      \node[xl] at (6.45,0)    {\zgb{pn}};
      \node[xl] at (8.15,0)    {\zgb{ir}};
      \node[xl] at (7.3,-0.85) {\zg{ni}};
    \end{scope}
  \end{scope}

  \begin{scope}[shift={(\latX,-7*\rowpitch)}, scale=1.45]
    \fill[plaqgray] (-1,-0.5) rectangle (0,0.5);
    \fmFrame
    \draw[red edge] (0,-0.5)--(0,0.5);
    \draw[red edge,-{Stealth[length=2.7mm,width=2.2mm]}] (0,-0.5)--(0,0.2);
  \end{scope}
  \begin{scope}[shift={(\stabX,-7*\rowpitch)}]
    \begin{scope}[shift={(\blshift,0)}]
      \fmDom{1.85}{0}
      \node[xl] at (1.85,0)    {\zgb{kq}};
      \node[xl] at (2.7,0.85)  {\zgb{mk}};
      \node[xl] at (3.55,0)    {\zg{om}};
      \node[xl] at (2.7,-0.85) {\zg{qo}};
    \end{scope}
    \fmCrossD{5.0}{1.8}{\opRg}
    \fmCrossU{5.0}{-1.8}{\opLg}
    \begin{scope}[shift={(\brshift,0)}]
      \fmDom{8.15}{0}
      \node[xl] at (7.3,0.85)  {\zgb{rp}};
      \node[xl] at (6.45,0)    {\zgb{pn}};
      \node[xl] at (8.15,0)    {\zg{ir}};
      \node[xl] at (7.3,-0.85) {\zg{ni}};
    \end{scope}
  \end{scope}

  \end{scope}

  \newcommand{\fmMini}{%
    \draw[lat] (-1,-0.5) rectangle (1,0.5);
    \draw[red edge] (0,-0.5)--(0,0.5);}
  \newcommand{\fmMvDot}[2]{\fill[opblue] (#1,#2) circle (2.2pt);}

  \node[font=\Huge] at (-3.9,{-1*\rowpitch+1.9}) {1.};
  \begin{scope}[shift={(-3.9,{-1*\rowpitch})},scale=1.45]
    \fmMini
    \draw[mv] (0,0)--(0.5,0.5);
    \fmMvDot{0.5}{0.5}
    \node[ml] at (0,0) {$R^{\dagger}$};
  \end{scope}
  \node[font=\Huge] at (-3.9,{-2*\rowpitch+1.9}) {2.};
  \begin{scope}[shift={(-3.9,{-2*\rowpitch})},scale=1.45]
    \fmMini
    \draw[mv] (0,0)--(-0.5,-0.5);
    \fmMvDot{-0.5}{-0.5}
    \node[ml] at (0,0) {$L$};
  \end{scope}
  \node[font=\Huge] at (-3.9,{-3*\rowpitch+1.9}) {3.};
  \begin{scope}[shift={(-3.9,{-3*\rowpitch})},scale=1.45]
    \fmMini
    \draw[mv] (0,0)--(1,0);
    \fmMvDot{1}{0}
    \node[ml] at (0,0) {$R$};
  \end{scope}
  \begin{scope}[shift={(\colshift,{4*\rowpitch})}]
  \node[font=\Huge] at (-3.9,{-4*\rowpitch+1.9}) {4.};
  \begin{scope}[shift={(-3.9,{-4*\rowpitch})},scale=1.45]
    \fmMini
    \draw[mv] (0,0)--(-1,0);
    \fmMvDot{-1}{0}
    \node[ml] at (0,0) {$L^{\dagger}$};
  \end{scope}
  \node[font=\Huge] at (-3.9,{-5*\rowpitch+1.9}) {5.};
  \begin{scope}[shift={(-3.9,{-5*\rowpitch})},scale=1.45]
    \fmMini
    \draw[mv] (0,0)--(0.5,-0.5);
    \fmMvDot{0.5}{-0.5}
    \node[ml] at (0,0) {$R$};
  \end{scope}
  \node[font=\Huge] at (-3.9,{-6*\rowpitch+1.9}) {6.};
  \begin{scope}[shift={(-3.9,{-6*\rowpitch})},scale=1.45]
    \fmMini
    \draw[mv] (0,0)--(-0.5,0.5);
    \fmMvDot{-0.5}{0.5}
    \node[ml] at (0,0) {$L^{\dagger}$};
  \end{scope}
  \node[font=\Huge] at (-3.9,{-7*\rowpitch+1.9}) {7.};
  \begin{scope}[shift={(-3.9,{-7*\rowpitch})},scale=1.45]
    \fmMini
    \node[ml] at (0,0) {$\mathcal{S}$};
  \end{scope}
  \end{scope}

\end{tikzpicture}}
\caption{Stabilizer tracking for the flux exchange circuit of Fig.~\ref{fig:flux_exchange_circuit}, with the two endpoints of the dual edge (red) slid alternately: the odd steps slide its head clockwise around $p_{2}$ by $CR$ or $CR^{\dagger}$ (blue) and the even steps slide its tail clockwise around $p_{1}$ by $CL$ or $CL^{\dagger}$, each controlled by a neighbouring edge (dot); the inverse gate $\mathcal{S}$ at step 7 restores the orientation of the dual edge.
Each row lists the move (Step), the stabilizers $B_{p_{1}}$, $A_{w}$ (top), $A_{u}$ (bottom) and $B_{p_{2}}$ with the dual edge in red (Stabilizers), and the rewired graph (Graph); the grey shading tracks the deformed plaquette, and the indices of untouched edges persist between steps.}
\label{fig:fb-stabilizers}
\end{figure}
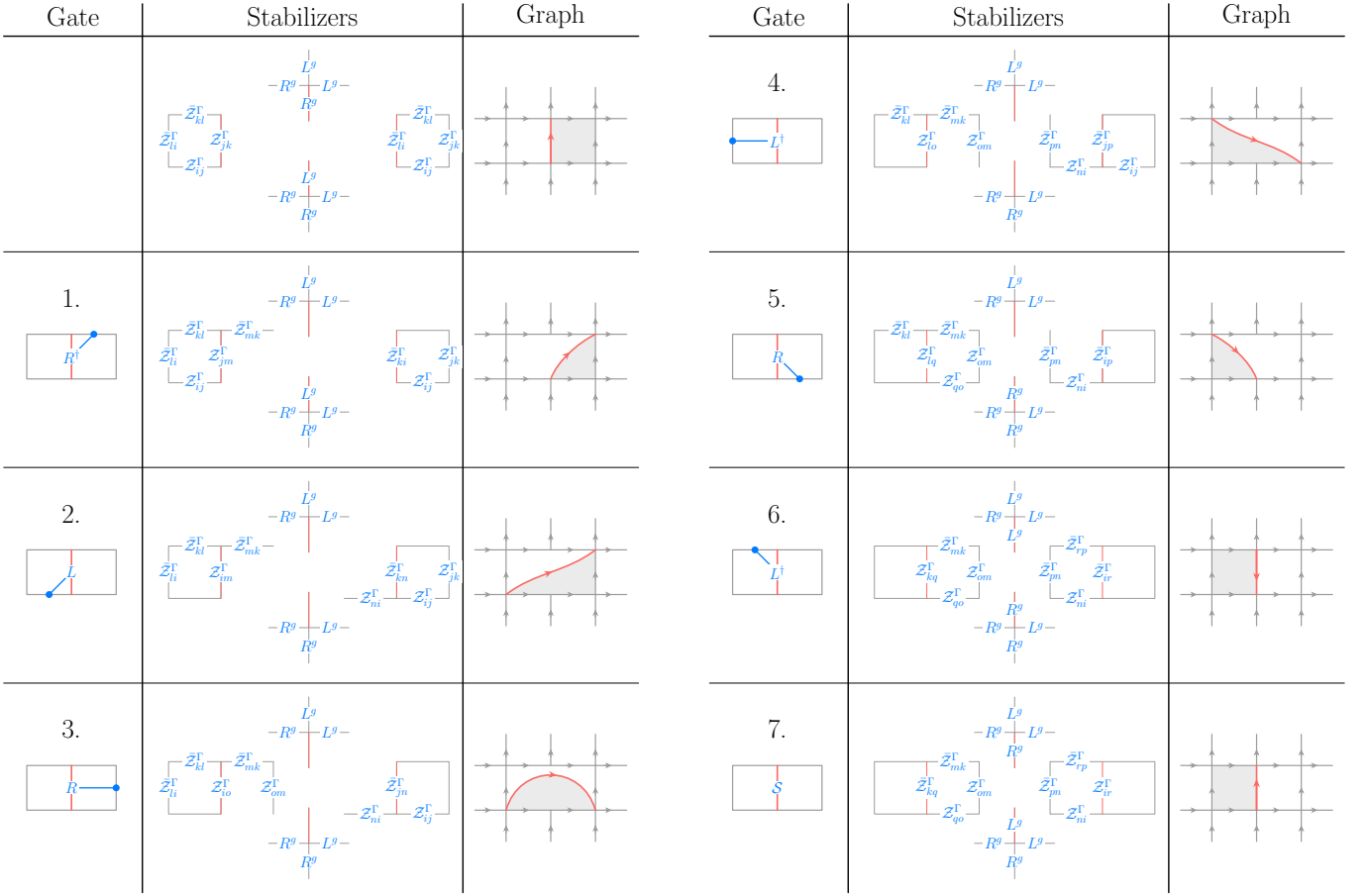

\subsubsection{Ungauging ordering}
\label{app:fb-ungauging}

The three head slides, right multiplications on $d$, commute with the three tail slides, left multiplications on $d$, so the circuit can be reordered freely.
Moving the tail slides of steps 2, 4 and 6 to after $\mathcal{S}_{d}$, which turns them into right multiplications by Eq.~\eqref{eq:controlled-target-inverse}, gives
\begin{equation}\label{eq:fb-unitary}
V_{\mathrm{ex}}=M\,\mathcal{S}_{d}\,\tilde{M}^{\dagger},\qquad
\tilde{M}^{\dagger}\equiv C_{g_{4}}R_{d}\,C_{g_{5}}R_{d}\,C_{g_{6}}R^{\dagger}_{d},\qquad
M\equiv C_{g_{2}}R^{\dagger}_{d}\,C_{g_{3}}R^{\dagger}_{d}\,C_{g_{1}}R_{d},
\end{equation}
where $M$ is the $G$ gauging circuit of the flux fusion protocol of Sec.~\ref{subsec:flux-fusion}, Eq.~\eqref{eq:ff-M-dagger}: by Eq.~\eqref{eq:ff-M-defining} it builds the flux projector of $p_{1}$ out of a single-edge projector on the dual edge alone,
\begin{equation}\label{eq:fb-gauging}
M\,T^{x}_{d}\,M^{\dagger}=B^{x}_{p_{1},u},
\end{equation}
so that $M^{\dagger}$ ungauges $p_{1}$ onto the dual edge and $M$ regauges it.
$\tilde{M}^{\dagger}$ acts on the edges of $p_{2}$ alone, as $g_{d}\mapsto g_{d}\,g_{6}\bar{g}_{5}\bar{g}_{4}=\bar{w}_{2}$ in the notation of Eq.~\eqref{eq:fb-initial-labels}, hence
\begin{equation}\label{eq:fb-ungauged}
\tilde{M}^{\dagger}\,B^{x}_{p_{2},u}\,\tilde{M}=T^{\bar{x}}_{d},\qquad
\tilde{M}^{\dagger}\,B^{y}_{p_{1},u}\,\tilde{M}=\delta_{g_{d}w_{2}w_{1},\,y}.
\end{equation}
The dual edge is now a detached loop based at $u$ that carries the flux of $p_{2}$ on a single qudit, in the sense of Fig.~\ref{fig:ancilla-edges}(a), while $p_{1}$ has grown into $p_{1}\cup p_{2}$; $A_{w}$ has no support on the loop and $A_{u}$ acts on it by the conjugation $L^{g}R^{g}$ of Eq.~\eqref{eq:conjugation-def}, which fixes $\ket{e}$.
This is the counterpart at $p_{2}$ of the ungauging step $M^{\dagger}$ of the flux fusion protocol.
The inverse gate maps $T^{\bar{x}}_{d}$ to $T^{x}_{d}$ and reverses the loop, and $M$ then regauges it onto $p_{1}$ by Eq.~\eqref{eq:fb-gauging}.
The flux exchange is thus an ungauging, an inverse gate on the ungauged edge, and a regauging, as the charge exchange in Eq.~\eqref{eq:ce-unitary}: the three-layer structure of the ribbon operators of Sec.~\ref{subsec:ribbon-rewiring} with $\mathcal{S}_{d}$ as the single-edge operator, except that $M\neq\tilde{M}$, which is what exchanges the two plaquettes instead of restoring them, and the ungauge--reset--regauge structure of the flux fusion channel of Eq.~\eqref{eq:flux-fusion-channel} with the inverse gate in place of the reset.
Ungauging a path of plaquettes one after another onto the successive dual edges leaves those edges as detached loops, each carrying the flux of its plaquette; a $\mathrm{SWAP}$ of two such loops, once slid to a common base point, exchanges the fluxes and regauging in the reverse order restores the lattice, which realizes the braiding of fluxes as a sequential circuit~\cite{chen_sequential_2024}, cf.\ Fig.~\ref{fig:resource-classes}.

\subsubsection{Local neutrality disentangles the dual edge}
\label{app:fb-neutrality}

By Eq.~\eqref{eq:fb-ungauged}, the flux projectors of Eqs.~\eqref{eq:Bg-MPO} and \eqref{eq:BC-MPO} at $p_{2}$ are mapped to the same expressions with $B^{x}_{p_{2},u}$ replaced by $T^{\bar{x}}_{d}$.
If $p_{2}$ carries no flux, as every plaquette does in the ground state, $B^{e}_{p_{2}}$ is mapped to $T^{e}_{d}$: the dual edge is disentangled into a product state, $\tilde{M}^{\dagger}$ is the three-gate plaquette disentangler of entanglement renormalization~\cite{aguado_vidal_2008,aguado_entanglement_2011}, and flux neutrality at $p_{2}$ is the condition for it to disentangle.
The inverse gate acts trivially on $\ket{e}$, and $M$ then builds a flux-free plaquette at $p_{1}$ out of it: exchanging a flux with the vacuum moves the flux across the dual edge.

\section{Charge fusion protocol: derivations}
\label{app:cf}

This appendix derives the operator mappings reported in Sec.~\ref{subsec:charge-fusion} for the charge fusion protocol.
Throughout, $\ell$ denotes the direct edge from $v_{1}$ to $v_{2}$, and $\tilde{A}^{g}$ and $\tilde{A}'^{g}$ denote the action of $A^{g}_{v_{1}}$ and $A^{g}_{v_{2}}$ on the three outer legs of $v_{1}$ and $v_{2}$ respectively, so that $A^{g}_{v_{1}}=L^{g}_{\ell}\otimes\tilde{A}^{g}$ and $A^{g}_{v_{2}}=R^{g}_{\ell}\otimes\tilde{A}'^{g}$ by the convention of Eq.~\eqref{eq:Av-convention}.
The group multiplication gates $L^{g},R^{g}$, the controlled gates $CL,CR$ and the inverse gate $\mathcal{S}$ are those of Appendix~\ref{app:mult-gates}, $P^{\Gamma}_{v}$ is the charge projector of Eq.~\eqref{eq:P-charge-type} (Eq.~\eqref{eq:charge-projector} of the main text), and $P_{v}=\frac{1}{|G|}\sum_{g}A^{g}_{v}$ is the vacuum projector of Eq.~\eqref{eq:P-trivial}.

\subsection{Lemma: closed flux ribbon can be deformed using neutrality}
\label{app:cf-lemma}

Given a state $\ket{\psi}$ with no charge at vertex $v_{1}$, we show that a charge projector that encircles both vertices acts on $\ket{\psi}$ exactly as the charge projector for only $v_{2}$.
No assumption on the charge at $v_{2}$ is needed; in particular, for such a state having a $\Gamma$ charge at $v_{2}$ is the same as the pair having total charge $\Gamma$.
We start with the charge projector (to the $\Gamma$ irrep) for two vertices $v_{1},v_{2}$, which is a closed ribbon
\begin{equation}\label{eq:cf-lemma-closed-ribbon}
P^{\Gamma}_{v_{1}v_{2}} = \frac{d_{\Gamma}}{|G|}\sum_{g\in G}\chi^{\Gamma}(\bar{g})\sum_{x\in G}\;
\vcenter{\hbox{%
\begin{tikzpicture}[line width=0.9pt,baseline]
  \tvlat
  \node[op] at (0,0.45) {$L^{g}$};
  \node[op] at (-0.5,0) {$R^{g}$};
  \node[op] at (0,-0.45) {$R^{g}$};
  \node[op] at (0.75,0) {$T^{x}$};
  \node[op] at (1.5,0.45) {$L^{\bar{x}gx}$};
  \node[op] at (2.0,0) {$L^{\bar{x}gx}$};
  \node[op] at (1.5,-0.45) {$R^{\bar{x}gx}$};
\end{tikzpicture}
}}
\end{equation}
This operator is a closed \textit{flux} ribbon: it transports a flux around the two vertices, and since fluxes braid nontrivially with charges, the closed ribbon detects the total charge enclosed by it, which is exactly why summing it against the character $\chi^{\Gamma}$ gives the charge projector of the region.
Since there is no charge at $v_{1}$, we also have the neutrality condition at the $v_{1}$ vertex
\begin{equation}\label{eq:cf-neutrality}
A^{g}_{v_{1}} \ket{\psi} =   \ket{\psi} , \quad \forall g \in G .
\end{equation}

\begin{lemma}[Deforming a closed flux ribbon using charge neutrality]
Using the charge neutrality condition of Eq.~\eqref{eq:cf-neutrality}, the charge projector over two vertices acts on $\ket{\psi}$ as that of only $v_{2}$, $P^{\Gamma}_{v_{1}v_{2}}\ket{\psi}=P^{\Gamma}_{v_{2}}\ket{\psi}$, or as an operator identity $P^{\Gamma}_{v_{1}v_{2}}P_{v_{1}}=P^{\Gamma}_{v_{2}}P_{v_{1}}$ with the vacuum projector $P_{v_{1}}=\frac{1}{|G|}\sum_{g}A^{g}_{v_{1}}$ of Eq.~\eqref{eq:P-trivial}, where
\begin{equation}\label{eq:cf-lemma-single-vertex}
P^{\Gamma}_{v_{2}} = \frac{d_{\Gamma}}{|G|}\sum_{g\in G}\chi^{\Gamma}(\bar{g})\;
\vcenter{\hbox{%
\begin{tikzpicture}[line width=0.9pt,baseline, op/.append style={fill=thmshade}]
  \tvlat
  \node[op] at (0.75,0) {$R^{g}$};
  \node[op] at (1.5,0.45) {$L^{g}$};
  \node[op] at (2.0,0) {$L^{g}$};
  \node[op] at (1.5,-0.45) {$R^{g}$};
\end{tikzpicture}
}}
\end{equation}
\end{lemma}

\begin{proof}
Since $A^{\bar{g}}_{v_{1}}\ket{\psi}=\ket{\psi}$ for every $g$, we may insert $A^{\bar{g}}_{v_{1}}$ separately inside each summand of $P^{\Gamma}_{v_{1}v_{2}}\ket{\psi}$, choosing in the $g$-th summand the same $g$ as the summation variable.
Writing $A^{\bar{g}}_{v_{1}}=L^{\bar{g}}_{\ell}\otimes\tilde{A}^{\bar{g}}$ by Eq.~\eqref{eq:Av-convention}, this cancels the action on the three outer legs of $v_{1}$ and leaves $L^{\bar{g}}$ on the direct edge, so the operator acting on $\ket{\psi}$ has support only on the direct edge and the edges incident on $v_{2}$:
\begin{equation}
\frac{d_{\Gamma}}{|G|}\sum_{g\in G}\chi^{\Gamma}(\bar{g})\sum_{x\in G}\;
\vcenter{\hbox{%
\begin{tikzpicture}[line width=0.9pt,baseline]
  \tvlat
  \node[op] at (0.75,0) {$T^{x}L^{\bar{g}}$};
  \node[op] at (1.5,0.45) {$L^{\bar{x}gx}$};
  \node[op] at (2.0,0) {$L^{\bar{x}gx}$};
  \node[op] at (1.5,-0.45) {$R^{\bar{x}gx}$};
\end{tikzpicture}
}}
\end{equation}
Next, we relabel $g\to xg\bar{x}$ to simplify the action on the three edges, and we use the fact that the trace $\chi^{\Gamma}(\cdot)$ is a class function, hence invariant under conjugation, $\chi^{\Gamma}(\overline{xg\bar{x}})=\chi^{\Gamma}(\bar{g})$:
\begin{equation}
\frac{d_{\Gamma}}{|G|}\sum_{g\in G}\chi^{\Gamma}(\bar{g})\sum_{x\in G}\;
\vcenter{\hbox{%
\begin{tikzpicture}[line width=0.9pt,baseline]
  \tvlat
  \node[op] at (0.75,0) {$T^{x}L^{x\bar{g}\bar{x}}$};
  \node[op] at (1.5,0.45) {$L^{g}$};
  \node[op] at (2.0,0) {$L^{g}$};
  \node[op] at (1.5,-0.45) {$R^{g}$};
\end{tikzpicture}
}}
\end{equation}
Let us focus on simplifying the operator $\sum_{x}T^{x}L^{x\bar{g}\bar{x}}$ on the direct edge,
\begin{align}
\sum_{x \in G}T^{x} L^{x\bar{g}\bar{x}} &= \sum_{x \in G}  \ketbra{x}{x} L^{x\bar{g}\bar{x}} =  \sum_{x \in G}  \ketbra{x}{xg} \xmapsto{x\to x\bar{g}} \sum_{x \in G}  \ketbra{x\bar{g}}{x} = R^{g},
\end{align}
where the action of $L^{h}$ on the right is given by $(L^{h}\ket{\psi})^{\dagger} =\bra{\psi} (L^{h})^{\dagger}=\bra{\psi} L^{\bar{h}}$, Eq.~\eqref{eq:LR-inverse}, and the last equality is the definition of $R^{g}$ in Eq.~\eqref{eq:LR-def}.
Incorporating the above simplification into the stabilizer, we obtain
\begin{equation}
\frac{d_{\Gamma}}{|G|}\sum_{g\in G}\chi^{\Gamma}(\bar{g})\;
\vcenter{\hbox{%
\begin{tikzpicture}[line width=0.9pt,baseline]
  \tvlat
  \node[op] at (0.75,0) {$R^{g}$};
  \node[op] at (1.5,0.45) {$L^{g}$};
  \node[op] at (2.0,0) {$L^{g}$};
  \node[op] at (1.5,-0.45) {$R^{g}$};
\end{tikzpicture}
}}
\end{equation}
i.e. $P^{\Gamma}_{v_{1}v_{2}}\ket{\psi}=P^{\Gamma}_{v_{2}}\ket{\psi}$ with $P^{\Gamma}_{v_{2}}$ exactly the charge projector on a single vertex, Eqs.~\eqref{eq:P-charge-type} and \eqref{eq:cf-lemma-single-vertex}, which is the claim.
The identity holds on the neutral subspace of $v_{1}$ only; the two projectors are not equal on the full Hilbert space.
The derivation is linear in the class function $\chi^{\Gamma}$, so it applies verbatim to any class function in place of $d_{\Gamma}\chi^{\Gamma}$.
\end{proof}

\subsection{Stabilizer mapping derivations}

\subsubsection{Mapping of single vertex projectors}
\label{app:cf-vertex}

We track a genuine stabilizer of the state (eigenvalue $+1$ at every step, with no truncation and no rescaling) built from the two single vertex charge projectors $P^{\Gamma_{1}}_{v_{1}}$ and $P^{\Gamma_{2}}_{v_{2}}$ of Eq.~\eqref{eq:P-charge-type} stabilizing the initial state,
\begin{equation}\label{eq:cf-P-v1}
P^{\Gamma_{1}}_{v_{1}} = \frac{d_{\Gamma_{1}}}{|G|}\sum_{g\in G}\chi^{\Gamma_{1}}(\bar{g})\;
\vcenter{\hbox{%
\begin{tikzpicture}[line width=0.9pt,baseline]
  \tvlat
  \node[op] at (0,0.45) {$L^{g}$};
  \node[op] at (-0.5,0) {$R^{g}$};
  \node[op] at (0,-0.45) {$R^{g}$};
  \node[op] at (0.75,0) {$L^{g}$};
\end{tikzpicture}
}}
\end{equation}
\begin{equation}\label{eq:cf-P-v2}
P^{\Gamma_{2}}_{v_{2}} = \frac{d_{\Gamma_{2}}}{|G|}\sum_{g\in G}\chi^{\Gamma_{2}}(\bar{g})\;
\vcenter{\hbox{%
\begin{tikzpicture}[line width=0.9pt,baseline]
  \tvlat
  \node[op] at (0.75,0) {$R^{g}$};
  \node[op] at (1.5,0.45) {$L^{g}$};
  \node[op] at (2.0,0) {$L^{g}$};
  \node[op] at (1.5,-0.45) {$R^{g}$};
\end{tikzpicture}
}}
\end{equation}

The stabilizer we carry is the support projector of the two-vertex diagonal gauge transformation $A^{g}_{v_{1}}A^{\bar{x}gx}_{v_{2}}$,
\begin{equation}\label{eq:cf-Pi-x}
\Pi_{x}\equiv\frac{1}{|G|}\sum_{g\in G}\phi(\bar{g})\,A^{g}_{v_{1}}A^{\bar{x}gx}_{v_{2}}
=\frac{1}{|G|}\sum_{g\in G}\phi(\bar{g})\;
\vcenter{\hbox{%
\begin{tikzpicture}[line width=0.9pt,baseline]
  \tvlat
  \node[op] at (0,0.45) {$L^{g}$};
  \node[op] at (-0.5,0) {$R^{g}$};
  \node[op] at (0,-0.45) {$R^{g}$};
  \node[op] at (0.75,0) {$L^{g}R^{\bar{x}gx}$};
  \node[op] at (1.5,0.45) {$L^{\bar{x}gx}$};
  \node[op] at (2.0,0) {$L^{\bar{x}gx}$};
  \node[op] at (1.5,-0.45) {$R^{\bar{x}gx}$};
\end{tikzpicture}
}}
\qquad
\phi\equiv\sum_{\Lambda\,:\,N^{\Lambda}_{\Gamma_{1}\Gamma_{2}}\neq0}d_{\Lambda}\chi^{\Lambda},
\end{equation}
which stabilizes the initial state for every $x\in G$.
It is the outcome-labelled form of the tensor product ribbon of Eq.~\eqref{eq:cf-tensor-ribbon}: since $L^{g}R^{\bar{x}gx}$ fixes $\ket{x}$ by Eq.~\eqref{eq:LR-def}, we have $T^{x}_{\ell}\,A^{g}_{v_{1}}A^{\bar{x}gx}_{v_{2}}=\tilde{A}^{g}\otimes T^{x}_{\ell}\otimes\tilde{A}'^{\bar{x}gx}$, hence
\begin{equation}\label{eq:cf-ribbon-vs-Pi-x}
P^{\Gamma_{1}\otimes\Gamma_{2}}_{v_{1}v_{2}}=\sum_{x\in G}T^{x}_{\ell}\,\Pi_{x}
=\frac{1}{|G|}\sum_{g\in G}\phi(\bar{g})\sum_{x\in G}\;
\vcenter{\hbox{%
\begin{tikzpicture}[line width=0.9pt,baseline]
  \tvlat
  \node[op] at (0,0.45) {$L^{g}$};
  \node[op] at (-0.5,0) {$R^{g}$};
  \node[op] at (0,-0.45) {$R^{g}$};
  \node[op] at (0.75,0) {$T^{x}$};
  \node[op] at (1.5,0.45) {$L^{\bar{x}gx}$};
  \node[op] at (2.0,0) {$L^{\bar{x}gx}$};
  \node[op] at (1.5,-0.45) {$R^{\bar{x}gx}$};
\end{tikzpicture}
}}
\end{equation}
whose $x$-th summand is $T^{x}_{\ell}\,\Pi_{x}$.
The two agree on each sector of the direct edge, $T^{x}_{\ell}\,P^{\Gamma_{1}\otimes\Gamma_{2}}_{v_{1}v_{2}}=T^{x}_{\ell}\,\Pi_{x}$, and both commute with the measurement projector $T^{x}_{\ell}$: the ribbon is diagonal in the group element basis of the direct edge, so $[P^{\Gamma_{1}\otimes\Gamma_{2}}_{v_{1}v_{2}},T^{x}_{\ell}]=0$ for every $x$, while $\Pi_{x}$ (with the $x$ of the measurement outcome) commutes with $T^{x}_{\ell}$ because its edge action $L^{g}R^{\bar{x}gx}$ fixes $\ket{x}$.
We carry $\Pi_{x}$ because it is built directly from the two-vertex gauge transformations, which makes the reduction to the fusion channels of $\Gamma_{1}\otimes\Gamma_{2}$ transparent; the ribbon itself can equally be tracked through the protocol, since $W^{\dagger}P^{\Gamma_{1}\otimes\Gamma_{2}}_{v_{1}v_{2}}W=\mathds{1}_{\ell}\otimes\Pi'$ has no support on the direct edge and passes through the reset unchanged (Eq.~\eqref{eq:cf-Pi-prime} and Appendix~\ref{app:cf-ribbon}).
Indeed, for fixed $x$ the operators $g\mapsto A^{g}_{v_{1}}A^{\bar{x}gx}_{v_{2}}$ form a representation of $G$ that commutes with $P^{\Gamma_{1}}_{v_{1}}$ and $P^{\Gamma_{2}}_{v_{2}}$ and acts on their joint image as $\Gamma_{1}(g)\otimes\Gamma_{2}(\bar{x}gx)\cong\Gamma_{1}(g)\otimes\Gamma_{2}(g)$, so $\Pi_{x}$ projects onto the channels $\Lambda\subset\Gamma_{1}\otimes\Gamma_{2}$.

\paragraph{Step 1: Ungauging circuit $W^{\dagger}$.}

Using $R^{g}_{\ell}T^{h}_{\ell}=T^{h\bar{g}}_{\ell}R^{g}_{\ell}$, which follows from Eq.~\eqref{eq:LR-def},
\begin{equation}\label{eq:cf-Av2-ungauged}
W^{\dagger}A^{g}_{v_{2}}W=\sum_{h\in G}\;
\vcenter{\hbox{%
\begin{tikzpicture}[line width=0.9pt,baseline]
  \tvlat
  \node[op] at (0,0.45) {$L^{g\bar{h}}L^{h}$};
  \draw[latgray] (-1.5,0) -- (-1.0,0);
  \node[op] at (-0.75,0) {$R^{g\bar{h}}R^{h}$};
  \node[op] at (0,-0.45) {$R^{g\bar{h}}R^{h}$};
  \node[op] at (0.75,0) {$T^{h\bar{g}}R^{g}$};
  \node[op] at (1.5,0.45) {$L^{g}$};
  \node[op] at (2.0,0) {$L^{g}$};
  \node[op] at (1.5,-0.45) {$R^{g}$};
\end{tikzpicture}
}}
=\;
\vcenter{\hbox{%
\begin{tikzpicture}[line width=0.9pt,baseline]
  \tvlat
  \node[op] at (0,0.45) {$L^{g}$};
  \node[op] at (-0.5,0) {$R^{g}$};
  \node[op] at (0,-0.45) {$R^{g}$};
  \node[op] at (0.75,0) {$R^{g}$};
  \node[op] at (1.5,0.45) {$L^{g}$};
  \node[op] at (2.0,0) {$L^{g}$};
  \node[op] at (1.5,-0.45) {$R^{g}$};
\end{tikzpicture}
}}
\end{equation}
i.e. $W^{\dagger}A^{g}_{v_{2}}W=R^{g}_{\ell}\otimes\tilde{A}^{g}\tilde{A}'^{g}$ since $\tilde{A}^{g\bar{h}}\tilde{A}^{h}=\tilde{A}^{g}$ and $\sum_{h}T^{h\bar{g}}_{\ell}=\mathds{1}_{\ell}$,
so the right vertex projector $P^{\Gamma_{2}}_{v_{2}}$ becomes
\begin{equation}
W^{\dagger}P^{\Gamma_{2}}_{v_{2}}W=\frac{d_{\Gamma_{2}}}{|G|}\sum_{g\in G}\chi^{\Gamma_{2}}(\bar{g})\;
\vcenter{\hbox{%
\begin{tikzpicture}[line width=0.9pt,baseline]
  \tvlat
  \node[op] at (0,0.45) {$L^{g}$};
  \node[op] at (-0.5,0) {$R^{g}$};
  \node[op] at (0,-0.45) {$R^{g}$};
  \node[op] at (0.75,0) {$R^{g}$};
  \node[op] at (1.5,0.45) {$L^{g}$};
  \node[op] at (2.0,0) {$L^{g}$};
  \node[op] at (1.5,-0.45) {$R^{g}$};
\end{tikzpicture}
}}
\end{equation}

Gate by gate, this follows from the commutation relations of $L^{g}$ and $R^{g}$ on the control qudit, Eqs.~\eqref{eq:comm-L-control} and \eqref{eq:comm-R-control}, together with the decomposition $C_{c}L^{\dagger}_{t}=S_{c}(C_{c}L_{t})S_{c}^{\dagger}$ of the inverse gates, Eq.~\eqref{eq:controlled-inverse}, and $\mathcal{S}R^{g}\mathcal{S}^{\dagger}=L^{g}$ from Eqs.~\eqref{eq:LR-def} and \eqref{eq:inverse-gate-def}, where $S_{c}$ is the inverse gate on the control qudit.
Its action $\tilde{A}^{g}\tilde{A}'^{g}$ on the six outer legs is the simultaneous gauge transformation of the two-vertex cluster, i.e. the global symmetry of the ungauged cluster, accompanied by $R^{g}$ on the direct edge; the original simultaneous transformation $A^{g}_{v_{1}}A^{g}_{v_{2}}$ ungauges to $L^{g}_{\ell}R^{g}_{\ell}\otimes\tilde{A}^{g}\tilde{A}'^{g}$ by Eq.~\eqref{eq:cf-Av1-ungauged}.

Whereas, using $L^{h}_{\ell}T^{k}_{\ell}=T^{hk}_{\ell}L^{h}_{\ell}$ from Eq.~\eqref{eq:LR-def},
\begin{equation}\label{eq:cf-Av1-ungauged}
W^{\dagger}A^{h}_{v_{1}}W=\sum_{k\in G}\;
\vcenter{\hbox{%
\begin{tikzpicture}[line width=0.9pt,baseline]
  \tvlat
  \node[op] at (0,0.45) {$L^{\overline{hk}}L^{h}L^{k}$};
  \draw[latgray] (-1.5,0) -- (-1.0,0);
  \node[op] at (-0.75,0) {$R^{\overline{hk}}R^{h}R^{k}$};
  \node[op] at (0,-0.45) {$R^{\overline{hk}}R^{h}R^{k}$};
  \node[op] at (0.75,0) {$L^{h}T^{k}$};
\end{tikzpicture}
}}
=\;
\vcenter{\hbox{%
\begin{tikzpicture}[line width=0.9pt,baseline]
  \tvlat
  \node[op] at (0.75,0) {$L^{h}$};
\end{tikzpicture}
}}
\end{equation}
i.e. $W^{\dagger}A^{h}_{v_{1}}W=L^{h}_{\ell}$ since $\tilde{A}^{\overline{hk}}\tilde{A}^{h}\tilde{A}^{k}=\mathds{1}$ and $\sum_{k}T^{k}_{\ell}=\mathds{1}_{\ell}$, equivalently $WL^{h}_{\ell}W^{\dagger}=A^{h}_{v_{1}}$, which is the commutation relation of $L^{h}$ on the control qudit, Eq.~\eqref{eq:comm-L-control}, applied to the three gates of $W$; so the left vertex projector $P^{\Gamma_{1}}_{v_{1}}$ becomes
\begin{equation}
W^{\dagger}P^{\Gamma_{1}}_{v_{1}}W=P^{\Gamma_{1}}_{L,\ell}=\frac{d_{\Gamma_{1}}}{|G|}\sum_{h\in G}\chi^{\Gamma_{1}}(\bar{h})\;
\vcenter{\hbox{%
\begin{tikzpicture}[line width=0.9pt,baseline]
  \tvlat
  \node[op] at (0.75,0) {$L^{h}$};
\end{tikzpicture}
}}
\end{equation}
i.e. the $\Gamma_{1}$ charge is now carried by the direct edge alone.

\paragraph{Step 2a: Measure the direct edge in the group element basis.}

We implement the reset of step 2 as a measurement followed by a flip and a preparation, Eq.~\eqref{eq:cf-reset-kraus}, and track the stabilizer through both parts.
Say we measure the edge to be in the $\ket{x}$ state.
A term $L^{h}R^{g}$ of the product of the two modified projectors sends $\ket{x}\mapsto\ket{hx\bar{g}}$ by Eq.~\eqref{eq:LR-def}, so only terms with $h=xg\bar{x}$ commute with $T^{x}_{\ell}$; keeping just those terms would be a truncation that does not preserve the eigenvalue.
An operator that commutes with $T^{x}_{\ell}$ as a whole is the ungauged $\Pi_{x}$, which after relabeling $g\to xg\bar{x}$ reads
\begin{equation}
W^{\dagger}\Pi_{x}W=\frac{1}{|G|}\sum_{g\in G}\phi(\bar{g})\;
\vcenter{\hbox{%
\begin{tikzpicture}[line width=0.9pt,baseline]
  \tvlat
  \node[op] at (0,0.45) {$L^{g}$};
  \node[op] at (-0.5,0) {$R^{g}$};
  \node[op] at (0,-0.45) {$R^{g}$};
  \node[op] at (0.75,0) {$L^{xg\bar{x}}R^{g}$};
  \node[op] at (1.5,0.45) {$L^{g}$};
  \node[op] at (2.0,0) {$L^{g}$};
  \node[op] at (1.5,-0.45) {$R^{g}$};
\end{tikzpicture}
}}
\end{equation}
Every term satisfies $h=xg\bar{x}$ automatically and its edge factor fixes the outcome, $L^{xg\bar{x}}R^{g}\ket{x}=\ket{x}$, so the operator commutes with $T^{x}_{\ell}$ term by term; the weighted sum, which is the operator with eigenvalue $+1$ on the state (individual summands are not), therefore survives the measurement as a stabilizer.
Note that $W^{\dagger}$ has moved the $x$-conjugation from the $v_{1}$ legs onto the direct edge, which is why no outcome-dependent correction is needed later.

\paragraph{Step 2b: Flip the outcome and prepare $\ket{+}$.}

The direct edge is reset by applying $R^{x}$ and then $U_{+}$.
Conjugating the edge action by $R^{x}$ gives $L^{xg\bar{x}}R^{xg\bar{x}}=C^{xg\bar{x}}$, the conjugation gate of Eq.~\eqref{eq:conjugation-def}, which fixes $\ket{e}$.
After $R^{x}$ the edge is in the product state $\ket{e}$, fixed by every edge factor, so the edge action can be dropped term by term at this point; $U_{+}$ then acts on a product state and prepares $\ket{+}$ independently (there is no need to conjugate the edge factor through $U_{+}$, and indeed $F^{\dagger}C^{g}F\neq C^{g}$ in general), leaving
\begin{equation}\label{eq:cf-Pi-prime}
\Pi'=\frac{1}{|G|}\sum_{g\in G}\phi(\bar{g})\;
\vcenter{\hbox{%
\begin{tikzpicture}[line width=0.9pt,baseline]
  \tvlat
  \node[op] at (0,0.45) {$L^{g}$};
  \node[op] at (-0.5,0) {$R^{g}$};
  \node[op] at (0,-0.45) {$R^{g}$};
  \node[op] at (1.5,0.45) {$L^{g}$};
  \node[op] at (2.0,0) {$L^{g}$};
  \node[op] at (1.5,-0.45) {$R^{g}$};
\end{tikzpicture}
}}
\end{equation}
which is $x$-independent, together with the reset stabilizer $\ketbra{+}{+}_{\ell}=\frac{1}{|G|}\sum_{g'}L^{g'}_{\ell}$.

The same conclusion follows without naming an outcome: by Peter--Weyl the image of $P^{\Gamma_{1}}_{L,\ell}$ is $V_{\Gamma_{1}}\otimes V^{*}_{\Gamma_{1}}$ with $L^{g}_{\ell}=\Gamma_{1}(g)\otimes\mathds{1}$ and $R^{g}_{\ell}=\mathds{1}\otimes\bar{\Gamma}_{1}(g)$, so the $\Gamma_{2}$ stabilizer of Step 1 says the six outer legs carry charge $\Lambda$ with $N^{\Gamma_{2}}_{\bar{\Gamma}_{1}\Lambda}=N^{\Lambda}_{\Gamma_{1}\Gamma_{2}}\neq0$ under $\tilde{A}^{g}\tilde{A}'^{g}$; tracing out the edge leaves $\Pi'$.

\paragraph{Step 3: Regauging circuit $W$.}

Using $W\tilde{A}^{g}W^{\dagger}=\sum_{h}T^{h}_{\ell}\otimes\tilde{A}^{hg\bar{h}}$, which is the commutation relation of $L^{g}$ and $R^{g}$ on the target qudit, Eqs.~\eqref{eq:comm-L-target} and \eqref{eq:comm-R-target}, applied to the three gates of $W$, the stabilizer becomes the closed flux ribbon with class function $\phi$,
\begin{equation}
W\Pi'W^{\dagger}=\frac{1}{|G|}\sum_{g\in G}\phi(\bar{g})\sum_{h}\;
\vcenter{\hbox{%
\begin{tikzpicture}[line width=0.9pt,baseline]
  \tvlat
  \node[op] at (0,0.45) {$L^{hg\bar{h}}$};
  \node[op] at (-0.5,0) {$R^{hg\bar{h}}$};
  \node[op] at (0,-0.45) {$R^{hg\bar{h}}$};
  \node[op] at (0.75,0) {$T^{h}$};
  \node[op] at (1.5,0.45) {$L^{g}$};
  \node[op] at (2.0,0) {$L^{g}$};
  \node[op] at (1.5,-0.45) {$R^{g}$};
\end{tikzpicture}
}}
\end{equation}
Relabeling $g\to\bar{h}gh$ ($\phi$ is a class function) and renaming $h\to x$, this is exactly the tensor product ribbon $P^{\Gamma_{1}\otimes\Gamma_{2}}_{v_{1}v_{2}}$ of Eq.~\eqref{eq:cf-tensor-ribbon}, which establishes the mapping stated in Sec.~\ref{subsec:charge-fusion}.
The reset stabilizer regauges to the vacuum projector $P_{v_{1}}=\frac{1}{|G|}\sum_{g'}A^{g'}_{v_{1}}$ of Eq.~\eqref{eq:P-trivial}, and by the lemma of Appendix~\ref{app:cf-lemma} (linear in the class function) the ribbon deforms to
\begin{equation}\label{eq:cf-Pi-v2}
\Pi_{v_{2}}=\sum_{\Lambda\,:\,N^{\Lambda}_{\Gamma_{1}\Gamma_{2}}\neq0}P^{\Lambda}_{v_{2}}
=\frac{1}{|G|}\sum_{g\in G}\phi(\bar{g})\;
\vcenter{\hbox{%
\begin{tikzpicture}[line width=0.9pt,baseline]
  \tvlat
  \node[op] at (0.75,0) {$R^{g}$};
  \node[op] at (1.5,0.45) {$L^{g}$};
  \node[op] at (2.0,0) {$L^{g}$};
  \node[op] at (1.5,-0.45) {$R^{g}$};
\end{tikzpicture}
}}
\end{equation}
the multiplicity-blind sum of orthogonal charge projectors over the fusion channels, with all weights equal to $1$.
The closed-ribbon tracking of Appendix~\ref{app:cf-ribbon} is the same computation with $d_{\Gamma}\chi^{\Gamma}$ in place of $\phi$.

\subsubsection{Mapping of closed loop ribbon operator}
\label{app:cf-ribbon}

\paragraph{Step 1: Ungauging circuit $W^{\dagger}$.}

The closed ribbon operator around the two vertices projecting to the charge sector $\Gamma$ is $P^{\Gamma}_{v_{1}v_{2}}$ of Eq.~\eqref{eq:cf-closed-ribbon}, i.e. $P^{\Gamma}_{v_{1}v_{2}}=\frac{d_{\Gamma}}{|G|}\sum_{g}\chi^{\Gamma}(\bar{g})\sum_{x}\tilde{A}^{g}\otimes T^{x}_{\ell}\otimes \tilde{A}'^{\bar{x}gx}$.
Conjugating by $W^{\dagger}$ of Eq.~\eqref{eq:cf-W-dagger}, the $T$'s force both controls to equal $x$, so
\begin{equation}
W^{\dagger}P^{\Gamma}_{v_{1}v_{2}}W=\frac{d_{\Gamma}}{|G|}\sum_{g\in G}\chi^{\Gamma}(\bar{g})\sum_{x\in G}\;
\vcenter{\hbox{%
\begin{tikzpicture}[line width=0.9pt,baseline]
  \tvlat
  \node[op] at (0,0.45) {$L^{\bar{x}}L^{g}L^{x}$};
  \draw[latgray] (-1.7,0) -- (-1.0,0);
  \node[op] at (-0.9,0) {$R^{\bar{x}}R^{g}R^{x}$};
  \node[op] at (0,-0.45) {$R^{\bar{x}}R^{g}R^{x}$};
  \node[op] at (0.75,0) {$T^{x}$};
  \node[op] at (1.5,0.45) {$L^{\bar{x}gx}$};
  \node[op] at (2.0,0) {$L^{\bar{x}gx}$};
  \node[op] at (1.5,-0.45) {$R^{\bar{x}gx}$};
\end{tikzpicture}
}}
=\frac{d_{\Gamma}}{|G|}\sum_{g\in G}\chi^{\Gamma}(\bar{g})\sum_{x\in G}\;
\vcenter{\hbox{%
\begin{tikzpicture}[line width=0.9pt,baseline]
  \tvlat
  \node[op] at (0,0.45) {$L^{\bar{x}gx}$};
  \node[op] at (-0.5,0) {$R^{\bar{x}gx}$};
  \node[op] at (0,-0.45) {$R^{\bar{x}gx}$};
  \node[op] at (0.75,0) {$T^{x}$};
  \node[op] at (1.5,0.45) {$L^{\bar{x}gx}$};
  \node[op] at (2.0,0) {$L^{\bar{x}gx}$};
  \node[op] at (1.5,-0.45) {$R^{\bar{x}gx}$};
\end{tikzpicture}
}}
\end{equation}
Relabeling $g\to xg\bar{x}$ (class function) and using $\sum_{x}T^{x}_{\ell}=\mathds{1}_{\ell}$, the direct edge drops out of the closed ribbon entirely:
\begin{equation}
\frac{d_{\Gamma}}{|G|}\sum_{g\in G}\chi^{\Gamma}(\bar{g})\;
\vcenter{\hbox{%
\begin{tikzpicture}[line width=0.9pt,baseline]
  \tvlat
  \node[op] at (0,0.45) {$L^{g}$};
  \node[op] at (-0.5,0) {$R^{g}$};
  \node[op] at (0,-0.45) {$R^{g}$};
  \node[op] at (1.5,0.45) {$L^{g}$};
  \node[op] at (2.0,0) {$L^{g}$};
  \node[op] at (1.5,-0.45) {$R^{g}$};
\end{tikzpicture}
}}
\end{equation}
In the same way the vertex operator at $v_{1}$ is moved onto the direct edge alone, Eq.~\eqref{eq:cf-Av1-ungauged}, which is the ungauging map: after $W^{\dagger}$, the charge of $v_{1}$ is carried by the direct edge as its left-regular representation content.
The direct edge drops out of the closed ribbon and of the plaquette operators (Appendix~\ref{app:cf-plaquette}), but it remains in the right vertex constraint through the factor $R^{g}_{\ell}$ in Eq.~\eqref{eq:cf-Av2-ungauged}.

\paragraph{Step 2: Reset the direct edge to $\ket{+}$.}

Since the closed ribbon no longer has support on the direct edge, it is untouched.
The reset adds the stabilizer $\ketbra{+}{+}_{\ell}=\frac{1}{|G|}\sum_{g'}L^{g'}_{\ell}$, which by the ungauging map is the statement that $v_{1}$ will carry trivial charge after regauging.
The stabilizer set now contains the product of the two, $\ketbra{+}{+}_{\ell}\otimes W^{\dagger}P^{\Gamma}_{v_{1}v_{2}}W$,
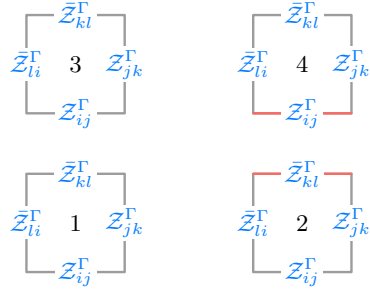
\begin{equation}
\frac{d_{\Gamma}}{|G|^{2}}\sum_{g,g'\in G}\chi^{\Gamma}(\bar{g})\;
\vcenter{\hbox{%
\begin{tikzpicture}[line width=0.9pt,baseline]
  \tvlat
  \node[op] at (0,0.45) {$L^{g}$};
  \node[op] at (-0.5,0) {$R^{g}$};
  \node[op] at (0,-0.45) {$R^{g}$};
  \node[op] at (0.75,0) {$L^{g'}$};
  \node[op] at (1.5,0.45) {$L^{g}$};
  \node[op] at (2.0,0) {$L^{g}$};
  \node[op] at (1.5,-0.45) {$R^{g}$};
\end{tikzpicture}
}}
\end{equation}

\paragraph{Step 3: Regauging circuit $W$.}

Using $W\,\tilde{A}^{g}\,W^{\dagger}=\sum_{h}T^{h}_{\ell}\otimes \tilde{A}^{hg\bar{h}}$ (Eqs.~\eqref{eq:comm-L-target} and \eqref{eq:comm-R-target}) and $W\,L^{g'}_{\ell}\,W^{\dagger}=\sum_{h}L^{g'}T^{h}_{\ell}\otimes \tilde{A}^{g'}$ (Eq.~\eqref{eq:comm-L-control}), the stabilizer transforms into (sum over $h$ implicit)
\begin{equation}
\frac{d_{\Gamma}}{|G|^{2}}\sum_{g,g'\in G}\chi^{\Gamma}(\bar{g})\;
\vcenter{\hbox{%
\begin{tikzpicture}[line width=0.9pt,baseline]
  \tvlat
  \node[op] at (0,0.45) {$L^{g'h g \bar{h}}$};
  \node[op] at (-0.5,0) {$R^{g'h g \bar{h}}$};
  \node[op] at (0,-0.45) {$R^{g'h g \bar{h}}$};
  \node[op] at (0.75,0) {$L^{g'}T^{h}$};
  \node[op] at (1.5,0.45) {$L^{g}$};
  \node[op] at (2.0,0) {$L^{g}$};
  \node[op] at (1.5,-0.45) {$R^{g}$};
\end{tikzpicture}
}}
\end{equation}
The $g'$-dependent part factors out as the vacuum projector $P_{v_{1}}=\frac{1}{|G|}\sum_{g'}A^{g'}_{v_{1}}$ of Eq.~\eqref{eq:P-trivial} (equivalently, $W$ maps the reset stabilizer $\frac{1}{|G|}\sum_{g'}L^{g'}_{\ell}$ back to $\frac{1}{|G|}\sum_{g'}A^{g'}_{v_{1}}$),
\begin{equation}
\left(\frac{1}{|G|}\sum_{g'} A^{g'}_{v_1}\right) \frac{d_{\Gamma}}{|G|}\sum_{g\in G}\chi^{\Gamma}(\bar{g})\;
\vcenter{\hbox{%
\begin{tikzpicture}[line width=0.9pt,baseline]
  \tvlat
  \node[op] at (0,0.45) {$L^{h g \bar{h}}$};
  \node[op] at (-0.5,0) {$R^{h g \bar{h}}$};
  \node[op] at (0,-0.45) {$R^{h g \bar{h}}$};
  \node[op] at (0.75,0) {$T^{h}$};
  \node[op] at (1.5,0.45) {$L^{g}$};
  \node[op] at (2.0,0) {$L^{g}$};
  \node[op] at (1.5,-0.45) {$R^{g}$};
\end{tikzpicture}
}}
\end{equation}
and by the lemma of Appendix~\ref{app:cf-lemma} the closed ribbon reduces to the single vertex charge projector $P^{\Gamma}_{v_{2}}$ of Eqs.~\eqref{eq:P-charge-type} and \eqref{eq:cf-single-vertex-v2}, so the post-protocol state is stabilized by the vacuum projector at $v_{1}$ together with $P^{\Gamma}_{v_{2}}$.
For a reducible total charge $\Gamma_{1}\otimes\Gamma_{2}$, the genuine stabilizer is the support projector of Eq.~\eqref{eq:cf-Pi-v2}.

\paragraph{Heisenberg picture.}

The Kraus operators of Eq.~\eqref{eq:cf-channel} give $\mathcal{E}^{\dagger}(O)=\sum_{x}K^{\dagger}_{x}OK_{x}=W\bigl(\mathds{1}_{\ell}\otimes\bra{+}W^{\dagger}OW\ket{+}_{\ell}\bigr)W^{\dagger}$ for any operator $O$.
For $O=P^{\Gamma}_{v_{2}}$, Eq.~\eqref{eq:cf-Av2-ungauged} and $\bra{+}R^{g}_{\ell}\ket{+}=1$ give $\bra{+}W^{\dagger}P^{\Gamma}_{v_{2}}W\ket{+}_{\ell}=\frac{d_{\Gamma}}{|G|}\sum_{g}\chi^{\Gamma}(\bar{g})\,\tilde{A}^{g}\tilde{A}'^{g}$.
Conjugating back with $W\tilde{A}^{g}W^{\dagger}=\sum_{h}T^{h}_{\ell}\otimes\tilde{A}^{hg\bar{h}}$ (Eqs.~\eqref{eq:comm-L-target} and \eqref{eq:comm-R-target}), while $\tilde{A}'^{g}$ commutes with $W$, and relabeling $g\to\bar{h}gh$ with $h\to x$ returns the closed ribbon of Eq.~\eqref{eq:cf-closed-ribbon}: $\mathcal{E}^{\dagger}(P^{\Gamma}_{v_{2}})=P^{\Gamma}_{v_{1}v_{2}}$, which is Eq.~\eqref{eq:cf-heisenberg}.
No neutrality assumption is used, so this identity holds on the full Hilbert space, in contrast to the lemma of Appendix~\ref{app:cf-lemma}.

\subsubsection{Mapping of plaquette operators}
\label{app:cf-plaquette}

Number the plaquettes around the direct edge (red) as in Fig.~\ref{fig:cf-plaquette-numbering}; horizontal edges are oriented to the right and vertical edges upward.
\begin{figure}[H]
\centering
\begin{tikzpicture}[scale=1.0, line cap=round, line join=round]
  \draw[lat] (-0.35, 1)--(3.35, 1);
  \draw[lat] (-0.35,-1)--(3.35,-1);
  \draw[lat] (-0.35, 0)--(3.35, 0);
  \foreach \x in {0,1,2,3}{\draw[lat] (\x,-1)--(\x,1);}
  \foreach \xa in {0.5,1.5,2.5}{
    \draw[lat,-{Stealth[length=2mm,width=1.7mm]}] (\xa-0.15,1)--(\xa+0.02,1);
    \draw[lat,-{Stealth[length=2mm,width=1.7mm]}] (\xa-0.15,-1)--(\xa+0.02,-1);}
  \foreach \xa in {0.5,2.5}{\draw[lat,-{Stealth[length=2mm,width=1.7mm]}] (\xa-0.15,0)--(\xa+0.02,0);}
  \foreach \xb in {0,1,2,3}{\foreach \yb in {0.5,-0.5}{
    \draw[lat,-{Stealth[length=2mm,width=1.7mm]}] (\xb,\yb-0.15)--(\xb,\yb+0.02);}}
  \draw[edgered, line width=1.5pt] (1,0)--(2,0);
  \draw[edgered, line width=1.5pt,-{Stealth[length=2.5mm,width=2mm]}] (1,0)--(1.7,0);
  \fill[gray!95!black] (1,0) circle (2pt);
  \fill[gray!95!black] (2,0) circle (2pt);
  \node[font=\small, fill=white, inner sep=1pt] at (1,-0.27) {$v_{1}$};
  \node[font=\small, fill=white, inner sep=1pt] at (2,-0.27) {$v_{2}$};
  \node[pn] at (0.5, 0.5) {3};
  \node[pn] at (1.5, 0.5) {4};
  \node[pn] at (0.5,-0.5) {1};
  \node[pn] at (1.5,-0.5) {2};
\end{tikzpicture}
\caption{Numbering of the four plaquettes affected by the charge fusion protocol; the direct edge $\ell$ is drawn in red.}
\label{fig:cf-plaquette-numbering}
\end{figure}
\newcommand{\ecPlaqLabels}[3]{%
  \node[xl] at (#1,#2+0.65)   {\zgb{kl}};
  \node[xl] at (#1-0.65,#2)   {\zgb{li}};
  \node[xl] at (#1+0.65,#2)   {\zg{jk}};
  \node[xl] at (#1,#2-0.65)   {\zg{ij}};
  \node[pn] at (#1,#2) {#3};}
\newcommand{\ecPlaq}[3]{%
  \draw[skel] (#1-0.65,#2-0.65) rectangle (#1+0.65,#2+0.65);
  \ecPlaqLabels{#1}{#2}{#3}}
\newcommand{\ecPlaqRB}[3]{%
  \draw[skel] (#1-0.65,#2-0.65) rectangle (#1+0.65,#2+0.65);
  \draw[rsk] (#1-0.65,#2-0.65)--(#1+0.65,#2-0.65);
  \ecPlaqLabels{#1}{#2}{#3}}
\newcommand{\ecPlaqRT}[3]{%
  \draw[skel] (#1-0.65,#2-0.65) rectangle (#1+0.65,#2+0.65);
  \draw[rsk] (#1-0.65,#2+0.65)--(#1+0.65,#2+0.65);
  \ecPlaqLabels{#1}{#2}{#3}}
\newcommand{\ecPlaqNB}[3]{%
  \draw[skel] (#1-0.65,#2-0.65) rectangle (#1+0.65,#2+0.65);
  \node[xl] at (#1,#2+0.65)   {\zgb{kl}};
  \node[xl] at (#1-0.65,#2)   {\zgb{li}};
  \node[xl] at (#1+0.65,#2)   {\zg{ik}};
  \node[pn] at (#1,#2) {#3};}
\newcommand{\ecPlaqNT}[3]{%
  \draw[skel] (#1-0.65,#2-0.65) rectangle (#1+0.65,#2+0.65);
  \node[xl] at (#1-0.65,#2)   {\zgb{li}};
  \node[xl] at (#1+0.65,#2)   {\zg{jl}};
  \node[xl] at (#1,#2-0.65)   {\zg{ij}};
  \node[pn] at (#1,#2) {#3};}
\newcommand{\ecPlaqFour}[7]{%
  \draw[skel] (#1-0.65,#2-0.65) rectangle (#1+0.65,#2+0.65);
  \node[xl] at (#1,#2+0.65)   {#4};
  \node[xl] at (#1-0.65,#2)   {#5};
  \node[xl] at (#1+0.65,#2)   {#6};
  \node[xl] at (#1,#2-0.65)   {#7};
  \node[pn] at (#1,#2) {#3};}
\newcommand{\ecPlaqExt}[5]{%
  \draw[rsk] (#1+#3*0.65,#2+#4*0.65) -- (#1+#3*1.95,#2+#4*0.65);
  \node[xl] at (#1+#3*1.3,#2+#4*0.65) {#5};}

A plaquette operator is the cyclic contraction of the group-valued $\mathcal{Z}^{\Gamma}_{ij}$ of Eq.~\eqref{eq:Z-def} around the plaquette ($\bar{\mathcal{Z}}^{\Gamma}_{ij}$ of Eq.~\eqref{eq:Zbar-def} on edges traversed against their orientation), $\chi^{\Gamma}(h_{p})=\sum_{ijkl}\mathcal{Z}^{\Gamma}_{ij}\mathcal{Z}^{\Gamma}_{jk}\bar{\mathcal{Z}}^{\Gamma}_{kl}\bar{\mathcal{Z}}^{\Gamma}_{li}$ (bottom, right, top, left), i.e. the Wilson character operator of the plaquette holonomy in the MPO form of Eq.~\eqref{eq:Be-MPO}.
$\chi^{\Gamma}(h_{p})$ is not itself a projector: it has eigenvalue $d_{\Gamma}$ on flat plaquettes and $\chi^{\Gamma}(h_{p})^{\dagger}=\chi^{\bar{\Gamma}}(h_{p})$; the flatness projector is $B^{e}_{p}=\frac{1}{|G|}\sum_{\Gamma}d_{\Gamma}\chi^{\Gamma}(h_{p})$ of Eq.~\eqref{eq:Be-MPO}, and since it and every flux class projector $B^{C}_{p}$ of Eq.~\eqref{eq:BC-MPO} are linear combinations of the $\chi^{\Gamma}(h_{p})$, tracking the character operators tracks them all.
The circuit acts on the up, left and down edges of $v_{1}$, so only plaquettes 1--4 are affected.
Initially,
\begin{center}
\begin{tikzpicture}[line cap=round, line join=round]
  \ecPlaq{0}{1.05}{3}    \ecPlaqRB{3.0}{1.05}{4}
  \ecPlaq{0}{-1.05}{1}   \ecPlaqRT{3.0}{-1.05}{2}
\end{tikzpicture}
\end{center}

\paragraph{Step 1: Ungauging circuit $W^{\dagger}$.}

The gates are $C_{\ell}L^{\dagger}_{\text{up}}$, $C_{\ell}R^{\dagger}_{\text{left}}$, $C_{\ell}R^{\dagger}_{\text{down}}$, controlled by the direct edge $\ell$.
They are the inverses of the controlled multiplications of Appendix~\ref{app:ctrl-diag}, so conjugating a stabilizer by them applies the identities of that appendix from right to left: a tensor on the control is invariant (Appendix~\ref{app:Z-on-control}), a single tensor on the target splits onto the control and the target by the absorption moves Eqs.~\eqref{eq:CL-absorb-Z}--\eqref{eq:CR-absorb-Zbar} (Table~\ref{tab:absorption}) read in reverse, and a contracted pair straddling the control and the target retracts onto the target by the splitting moves Eqs.~\eqref{eq:CL-Z}--\eqref{eq:CR-Zbar} (Table~\ref{tab:splitting}) read in reverse.

\noindent\textbf{(a) $C_{\ell}L^{\dagger}_{\text{up}}$.} Plaquette 4 carries $\mathcal{Z}^{\Gamma}_{ij}$ on $\ell$ and $\bar{\mathcal{Z}}^{\Gamma}_{li}$ on the up edge; by Eq.~\eqref{eq:CL-Zbar} read in reverse the pair retracts onto the up edge, $\sum_{i}\bar{\mathcal{Z}}^{\Gamma}_{li}\big|_{\text{up}}\mathcal{Z}^{\Gamma}_{ij}\big|_{\ell}\to\bar{\mathcal{Z}}^{\Gamma}_{lj}\big|_{\text{up}}$, so plaquette 4 loses its support on $\ell$.
Plaquette 3 carries $\mathcal{Z}^{\Gamma}_{jk}$ on the up edge and nothing on $\ell$; by Eq.~\eqref{eq:CL-absorb-Z} read in reverse it splits into $\sum_{m}\mathcal{Z}^{\Gamma}_{jm}\big|_{\ell}\mathcal{Z}^{\Gamma}_{mk}\big|_{\text{up}}$, extending plaquette 3 onto $\ell$:
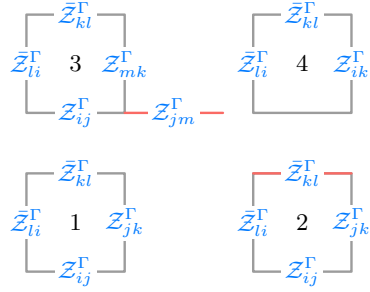
\begin{center}
\begin{tikzpicture}[line cap=round, line join=round]
  \ecPlaqFour{0}{1.05}{3}{\zgb{kl}}{\zgb{li}}{\zg{mk}}{\zg{ij}}
  \ecPlaqExt{0}{1.05}{1}{-1}{\zg{jm}}
  \ecPlaqNB{3.0}{1.05}{4}
  \ecPlaq{0}{-1.05}{1}   \ecPlaqRT{3.0}{-1.05}{2}
\end{tikzpicture}
\end{center}

\noindent\textbf{(b) $C_{\ell}R^{\dagger}_{\text{left}}$.} Plaquette 3 carries $\mathcal{Z}^{\Gamma}_{ij}$ on the left edge and the extension $\mathcal{Z}^{\Gamma}_{jm}$ on $\ell$; by Eq.~\eqref{eq:CR-Z} read in reverse the pair retracts onto the left edge, $\sum_{j}\mathcal{Z}^{\Gamma}_{ij}\big|_{\text{left}}\mathcal{Z}^{\Gamma}_{jm}\big|_{\ell}\to\mathcal{Z}^{\Gamma}_{im}\big|_{\text{left}}$, which removes the extension: plaquette 3 is restored.
Plaquette 1 carries $\bar{\mathcal{Z}}^{\Gamma}_{kl}$ on the left edge; by Eq.~\eqref{eq:CR-absorb-Zbar} read in reverse it splits into $\sum_{m}\mathcal{Z}^{\Gamma}_{km}\big|_{\ell}\bar{\mathcal{Z}}^{\Gamma}_{ml}\big|_{\text{left}}$, extending plaquette 1 onto $\ell$:
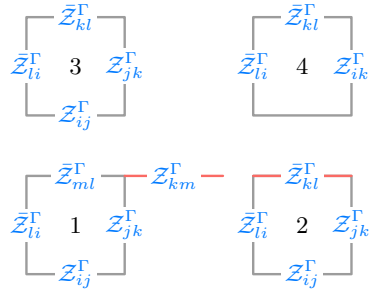
\begin{center}
\begin{tikzpicture}[line cap=round, line join=round]
  \ecPlaq{0}{1.05}{3}    \ecPlaqNB{3.0}{1.05}{4}
  \ecPlaqFour{0}{-1.05}{1}{\zgb{ml}}{\zgb{li}}{\zg{jk}}{\zg{ij}}
  \ecPlaqExt{0}{-1.05}{1}{1}{\zg{km}}
  \ecPlaqRT{3.0}{-1.05}{2}
\end{tikzpicture}
\end{center}

\noindent\textbf{(c) $C_{\ell}R^{\dagger}_{\text{down}}$.} Plaquette 1 carries $\mathcal{Z}^{\Gamma}_{jk}$ on the down edge and the extension $\mathcal{Z}^{\Gamma}_{km}$ on $\ell$; by Eq.~\eqref{eq:CR-Z} read in reverse the pair retracts onto the down edge, $\sum_{k}\mathcal{Z}^{\Gamma}_{jk}\big|_{\text{down}}\mathcal{Z}^{\Gamma}_{km}\big|_{\ell}\to\mathcal{Z}^{\Gamma}_{jm}\big|_{\text{down}}$, which restores plaquette 1.
Plaquette 2 carries $\bar{\mathcal{Z}}^{\Gamma}_{kl}$ on $\ell$ and $\bar{\mathcal{Z}}^{\Gamma}_{li}$ on the down edge; by Eq.~\eqref{eq:CR-Zbar} read in reverse the pair retracts onto the down edge, $\sum_{l}\bar{\mathcal{Z}}^{\Gamma}_{kl}\big|_{\ell}\bar{\mathcal{Z}}^{\Gamma}_{li}\big|_{\text{down}}\to\bar{\mathcal{Z}}^{\Gamma}_{ki}\big|_{\text{down}}$, so plaquette 2 loses its support on $\ell$:
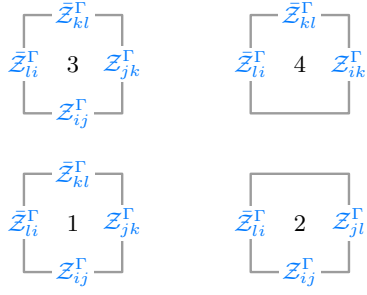
\begin{center}
\begin{tikzpicture}[line cap=round, line join=round]
  \ecPlaq{0}{1.05}{3}    \ecPlaqNB{3.0}{1.05}{4}
  \ecPlaq{0}{-1.05}{1}   \ecPlaqNT{3.0}{-1.05}{2}
\end{tikzpicture}
\end{center}

The character operators of the flanking plaquettes 4 and 2 have become the characters of the three-edge holonomy with no support on the direct edge (so $B^{e}_{p}$ becomes the three-edge flatness condition), while plaquettes 3 and 1 were transiently extended onto the direct edge and are back to their original form.

\paragraph{Step 2: Reset the direct edge.}

None of the four operators has support on $\ell$, so they pass through the replacement channel unchanged.

\paragraph{Step 3: Regauging circuit $W$.}

The gates $C_{\ell}L_{\text{up}}$, $C_{\ell}R_{\text{left}}$, $C_{\ell}R_{\text{down}}$ retrace (a)--(c) with the identities of Appendix~\ref{app:ctrl-diag} now read forwards: $C_{\ell}L_{\text{up}}$ regrows the support of plaquette 4 on $\ell$ by Eq.~\eqref{eq:CL-Zbar} and extends plaquette 3 by Eq.~\eqref{eq:CL-Z}, $C_{\ell}R_{\text{left}}$ restores 3 by Eq.~\eqref{eq:CR-absorb-Z} and extends 1 by Eq.~\eqref{eq:CR-Zbar}, and $C_{\ell}R_{\text{down}}$ restores 1 by Eq.~\eqref{eq:CR-absorb-Z} and regrows plaquette 2 by Eq.~\eqref{eq:CR-Zbar}; the factor emitted onto $\ell$ is now $\bar{\mathcal{Z}}^{\Gamma}$ where it was $\mathcal{Z}^{\Gamma}$ in Step 1, and vice versa.
All plaquette operators return to their original form.

\section{Flux fusion protocol: derivations}
\label{app:ff}

This appendix derives the operator mappings reported in Sec.~\ref{subsec:flux-fusion} for the flux fusion protocol.
Appendix~\ref{app:M-gauging} establishes two facts (F1) and (F2) about the gauging circuit $M=U_{c}U_{b}U_{a}$ of Eq.~\eqref{eq:M-gauging} and writes the channel $\Phi$ of Eq.~\eqref{eq:flux-fusion-channel} in the Heisenberg form of Eq.~\eqref{eq:ff-channel-dual}, with the reset rule of Eq.~\eqref{eq:ff-reset-rule}.
The flux projectors $B^{e}_{p}$, $B^{x}_{p,v}$ and $B^{C}_{p}$ are the MPO contractions of Eqs.~\eqref{eq:Be-MPO}, \eqref{eq:Bg-MPO} and \eqref{eq:BC-MPO} of the group-valued tensors $\mathcal{Z}^{\Gamma}_{ij}$ and $\bar{\mathcal{Z}}^{\Gamma}_{ij}$ of Eqs.~\eqref{eq:Z-def} and \eqref{eq:Zbar-def}, drawn on the skeleton of the two plaquettes with the shared dual edge in red.
Conjugating a stabilizer by a gate $U$ maps it to $UOU^{\dagger}$, so the splitting and absorption moves of Appendix~\ref{app:ctrl-diag} (Tables~\ref{tab:splitting} and \ref{tab:absorption}) are applied forwards for the un-inverted gates $CR$ and read in reverse for the inverted gates $CR^{\dagger}$, as in Appendix~\ref{app:cf-plaquette}; a tensor on a control is invariant in either case (Appendix~\ref{app:Z-on-control}), which is (F1) in tensor language.

\paragraph{Based holonomies.}
As in Sec.~\ref{subsec:flux-fusion}, flux labels are based at the bottom vertex $u$ of the dual edge $d$, and $g_{d}$ denotes the group element on $d$.
Traversed counterclockwise from $u$, the perimeter of $p_{1}\cup p_{2}$ is the loop around $p_{2}$ followed by the loop around $p_{1}$, the two traversals of $d$ cancelling, so the perimeter holonomy does not involve $g_{d}$ and the labels based at $u$ multiply, $h_{p_{1}p_{2},u}=h_{p_{2},u}\,h_{p_{1},u}$.
With $B^{x}_{p_{1}p_{2},u}=\delta_{h_{p_{1}p_{2},u},x}$ the labelled perimeter projector, whose class sum is the closed charge ribbon $B^{C}_{p_{1}p_{2}}$ of Eq.~\eqref{eq:ff-closed-loop}, this reads
\begin{align}
B^{x}_{p_{1},u}\,B^{y}_{p_{2},u}&=B^{x}_{p_{1},u}\,B^{y}_{p_{2},u}\,B^{yx}_{p_{1}p_{2},u},\label{eq:ff-label-composition}\\
\sum_{x,y:\;yx=h}B^{x}_{p_{1},u}\,B^{y}_{p_{2},u}&=B^{h}_{p_{1}p_{2},u}.\nonumber
\end{align}

\subsection{The gauging circuit and the flux fusion channel}
\label{app:M-gauging}

The circuit $M$ of the flux fusion protocol, the inverse of the step-1 circuit $M^{\dagger}$ of Eq.~\eqref{eq:ff-M-dagger}, consists of three controlled multiplications targeting the dual edge $d$, applied in the order (a) $CR$ controlled by the bottom edge of $p_{1}$, (b) $CR^{\dagger}$ controlled by its left edge, (c) $CR^{\dagger}$ controlled by its top edge, so that $M=U_{c}U_{b}U_{a}$ and
\begin{equation}
M=\sum_{a,b,c\in G}T^{a}_{\text{bottom}}T^{b}_{\text{left}}T^{c}_{\text{top}}\otimes R^{\bar{c}\bar{b}a}_{d},\qquad
M\ket{g_{d}}=\ket{g_{d}\,\bar{a}bc},\qquad M^{\dagger}\ket{g_{d}}=\ket{g_{d}\,\bar{c}\bar{b}a}=\ket{h_{p_{1},u}},
\label{eq:M-gauging}
\end{equation}
where $T^{a}$, $T^{b}$, $T^{c}$ are the single-edge projectors on the bottom, left and top edge of $p_{1}$ as in Eq.~\eqref{eq:ff-M-dagger}, so that for group elements $a$, $b$, $c$ on these edges $\bar{c}\bar{b}a$ is the ordered product of the three edges of $p_{1}$ other than $d$ and $h_{p_{1},u}=g_{d}\bar{c}\bar{b}a$ is the flux of $p_{1}$ based at $u$; $M$ acts as the identity on the three outer edges of $p_{2}$.
Two facts about $M$ carry the stabilizer bookkeeping of the protocol.

(F1) $M$ is diagonal on every edge except $d$, so any operator diagonal in the group basis with no support on $d$ commutes with $M$ and $M^{\dagger}$; in tensor language this is the invariance of a $\mathcal{Z}^{\Gamma}$ on a control (Appendix~\ref{app:Z-on-control}).
The restriction to diagonal operators is essential, since for $G=\mathbb{Z}_{2}$ the Pauli $X$ on the bottom edge of $p_{1}$ is mapped to $MX_{\text{bottom}}M^{\dagger}=X_{\text{bottom}}X_{d}$.

(F2) For a permutation $U\ket{g}=\ket{\pi(g)}$ of the group basis one has $Uf(g)U^{\dagger}=f(\pi^{-1}(g))$, hence for any diagonal $F$
\begin{equation}
M^{\dagger}FM=F\big|_{g_{d}\to g_{d}\bar{a}bc},\qquad MFM^{\dagger}=F\big|_{g_{d}\to g_{d}\bar{c}\bar{b}a=h_{p_{1},u}},
\label{eq:M-substitution}
\end{equation}
and in particular $MT^{x}_{d}M^{\dagger}=B^{x}_{p_{1},u}$ and $M^{\dagger}B^{x}_{p_{1},u}M=T^{x}_{d}$: $M^{\dagger}$ writes the flux of $p_{1}$, based at $u$, onto the dual edge alone, and $M$ reads it back.
These identities fix the action of $M$ on group-basis-diagonal operators and identify it as the local step of the $\mathrm{Rep}(G)$ gauging map at $p_{1}$.
In tensor language Eq.~\eqref{eq:M-substitution} is the splitting of a target tensor onto the controls by the rules of Table~\ref{tab:splitting}; applied to $T^{e}_{d}=\frac{1}{|G|}\sum_{\Gamma}d_{\Gamma}\sum_{i}\mathcal{Z}^{\Gamma}_{ii}(d)$, the three gates of $M$ grow the single-edge trace into the closed contraction around $p_{1}$ (Fig.~\ref{fig:M-splitting-chain}).

\newcommand{\ecPone}[2]{%
  \draw[skel] (#1-0.65,#2-0.65) rectangle (#1+0.65,#2+0.65);
  \draw[rsk] (#1+0.65,#2-0.65)--(#1+0.65,#2+0.65);
  \node[pn] at (#1,#2) {$p_{1}$};}
\begin{figure}[H]
\centering
\begin{tikzpicture}[scale=1.25, transform shape, line cap=round, line join=round,
    skel/.style={latgray, line width=0.9pt},
    rsk/.style={edgered, line width=0.9pt},
    xl/.style={opblue, font=\small, fill=white, inner sep=1.2pt},
    pn/.style={latgray, font=\small},
    eq/.style={black, font=\small},
    arr/.style={-{Stealth[length=2mm,width=1.6mm]}, black, line width=0.7pt}]
  \node[eq] at (0,1.05) {$\sum_{i}$};
  \ecPone{0}{0}
  \node[xl] at (0.65,0) {\zg{ii}};
  \draw[arr] (1.1,0)--(1.9,0); \node[eq] at (1.5,0.3) {(a)};
  \node[eq] at (3.0,1.05) {$\sum_{ij}$};
  \ecPone{3.0}{0}
  \node[xl] at (3.0,-0.65) {\zg{ij}};
  \node[xl] at (3.65,0) {\zg{ji}};
  \draw[arr] (4.1,0)--(4.9,0); \node[eq] at (4.5,0.3) {(b)};
  \node[eq] at (6.0,1.05) {$\sum_{ijk}$};
  \ecPone{6.0}{0}
  \node[xl] at (6.0,-0.65) {\zg{ij}};
  \node[xl] at (6.65,0) {\zg{jk}};
  \node[xl] at (5.35,0) {\zgb{ki}};
  \draw[arr] (7.1,0)--(7.9,0); \node[eq] at (7.5,0.3) {(c)};
  \node[eq] at (9.0,1.05) {$\sum_{ijkl}$};
  \ecPone{9.0}{0}
  \node[xl] at (9.0,-0.65) {\zg{ij}};
  \node[xl] at (9.65,0) {\zg{jk}};
  \node[xl] at (9.0,0.65) {\zgb{kl}};
  \node[xl] at (8.35,0) {\zgb{li}};
\end{tikzpicture}
\caption{Splitting chain of the gauging circuit.
The single-edge trace $\sum_{i}\mathcal{Z}^{\Gamma}_{ii}(d)$ on the dual edge is split onto the controls by gates (a), (b), (c) of $M$; the result is the closed counterclockwise contraction of Eq.~\eqref{eq:Be-MPO} around $p_{1}$.}
\label{fig:M-splitting-chain}
\end{figure}
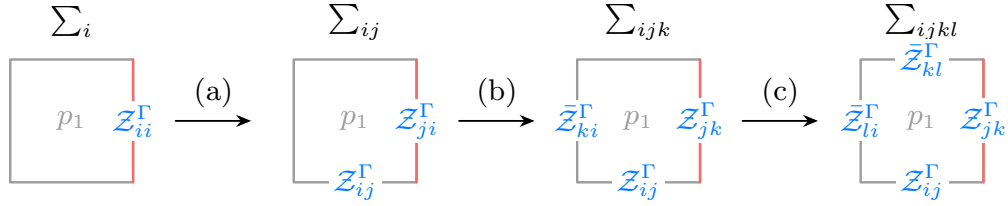
Gate (a) splits $\sum_{i}\mathcal{Z}^{\Gamma}_{ii}(d)$ into $\sum_{ij}\mathcal{Z}^{\Gamma}_{ij}(\text{bottom})\mathcal{Z}^{\Gamma}_{ji}(d)$, gate (b) emits $\bar{\mathcal{Z}}^{\Gamma}(\text{left})$ on the anti-aligned left edge, and gate (c) emits $\bar{\mathcal{Z}}^{\Gamma}(\text{top})$ on the anti-aligned top edge; the result is $B^{e}_{p_{1}}$ of Eq.~\eqref{eq:Be-MPO} with the same prefactor.

\paragraph{The channel.}
The flux fusion protocol is the channel $\Phi$ of Eq.~\eqref{eq:flux-fusion-channel} in Sec.~\ref{subsec:flux-fusion}, with Kraus operators $K_{x}=M\ketbra{e}{x}_{d}M^{\dagger}=L^{\bar{x}}_{d}B^{x}_{p_{1},u}$.
The Kraus form follows from $\ketbra{e}{x}_{d}=L^{\bar{x}}_{d}T^{x}_{d}$, Eq.~\eqref{eq:M-substitution}, and the fact that $M$ acts on $d$ by right multiplications, which commute with $L^{\bar{x}}_{d}$.
No classical register and no outcome-dependent gate beyond the local reset of the dual edge are needed.
Two consequences are used throughout.
First, for any operator $F$ diagonal in the group basis the Heisenberg action is the substitution $\Phi^{*}(F)=F|_{g_{d}\to\bar{a}bc}$ with $\bar{c}\bar{b}a$ the ordered product of the three edges of $p_{1}$ other than $d$ as in Eq.~\eqref{eq:M-gauging}: $M^{\dagger}FM$ substitutes $g_{d}\to g_{d}\bar{a}bc$ by Eq.~\eqref{eq:M-substitution}, the reset adjoint evaluates at $g_{d}=e$ and extends by $\mathds{1}_{d}$, and the final $M$ acts trivially because $g_{d}$ no longer appears; the value $\bar{a}bc$ is the unique dual-edge value with $h_{p_{1},u}=e$.
Second, on group-basis configurations $\Phi(\ketbra{\{g\}}{\{g'\}})=\delta_{h_{p_{1},u}(g),h_{p_{1},u}(g')}\ketbra{f(g)}{f(g')}$, where $f$ replaces $g_{d}$ by $\bar{a}bc$: coherence between different based source labels never survives, and everything else is carried along isometrically.
In the Heisenberg picture
\begin{equation}\label{eq:ff-channel-dual}
\Phi^{*}(O)=M\big[\mathds{1}_{d}\otimes\bra{e}_{d}\,M^{\dagger}OM\,\ket{e}_{d}\big]M^{\dagger},
\end{equation}
where $\bra{e}_{d}\cdot\ket{e}_{d}$ is the partial matrix element on the dual edge.
On the tensors of a diagram it acts by Eqs.~\eqref{eq:Z-def} and \eqref{eq:Zbar-def} as
\begin{equation}\label{eq:ff-reset-rule}
\bra{e}\mathcal{Z}^{\Gamma}_{ij}\ket{e}=\bra{e}\bar{\mathcal{Z}}^{\Gamma}_{ij}\ket{e}=\delta_{ij},\qquad
\bra{e}T^{x}\ket{e}=\delta_{x,e}:
\end{equation}
evaluating the dual-edge tensor at $e$ joins its two bond indices, i.e.\ removes the tensor from the diagram and connects its two neighbours directly, while a single-edge projector on $d$ is replaced by its value at $e$.
Together with (F1), (F2) and the splitting and absorption moves of Appendix~\ref{app:ctrl-diag}, Eq.~\eqref{eq:ff-reset-rule} reduces every Heisenberg-picture computation $\Phi^{*}(O)$ to moves on MPO diagrams.

\subsection{Lemma: closed charge ribbon can be deformed using neutrality}
\label{app:ff-lemma}

Given a state $\ket{\psi}$ with no flux at plaquette $p_{1}$, we show that the closed charge ribbon that encircles both plaquettes acts on $\ket{\psi}$ exactly as the class projector of $p_{2}$ alone.
No assumption on the flux of $p_{2}$ is needed; in particular, for such a state having flux class $C$ at $p_{2}$ is the same as the two plaquettes having total flux class $C$.
We start with the class projector (to the conjugacy class $C$) for the region $p_{1}\cup p_{2}$, which is the closed charge ribbon of Eq.~\eqref{eq:ff-closed-loop},
\begin{equation}\label{eq:ff-lemma-closed-ribbon}
B^{C}_{p_{1}p_{2}}=\frac{|C|}{|G|}\sum_{\Gamma}\bar{\chi}^{\Gamma}(C)\sum_{ijklmn}\;
\vcenter{\hbox{%
\begin{tikzpicture}[line cap=round, line join=round, baseline]
  \dmlat \dmlab
  \node[xl] at (-0.75,-0.75) {\zg{ij}};
  \node[xl] at (0.75,-0.75)  {\zg{jk}};
  \node[xl] at (1.5,0)       {\zg{kl}};
  \node[xl] at (0.75,0.75)   {\zgb{lm}};
  \node[xl] at (-0.75,0.75)  {\zgb{mn}};
  \node[xl] at (-1.5,0)      {\zgb{ni}};
\end{tikzpicture}
}}
\end{equation}
This operator is a closed \textit{charge} ribbon: it transports a charge around the two plaquettes, and since charges braid nontrivially with fluxes, the closed ribbon detects the total flux class enclosed by it, which is exactly why weighting it by the class character $\bar{\chi}^{\Gamma}(C)$ gives the flux projector of the region.
Since there is no flux at $p_{1}$, we also have the flatness condition at $p_{1}$, with $B^{e}_{p_{1}}$ the trivial flux projector of Eq.~\eqref{eq:Be-MPO},
\begin{equation}\label{eq:ff-flatness}
B^{e}_{p_{1}}\ket{\psi}=\ket{\psi},\qquad
B^{e}_{p_{1}}=\frac{1}{|G|}\sum_{\Gamma}d_{\Gamma}\sum_{ijkl}\;
\vcenter{\hbox{%
\begin{tikzpicture}[line cap=round, line join=round, baseline]
  \dmlat \dmlab
  \node[xl] at (-0.75,-0.75) {\zg{ij}};
  \node[xl] at (0,0)         {\zg{jk}};
  \node[xl] at (-0.75,0.75)  {\zgb{kl}};
  \node[xl] at (-1.5,0)      {\zgb{li}};
\end{tikzpicture}
}}
\end{equation}

\begin{lemma}[Deforming a closed charge ribbon using flux neutrality]
Using the flatness condition of Eq.~\eqref{eq:ff-flatness}, the class projector over the two plaquettes acts on $\ket{\psi}$ as that of only $p_{2}$, $B^{C}_{p_{1}p_{2}}\ket{\psi}=B^{C}_{p_{2}}\ket{\psi}$, or as an operator identity
\begin{equation}\label{eq:ff-lemma}
B^{e}_{p_{1}}\,B^{C}_{p_{1}p_{2}}=B^{e}_{p_{1}}\,B^{C}_{p_{2}},
\end{equation}
for every conjugacy class $C$, where $B^{C}_{p_{2}}$ is the class projector of Eq.~\eqref{eq:BC-MPO} at $p_{2}$,
\begin{equation}\label{eq:ff-lemma-single-plaquette}
B^{C}_{p_{2}}=\frac{|C|}{|G|}\sum_{\Gamma}\bar{\chi}^{\Gamma}(C)\sum_{jklm}\;
\vcenter{\hbox{%
\begin{tikzpicture}[line cap=round, line join=round, baseline, xl/.append style={fill=thmshade}]
  \dmlat \dmlab
  \node[xl] at (0.75,-0.75)  {\zg{jk}};
  \node[xl] at (1.5,0)       {\zg{kl}};
  \node[xl] at (0.75,0.75)   {\zgb{lm}};
  \node[xl] at (0,0)         {\zgb{mj}};
\end{tikzpicture}
}}
\end{equation}
\end{lemma}

\begin{proof}
In terms of labels the claim is one line: Eq.~\eqref{eq:ff-label-composition} with $x=e$ gives $B^{e}_{p_{1}}B^{h}_{p_{1}p_{2},u}=B^{e}_{p_{1}}B^{h}_{p_{2},u}$ for every $h$, and summing over $h\in C$ gives Eq.~\eqref{eq:ff-lemma}.
In tensor language the same move is a deformation of the MPO string across the flux-free plaquette.
On the support of $B^{e}_{p_{1}}$ the holonomy around $p_{1}$ is trivial, $h_{p_{1},u}=g_{d}\bar{c}\bar{b}a=e$ with $a$, $b$, $c$ the group elements on the bottom, left and top edge of $p_{1}$ as in Eq.~\eqref{eq:M-gauging}, hence $\bar{c}\bar{b}=\bar{g}_{d}\bar{a}$: the irrep matrices along the two paths from the top end of $d$ to the bottom-left corner of $p_{1}$ coincide, and the tensors of the loop on the top and left edges may be replaced by tensors on the dual and bottom edges,
\begin{equation}\label{eq:ff-slide-rule}
B^{e}_{p_{1}}\sum_{n}\bar{\mathcal{Z}}^{\Gamma}_{mn}(\text{top})\bar{\mathcal{Z}}^{\Gamma}_{ni}(\text{left})=B^{e}_{p_{1}}\sum_{n}\bar{\mathcal{Z}}^{\Gamma}_{mn}(d)\bar{\mathcal{Z}}^{\Gamma}_{ni}(\text{bottom}).
\end{equation}
The string then runs along the bottom edge of $p_{1}$ twice, once in each direction:
\begin{equation}
B^{e}_{p_{1}}\,\frac{|C|}{|G|}\sum_{\Gamma}\bar{\chi}^{\Gamma}(C)\sum_{ijklmn}\;
\vcenter{\hbox{%
\begin{tikzpicture}[line cap=round, line join=round, baseline]
  \dmlat \dmlab
  \node[xl] at (-0.75,-0.75) {\zgb{ni}\,\zg{ij}};
  \node[xl] at (0.75,-0.75)  {\zg{jk}};
  \node[xl] at (1.5,0)       {\zg{kl}};
  \node[xl] at (0.75,0.75)   {\zgb{lm}};
  \node[xl] at (0,0)         {\zgb{mn}};
\end{tikzpicture}
}}
\end{equation}
The two tensors on the bottom edge contract by the orthogonality relation of Eq.~\eqref{eq:Zbar-orthogonality}, $\sum_{i}\bar{\mathcal{Z}}^{\Gamma}_{ni}\mathcal{Z}^{\Gamma}_{ij}=\delta_{nj}$, and the closed contraction around $p_{2}$ remains,
\begin{equation}
B^{e}_{p_{1}}\,\frac{|C|}{|G|}\sum_{\Gamma}\bar{\chi}^{\Gamma}(C)\sum_{jklm}\;
\vcenter{\hbox{%
\begin{tikzpicture}[line cap=round, line join=round, baseline]
  \dmlat \dmlab
  \node[xl] at (0.75,-0.75)  {\zg{jk}};
  \node[xl] at (1.5,0)       {\zg{kl}};
  \node[xl] at (0.75,0.75)   {\zgb{lm}};
  \node[xl] at (0,0)         {\zgb{mj}};
\end{tikzpicture}
}}
\end{equation}
i.e. $B^{C}_{p_{1}p_{2}}\ket{\psi}=B^{C}_{p_{2}}\ket{\psi}$ with $B^{C}_{p_{2}}$ exactly the class projector on a single plaquette, Eqs.~\eqref{eq:BC-MPO} and \eqref{eq:ff-lemma-single-plaquette}, which is the claim.
The identity holds on the flux-free subspace of $p_{1}$ only; the two projectors are not equal on the full Hilbert space.
The derivation is linear in the class function, so it applies term by term to any sum of class projectors in place of $B^{C}$.
\end{proof}

\subsection{Stabilizer mapping derivations}

\subsubsection{Mapping of closed loop ribbon operator}
\label{app:ff-ribbon}

Hypothesis: the input is stabilized by $B^{C}_{p_{1}p_{2}}$.
It is needed for the stabilizer statements; the Heisenberg identities at the end hold for every input.

\paragraph{Step 1: Ungauging circuit $M^{\dagger}$.}

The three gates of $M^{\dagger}$ all target the dual edge $d$, on which the closed ribbon of Eq.~\eqref{eq:ff-closed-loop} carries no tensor; its tensors on the bottom, left and top edge of $p_{1}$ sit on the controls and are invariant (Appendix~\ref{app:Z-on-control}), and $M^{\dagger}$ does not touch the three outer edges of $p_{2}$ at all.
By (F1) the loop therefore passes through unchanged,
\begin{equation}\label{eq:ff-loop-ungauged}
M^{\dagger}B^{C}_{p_{1}p_{2}}M=\frac{|C|}{|G|}\sum_{\Gamma}\bar{\chi}^{\Gamma}(C)\sum_{ijklmn}\;
\vcenter{\hbox{%
\begin{tikzpicture}[line cap=round, line join=round, baseline]
  \dmlat \dmlab
  \node[xl] at (-0.75,-0.75) {\zg{ij}};
  \node[xl] at (0.75,-0.75)  {\zg{jk}};
  \node[xl] at (1.5,0)       {\zg{kl}};
  \node[xl] at (0.75,0.75)   {\zgb{lm}};
  \node[xl] at (-0.75,0.75)  {\zgb{mn}};
  \node[xl] at (-1.5,0)      {\zgb{ni}};
\end{tikzpicture}
}}
\end{equation}
The dual edge stays out of the closed ribbon; it carries the ungauged flux of $p_{1}$ (Appendix~\ref{app:ff-plaquette}) and enters the ungauged target plaquette projector through the tail of Fig.~\ref{fig:ff-tail}.

\paragraph{Step 2: Reset the dual edge to $\ket{e}$.}

Since the closed ribbon has no support on $d$, the Kraus operators $\ketbra{e}{x}_{d}$ of the reset commute with it, and it remains a stabilizer.
The reset adds the stabilizer
\begin{equation}\label{eq:ff-reset-stabilizer}
T^{e}_{d}=\frac{1}{|G|}\sum_{\Gamma'}d_{\Gamma'}\sum_{i'}\mathcal{Z}^{\Gamma'}_{i'i'}(d),
\end{equation}
the single-edge form of Eq.~\eqref{eq:Be-MPO}, which by (F2) is the statement that $p_{1}$ will carry trivial flux after regauging.
The stabilizer set now contains the product of the two, $T^{e}_{d}\otimes M^{\dagger}B^{C}_{p_{1}p_{2}}M$,
\begin{equation}\label{eq:ff-loop-times-reset}
\frac{|C|}{|G|^{2}}\sum_{\Gamma,\Gamma'}\bar{\chi}^{\Gamma}(C)\,d_{\Gamma'}\sum_{ijklmn}\sum_{i'}\;
\vcenter{\hbox{%
\begin{tikzpicture}[line cap=round, line join=round, baseline]
  \dmlat \dmlab
  \node[xl] at (-0.75,-0.75) {\zg{ij}};
  \node[xl] at (0.75,-0.75)  {\zg{jk}};
  \node[xl] at (1.5,0)       {\zg{kl}};
  \node[xl] at (0.75,0.75)   {\zgb{lm}};
  \node[xl] at (-0.75,0.75)  {\zgb{mn}};
  \node[xl] at (-1.5,0)      {\zgb{ni}};
  \node[xl] at (0,0)         {\zgp{i'i'}};
\end{tikzpicture}
}}
\end{equation}

\paragraph{Step 3: Regauging circuit $M$.}

The loop tensors sit on controls and are again invariant under $M$, while the dual-edge tensor is the target of all three gates and splits onto the controls: gate (a), $CR$ controlled by the bottom edge, sends $\mathcal{Z}^{\Gamma'}_{ij}(d)\to\sum_{m}\mathcal{Z}^{\Gamma'}_{im}(d)\mathcal{Z}^{\Gamma'}_{mj}(\text{bottom})$ by the splitting move of Eq.~\eqref{eq:CR-Z}, and gates (b) and (c), $CR^{\dagger}$ controlled by the left and top edge, send $\mathcal{Z}^{\Gamma'}_{ij}(d)\to\sum_{m}\mathcal{Z}^{\Gamma'}_{im}(d)\bar{\mathcal{Z}}^{\Gamma'}_{mj}(\text{left})$ and likewise with the top edge by the absorption move of Eq.~\eqref{eq:CR-absorb-Z} read in reverse.
Gate by gate, the single-edge trace grows into the closed contraction around $p_{1}$, which is the splitting chain of Fig.~\ref{fig:M-splitting-chain} with $\Gamma'$ in place of $\Gamma$ (the loop tensors ride along unchanged), so that $MT^{e}_{d}M^{\dagger}=B^{e}_{p_{1}}$, which is (F2).
The stabilizer of Eq.~\eqref{eq:ff-loop-times-reset} transforms into
\begin{equation}\label{eq:ff-loop-times-Be}
\frac{|C|}{|G|^{2}}\sum_{\Gamma,\Gamma'}\bar{\chi}^{\Gamma}(C)\,d_{\Gamma'}\sum_{ijklmn}\sum_{i'j'k'l'}\;
\vcenter{\hbox{%
\begin{tikzpicture}[line cap=round, line join=round, baseline]
  \dmlat
  \node[xl] at (-0.75,-0.75) {\zg{ij}\,\zgp{i'j'}};
  \node[xl] at (0.75,-0.75)  {\zg{jk}};
  \node[xl] at (1.5,0)       {\zg{kl}};
  \node[xl] at (0.75,0.75)   {\zgb{lm}};
  \node[xl] at (-0.75,0.75)  {\zgb{mn}\,\zgpb{k'l'}};
  \node[xl] at (-1.5,0)      {\zgb{ni}\,\zgpb{l'i'}};
  \node[xl] at (0,0)         {\zgp{j'k'}};
\end{tikzpicture}
}}
\end{equation}
The $\Gamma'$-dependent part factors out as the trivial flux projector $B^{e}_{p_{1}}$ of Eq.~\eqref{eq:ff-flatness} (equivalently, $M$ maps the reset stabilizer $T^{e}_{d}$ back to $B^{e}_{p_{1}}$), so the stabilizer is $B^{e}_{p_{1}}B^{C}_{p_{1}p_{2}}$, and by the lemma of Appendix~\ref{app:ff-lemma} the closed ribbon reduces to the single plaquette class projector $B^{C}_{p_{2}}$ of Eq.~\eqref{eq:ff-single-plaquette-p2},
\begin{equation}
B^{e}_{p_{1}}\,\frac{|C|}{|G|}\sum_{\Gamma}\bar{\chi}^{\Gamma}(C)\sum_{jklm}\;
\vcenter{\hbox{%
\begin{tikzpicture}[line cap=round, line join=round, baseline]
  \dmlat \dmlab
  \node[xl] at (0.75,-0.75)  {\zg{jk}};
  \node[xl] at (1.5,0)       {\zg{kl}};
  \node[xl] at (0.75,0.75)   {\zgb{lm}};
  \node[xl] at (0,0)         {\zgb{mj}};
\end{tikzpicture}
}}
\end{equation}
Hence the output $\rho'$ satisfies $B^{C}_{p_{2}}\rho'=B^{C}_{p_{2}}B^{e}_{p_{1}}\rho'=B^{C}_{p_{1}p_{2}}B^{e}_{p_{1}}\rho'=\rho'$: the class measured by the loop around $p_{1}\cup p_{2}$ is measured by $p_{2}$ alone, and $p_{1}$ carries trivial magnetic flux.
For a pair of single plaquette fluxes the genuine stabilizer is the support projector $\Pi_{p_{2}}$ of Eq.~\eqref{eq:ff-Pi-p2}.

\paragraph{Heisenberg identities (every input).}
Eq.~\eqref{eq:ff-channel-dual} and the reset rule \eqref{eq:ff-reset-rule} give the three identities of Eq.~\eqref{eq:ff-heisenberg-loop} directly.
For the loop, $M^{\dagger}B^{C}_{p_{1}p_{2}}M=B^{C}_{p_{1}p_{2}}$ has no tensor on $d$, the partial matrix element leaves it unchanged, and the final conjugation by $M$ is again trivial.
For the source plaquette, $M^{\dagger}B^{e}_{p_{1}}M=T^{e}_{d}$ by (F2) and $\bra{e}T^{e}_{d}\ket{e}=1$, so $\Phi^{*}(B^{e}_{p_{1}})=\mathds{1}$: the output is stabilized by $B^{e}_{p_{1}}$ whatever the input.
For the target plaquette, $M^{\dagger}B^{C}_{p_{2}}M$ is the tail operator of Fig.~\ref{fig:ff-tail}, the closed contraction around the two plaquettes with an additional barred tensor on $d$; the reset rule removes that tensor and joins its neighbours, leaving exactly $B^{C}_{p_{1}p_{2}}\otimes\mathds{1}_{d}$, which is invariant under the final $M$.
The labelled version $\Phi^{*}(B^{x}_{p_{2},u})=B^{x}_{p_{1}p_{2},u}$ follows in the same way, the insertion $\Gamma(\bar{x})$ at $u$ riding along (Eq.~\eqref{eq:ff-tail-label-diagram}).
Since $\Tr[P\rho]=1$ is equivalent to $P\rho=\rho$ for a projector, the first identity is the stabilizer statement under the hypothesis and the expectation-value identity without it.

\subsubsection{Mapping of single plaquette projectors}
\label{app:ff-plaquette}

Hypothesis: the input is stabilized by the single plaquette class projectors $B^{C_{1}}_{p_{1}}$ and $B^{C_{2}}_{p_{2}}$ of Eq.~\eqref{eq:BC-MPO},
\begin{equation}\label{eq:ff-B-p1}
B^{C_{1}}_{p_{1}}=\frac{|C_{1}|}{|G|}\sum_{\Gamma}\bar{\chi}^{\Gamma}(C_{1})\sum_{ijkl}\;
\vcenter{\hbox{%
\begin{tikzpicture}[line cap=round, line join=round, baseline]
  \dmlat \dmlab
  \node[xl] at (-0.75,-0.75) {\zg{ij}};
  \node[xl] at (0,0)         {\zg{jk}};
  \node[xl] at (-0.75,0.75)  {\zgb{kl}};
  \node[xl] at (-1.5,0)      {\zgb{li}};
\end{tikzpicture}
}}
\end{equation}
\begin{equation}\label{eq:ff-B-p2}
B^{C_{2}}_{p_{2}}=\frac{|C_{2}|}{|G|}\sum_{\Gamma}\bar{\chi}^{\Gamma}(C_{2})\sum_{ijkl}\;
\vcenter{\hbox{%
\begin{tikzpicture}[line cap=round, line join=round, baseline]
  \dmlat \dmlab
  \node[xl] at (0.75,-0.75)  {\zg{ij}};
  \node[xl] at (1.5,0)       {\zg{jk}};
  \node[xl] at (0.75,0.75)   {\zgb{kl}};
  \node[xl] at (0,0)         {\zgb{li}};
\end{tikzpicture}
}}
\end{equation}
We track a genuine stabilizer of the state (eigenvalue $+1$ at every step, with no truncation and no rescaling), the product $B^{C_{1}}_{p_{1}}B^{C_{2}}_{p_{2}}$ of the two commuting projectors.
It is the labelled form of the fusion-channel projector $\Pi_{p_{1}p_{2}}$ of Eq.~\eqref{eq:ff-Pi-p1p2}: by Eq.~\eqref{eq:ff-label-composition},
\begin{equation}\label{eq:ff-product-vs-Pi}
B^{C_{1}}_{p_{1}}B^{C_{2}}_{p_{2}}=\sum_{x\in C_{1}}\sum_{y\in C_{2}}B^{x}_{p_{1},u}B^{y}_{p_{2},u}B^{yx}_{p_{1}p_{2},u}
=B^{C_{1}}_{p_{1}}B^{C_{2}}_{p_{2}}\,\Pi_{p_{1}p_{2}},
\end{equation}
since every product $yx$ lies in a class with $a^{C}_{C_{1}C_{2}}\neq0$; so $\Pi_{p_{1}p_{2}}$ stabilizes the input as well, and being supported on the perimeter only it is invariant under every step of the protocol by (F1).
We carry the product because it is built directly from the two plaquette projectors, which makes the reduction to the fusion channels transparent, and because after ungauging it displays the measured label explicitly.

\paragraph{Step 1: Ungauging circuit $M^{\dagger}$, source plaquette.}

With the label based at $u$, the insertion $\Gamma(\bar{x})$ of $B^{x}_{p_{1},u}$ of Eq.~\eqref{eq:Bg-MPO} sits at the corner $u$ between the tensors on the bottom edge and $d$, and $M^{\dagger}=U_{a}^{\dagger}U_{b}^{\dagger}U_{c}^{\dagger}$ applies (c)$^{\dagger}$ ($CR$ controlled by the top edge), then (b)$^{\dagger}$ ($CR$ controlled by the left edge), then (a)$^{\dagger}$ ($CR^{\dagger}$ controlled by the bottom edge), each meeting exactly the pair its absorption rule retracts; the insertion is a $c$-number and rides along:
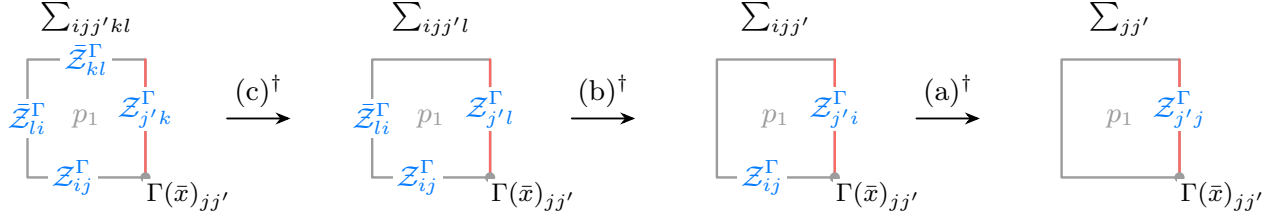
\begin{figure}[H]
\centering
\newcommand{\ecPoneU}[2]{%
  \draw[skel] (#1-0.65,#2-0.65) rectangle (#1+0.65,#2+0.65);
  \draw[rsk] (#1+0.65,#2-0.65)--(#1+0.65,#2+0.65);
  \fill[latgray] (#1+0.65,#2-0.65) circle (1.6pt);
  \node[latgray, font=\small] at (#1,#2) {$p_{1}$};}
\newcommand{\ins}[2]{\node[il, anchor=north west] at (#1+0.62,#2-0.66) {$\Gamma(\bar{x})_{jj'}$};}
\begin{tikzpicture}[scale=1.2, transform shape, line cap=round, line join=round]
  \node[eq] at (0,1.1) {$\sum_{ijj'kl}$};
  \ecPoneU{0}{0}
  \node[xl] at (-0.15,-0.65) {\zg{ij}};
  \node[xl] at (0.65,0.1) {\zg{j'k}};
  \node[xl] at (0,0.65) {\zgb{kl}};
  \node[xl] at (-0.65,0) {\zgb{li}};
  \ins{0}{0}
  \draw[arr] (1.55,0)--(2.25,0); \node[eq] at (1.9,0.3) {(c)$^{\dagger}$};
  \node[eq] at (3.8,1.1) {$\sum_{ijj'l}$};
  \ecPoneU{3.8}{0}
  \node[xl] at (3.65,-0.65) {\zg{ij}};
  \node[xl] at (4.45,0.1) {\zg{j'l}};
  \node[xl] at (3.15,0) {\zgb{li}};
  \ins{3.8}{0}
  \draw[arr] (5.35,0)--(6.05,0); \node[eq] at (5.7,0.3) {(b)$^{\dagger}$};
  \node[eq] at (7.6,1.1) {$\sum_{ijj'}$};
  \ecPoneU{7.6}{0}
  \node[xl] at (7.45,-0.65) {\zg{ij}};
  \node[xl] at (8.25,0.1) {\zg{j'i}};
  \ins{7.6}{0}
  \draw[arr] (9.15,0)--(9.85,0); \node[eq] at (9.5,0.3) {(a)$^{\dagger}$};
  \node[eq] at (11.4,1.1) {$\sum_{jj'}$};
  \ecPoneU{11.4}{0}
  \node[xl] at (12.05,0.1) {\zg{j'j}};
  \ins{11.4}{0}
\end{tikzpicture}
\caption{Ungauging of the labelled source projector $B^{x}_{p_{1},u}$ of Eq.~\eqref{eq:Bg-MPO}.
The gates (c)$^{\dagger}$, (b)$^{\dagger}$, (a)$^{\dagger}$ absorb the tensors on the top, left and bottom edge of $p_{1}$ onto the dual edge, leaving a single tensor on $d$ with the insertion $\Gamma(\bar{x})$ at the corner $u$; it is the absorption chain inverse to the splitting chain of Fig.~\ref{fig:M-splitting-chain}.}
\label{fig:ff-contraction}
\end{figure}
Gate (c)$^{\dagger}$ is the un-inverted $CR$, so the absorption move of Eq.~\eqref{eq:CR-absorb-Z} applies forwards to the pair on $d$ and the top edge, $\sum_{k}\mathcal{Z}^{\Gamma}_{j'k}(d)\bar{\mathcal{Z}}^{\Gamma}_{kl}(\text{top})\to\mathcal{Z}^{\Gamma}_{j'l}(d)$; gate (b)$^{\dagger}$ applies the same move to the pair on $d$ and the left edge; gate (a)$^{\dagger}$ is $CR^{\dagger}$, so the splitting move of Eq.~\eqref{eq:CR-Z} read in reverse retracts the pair on $d$ and the bottom edge, $\sum_{i}\mathcal{Z}^{\Gamma}_{j'i}(d)\mathcal{Z}^{\Gamma}_{ij}(\text{bottom})\to\mathcal{Z}^{\Gamma}_{j'j}(d)$.
The result is a single tensor on $d$ carrying the insertion, $\frac{1}{|G|}\sum_{\Gamma}d_{\Gamma}\sum_{jj'}[\Gamma(\bar{x})]_{jj'}\mathcal{Z}^{\Gamma}_{j'j}(d)$, whose diagonal entry on $\ket{g}_{d}$ is the character expansion $\frac{1}{|G|}\sum_{\Gamma}d_{\Gamma}\chi^{\Gamma}(\bar{x}g)=\delta_{g,x}$: it is the single-edge projector $T^{x}_{d}$, the labelled form of (F2).
Summed over the class, using $\sum_{x\in C_{1}}\Gamma(\bar{x})=\frac{|C_{1}|}{d_{\Gamma}}\bar{\chi}^{\Gamma}(C_{1})\,\mathds{1}$ as in Eq.~\eqref{eq:BC-MPO},
\begin{equation}\label{eq:ff-source-ungauged}
M^{\dagger}B^{C_{1}}_{p_{1}}M=T^{C_{1}}_{d}:=\sum_{x\in C_{1}}T^{x}_{d}=\frac{|C_{1}|}{|G|}\sum_{\Gamma}\bar{\chi}^{\Gamma}(C_{1})\sum_{i}\;
\vcenter{\hbox{%
\begin{tikzpicture}[line cap=round, line join=round, baseline]
  \dmlat \dmlab
  \node[xl] at (0,0) {\zg{ii}};
\end{tikzpicture}
}}
\end{equation}
i.e. the $C_{1}$ flux is now carried by the dual edge alone.
(For a label based at the bottom-left corner of $p_{1}$ the insertion precedes the tensor on the bottom edge, which then cannot be absorbed, and the ungauged operator retains support on it: the variable written onto the dual edge is the flux based at $u$.)

\paragraph{Step 1: Ungauging circuit $M^{\dagger}$, target plaquette.}

$B^{y}_{p_{2},u}$ carries $\bar{\mathcal{Z}}^{\Gamma}(d)$ on the dual edge, the target of all three gates, whose controls, the bottom, left and top edge of $p_{1}$, carry nothing; the barred-target moves emit one tensor per gate onto its control:
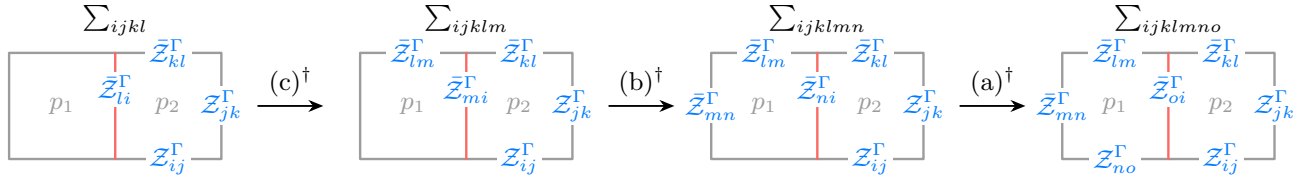
\begin{figure}[H]
\centering
\newcommand{\ecDomTwo}[2]{%
  \draw[skel] (#1-1.3,#2-0.65) rectangle (#1+1.3,#2+0.65);
  \draw[rsk] (#1,#2-0.65)--(#1,#2+0.65);
  \node[xl] at (#1+0.65,#2-0.65) {\zg{ij}};
  \node[xl] at (#1+1.3,#2)       {\zg{jk}};
  \node[xl] at (#1+0.65,#2+0.65) {\zgb{kl}};
  \node[latgray, font=\small] at (#1-0.65,#2) {$p_{1}$};
  \node[latgray, font=\small] at (#1+0.65,#2) {$p_{2}$};}
\begin{tikzpicture}[scale=1.08, transform shape, line cap=round, line join=round]
  \node[eq] at (0,1.05) {$\sum_{ijkl}$};
  \ecDomTwo{0}{0}
  \node[xl] at (0,0.2) {\zgb{li}};
  \draw[arr] (1.75,0)--(2.55,0); \node[eq] at (2.15,0.3) {(c)$^{\dagger}$};
  \node[eq] at (4.3,1.05) {$\sum_{ijklm}$};
  \ecDomTwo{4.3}{0}
  \node[xl] at (3.65,0.65) {\zgb{lm}};
  \node[xl] at (4.3,0.2) {\zgb{mi}};
  \draw[arr] (6.05,0)--(6.85,0); \node[eq] at (6.45,0.3) {(b)$^{\dagger}$};
  \node[eq] at (8.6,1.05) {$\sum_{ijklmn}$};
  \ecDomTwo{8.6}{0}
  \node[xl] at (7.95,0.65) {\zgb{lm}};
  \node[xl] at (7.3,0)    {\zgb{mn}};
  \node[xl] at (8.6,0.2) {\zgb{ni}};
  \draw[arr] (10.35,0)--(11.15,0); \node[eq] at (10.75,0.3) {(a)$^{\dagger}$};
  \node[eq] at (12.9,1.05) {$\sum_{ijklmno}$};
  \ecDomTwo{12.9}{0}
  \node[xl] at (12.25,0.65) {\zgb{lm}};
  \node[xl] at (11.6,0)     {\zgb{mn}};
  \node[xl] at (12.25,-0.65){\zg{no}};
  \node[xl] at (12.9,0.2)  {\zgb{oi}};
\end{tikzpicture}
\caption{Ungauging of the target projector $B^{C_{2}}_{p_{2}}$.
The barred tensor on $d$ is the target of all three gates, each of which emits one tensor onto its control, so the projector grows a tail around $p_{1}$ (common prefactor suppressed).}
\label{fig:ff-tail}
\end{figure}
Gates (c)$^{\dagger}$ and (b)$^{\dagger}$ are un-inverted $CR$ gates and split the barred target tensor by Eq.~\eqref{eq:CR-Zbar}, $\bar{\mathcal{Z}}^{\Gamma}_{li}(d)\to\sum_{m}\bar{\mathcal{Z}}^{\Gamma}_{lm}(\text{top})\bar{\mathcal{Z}}^{\Gamma}_{mi}(d)$ and then $\bar{\mathcal{Z}}^{\Gamma}_{mi}(d)\to\sum_{n}\bar{\mathcal{Z}}^{\Gamma}_{mn}(\text{left})\bar{\mathcal{Z}}^{\Gamma}_{ni}(d)$; gate (a)$^{\dagger}$ is $CR^{\dagger}$ and splits it by the absorption move of Eq.~\eqref{eq:CR-absorb-Zbar} read in reverse, $\bar{\mathcal{Z}}^{\Gamma}_{ni}(d)\to\sum_{o}\mathcal{Z}^{\Gamma}_{no}(\text{bottom})\bar{\mathcal{Z}}^{\Gamma}_{oi}(d)$.
The ungauged projector is the closed contraction around the two plaquettes with an additional barred tensor on $d$ and, for the labelled projector, the insertion $\Gamma(\bar{y})$ at $u$.
Decomposing the dual-edge tensor in the group basis by Eq.~\eqref{eq:Zbar-def}, $\bar{\mathcal{Z}}^{\Gamma}_{oi'}(d)=\sum_{x}T^{x}_{d}\,[\Gamma(\bar{x})]_{oi'}$, the $c$-number $\Gamma(\bar{x})$ sits at $u$ next to $\Gamma(\bar{y})$ and combines into $\Gamma(\bar{x})\Gamma(\bar{y})=\Gamma(\overline{yx})$, the insertion of the loop projector around both plaquettes with label $yx$:
\begin{equation}\label{eq:ff-tail-label-diagram}
M^{\dagger}B^{y}_{p_{2},u}M=\frac{1}{|G|}\sum_{\Gamma}d_{\Gamma}\sum_{\{i\}}\;
\vcenter{\hbox{%
\begin{tikzpicture}[line cap=round, line join=round, baseline]
  \dmlat \dmlab
  \node[xl] at (0.95,-0.75)  {\zg{ij}};
  \node[xl] at (1.5,0)       {\zg{jk}};
  \node[xl] at (0.75,0.75)   {\zgb{kl}};
  \node[xl] at (-0.75,0.75)  {\zgb{lm}};
  \node[xl] at (-1.5,0)      {\zgb{mn}};
  \node[xl] at (-0.95,-0.75) {\zg{no}};
  \node[xl] at (0,0)         {\zgb{oi'}};
  \node[il] at (0,-0.75)     {$\Gamma(\bar{y})_{i'i}$};
\end{tikzpicture}
}}
=\frac{1}{|G|}\sum_{\Gamma}d_{\Gamma}\sum_{x\in G}\sum_{\{i\}}\;
\vcenter{\hbox{%
\begin{tikzpicture}[line cap=round, line join=round, baseline]
  \dmlat \dmlab
  \node[xl] at (0.95,-0.75)  {\zg{ij}};
  \node[xl] at (1.5,0)       {\zg{jk}};
  \node[xl] at (0.75,0.75)   {\zgb{kl}};
  \node[xl] at (-0.75,0.75)  {\zgb{lm}};
  \node[xl] at (-1.5,0)      {\zgb{mn}};
  \node[xl] at (-0.95,-0.75) {\zg{no}};
  \node[xl, text=edgered] at (0,0) {$T^{x}$};
  \node[il] at (0,-0.75)     {$\Gamma(\overline{yx})_{oi}$};
\end{tikzpicture}
}}
\end{equation}
that is,
\begin{equation}\label{eq:ff-tail-label}
M^{\dagger}B^{y}_{p_{2},u}M=\sum_{x\in G}T^{x}_{d}\otimes B^{yx}_{p_{1}p_{2},u},\qquad
M^{\dagger}B^{C_{2}}_{p_{2}}M=\sum_{x\in G}T^{x}_{d}\otimes B^{C_{2}x}_{p_{1}p_{2},u},\qquad
B^{C_{2}x}_{p_{1}p_{2},u}:=\sum_{y\in C_{2}}B^{yx}_{p_{1}p_{2},u}.
\end{equation}
After ungauging, the dual edge holds the source label and the target projector measures the total flux of the two plaquettes relative to it.
Multiplying with the ungauged source projector $T^{C_{1}}_{d}$ of Eq.~\eqref{eq:ff-source-ungauged} restricts the dual-edge label to $C_{1}$, so the ungauged product stabilizer reads
\begin{equation}\label{eq:ff-product-ungauged}
M^{\dagger}B^{C_{1}}_{p_{1}}B^{C_{2}}_{p_{2}}M=\sum_{x\in C_{1}}T^{x}_{d}\otimes B^{C_{2}x}_{p_{1}p_{2},u}
=\frac{1}{|G|}\sum_{\Gamma}d_{\Gamma}\sum_{x\in C_{1}}\sum_{y\in C_{2}}\sum_{\{i\}}\;
\vcenter{\hbox{%
\begin{tikzpicture}[line cap=round, line join=round, baseline]
  \dmlat \dmlab
  \node[xl] at (0.95,-0.75)  {\zg{ij}};
  \node[xl] at (1.5,0)       {\zg{jk}};
  \node[xl] at (0.75,0.75)   {\zgb{kl}};
  \node[xl] at (-0.75,0.75)  {\zgb{lm}};
  \node[xl] at (-1.5,0)      {\zgb{mn}};
  \node[xl] at (-0.95,-0.75) {\zg{no}};
  \node[xl, text=edgered] at (0,0) {$T^{x}$};
  \node[il] at (0,-0.75)     {$\Gamma(\overline{yx})_{oi}$};
\end{tikzpicture}
}}
\end{equation}
Every term is diagonal in the group element basis of $d$ and fixes the dual-edge label, so the operator commutes with each measurement projector $T^{x}_{d}$ term by term; the weighted sum, which is the operator with eigenvalue $+1$ on the state, therefore survives the measurement as a stabilizer.
Note that $M^{\dagger}$ has moved the source label from the tensors around $p_{1}$ onto the dual edge, where it enters the loop projector only through the total label $yx$ of the two plaquettes: on the joint support the dual-edge label lies in $C_{1}$ and the loop label in $C_{2}x$, hence in the classes $C$ with $a^{C}_{C_{1}C_{2}}\neq0$, so
\begin{equation}
\sum_{x\in C_{1}}T^{x}_{d}\otimes B^{C_{2}x}_{p_{1}p_{2},u}
=\Big(\sum_{x\in C_{1}}T^{x}_{d}\otimes B^{C_{2}x}_{p_{1}p_{2},u}\Big)\,\Pi_{p_{1}p_{2}},
\end{equation}
the ungauged form of Eq.~\eqref{eq:ff-product-vs-Pi}.

\paragraph{Step 2a: Measure the dual edge in the group element basis.}

We implement the reset of step 2 as a measurement followed by a flip and track the stabilizer through both parts.
Say we measure the edge to be in the $\ket{x}$ state; outcomes outside $C_{1}$ have probability zero, since $T^{C_{1}}_{d}$ stabilizes the state.
Projecting Eq.~\eqref{eq:ff-product-ungauged} with $T^{x}_{d}$ keeps the terms with the measured label,
\begin{equation}\label{eq:ff-branch-x}
T^{x}_{d}\otimes B^{C_{2}x}_{p_{1}p_{2},u}=\frac{1}{|G|}\sum_{\Gamma}d_{\Gamma}\sum_{y\in C_{2}}\sum_{\{i\}}\;
\vcenter{\hbox{%
\begin{tikzpicture}[line cap=round, line join=round, baseline]
  \dmlat \dmlab
  \node[xl] at (0.95,-0.75)  {\zg{ij}};
  \node[xl] at (1.5,0)       {\zg{jk}};
  \node[xl] at (0.75,0.75)   {\zgb{kl}};
  \node[xl] at (-0.75,0.75)  {\zgb{lm}};
  \node[xl] at (-1.5,0)      {\zgb{mn}};
  \node[xl] at (-0.95,-0.75) {\zg{no}};
  \node[xl, text=edgered] at (0,0) {$T^{x}$};
  \node[il] at (0,-0.75)     {$\Gamma(\overline{yx})_{oi}$};
\end{tikzpicture}
}}
\end{equation}
the loop around both plaquettes measured relative to the outcome: the outcome-$x$ branch is stabilized by the total flux projector whose label is the product of the target label $y\in C_{2}$ with the source label $x$.

\paragraph{Step 2b: Flip the outcome to $\ket{e}$.}

The reset applies $L^{\bar{x}}_{d}$ of Eq.~\eqref{eq:LR-def}, which takes the edge to $\ket{e}$ and the factor $T^{x}_{d}$ of Eq.~\eqref{eq:ff-branch-x} to the reset stabilizer $T^{e}_{d}$ of Eq.~\eqref{eq:ff-reset-stabilizer}, while the loop projector has no support on $d$ and is untouched.
The outcome-$x$ branch is therefore stabilized by $T^{e}_{d}$ together with
\begin{equation}\label{eq:ff-branch-x-reset}
B^{C_{2}x}_{p_{1}p_{2},u}=\frac{1}{|G|}\sum_{\Gamma}d_{\Gamma}\sum_{y\in C_{2}}\sum_{\{i\}}\;
\vcenter{\hbox{%
\begin{tikzpicture}[line cap=round, line join=round, baseline]
  \dmlat \dmlab
  \node[xl] at (0.95,-0.75)  {\zg{ij}};
  \node[xl] at (1.5,0)       {\zg{jk}};
  \node[xl] at (0.75,0.75)   {\zgb{kl}};
  \node[xl] at (-0.75,0.75)  {\zgb{lm}};
  \node[xl] at (-1.5,0)      {\zgb{mn}};
  \node[xl] at (-0.95,-0.75) {\zg{no}};
  \node[il] at (0,-0.75)     {$\Gamma(\overline{yx})_{oi}$};
\end{tikzpicture}
}}
\end{equation}
The label $x$ is discarded, so the branches are mixed and only what is common to all $x\in C_{1}$ is guaranteed: summing the branch stabilizers with their multiplicities gives $\sum_{x\in C_{1}}B^{C_{2}x}_{p_{1}p_{2},u}=\sum_{C}a^{C}_{C_{1}C_{2}}B^{C}_{p_{1}p_{2}}$, whose support projector is $\Pi_{p_{1}p_{2}}$, and $\Pi_{p_{1}p_{2}}$ commutes with every Kraus operator $\ketbra{e}{x}_{d}$.
Stabilizers: $\{\Pi_{p_{1}p_{2}},\,T^{e}_{d}\}$, the closed-loop situation of Appendix~\ref{app:ff-ribbon} with $B^{C}_{p_{1}p_{2}}$ replaced by the sum over the allowed classes.

\paragraph{Step 3: Regauging circuit $M$.}

$\Pi_{p_{1}p_{2}}$ has no support on $d$ and is invariant by (F1), and the reset stabilizer regauges to the trivial flux projector, $MT^{e}_{d}M^{\dagger}=B^{e}_{p_{1}}$, by the splitting chain of Fig.~\ref{fig:M-splitting-chain}.
Stabilizers: $\{\Pi_{p_{1}p_{2}},\,B^{e}_{p_{1}}\}$.
The lemma of Appendix~\ref{app:ff-lemma} is linear in the class function, so it applies term by term, $B^{e}_{p_{1}}\Pi_{p_{1}p_{2}}=B^{e}_{p_{1}}\Pi_{p_{2}}$, and the closed ribbon deforms to
\begin{equation}\label{eq:ff-Pi-p2}
\Pi_{p_{2}}=\sum_{C:\,a^{C}_{C_{1}C_{2}}\neq0}B^{C}_{p_{2}}
=\sum_{C:\,a^{C}_{C_{1}C_{2}}\neq0}\frac{|C|}{|G|}\sum_{\Gamma}\bar{\chi}^{\Gamma}(C)\sum_{jklm}\;
\vcenter{\hbox{%
\begin{tikzpicture}[line cap=round, line join=round, baseline]
  \dmlat \dmlab
  \node[xl] at (0.75,-0.75)  {\zg{jk}};
  \node[xl] at (1.5,0)       {\zg{kl}};
  \node[xl] at (0.75,0.75)   {\zgb{lm}};
  \node[xl] at (0,0)         {\zgb{mj}};
\end{tikzpicture}
}}
\end{equation}
the multiplicity-blind sum of orthogonal flux-class projectors over the fusion channels, with all weights equal to $1$, which establishes the mapping stated in Sec.~\ref{subsec:flux-fusion}.
Output stabilizers: $\{\Pi_{p_{2}},\,B^{e}_{p_{1}}\}$.
If the outcome $x$ is kept, the same steps carry the branch stabilizer of Eq.~\eqref{eq:ff-branch-x-reset} to the labelled single plaquette projector
\begin{equation}\label{eq:ff-branch-x-output}
B^{C_{2}x}_{p_{2},u}=\frac{1}{|G|}\sum_{\Gamma}d_{\Gamma}\sum_{y\in C_{2}}\sum_{\{i\}}\;
\vcenter{\hbox{%
\begin{tikzpicture}[line cap=round, line join=round, baseline]
  \dmlat \dmlab
  \node[xl] at (0.95,-0.75)  {\zg{ij}};
  \node[xl] at (1.5,0)       {\zg{jk}};
  \node[xl] at (0.75,0.75)   {\zgb{kl}};
  \node[xl] at (0,0)         {\zgb{lo}};
  \node[il] at (0,-0.75)     {$\Gamma(\overline{yx})_{oi}$};
\end{tikzpicture}
}}
\end{equation}
by the labelled form $B^{e}_{p_{1}}B^{yx}_{p_{1}p_{2},u}=B^{e}_{p_{1}}B^{yx}_{p_{2},u}$ of the lemma: the target label $y$ and the source label $x$ fuse to $yx$ at $p_{2}$.

\paragraph{Heisenberg identities.}
Applying the reset rule \eqref{eq:ff-reset-rule} to Eq.~\eqref{eq:ff-tail-label} keeps only $x=e$, so $\Phi^{*}(B^{C}_{p_{2}})=B^{C}_{p_{1}p_{2}}$ and $\Phi^{*}(B^{y}_{p_{2},u})=B^{y}_{p_{1}p_{2},u}$ for every input.
The weight of class $C$ in the output is therefore the input's total-flux weight $\Tr[B^{C}_{p_{1}p_{2}}\rho]$; it vanishes identically when $a^{C}_{C_{1}C_{2}}=0$, and an input stabilized by $B^{C}_{p_{1}p_{2}}$ with $a^{C}_{C_{1}C_{2}}\neq0$ gives weight one, so no smaller sum of flux-class projectors than $\Pi_{p_{2}}$ stabilizes all outputs.
In labelled form, an input stabilized by $B^{x_{1}}_{p_{1},u}$ and $B^{x_{2}}_{p_{2},u}$ is stabilized by $B^{x_{2}x_{1}}_{p_{1}p_{2},u}$ by Eq.~\eqref{eq:ff-label-composition}, hence its output by $B^{x_{2}x_{1}}_{p_{2},u}$: the two labels fuse to $x_{2}x_{1}$, and the order matters for non-Abelian $G$.

\subsection{Mapping of the closed dyonic ribbon}
\label{app:ff-dyon}

We follow the fusion projector through $M^\dagger$, reset, and $M$.
The closed ribbon survives every step; reset makes the source plaquette flux free, allowing its conserved total type to be identified at $(p_2,u)$.
We first track a definite fusion channel, then establish the probability-transfer identity for general inputs.

\subsubsection*{The fusion projector and the input state}

We now describe the structure of the dyonic fusion projector shown in the main text.
The shared edge $d$ between the two plaquettes points from $u$ to $w$.
The plaquette $p_1$ is bounded by the edges $b_1$ (bottom), $d$ (right), $t_1$ (top) and $l_1$ (left), the plaquette $p_2$ by $d$ (left), $b_2$ (bottom), $r_2$ (right) and $t_2$ (top), the six perimeter vertices counterclockwise from the bottom-left corner are $v_1,u,v_3,v_4,w,v_6$, and $g_x$ denotes the group element on edge $x$ with $\bar g=g^{-1}$.
We follows the square lattice edge orientation convention, and define
\begin{equation}
 k=\bar g_{t_1}\bar g_{l_1}g_{b_1},\qquad
 a=g_{b_2}g_{r_2}\bar g_{t_2},\qquad
 x=h_{p_1,u}=g_dk,\qquad
 H=h_{p_1p_2,u}=ak=h_{p_2,u}h_{p_1,u}.
 \label{eq:df-holonomies}
\end{equation}
Thus $x$ is the source flux and $H$ is the total flux, both based at $u$.
Importantly, $H$ depends only on the perimeter, not on $d$.

The fusion projector must resolve both $H$ and its centralizer charge.
Let $\sigma$ run counterclockwise around $p_1\cup p_2$, starting at $u$ along $b_2$.
Write $\mathcal P_v$ for the ordered perimeter path product from $u$ to a boundary vertex $v$, so $\mathcal P_u=e$ and $\mathcal P_w=a$.
On each perimeter configuration, define the transported outer-leg action
\begin{equation}
 z_v=\overline{\mathcal P}_v z\mathcal P_v,\qquad
 W_z=\prod_{v\in\partial(p_1\cup p_2)}
          \bigl(A_v^{z_v}\bigr)_{\mathrm{outer}},\qquad
 F_\sigma^{(h,z)}=B^h_{p_1p_2,u}W_z.
 \label{eq:ff-dyonic-ribbon}
\end{equation}
``Outer'' includes every edge leaving the two-plaquette region, including the two outer legs at each corner.
All ten outer legs are shown in the diagrams below.
The perimeter variables are diagonal controls; $W_z$ changes only these outer legs, leaving both the perimeter and $d$ untouched.

Fix $h_0\in C$ and an irrep $\pi$ of $Z(h_0)=\{z:zh_0=h_0z\}$.
For $h=q_hh_0\bar q_h$, set $\chi_\pi^{(h)}(z)=\chi_\pi(\bar q_hzq_h)$ for $z\in Z(h)$.
The choice of $q_h$ does not affect the character.
With $\bar\chi$ denoting complex conjugation, the fusion projector is
\begin{equation}
 \boxed{
 K_\sigma^{(C,\pi)}
   =\sum_{h\in C}B^h_{p_1p_2,u}\,
      \frac{d_\pi}{|Z(h)|}\sum_{z\in Z(h)}
          \bar\chi_\pi^{(h)}(z)W_z.}
 \label{eq:ff-K}
\end{equation}
Remark that $W_zW_y=W_{zy}$, $W_z^\dagger=W_{\bar z}$, and $[W_z,B^h_{p_1p_2,u}]=0$.
The character sum is therefore the centralizer's representation projector; the different $(C,\pi)$ give orthogonal Hermitian projectors resolving the identity~\cite{bombin_family_2008}.

The diagrams use the same blue edge operators and red shared edge as the other stabilizer calculations.
On an edge with label $q$, $\mathcal Z^\Gamma_{ab}\ket q=\Gamma(q)_{ab}\ket q$ and $\bar{\mathcal Z}^\Gamma_{ab}\ket q=\Gamma(\bar q)_{ab}\ket q$.
The insertion $[\Gamma(\bar h)]_{jj'}$, written outside a diagram, opens the magnetic trace at $u$: the loop starts along $b_2$ with index $j'$ and returns along $b_1$ with index $j$.
The blue outer-leg shifts use $z_i=z_{v_i}$ and $z_w=\bar a z a$.
All drawings use the fixed hardware skeleton; $I_d$ explicitly marks an edge outside the operator's support.
Each tensor diagram is a term in the displayed sums, not a separate stabilizer.
Summing over $\Gamma$ and the bond indices gives the fixed-$(h,z)$ operator $F_\sigma^{(h,z)}$, and the centralizer sum in Eq.~\eqref{eq:ff-K} gives the fusion projector.

For the rest of the derivation, abbreviate the fixed projector as $K=K_\sigma^{(C,\pi)}$.
Assume that the five perimeter vertices other than $u$ are neutral, and write
\begin{equation}
 \mathcal N=\{v_1,v_3,v_4,w,v_6\},\qquad
 P_v=\frac1{|G|}\sum_g A_v^g,\qquad
 \Pi_{\mathrm{in}}=\prod_{v\in\mathcal N}P_v.
 \label{eq:df-subspaces}
\end{equation}
The input has $K\rho_0=\rho_0$ and $\Pi_{\mathrm{in}}\rho_0=\rho_0$.
These projectors commute.
No neutrality is required at $u$, although initially charge-free vertices are included in this input class.

\subsubsection*{Step 1: Rewire and isolate the shared edge}

On a group-basis configuration, the circuit acts as
\begin{equation}
 M\ket{g_d}=\ket{g_d\bar k},\qquad
 M^\dagger\ket{g_d}=\ket{g_dk}=\ket{x}_d,\qquad
 \rho_1=M^\dagger\rho_0M.
 \label{eq:df-rewire}
\end{equation}
All other edge labels are unchanged.
Thus the source flux is written onto $d$, while the total flux and the outer-leg actions remain on the rest of the graph.
Indeed, $M$ and the fusion projector share only diagonal perimeter controls, and act on disjoint targets.
Each term of Eq.~\eqref{eq:ff-K} therefore obeys
\begin{equation}
 M^\dagger\bigl(B^h_{p_1p_2,u}W_z\bigr)M
   =I_d\otimes\bigl(B^h_{p_1p_2,u}W_z\bigr),\qquad
 M^\dagger K M=I_d\otimes\widehat K,
 \label{eq:df-step1-projector}
\end{equation}
where the hatted projector acts on all edges except $d$ and has exactly the same character sum as $K$.
The explicit identity factor records the absence of support on the extracted qudit.

In the tensor notation this unchanged retained operator is
\begin{equation}
 I_d\otimes\widehat F_\sigma^{(h,z)}
 =\frac1{|G|}\sum_\Gamma d_\Gamma
       \sum_{ijj'klmn}[\Gamma(\bar h)]_{jj'}\;
       \vcenter{\hbox{\fdDyonBoundary{2}}}.
 \label{eq:dg-step1-tensors}
\end{equation}
The six magnetic tensors are diagonal on the perimeter; all electric shifts lie on outer legs.
Neither has support on the reset target.
This is also the input diagram: rewiring does not move or alter any of these factors.

We also track the neutrality constraints, since they will be needed to identify the final site.
Let $\widetilde A_u^g$ and $\widetilde A_w^g$ omit the factor on $d$ from $A_u^g$ and $A_w^g$.
Under $A_u^g$, $(g_d,k)$ changes to $(gg_d,k\bar g)$, so $x\mapsto gx\bar g$; under $A_w^g$, it changes to $(g_d\bar g,gk)$, so $x$ is unchanged.
Hence
\begin{equation}
 M^\dagger A_u^g M=C_d^g\otimes\widetilde A_u^g,\qquad
 M^\dagger A_w^g M=I_d\otimes\widetilde A_w^g,\qquad
 [M,A_v^g]=0\quad(v=v_1,v_3,v_4,v_6),
 \label{eq:ff-vertex-ungauged}
\end{equation}
with $C_d^g=L_d^gR_d^g$.
The two nontrivial vertex transformations are displayed on the same skeleton:
\begin{equation}
 \begin{aligned}
 A_u^g:\quad
 \vcenter{\hbox{\fdDyonVertex[.85]{0}}}
 &\xrightarrow{\ M^\dagger(\cdot)M\ }
 \vcenter{\hbox{\fdDyonVertex[.85]{1}}},\\[3pt]
 A_w^g:\quad
 \vcenter{\hbox{\fdDyonVertex[.85]{2}}}
 &\xrightarrow{\ M^\dagger(\cdot)M\ }
 \vcenter{\hbox{\fdDyonVertex[.85]{3}}}.
 \end{aligned}
 \label{eq:dg-vertex-diagrams}
\end{equation}
Corner legs, on which these particular operators act trivially, are omitted here.
The corner transformations cancel inside $k$ or have disjoint support.
Thus $M^\dagger\Pi_{\mathrm{in}}M=I_d\otimes\widehat\Pi$, where $\widehat\Pi$ is the product of the four corner vacuum projectors and $|G|^{-1}\sum_g\widetilde A_w^g$.
Both constraints now act only on retained edges: $(I_d\otimes\widehat K)\rho_1=(I_d\otimes\widehat\Pi)\rho_1=\rho_1$.
No disentanglement of $d$ is assumed.

\subsubsection*{Step 2: Reset without changing the fusion label}

Reset discards the old state of $d$ and prepares $\ket{e}$:
\begin{equation}
 \rho_2=T_d^e\otimes\tau,\qquad
 T_d^e=\ketbra{e}{e}_d,\qquad
 \tau=\operatorname{Tr}_d\rho_1.
 \label{eq:df-reset}
\end{equation}
Because the fusion projector has no support on $d$, it passes through the partial trace:
\begin{equation}
 \widehat K\tau
  =\operatorname{Tr}_d\!\left[(I_d\otimes\widehat K)\rho_1\right]
  =\tau,\qquad
 \widehat\Pi\tau=\tau.
 \label{eq:df-reset-projector}
\end{equation}
The new single-edge stabilizer is $T_d^e=|G|^{-1}\sum_{\Gamma'}d_{\Gamma'}\sum_\alpha \mathcal Z^{\Gamma'}_{\alpha\alpha}(d)$.
Its product with each retained ribbon term is
\begin{equation}
 T_d^e\otimes\widehat F_\sigma^{(h,z)}
 =\frac1{|G|^2}\sum_{\Gamma,\Gamma'}d_\Gamma d_{\Gamma'}
    \sum_{ijj'klmn,\alpha}[\Gamma(\bar h)]_{jj'}\;
    \vcenter{\hbox{\fdDyonBoundary{3}}}.
 \label{eq:dg-step2-tensors}
\end{equation}
Only the new $\Gamma'$ trace has appeared on $d$; every tensor of the fusion projector is unchanged.
The $\Gamma'$ indices are independent of the indices on the original $\Gamma$ loop.

The post-reset state is therefore supported on $T_d^e\otimes\widehat K\widehat\Pi$.
The old fusion constraint is unchanged; the new constraint is $T_d^e=1$.
In physical terms, reset removes the source flux, not the total anyon label recorded on the retained graph.
For an input with several fusion channels, the same trace preserves each channel's probability; it need not preserve coherence between channels.

\subsubsection*{Step 3: Restore the graph and identify the local dyon}

Finally $\rho_3=M\rho_2M^\dagger=\Phi(\rho_0)$.
Reversing the rewiring restores the old constraints and turns the reset constraint into trivial source flux:
\begin{equation}
 M(I_d\otimes\widehat K)M^\dagger=K,\qquad
 M(I_d\otimes\widehat\Pi)M^\dagger=\Pi_{\mathrm{in}},\qquad
 M(T_d^e\otimes I)M^\dagger=B^e_{p_1}.
 \label{eq:df-step3-projector}
\end{equation}
In the drawing, $M$ expands the $\Gamma'$ trace around $p_1$ while leaving the original boundary tensors and outer-leg shifts fixed:
\begin{equation}
 B^e_{p_1}F_\sigma^{(h,z)}
 =\frac1{|G|^2}\sum_{\Gamma,\Gamma'}d_\Gamma d_{\Gamma'}
    \sum_{ijj'klmn}\sum_{\alpha\beta\gamma\delta}
       [\Gamma(\bar h)]_{jj'}\;
    \vcenter{\hbox{\fdDyonBoundary{4}}}.
 \label{eq:dg-step3-tensors}
\end{equation}

\ifdefined\dfAfterFThirtyEightLoaded\else
\def\dfAfterFThirtyEightLoaded{1}

\makeatletter
\@ifundefined{color@opblue}
  {\definecolor{opblue}{RGB}{0,122,255}}{}
\@ifundefined{color@latgray}
  {\definecolor{latgray}{RGB}{156,156,156}}{}
\@ifundefined{color@edgered}
  {\definecolor{edgered}{RGB}{250,106,100}}{}
\@ifundefined{color@ribgreen}
  {\definecolor{ribgreen}{RGB}{52,168,83}}{}
\makeatother

\tikzset{
  dfa edge/.style={
    draw=latgray,
    line width=.75pt
  },
  dfa shared/.style={
    draw=edgered,
    line width=1.15pt
  },
  dfa op/.style={
    text=opblue,
    fill=white,
    inner sep=1.1pt,
    font=\small
  },
  dfa label/.style={
    text=latgray!80!black,
    font=\scriptsize,
    inner sep=.8pt
  },
  dfa identity/.style={
    text=latgray,
    font=\small,
    fill=white,
    inner sep=1pt
  },
  dfa flux/.style={
    draw=ribgreen,
    line width=.8pt,
    rounded corners=3pt,
    densely dashed
  },
  dfa fluxlabel/.style={
    text=ribgreen!60!black,
    fill=white,
    font=\scriptsize,
    inner sep=1pt
  }
}

\newcommand{\dfAfterCore}{%
  \draw[dfa edge]
    (-1.55,-.85) rectangle (1.55,.85);
  \draw[dfa shared]
    (0,-.85)--(0,.85);
  \foreach \xx in {-1.55,0,1.55}{
    \foreach \yy in {-.85,.85}{
      \fill[latgray] (\xx,\yy) circle(1.2pt);
    }
  }
  \node[dfa label,anchor=south east]
    at (-.03,-.82) {$u$};
  \node[dfa label,anchor=north east]
    at (-.03,.82) {$w$};
  \node[dfa label]
    at (-.75,-.22) {$p_1$};
}

\newcommand{\dfAfterLegs}{%
  \foreach \xx in {-1.55,0,1.55}{
    \draw[dfa edge] (\xx,-.85)--(\xx,-1.78);
    \draw[dfa edge] (\xx,.85)--(\xx,1.78);
  }
  \foreach \yy in {-.85,.85}{
    \draw[dfa edge] (-2.52,\yy)--(-1.55,\yy);
    \draw[dfa edge] (1.55,\yy)--(2.52,\yy);
  }
  \node[dfa label,anchor=south west]
    at (-1.52,-.82) {$v_1$};
  \node[dfa label,anchor=south east]
    at (1.52,-.82) {$v_3$};
  \node[dfa label,anchor=north east]
    at (1.52,.82) {$v_4$};
  \node[dfa label,anchor=north west]
    at (-1.52,.82) {$v_6$};
}

\newcommand{\dfAfterOuter}{%
  \node[dfa op] at (-2.07,-.85) {$R^{z_1}$};
  \node[dfa op] at (-1.55,-1.37) {$R^{z_1}$};
  \node[dfa op] at (0,-1.37) {$R^{z}$};
  \node[dfa op] at (1.55,-1.37) {$R^{z_3}$};
  \node[dfa op] at (2.07,-.85) {$L^{z_3}$};
  \node[dfa op] at (2.07,.85) {$L^{z_4}$};
  \node[dfa op] at (1.55,1.37) {$L^{z_4}$};
  \node[dfa op] at (0,1.37) {$L^{z_w}$};
  \node[dfa op] at (-1.55,1.37) {$L^{z_6}$};
  \node[dfa op] at (-2.07,.85) {$R^{z_6}$};
}

\newcommand{\dfAfterMagnetic}[1]{%
  \begin{tikzpicture}[
    x=1cm,y=1cm,
    line cap=round,
    line join=round,
    baseline=(current bounding box.center)
  ]
    \fill[ribgreen!5]
      (-1.55,-.85) rectangle (0,.85);
    \dfAfterCore
    \node[dfa label] at (.78,-.23) {$p_2$};
    \node[dfa fluxlabel]
      at (-.77,.22) {$B^e_{p_1}=1$};

    \node[dfa op]
      at (.77,-.85) {$\mathcal Z^\Gamma_{j'k}$};
    \node[dfa op]
      at (1.55,0) {$\mathcal Z^\Gamma_{kl}$};
    \node[dfa op]
      at (.77,.85) {$\bar{\mathcal Z}^\Gamma_{lm}$};

    \ifnum#1=0\relax
      \node[dfa op]
        at (-.77,-.85) {$\mathcal Z^\Gamma_{ij}$};
      \node[dfa op]
        at (-1.55,0) {$\bar{\mathcal Z}^\Gamma_{ni}$};
      \node[dfa op]
        at (-.77,.85) {$\bar{\mathcal Z}^\Gamma_{mn}$};
      \node[dfa identity] at (0,0) {$I_d$};
    \else
      \ifnum#1=1\relax
        \node[dfa op] at (-.77,-.85)
          {$\bar{\mathcal Z}^\Gamma_{ni}
            \mathcal Z^\Gamma_{ij}$};
        \node[dfa op]
          at (0,0) {$\bar{\mathcal Z}^\Gamma_{mn}$};
      \else
        \node[dfa op]
          at (0,0) {$\bar{\mathcal Z}^\Gamma_{mj}$};
      \fi
    \fi

    \path (-2.05,-1.12) rectangle (2.05,1.12);
  \end{tikzpicture}%
}

\newcommand{\dfAfterElectric}[1]{%
  \begin{tikzpicture}[
    x=.9cm,y=.9cm,
    line cap=round,
    line join=round,
    baseline=(current bounding box.center)
  ]
    \dfAfterCore
    \dfAfterLegs
    \draw[dfa flux]
      (.28,-.51) rectangle (1.25,.51);
    \node[dfa fluxlabel]
      at (.77,.19) {$B^h_{p_2,u}$};

    \ifnum#1=0\relax
      \dfAfterOuter
      \node[dfa identity] at (0,0) {$I_d$};
    \else
      \node[dfa op] at (0,-1.37) {$R^z$};
      \ifnum#1=1\relax
        \node[dfa op] at (-.77,-.85) {$L^{\bar z_1}$};
        \node[dfa op] at (.77,-.85) {$R^{\bar z_3}$};
        \node[dfa op] at (0,0) {$R^{\bar z_w}$};

        \foreach \xx/\yy in {
          -2.07/-.85,
          -1.55/-1.37,
          1.55/-1.37,
          2.07/-.85,
          2.07/.85,
          1.55/1.37,
          0/1.37,
          -1.55/1.37,
          -2.07/.85
        }{
          \node[dfa identity] at (\xx,\yy) {$I$};
        }
        \foreach \xx/\yy in {
          -1.55/0,
          -.77/.85,
          .77/.85,
          1.55/0
        }{
          \node[dfa identity] at (\xx,\yy) {$I$};
        }
      \else
        \node[dfa op] at (-.77,-.85) {$R^z$};
        \node[dfa op] at (.77,-.85) {$L^z$};
        \node[dfa op] at (0,0) {$L^z$};
      \fi
    \fi

    \path (-2.65,-1.84) rectangle (2.65,1.84);
  \end{tikzpicture}%
}

\newcommand{\dfAfterLocalTensor}{%
  \begin{tikzpicture}[
    x=1cm,y=.88cm,
    line cap=round,
    line join=round,
    baseline=(current bounding box.center)
  ]
    \dfAfterCore
    \dfAfterLegs
    \node[dfa label] at (.77,-.2) {$p_2$};
    \node[dfa op] at (-.77,-.85) {$R^z$};
    \node[dfa op] at (0,-1.37) {$R^z$};
    \node[dfa op]
      at (.77,-.85) {$\mathcal Z^\Gamma_{j'k}L^z$};
    \node[dfa op]
      at (1.55,0) {$\mathcal Z^\Gamma_{kl}$};
    \node[dfa op]
      at (.77,.85) {$\bar{\mathcal Z}^\Gamma_{lm}$};
    \node[dfa op]
      at (0,0) {$\bar{\mathcal Z}^\Gamma_{mj}L^z$};
  \end{tikzpicture}%
}

\fi

The $\Gamma'$ loop in Eq.~\eqref{eq:dg-step3-tensors} is $B^e_{p_1}$; its indices are independent of the original $\Gamma$ loop.
The final state satisfies
\begin{equation}
  K\rho_3=\rho_3,\qquad
  \Pi_{\mathrm{out}}\rho_3=\rho_3,\qquad
  \Pi_{\mathrm{out}}=B^e_{p_1}\Pi_{\mathrm{in}}.
  \label{eq:df-output-support}
\end{equation}
Two moves turn the boundary ribbon into a single-site projector: $B^e_{p_1}$ moves the magnetic loop to $p_2$, and the five neutral vertices move the electric action to $u$.
These rewrite operators on the output subspace, not add circuit steps.
Work first with $F_\sigma^{(h,z)}=B^h_{p_1p_2,u}W_z$, $z\in Z(h)$; the character sum will be restored at the end.

\paragraph*{1. Move the magnetic loop across the flux-free plaquette.}

The source constraint gives $g_d\bar g_{t_1}\bar g_{l_1}g_{b_1}=e$, hence $\bar g_{t_1}\bar g_{l_1}=\bar g_d\bar g_{b_1}$.
Thus the segment of the magnetic loop along $t_1,l_1$ can be rerouted along $d,b_1$, with both new edges traversed backwards:
\begin{equation}
  B^e_{p_1}\sum_n
    \bar{\mathcal Z}^\Gamma_{mn}(t_1)
    \bar{\mathcal Z}^\Gamma_{ni}(l_1)
  =
  B^e_{p_1}\sum_n
    \bar{\mathcal Z}^\Gamma_{mn}(d)
    \bar{\mathcal Z}^\Gamma_{ni}(b_1).
  \label{eq:dfa-magnetic-slide}
\end{equation}
The loop now traverses $b_1$ in both directions.
Its two tensors cancel, $\sum_i\bar{\mathcal Z}^\Gamma_{ni}(b_1) \mathcal Z^\Gamma_{ij}(b_1)=\delta_{nj}$, leaving the four tensors around $p_2$:

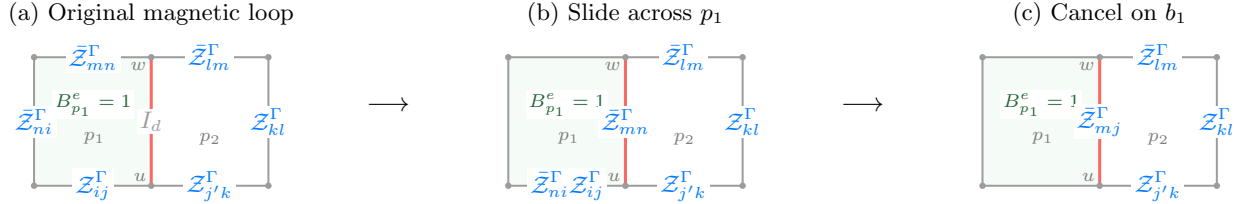
\begin{center}
  \begin{minipage}[c]{.30\linewidth}
    \centering
    \small (a) Original magnetic loop
    \par\smallskip
    \dfAfterMagnetic{0}
  \end{minipage}\hfill
  $\longrightarrow$\hfill
  \begin{minipage}[c]{.30\linewidth}
    \centering
    \small (b) Slide across $p_1$
    \par\smallskip
    \dfAfterMagnetic{1}
  \end{minipage}\hfill
  $\longrightarrow$\hfill
  \begin{minipage}[c]{.30\linewidth}
    \centering
    \small (c) Cancel on $b_1$
    \par\smallskip
    \dfAfterMagnetic{2}
  \end{minipage}

  \par\smallskip
  \begin{minipage}{.97\linewidth}
    \small
    The common factor $B^e_{p_1}|G|^{-1}\sum_\Gamma d_\Gamma [\Gamma(\bar h)]_{jj'}$ and the bond-index sums are implicit.
    $W_z$ and its original perimeter transporters are unchanged and not drawn in this magnetic move.
  \end{minipage}
\end{center}

The unchanged insertion at $u$ gives the labelled receiver projector:
\begin{equation}
  B^h_{p_1p_2,u}\Pi_{\mathrm{out}}
    =B^h_{p_2,u}\Pi_{\mathrm{out}},\qquad
  F_\sigma^{(h,z)}\Pi_{\mathrm{out}}
    =B^h_{p_2,u}W_z\Pi_{\mathrm{out}}.
  \label{eq:df-local-magnetic}
\end{equation}
The magnetic loop is now local; the blue shifts still act on all ten outer legs.

\paragraph*{2. Cancel the outer-leg shifts using the neutral vertices.}

The transport rule $z_v=\mathcal P_v^{-1}z\mathcal P_v$ implies $z_{v'}=q^{-1}z_vq$ on a perimeter edge $v\to v'$ with label $q$.
The following diagram shows how to turn the boundary action into $A_u^z$.
For the cancellation, fix a group-basis configuration $\ket{\mathbf g}$ with $H=h$ and $h_{p_1,u}=e$, evaluate all $z_v$ on that configuration, and use the inverse full vertex transformations
\begin{equation}
  U_z(\mathbf g)
    =\prod_{v\in\mathcal N}A_v^{\bar z_v(\mathbf g)},
  \qquad
  \mathcal N=\{v_1,v_3,v_4,w,v_6\}.
  \label{eq:dfa-compensating-stars}
\end{equation}

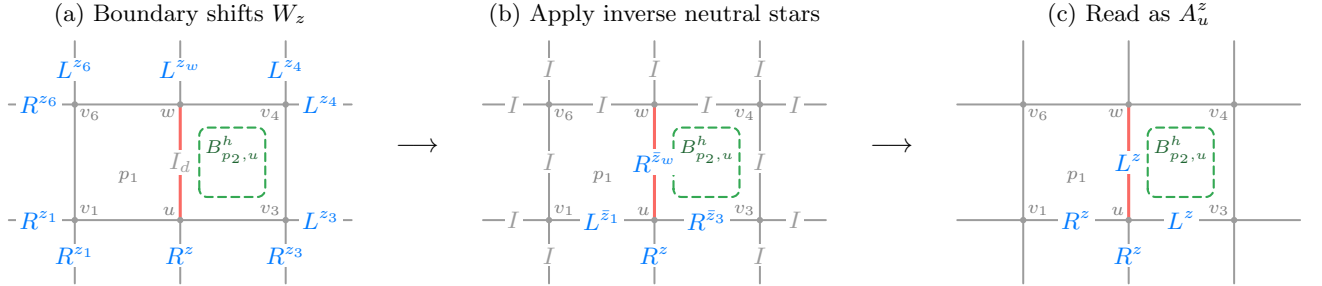
\begin{center}
  \begin{minipage}[c]{.30\linewidth}
    \centering
    \small (a) Boundary shifts $W_z$
    \par\smallskip
    \dfAfterElectric{0}
  \end{minipage}\hfill
  $\longrightarrow$\hfill
  \begin{minipage}[c]{.30\linewidth}
    \centering
    \small (b) Apply inverse neutral stars
    \par\smallskip
    \dfAfterElectric{1}
  \end{minipage}\hfill
  $\longrightarrow$\hfill
  \begin{minipage}[c]{.30\linewidth}
    \centering
    \small (c) Read as $A_u^z$
    \par\smallskip
    \dfAfterElectric{2}
  \end{minipage}

  \par\smallskip
  \begin{minipage}{.97\linewidth}
    \small
    In (b), $z_v$ is evaluated on the initial configuration.
    Inverse neutral stars cancel the outer shifts and the four perimeter-edge actions away from $u$, leaving the three internal shifts shown.
    The green projector multiplies the shifts on the left; unlabelled edges in (c) carry the identity.
    These cancellations are detailed below.
  \end{minipage}
\end{center}

There are three elementary cancellations.
On an outer leg at $v\ne u$, the original shift cancels against the inverse star: $L^{\bar z_v}L^{z_v}=I$ or $R^{\bar z_v}R^{z_v}=I$.
On any of the four perimeter edges not incident to $u$, the two inverse endpoint transformations cancel because $\bar z_v q z_{v'}=q$.
Thus only the three internal edges incident to $u$ can still change.
Their transport labels are
\begin{equation}
  z_3=\bar g_{b_2}z g_{b_2},\qquad
  z_1=g_{b_1}z\bar g_{b_1},\qquad
  z_w=\bar g_d z g_d.
  \label{eq:dfa-three-transporters}
\end{equation}
For the second equality use $\mathcal P_{v_1}g_{b_1}=H=h$ and $z\in Z(h)$.
For the third, the right-hand path is $\mathcal P_w=a=Hg_d=hg_d$, where $a=Hg_d$ uses the source-flatness constraint.
This is precisely where trivial flux at $p_1$ is needed to replace the perimeter transport to $w$ by transport along $d$.

The three remaining shifts in panel (b) consequently act as
\begin{equation}
  \begin{aligned}
    b_1:\quad
      &\bar z_1 g_{b_1}=g_{b_1}\bar z
      &&\Longrightarrow\quad R^z_{b_1},\\
    b_2:\quad
      &g_{b_2}z_3=z g_{b_2}
      &&\Longrightarrow\quad L^z_{b_2},\\
    d:\quad
      &g_dz_w=z g_d
      &&\Longrightarrow\quad L^z_d.
  \end{aligned}
  \label{eq:dfa-edge-cancellation}
\end{equation}
Together with the unchanged $R^z_{s_u}$, these are exactly $A_u^z=R^z_{b_1}L^z_dL^z_{b_2}R^z_{s_u}$:
\[
  U_z(\mathbf g)W_z\ket{\mathbf g}
    =A_u^z\ket{\mathbf g}.
\]

To turn this configuration-wise cancellation into an operator identity, apply $\Pi_{\mathrm{in}}$ on the left.
For each fixed configuration the $z_v(\mathbf g)$ are ordinary group elements, so $\Pi_{\mathrm{in}}U_z(\mathbf g)=\Pi_{\mathrm{in}}$.
Moreover, $A_u^z$ and $W_z$ commute with the neutral-vertex projectors: in the latter case the paths and their outer-leg endpoints transform together.
These projectors preserve both $H=h$ and $h_{p_1,u}=e$.
Therefore
\begin{equation}
  B^h_{p_2,u}W_z\Pi_{\mathrm{out}}
    =B^h_{p_2,u}A_u^z\Pi_{\mathrm{out}},
  \qquad z\in Z(h).
  \label{eq:df-boundary-reduction}
\end{equation}
The inverse stars are used only for this fixed-configuration comparison and then absorbed by the neutral projector; no extra operation or gauge fixing is required in the protocol.

\paragraph*{3. Restore the character sum: the local dyon projector.}

Combining the two moves gives, term by term,
\begin{equation}
  F_\sigma^{(h,z)}\Pi_{\mathrm{out}}
    =B^h_{p_2,u}A_u^z\Pi_{\mathrm{out}}.
  \label{eq:dfa-fixed-term-local}
\end{equation}
The resulting local tensor diagram is
\begin{equation}
  B^h_{p_2,u}A_u^z
    =\frac1{|G|}\sum_\Gamma d_\Gamma
       \sum_{jj'klm}[\Gamma(\bar h)]_{jj'}\;
       \vcenter{\hbox{\dfAfterLocalTensor}}.
  \label{eq:dg-local-tensors}
\end{equation}
The magnetic loop closes through $d$ and the shifts form $A_u^z$.
In products such as $\bar{\mathcal Z}^\Gamma_{mj}L^z$, the shift acts first.
This is precisely the term entering the single-site projector
\begin{equation}
  D^{(C,\pi)}_{(p_2,u)}
    =\sum_{h\in C}\frac{d_\pi}{|Z(h)|}
       \sum_{z\in Z(h)}
         \bar\chi_\pi^{(h)}(z)B^h_{p_2,u}A_u^z.
  \label{eq:ff-site-projector}
\end{equation}
The weights are unchanged throughout the two moves.
Summing Eq.~\eqref{eq:dfa-fixed-term-local} with those weights proves
\begin{equation}
  \boxed{
    K_\sigma^{(C,\pi)}\Pi_{\mathrm{out}}
      =D^{(C,\pi)}_{(p_2,u)}\Pi_{\mathrm{out}}.
  }
  \label{eq:ff-localization}
\end{equation}
The local dyon projector therefore stabilizes $\rho_3$.
The source-vacuum and five neutral-vertex constraints remain present; they have been used to rewrite the ribbon, not discarded.

Finally, each Kraus operator commutes with $K_\sigma^{(C,\pi)}$ and maps $\Pi_{\mathrm{in}}$ into $\Pi_{\mathrm{out}}$.
Thus $D^{(C,\pi)}_{(p_2,u)}K_x\Pi_{\mathrm{in}} =K_xK_\sigma^{(C,\pi)}\Pi_{\mathrm{in}}$; multiplying by $K_x^\dagger$ and summing proves
\begin{equation}
  \Phi^*\!\left(D^{(C,\pi)}_{(p_2,u)}\right)
    \Pi_{\mathrm{in}}
    =K_\sigma^{(C,\pi)}\Pi_{\mathrm{in}},\qquad
  \operatorname{Tr}\!\left[
    D^{(C,\pi)}_{(p_2,u)}\Phi(\rho)\right]
    =\operatorname{Tr}\!\left[
      K_\sigma^{(C,\pi)}\rho\right],
  \quad
  \rho=\Pi_{\mathrm{in}}\rho\Pi_{\mathrm{in}}.
  \label{eq:df-restricted-pullback}
\end{equation}
For $C\ne\{e\}$ and nontrivial $\pi$, the outcome is a dyon at $(p_2,u)$; for $C=\{e\}$, the same expression reduces to the pure-charge projector $B^e_{p_2}P_u^\pi$.

\section{Charge-fusion decoherence details}
\label{app:cc-details}

We specify the four-anyon preparation of Sec.~\ref{subsec:cc-example} and show how rewiring exposes the information removed by reset.
The local factorization identifies both the pair-fusion channel and the orthogonal edge record discarded by the protocol.

\subsubsection{State construction and local charges}

Work on a connected, simply connected square-lattice patch, with horizontal edges oriented right and vertical edges up.
The four marked vertices form the $2\times2$ block of Fig.~\ref{fig:cc-setup}, with $v_3,v_4$ directly below $v_1,v_2$.
For $S_3=\langle r,s\mid r^3=s^2=e,\ srs=r^{-1}\rangle$, choose
\begin{equation}
 D(r)=\begin{pmatrix}-\tfrac12&-\tfrac{\sqrt3}{2}\\[2pt]
                      \tfrac{\sqrt3}{2}&-\tfrac12\end{pmatrix},\qquad
 D(s)=\begin{pmatrix}1&0\\0&-1\end{pmatrix},\qquad
 \ket{\mathbf{2};i,j}=\frac1{\sqrt3}\sum_tD(t)_{ij}\ket{t}.
 \label{eq:cc-D}
\end{equation}
Here $L^g\ket{t}=\ket{gt}$ and $R^g\ket{t}=\ket{tg^{-1}}$; $L^g$ acts within the first representation index and $R^g$ within the second.
Fix the first index to $1$ at each marked vertex, and contract the second indices of each pair as
\begin{equation}
 \ket{B_\lambda}_{pq}
 =\frac1{\sqrt2}\sum_{j,k=1}^2(C_\lambda)_{jk}
       \ket{\mathbf{2};1,j}_{v_p}\ket{\mathbf{2};1,k}_{v_q},\qquad
 C_{\mathbf{1}}=I_2,\quad
 C_{\mathbf{1}'}=\begin{pmatrix}0&1\\-1&0\end{pmatrix}.
 \label{eq:cc-bell-main}
\end{equation}
Thus the two choices are the normalized contractions $(\ket{11}+\ket{22})/\sqrt2$ and $(\ket{12}-\ket{21})/\sqrt2$ of the second indices.
At every unmarked vertex $w\in\mathcal U$, prepare $\ket{+}=6^{-1/2}\sum_t\ket{t}$.

To specify the resulting physical edge state, use the normalized vertex-to-edge map
\begin{equation}
 \mathcal G\ket{\{t_v\}}
   =\frac1{\sqrt6}\bigotimes_{u\to v}\ket{t_ut_v^{-1}}_{uv},\qquad
 \ket{\Psi_\lambda}=\mathcal G\ket{\Phi_\lambda},\qquad
 \ket{\Phi_\lambda}=\ket{B_\lambda}_{12}\ket{B_\lambda}_{34}
                    \bigotimes_{w\in\mathcal U}\ket{+}_w.
 \label{eq:cc-input-preparation}
\end{equation}
This gauging map gives us a convenient specification of the input state with four $\mathbf{2}$ charges.

The following identities allow us to see why this state satisfies the desired properties:
\begin{equation}
 D(g)C_\lambda D(g)^T=\chi_\lambda(g)C_\lambda,\qquad
 \mathcal G^\dagger\mathcal G=\frac16\sum_g\bigotimes_v R_v^g,
 \qquad \mathcal G L_v^g=A_v^g\mathcal G
 \label{eq:cc-G-isometry}
\end{equation}

Each pair transforms as $\lambda$ under simultaneous right shifts, so $\bigotimes_vR_v^g\ket{\Phi_\lambda} =\chi_\lambda(g)^2\ket{\Phi_\lambda}=\ket{\Phi_\lambda}$ for $\lambda\in\{\mathbf{1},\mathbf{1}'\}$, where $\mathbf{1}'$ denotes the sign representation: $\chi_{\mathbf{1}}=1$ and $\chi_{\mathbf{1}'}=\operatorname{sgn}$.
Consequently $\mathcal G$ preserves the norms and orthogonality of these inputs; it is not an isometry on arbitrary vertex states.
For $P_{L,v}^\mu=(d_\mu/6)\sum_g\chi_\mu(g^{-1})L_v^g$, the intertwining identity gives
\begin{equation}
 P_{v_i}^{\mathbf{2}}\ket{\Psi_\lambda}
   =\mathcal G P_{L,v_i}^{\mathbf{2}}\ket{\Phi_\lambda}
   =\ket{\Psi_\lambda},\qquad i=1,2,3,4.
 \label{eq:cc-local-proof}
\end{equation}
Unmarked vertices are neutral because $L^g\ket{+}=\ket{+}$, and all plaquettes are flux free because loop products of $t_ut_v^{-1}$ telescope to $e$.
Both pair contractions are changed in the $\mathbf{1}'$ preparation, so no uncompensated charge is left.

\subsubsection{Rewiring and pair-fusion channels}

Let $\ell:v_1\to v_2$ be the shared edge and $W=\sum_x T_\ell^x\widetilde A_1^x$, where $T_\ell^x=\ketbra{x}{x}_\ell$ and $\widetilde A_i^g$ acts on the three outer legs of $v_i$.
The first rewiring step is $W^\dagger$.
It replaces the potential $a$ at $v_1$ by the potential $b$ at $v_2$ on every outer leg, while leaving the shared edge in $x=ab^{-1}$.
With $\mathcal G_{\mathrm m}$ the same normalized vertex-to-edge map on the graph with $v_1$ merged into $v_2$ and $\ell$ removed, this is
\begin{equation}
 \mathcal Q_{12}\ket{a,b}=\ket{ab^{-1}}_\ell\ket{b}_{v_2},\qquad
 W^\dagger\mathcal G=(I_\ell\otimes\mathcal G_{\mathrm m})\mathcal Q_{12}.
 \label{eq:cc-ungauging-map}
\end{equation}
The auxiliary change of variables $\mathcal Q_{12}$ only relabels two registers.
The retained graph includes every other edge and the spectator pair.
The single local contraction
\begin{equation}
 \langle x,b|\mathcal Q_{12}|B_\lambda\rangle
   =\frac{[D(xb)C_\lambda D(b)^T]_{11}}{3\sqrt2}
   =\frac{[D(x)C_\lambda]_{11}}{\sqrt3}
      \frac{\chi_\lambda(b)}{\sqrt6}
 \label{eq:cc-local-contraction}
\end{equation}
therefore factorizes the entire state:
\begin{align}
 W^\dagger\ket{\Psi_\lambda}
  &=\ket{\xi_\lambda}_\ell\ket{\eta_\lambda},&
 \ket{\xi_{\mathbf{1}}}&=\ket{\mathbf{2};1,1},&
 \ket{\xi_{\mathbf{1}'}}&=-\ket{\mathbf{2};1,2},
 \label{eq:cc-factorization}\\
 \ket{\eta_\lambda}
  &=\mathcal G_{\mathrm m}\!\left[
       \ket{\lambda}_{v_2}\ket{B_\lambda}_{34}
       \bigotimes_{w\in\mathcal U}\ket{+}_w\right],&
 \ket{\lambda}&=\frac1{\sqrt6}\sum_b\chi_\lambda(b)\ket{b}.&&
 \label{eq:cc-retained-state}
\end{align}
All factors are normalized: the retained vertex input is globally right invariant, and the extracted states are orthonormal Fourier states.
The minus sign follows from $[D(x)C_{\mathbf{1}'}]_{11}=-D(x)_{12}$.

The transported pair representation and its projectors are
\begin{equation}
 \mathcal A_{12}(g)=\sum_xT_\ell^x
      \widetilde A_1^{xgx^{-1}}\widetilde A_2^g,\qquad
 P_{v_1v_2}^\mu=\frac{d_\mu}{6}\sum_g
      \chi_\mu(g^{-1})\mathcal A_{12}(g).
 \label{eq:cc-pair-representation}
\end{equation}
This is the pair projector of Eq.~\eqref{eq:cf-closed-ribbon}, with the group variable based at $v_2$.
To see this, replace the group variable in that equation by $xgx^{-1}$ and use $\chi_\mu(xg^{-1}x^{-1})=\chi_\mu(g^{-1})$.
Rewiring removes the transport factor: $W^\dagger\mathcal A_{12}(g)W=I_\ell\otimes A_{\mathrm m}^g$, where $A_{\mathrm m}^g=\widetilde A_1^g\widetilde A_2^g$.
Since $A_{\mathrm m}^g\mathcal G_{\mathrm m} =\mathcal G_{\mathrm m}L_{v_2}^g$ and $L^g\ket{\lambda}=\chi_\lambda(g)\ket{\lambda}$, $A_{\mathrm m}^g\ket{\eta_\lambda} =\chi_\lambda(g)\ket{\eta_\lambda}$.
Character orthogonality gives
\begin{equation}
 P_{v_1v_2}^\mu\ket{\Psi_\lambda}
  =\left[\frac{d_\mu}{6}\sum_g
      \chi_\mu(g^{-1})\chi_\lambda(g)\right]\ket{\Psi_\lambda}
  =\delta_{\mu\lambda}\ket{\Psi_\lambda}.
 \label{eq:cc-pair-projector-proof}
\end{equation}
Exchanging the two pairs gives the identical result for $P_{v_3v_4}^\mu$.
Thus both pairs have charge $\lambda$, and their combined charge is vacuum because $\lambda\otimes\lambda=\mathbf{1}$.

This adjacent placement does not require disjoint edge support.
On the connecting edge $v_3\to v_1$, $W$ acts by right multiplication at $v_1$, whereas the spectator action at $v_3$ is a left multiplication; these commute.
The control edge $v_3\to v_4$ of the spectator ribbon is untouched, so
\begin{equation}
 [W,A_{v_3}^g]=[W,A_{v_4}^g]=[W,P_{v_3v_4}^{\mu}]=0.
 \label{eq:cc-spectator-commutation}
\end{equation}
These spectator operators also have no support on the reset edge $\ell$.
Thus their charge labels and pair-fusion probabilities are preserved even for the adjacent placement of the two pairs.
These conserved observables do not determine the coherence between different pair-charge sectors of the joint state, which we track below.

\subsubsection{Reset and loss of coherence}

Reset replaces the extracted qudit by $\ket{+}$, after which $W$ restores the original graph.
Hence
\begin{equation}
 \mathcal E(X)=W\!\left[\ketbra{+}{+}_\ell\otimes
     \operatorname{Tr}_\ell(W^\dagger XW)\right]W^\dagger,\qquad
 \ket{\Omega_\lambda}=W(\ket{+}_\ell\ket{\eta_\lambda}).
 \label{eq:cc-channel-appendix}
\end{equation}
Since $\mathcal Q_{12}(\ket{+}_{v_1}\ket{\lambda}_{v_2}) =\ket{+}_\ell\ket{\lambda}_{v_2}$, the output also has the preparation
\begin{equation}
 \ket{\Omega_\lambda}=\mathcal G\!\left[
   \ket{+}_{v_1}\ket{\lambda}_{v_2}\ket{B_\lambda}_{34}
   \bigotimes_{w\in\mathcal U}\ket{+}_w\right],\qquad
 A_{v_1}^g\ket{\Omega_\lambda}=\ket{\Omega_\lambda},\quad
 A_{v_2}^g\ket{\Omega_\lambda}=\chi_\lambda(g)\ket{\Omega_\lambda}.
 \label{eq:cc-output-preparation}
\end{equation}
Thus $v_1$ is vacuum, $v_2$ carries the fused charge, and both spectator charges remain $\mathbf{2}$ with no magnetic flux.
The two outputs are normalized and have distinct receiver charges.

Substituting Eq.~\eqref{eq:cc-factorization} into the replacement channel gives its full action on the prepared subspace:
\begin{equation}
 \mathcal E\bigl(\ketbra{\Psi_\lambda}{\Psi_\mu}\bigr)
 =\underbrace{\braket{\xi_\mu}{\xi_\lambda}}_{\delta_{\lambda\mu}}
      \ketbra{\Omega_\lambda}{\Omega_\mu}.
 \label{eq:cc-general-cross}
\end{equation}
The diagonal terms retain their unit weight, whereas the cross term vanishes because the discarded states are orthogonal.
This proves Eq.~\eqref{eq:cc-output-main}: the fusion probabilities survive but the relative phase does not.
No explicit measurement or classical readout of $\lambda$ is required; reset transfers the information distinguishing the components to the discarded degrees of freedom.
On the prepared subspace, the channel can equivalently be written with Kraus operators $F_\lambda=\ketbra{\Omega_\lambda}{\Psi_\lambda}$.
It therefore has the same action as a measurement in the fusion basis with the outcome forgotten, followed by the corresponding map to the output states.
In particular, the two orthogonal inputs $(\ket{\Psi_{\mathbf1}}\pm\ket{\Psi_{\mathbf1'}})/\sqrt2$ give the same output mixture, so preserving the fusion probabilities does not preserve an arbitrary encoded state.

The complete dephasing depends on the internal states chosen in Eq.~\eqref{eq:cc-bell-main}.
For example, replace the fixed first-index state at $v_2$ by a normalized vector $u=(u_1,u_2)^T$, leaving the spectator pair unchanged:
\begin{equation*}
 \ket{B_\lambda(u)}_{12}
 =\frac1{\sqrt2}\sum_{i,j,k=1}^2u_i(C_\lambda)_{jk}
   \ket{\mathbf2;1,j}_{v_1}\ket{\mathbf2;i,k}_{v_2}.
\end{equation*}
The same contraction as in Eq.~\eqref{eq:cc-local-contraction} gives the retained state $\ket{\eta_\lambda}$ and a normalized edge state with amplitude $\langle x|\xi_\lambda(u)\rangle=[D(x)C_\lambda u]_1/\sqrt3$.
Schur orthogonality then gives the coherence multiplier
\begin{equation*}
 \braket{\xi_{\mathbf1'}(u)}{\xi_{\mathbf1}(u)}
 =(C_{\mathbf1'}u)^\dagger u
 =-2i\,\operatorname{Im}(u_1^*u_2).
\end{equation*}
The original preparation has $u=(1,0)^T$ and zero overlap.
For $u=(1,i)^T/\sqrt2$, instead, $C_{\mathbf1'}u=iu$, so the two discarded states differ only by a phase and the output is the pure state $\alpha\ket{\Omega_{\mathbf1}}+i\beta\ket{\Omega_{\mathbf1'}}$.
The four local anyon types and the two pair charges are unchanged, but the coherence now survives with a fixed relative phase.

For flux fusion, replace $W$ by the dual-edge rewiring circuit $M$ and $\ket{+}$ by $\ket{e}$.
Suppose that $M^\dagger\ket{\Theta_\lambda}=\ket{\zeta_\lambda}_d\ket{\nu_\lambda}$ with normalized factors, and write $\ket{\Upsilon_\lambda}=M(\ket{e}_d\ket{\nu_\lambda})$ for the corresponding output.
The same partial trace gives
\begin{equation*}
 \Phi\bigl(\ketbra{\Theta_\lambda}{\Theta_\mu}\bigr)
 =\braket{\zeta_\mu}{\zeta_\lambda}
   \ketbra{\Upsilon_\lambda}{\Upsilon_\mu}.
\end{equation*}
Orthogonal discarded states erase that coherence, while a nonzero overlap retains a corresponding cross term, without altering the conserved total-flux probabilities.
If an ungauged component is entangled between the discarded edge and the retained lattice, its cross terms must be evaluated using the full partial trace.
Resolving the full anyon type additionally requires the closed dyonic ribbon, whose identification with the local output type uses the neutral-vertex assumptions of Eq.~\eqref{eq:ff-dyon-claim}.

\section{Ancilla-based moving protocols from graph rewiring}\label{app:ancilla-move}

Graph rewiring also clarifies previous approach to coherently moving non-Abelions with an ancilla, as used to move the $C_{2}$ fluxes of the $S_{3}$ quantum double in Ref.~\cite{lo2026universal}. The idea is to copy the internal state of the anyon into an ancilla, apply the ribbon operator conditioned on the ancilla, and disentangle the ancilla at the end. In order to be able to disentangle the ancilla, this is restricted to states where there is no flux at the site where it is moved to. Otherwise, the ancilla is entangled with the internal state of the remnant anyon.
Figs.~\ref{fig:ancilla-moving-setup} shows the input state, where there is a flux at the leftmost plaquette with flux label $g=a_1 \alpha_1 \bar{t} \bar{\alpha}_0$ based at the vertex $v$, and an unentangled ancilla in state $\ket{e}$. Fig.~\ref{fig:ancilla-moving} shows the circuit to move the flux $g$, using the flux ribbon of Fig.~\ref{fig:gr-flux-ribbon} as a subroutine.

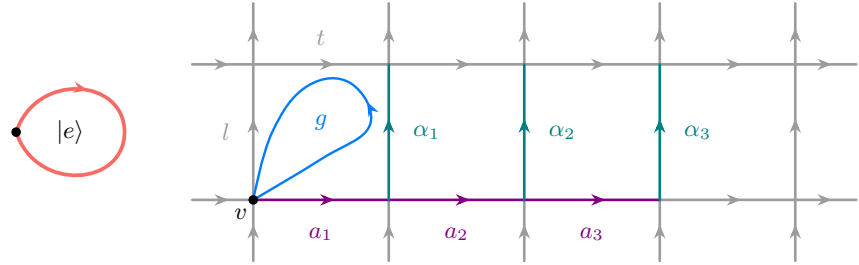
\begin{figure}[h]
\centering
\begin{tikzpicture}[x=0.1\columnwidth, y=0.1\columnwidth, line cap=round, line join=round,
    every node/.style={font=\small},
    oriented/.style 2 args={postaction={decorate, decoration={markings,
        mark=at position #1 with {\arrow[#2]{Stealth[length=1.9mm,width=1.5mm]}}}}},
    lat/.style ={latgray, line width=1.0pt, oriented={0.6}{latgray}},
    dir/.style ={violet,  line width=1.0pt, oriented={0.6}{violet}},
    dual/.style={teal,    line width=1.0pt, oriented={0.6}{teal}},
    anc/.style ={edgered, line width=1.4pt, oriented={0.3}{edgered}},
    flx/.style ={opblue,  line width=0.9pt, oriented={0.45}{opblue}},
    vx/.style  ={circle, fill=black, inner sep=0pt, minimum size=3.6pt},
    pl/.style  ={anchor=south west, inner sep=1pt}]
  \newcommand{\asLat}{%
    \foreach \i in {0,...,4}{\draw[lat] (\i,-0.45)--(\i,0); \draw[lat] (\i,1)--(\i,1.45);}
    \draw[lat] (0,0)--(0,1); \draw[lat] (4,0)--(4,1);
    \draw[lat] (-0.45,1)--(0,1); \foreach \i in {0,...,3}{\pgfmathtruncatemacro{\j}{\i+1}\draw[lat] (\i,1)--(\j,1);} \draw[lat] (4,1)--(4.45,1);
    \draw[lat] (-0.45,0)--(0,0); \draw[lat] (3,0)--(4,0); \draw[lat] (4,0)--(4.45,0);
    \foreach \i in {0,1,2}{\pgfmathtruncatemacro{\j}{\i+1}\draw[dir] (\i,0)--(\j,0);}
    \foreach \i in {1,2,3}{\draw[dual] (\i,0)--(\i,1);}}
  \newcommand{\asFlux}[2]{\draw[flx] plot[smooth, tension=0.7] coordinates {(#1,0) (#1+0.5,0.3) (#1+0.86,0.55) (#1+0.72,0.86) (#1+0.4,0.84) (#1+0.14,0.5) (#1,0)};
    \node[opblue] at (#1+0.5,0.57) {#2};}
  \newcommand{\asAnc}[1]{\draw[anc] (-1.75,0.5) .. controls (-1.6,0.95) and (-0.95,0.9) .. (-0.95,0.5) .. controls (-0.95,0.1) and (-1.6,0.05) .. (-1.75,0.5);
    \node[fill=white, inner sep=1pt] at (-1.35,0.5) {#1};
    \node[vx] at (-1.75,0.5) {};}
  \begin{scope}
    \asLat
    \foreach \i in {1,...,3}{\node[violet, inner sep=1pt] at ({\i-0.5},-0.25) {$a_{\i}$};
      \node[teal, inner sep=1pt] at ({\i+0.28},0.5) {$\alpha_{\i}$};}
    \node[latgray, inner sep=1pt] at (0.5,1.2) {$t$};
    \node[latgray, inner sep=1pt] at (-0.2,0.5) {$l$};
    \asFlux{0}{$g$}
    \node[vx] at (0,0) {}; \node[anchor=north east, inner sep=1pt] at (-0.02,-0.04) {$v$};
    \asAnc{$\ket{e}$}
  \end{scope}
\end{tikzpicture}
\caption{Initial configuration of the ancilla-based moving protocol of Fig.~\ref{fig:ancilla-moving}.
The leftmost plaquette of the lattice of Fig.~\ref{fig:gr-flux-ribbon} carries a flux $g$, marked by the counterclockwise loop (blue) based at its lower-left corner $v$: with the edges oriented as drawn (horizontal edges point right, vertical edges point up), $g=a_1 \alpha_1 \bar{t} \bar{\alpha}_0$ in terms of the group elements on its boundary edges $a_{1}$, $\alpha_{1}$, $t$ and $l$.
The ancilla qudit in the state $\ket{e}$ is an unentangled loop (red) not yet attached to the lattice.}
\label{fig:ancilla-moving-setup}
\end{figure}

\paragraph{Entangling (steps 1a--1d).}
The ancilla qudit in $\ket{e}$ is attached to $v$ as an unentangled loop (Fig.~\ref{fig:ancilla-edges}(a)).
Four controlled multiplications, controlled by the boundary edges of the plaquette and targeting the ancilla, slide the head of the ancilla edge counterclockwise around the plaquette.
The ancilla loop then encloses the flux, and its group element is the flux label: $\ket{e}\mapsto\ket{g}$.

\paragraph{Controlled ribbon (step 2).}
Steps 1--3 of Fig.~\ref{fig:gr-flux-ribbon} attach the dual edges of the ribbon to $v$.
The $L^{g}$ layer is replaced by a $CL^{\dagger}$ gate on each dual edge, controlled by the ancilla.
Each gate slides the tail of a dual edge around the ancilla loop, so the enclosed flux hops across the dual edges into the destination plaquette, which now carries $g$ based at $v$, while the original plaquette is flux-free.

\paragraph{Disentangling (steps 3a--3e).}
In the rewired graph the destination plaquette is a pentagon with corner $v$.
Five controlled multiplications, controlled by its boundary edges and targeting the ancilla, slide the head of the ancilla edge clockwise around it, which multiplies the ancilla by $\bar{g}$: $\ket{g}\mapsto\ket{e}$.
Steps 5--7 of Fig.~\ref{fig:gr-flux-ribbon} restore the lattice.

If the destination plaquette is already occupied, the disentangling step reads the fused flux rather than $g$, and the ancilla remains entangled with the internal state of the fusion product (for the $C_{2}$ fluxes of Ref.~\cite{lo2026universal}, the remnant charge).

\begin{figure}[h]
\centering
\adjustbox{max width=\columnwidth, max totalheight=0.75\textheight}{%
\begin{tikzpicture}[
      line cap=round, line join=round,
      oriented/.style 2 args={postaction={decorate, decoration={markings,
          mark=at position #1 with {\arrow[#2]{Stealth[length=1.9mm,width=1.5mm]}}}}},
      lat/.style    ={latgray, line width=1.1pt},
      lato/.style   ={latgray, line width=1.1pt, oriented={0.55}{latgray}},
      dir/.style    ={violet,  line width=1.1pt, oriented={\amArrPos}{violet}},
      dirarc/.style ={violet,  line width=1.1pt, oriented={0.6}{violet}},
      dual/.style   ={teal,    line width=1.1pt, oriented={\amArrPos}{teal}},
      dualarc/.style={teal,    line width=1.1pt, oriented={0.6}{teal}},
      anc/.style    ={edgered, line width=1.5pt, oriented={#1}{edgered}},
      anc/.default  =0.5,
      rib/.style    ={opblue,  line width=1.7pt},
      gl/.style     ={opblue,  font=\normalsize, fill=white, inner sep=1pt},
      sl/.style     ={black,   font=\normalsize, fill=white, inner sep=1pt},
      fl/.style     ={opblue,  font=\large},
      rule/.style   ={black, line width=0.9pt},
      hd/.style     ={font=\LARGE},
      sn/.style     ={font=\Large},
      sub/.style    ={font=\Large, align=center},
    ]
    \def\rowpitch{3.0}
    \def\stepX{-5.0}
    \def\latX{1.0}
    \def\xleft{-6.9}  \def\xmid{0.0}  \def\xright{5.9}
    \def\mcw{13.6}
    \def\ytop{1.5}
    \def\amArrPos{0.25}
    \newcommand{\amSpect}{\foreach \i in {0,...,4}{\draw[lat] (\i,-0.5)--(\i,0); \draw[lat] (\i,1)--(\i,1.5);}
      \draw[lato] (0,0)--(0,1); \draw[lato] (4,0)--(4,1);
      \draw[lat] (-0.45,1)--(0,1); \draw[lato] (0,1)--(1,1); \draw[lat] (1,1)--(3,1); \draw[lato] (3,1)--(4,1); \draw[lat] (4,1)--(4.45,1);
      \draw[lat] (-0.45,0)--(0,0); \draw[lato] (3,0)--(4,0); \draw[lat] (4,0)--(4.45,0);}
    \newcommand{\amUps}[1]{\ifnum#1<3 \pgfmathtruncatemacro{\amFirst}{#1+1}\foreach \i in {\amFirst,...,3}{\draw[dual] (\i,0)--(\i,1);}\fi}
    \newcommand{\amFan}[1]{\ifnum#1>0 \foreach \i in {1,...,#1}{\pgfmathsetmacro{\amS}{1.4-0.2*\i}%
      \ifnum\i=1 \draw[dualarc] (0,0) .. controls ({0.5*\amS},{0.6*\amS}) and (0.7,0.827) .. (1,1); \fi
      \ifnum\i=2 \draw[dualarc] (0,0) .. controls ({0.5*\amS},{0.6*\amS}) and (1.4,0.86)  .. (2,1); \fi
      \ifnum\i=3 \draw[dualarc] (0,0) .. controls ({0.5*\amS},{0.6*\amS}) and (2.7,1.0)   .. (3,1); \fi}\fi}
    \newcommand{\amStr}[1]{\draw[dir] (0,0)--(1,0);
      \ifnum#1<2 \pgfmathtruncatemacro{\amFirst}{#1+1}\foreach \i in {\amFirst,...,2}{\pgfmathtruncatemacro{\amNext}{\i+1}\draw[dir] (\i,0)--(\amNext,0);}\fi
      \ifnum#1>0 \foreach \i in {1,...,#1}{\pgfmathsetmacro{\amEX}{\i+1}\pgfmathsetmacro{\amS}{0.4+0.15*\i}%
        \draw[dirarc] (0,0) .. controls ({0.5*\amS},{0.6*\amS}) and ({\amEX-0.85},{0.6*\amS}) .. (\amEX,0);}\fi}
    \newcommand{\amBase}{\amSpect \amUps{0} \amStr{0}}
    \newcommand{\amColl}{\amSpect \amUps{3} \amFan{3} \amStr{2}}
    \newcommand{\amVert}{\fill[black] (0,0) circle (2.2pt); \node[font=\large, anchor=north east, inner sep=1pt] at (-0.02,-0.04) {$v$};}
    \newcommand{\amAncDrop}{\draw[anc=0.3] (0,0) .. controls (0.05,0.5) and (0.62,0.62) .. (0.68,0.35) .. controls (0.72,0.12) and (0.3,0.04) .. (0,0);}
    \newcommand{\amAncDropB}{\draw[anc=0.35] (0,0) .. controls (0.45,0.34) and (1.05,0.56) .. (1.35,0.47) .. controls (1.48,0.4) and (0.7,0.2) .. (0,0);}
    \newcommand{\amAncAone}{\draw[anc=0.5] (0,0) .. controls (0.25,0.34) and (0.75,0.34) .. (1,0);}
    \newcommand{\amAncChord}{\draw[anc=0.55] (0,0)--(1,1);}
    \newcommand{\amAncTop}{\draw[anc=0.5] (0,0) .. controls (0.95,0.05) and (0.95,0.95) .. (0,1);}
    \newcommand{\amAncWrapA}{\draw[anc=0.42] plot[smooth, tension=0.6] coordinates {(0,0) (0.5,0.07) (0.9,0.3) (0.93,0.7) (0.65,0.93) (0.25,0.9) (0.08,0.5) (0,0)};}
    \newcommand{\amAncWrapT}{\draw[anc=0.6] plot[smooth, tension=0.6] coordinates {(0,0) (0.3,0.46) (0.6,0.78) (0.72,0.9) (0.45,0.93) (0.12,0.9) (0.06,0.5) (0,0)};}
    \newcommand{\amAncWrapB}{\draw[anc=0.47] plot[smooth, tension=0.5] coordinates {(0,0) (0.8,0.36) (1.8,0.44) (2.6,0.3) (3.15,0.18) (3.82,0.18) (3.82,0.82) (3.15,0.84) (2.6,0.76) (2.0,0.63) (1.4,0.52) (0.8,0.4) (0,0)};}
    \newcommand{\amAncUa}{\draw[anc=0.35] (0,0) .. controls (0.45,0.45) and (2.4,0.45) .. (3.1,0.3) .. controls (3.6,0.2) and (3.85,0.25) .. (3.85,0.55) .. controls (3.85,0.85) and (3.5,0.95) .. (3,1);}
    \newcommand{\amAncUb}{\draw[anc=0.35] (0,0) .. controls (0.45,0.45) and (2.4,0.45) .. (3.1,0.3) .. controls (3.6,0.2) and (3.85,0.25) .. (3.85,0.55) .. controls (3.85,0.8) and (3.95,0.9) .. (4,1);}
    \newcommand{\amAncUc}{\draw[anc=0.4] (0,0) .. controls (0.45,0.45) and (2.4,0.45) .. (3.1,0.3) .. controls (3.5,0.45) and (3.9,0.4) .. (4,0);}
    \newcommand{\amAncUd}{\draw[anc=0.4] (0,0) .. controls (0.45,0.45) and (2.4,0.45) .. (3.1,0.3) .. controls (3.5,0.45) and (3.7,0.4) .. (3.55,0.2) .. controls (3.45,0.05) and (3.2,0.02) .. (3,0);}
    \newcommand{\amFlux}[3]{\node[fl] at (#1,#2) {#3};}
    \newcommand{\amDot}[2]{\fill[opblue] (#1,#2) circle (2.6pt);}
    \newcommand{\amFrame}[2]{\begin{scope}[shift={(#2,0)}]
      \node[hd] at ({(\xleft+\xmid)/2}, {\ytop+0.5}) {Gate};
      \node[hd] at ({(\xmid+\xright)/2}, {\ytop+0.5}) {Graph};
      \draw[rule, line width=1.3pt] (\xleft,\ytop) -- (\xright,\ytop);
      \draw[rule, line width=1.3pt] (\xmid,{\ytop+0.95}) -- (\xmid,{-(#1-1)*\rowpitch-\ytop});
      \foreach \r in {1,...,#1}{\pgfmathtruncatemacro{\q}{#1-1}\ifnum\r>\q\else\draw[rule] (\xleft,{-\r*\rowpitch+\ytop}) -- (\xright,{-\r*\rowpitch+\ytop});\fi}
      \end{scope}}
    \amFrame{6}{0} \amFrame{7}{\mcw}
    \newcommand{\amStepNo}[3]{\node[sn] at ({#1-6.3},{-#2*\rowpitch}) {#3};}
    \newcommand{\amSubNo}[3]{\node[sub] at ({#1+(\xleft+\xmid)/2},{-#2*\rowpitch}) {#3};}
    \begin{scope}[shift={(\latX,-0.5)}] \def\amArrPos{0.6} \amBase
      \foreach \i in {1,...,3}{\node[violet, font=\large, inner sep=1pt] at ({\i-0.5},-0.3) {$a_{\i}$};
        \node[teal, font=\large, inner sep=1pt] at ({\i+0.3},0.5) {$\alpha_{\i}$};}
      \amFlux{0.78}{0.78}{$g$} \amAncDrop \amVert
      \node[sl] at (0.33,0.52) {$\ket{e}$};
    \end{scope}
    \amStepNo{0}{1}{1a.}
    \begin{scope}[shift={(\stepX,{-1*\rowpitch-0.5})}] \amBase \amAncDrop \amVert
      \amDot{0.85}{0} \draw[rib] (0.85,0) .. controls (0.95,0.3) and (0.7,0.55) .. (0.42,0.47);
      \node[gl] at (0.34,0.46) {$R^{\dagger}$};
    \end{scope}
    \begin{scope}[shift={(\latX,{-1*\rowpitch-0.5})}] \def\amArrPos{0.6} \amBase \amAncAone \amFlux{0.5}{0.6}{$g$} \amVert \end{scope}
    \amStepNo{0}{2}{1b.}
    \begin{scope}[shift={(\stepX,{-2*\rowpitch-0.5})}] \amBase \amAncAone \amVert
      \amDot{1}{0.5} \draw[rib] (1,0.5) .. controls (0.85,0.5) and (0.65,0.4) .. (0.52,0.28);
      \node[gl] at (0.5,0.26) {$R^{\dagger}$};
    \end{scope}
    \begin{scope}[shift={(\latX,{-2*\rowpitch-0.5})}] \def\amArrPos{0.6} \amBase \amAncChord \amFlux{0.3}{0.64}{$g$} \amVert \end{scope}
    \amStepNo{0}{3}{1c.}
    \begin{scope}[shift={(\stepX,{-3*\rowpitch-0.5})}] \amBase \amAncChord \amVert
      \amDot{0.5}{1} \draw[rib] (0.5,1)--(0.5,0.55);
      \node[gl] at (0.5,0.5) {$R$};
    \end{scope}
    \begin{scope}[shift={(\latX,{-3*\rowpitch-0.5})}] \def\amArrPos{0.6} \amBase \amAncTop \amFlux{0.32}{0.5}{$g$} \amVert \end{scope}
    \amStepNo{0}{4}{1d.}
    \begin{scope}[shift={(\stepX,{-4*\rowpitch-0.5})}] \amBase \amAncTop \amVert
      \amDot{0}{0.5} \draw[rib] (0,0.5)--(0.68,0.5);
      \node[gl] at (0.71,0.5) {$R$};
    \end{scope}
    \begin{scope}[shift={(\latX,{-4*\rowpitch-0.5})}] \def\amArrPos{0.6} \amBase \amAncWrapA \amFlux{0.45}{0.47}{$g$} \amVert
      \node[sl] at (0.3,0.89) {$\ket{g}$};
    \end{scope}
    \amSubNo{0}{5}{Fig.~\ref{fig:gr-flux-ribbon},\\ steps 1--3}
    \begin{scope}[shift={(\latX,{-5*\rowpitch-0.5})}] \def\amArrPos{0.6} \amColl \amAncWrapT \amVert
      \amFlux{0.29}{0.7}{$g$}
      \node[sl] at (0.5,1.25) {$\ket{g}$};
    \end{scope}
    \begin{scope}[shift={(\mcw,0)}]
    \amStepNo{0}{0}{2.}
    \begin{scope}[shift={(\stepX,-0.5)}] \amColl \amAncWrapT \amVert
      \amDot{0.45}{0.62}
      \draw[rib] (0.45,0.62) .. controls (0.65,0.62) and (0.8,0.75) .. (0.88,0.92);
      \node[gl] at (0.9,0.95) {$L^{\dagger}$};
      \draw[rib] (0.45,0.62) .. controls (0.8,0.52) and (1.2,0.6) .. (1.45,0.86);
      \node[gl] at (1.5,0.88) {$L^{\dagger}$};
      \draw[rib] (0.45,0.62) .. controls (1.0,0.38) and (1.9,0.5) .. (2.25,0.84);
      \node[gl] at (2.3,0.87) {$L^{\dagger}$};
    \end{scope}
    \begin{scope}[shift={(\latX,-0.5)}] \def\amArrPos{0.6} \amColl \amAncWrapB \amVert
      \amFlux{3.5}{0.5}{$g$}
      \node[sl] at (3.2,0.84) {$\ket{g}$};
    \end{scope}
    \amStepNo{0}{1}{3a.}
    \begin{scope}[shift={(\stepX,{-1*\rowpitch-0.5})}] \amColl \amAncWrapB \amVert
      \amDot{2.1}{0.81} \draw[rib] (2.1,0.81)--(2.42,0.72); \node[gl] at (2.6,0.7) {$R^{\dagger}$};
    \end{scope}
    \begin{scope}[shift={(\latX,{-1*\rowpitch-0.5})}] \def\amArrPos{0.6} \amColl \amAncUa \amVert \amFlux{3.4}{0.62}{$g$} \end{scope}
    \amStepNo{0}{2}{3b.}
    \begin{scope}[shift={(\stepX,{-2*\rowpitch-0.5})}] \amColl \amAncUa \amVert
      \amDot{3.5}{1} \draw[rib] (3.5,1)--(3.5,0.34); \node[gl] at (3.5,0.27) {$R^{\dagger}$};
    \end{scope}
    \begin{scope}[shift={(\latX,{-2*\rowpitch-0.5})}] \def\amArrPos{0.6} \amColl \amAncUb \amVert \amFlux{3.4}{0.65}{$g$} \end{scope}
    \amStepNo{0}{3}{3c.}
    \begin{scope}[shift={(\stepX,{-3*\rowpitch-0.5})}] \amColl \amAncUb \amVert
      \amDot{4}{0.5} \draw[rib] (4,0.5) .. controls (3.85,0.5) and (3.7,0.4) .. (3.55,0.3); \node[gl] at (3.5,0.27) {$R$};
    \end{scope}
    \begin{scope}[shift={(\latX,{-3*\rowpitch-0.5})}] \def\amArrPos{0.6} \amColl \amAncUc \amVert \amFlux{3.5}{0.68}{$g$} \end{scope}
    \amStepNo{0}{4}{3d.}
    \begin{scope}[shift={(\stepX,{-4*\rowpitch-0.5})}] \amColl \amAncUc \amVert
      \amDot{3.5}{0} \draw[rib] (3.5,0)--(3.48,0.32); \node[gl] at (3.45,0.42) {$R$};
    \end{scope}
    \begin{scope}[shift={(\latX,{-4*\rowpitch-0.5})}] \def\amArrPos{0.6} \amColl \amAncUd \amVert \amFlux{3.5}{0.72}{$g$} \end{scope}
    \amStepNo{0}{5}{3e.}
    \begin{scope}[shift={(\stepX,{-5*\rowpitch-0.5})}] \amColl \amAncUd \amVert
      \amDot{2.5}{0.19} \draw[rib] (2.5,0.19)--(2.5,0.36); \node[gl] at (2.5,0.45) {$R$};
    \end{scope}
    \begin{scope}[shift={(\latX,{-5*\rowpitch-0.5})}] \def\amArrPos{0.6} \amColl \amAncDropB \amVert
      \amFlux{3.5}{0.5}{$g$}
      \node[sl] at (1.72,0.5) {$\ket{e}$};
    \end{scope}
    \amSubNo{0}{6}{Fig.~\ref{fig:gr-flux-ribbon},\\ steps 5--7}
    \begin{scope}[shift={(\latX,{-6*\rowpitch-0.5})}] \def\amArrPos{0.6} \amBase \amAncDrop \amVert
      \amFlux{3.5}{0.5}{$g$}
      \node[sl] at (0.33,0.52) {$\ket{e}$};
    \end{scope}
    \end{scope}
  \end{tikzpicture}}
\caption{Ancilla-based coherent moving of a flux as graph rewiring, on the lattice of Fig.~\ref{fig:gr-flux-ribbon}, read down the left half and then the right half.
Each row shows one elementary move (Gate) acting on the graph of the previous row and the resulting graph (Graph); the dot marks the control and the label the target, as in Fig.~\ref{fig:gr-four-cases}.
The ancilla edge (red) is based at $v$ and its state is indicated whenever it is a closed loop; the flux label $g$ (blue) is based at $v$.
Steps 1a--1d wrap the ancilla loop around the flux, step 2 hops the enclosed flux across the dual edges attached to $v$ by steps 1--3 of Fig.~\ref{fig:gr-flux-ribbon}, and steps 3a--3e unwrap the ancilla loop around the destination face; steps 5--7 of Fig.~\ref{fig:gr-flux-ribbon} restore the lattice.}
\label{fig:ancilla-moving}
\end{figure}
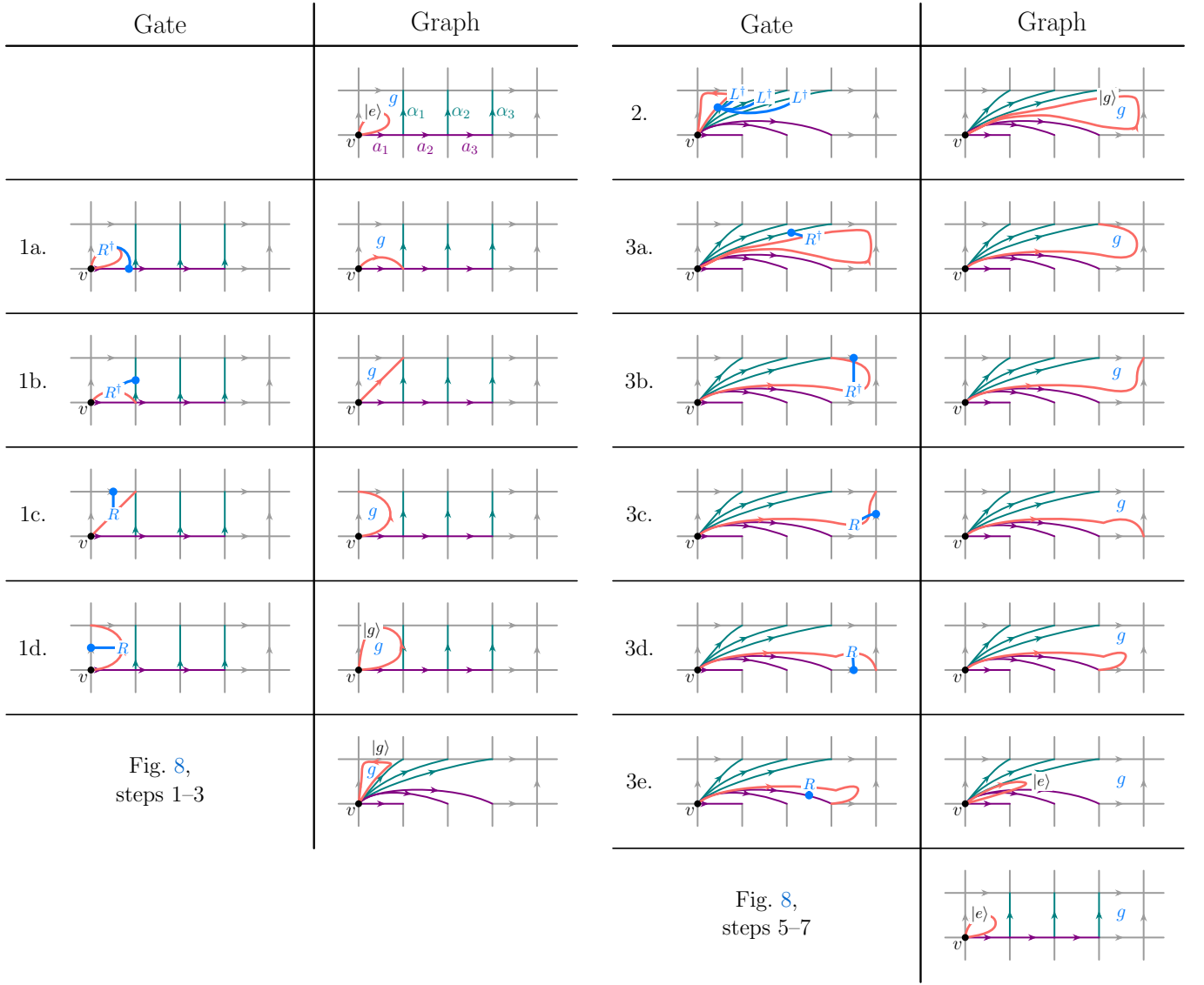

\end{document}